\documentclass[ejsv2,noshowframe]{imsart}

\RequirePackage[authoryear]{natbib}
\RequirePackage[colorlinks,citecolor=blue,urlcolor=blue]{hyperref}
\RequirePackage{graphicx}
\RequirePackage{color, amsmath, amsthm, graphicx, bbm, footnote, mathtools, bm, enumitem, longtable, subcaption}

\usepackage{apptools}
\AtAppendix{\counterwithin{theorem}{section}}
\AtAppendix{\counterwithin{remark}{section}}

\arxiv{2309.06429}
\startlocaldefs
\theoremstyle{plain}

\newtheorem{theorem}{Theorem}[section]
\newtheorem{lemma}[theorem]{Lemma}
\newtheorem{proposition}[theorem]{Proposition}
\newtheorem{assumption}{Assumption}

\theoremstyle{definition}
\newtheorem{definition}[theorem]{Definition}
\newtheorem*{example}{Example}

\theoremstyle{remark}

\newtheorem{remark}{Remark}

\DeclareMathOperator*{\argmax}{arg\,max}
\DeclareMathOperator*{\argmin}{arg\,min}

\newcommand{\ind}{{\perp\!\!\!\perp}}

\newcommand{\norm}[1]{\left|\left| #1 \right|\right|}
\newcommand{\Diag}{\mathtt{Diag}}
\renewcommand{\hat}{\widehat}
\renewcommand{\tilde}{\widetilde}

\newtheorem{innercustomthm}{Theorem}
\newenvironment{customthm}[1]
{\renewcommand\theinnercustomthm{#1}\innercustomthm}
{\endinnercustomthm}

\newtheorem{innercustomprop}{Proposition}
\newenvironment{customprop}[1]
{\renewcommand\theinnercustomprop{#1}\innercustomprop}
{\endinnercustomprop}

\newtheorem{innercustomlem}{Lemma}
\newenvironment{customlem}[1]
{\renewcommand\theinnercustomlem{#1}\innercustomlem}
{\endinnercustomlem}

\allowdisplaybreaks
\endlocaldefs

\begin{document}
\begin{frontmatter}
\title{Efficient Inference on High-Dimensional Linear Models With Missing Outcomes}
\runtitle{High-Dimensional Inference With Missing Outcomes}

\begin{aug}
\author[A]{\fnms{Yikun} \snm{Zhang}\ead[label=e1]{yikun@uw.edu}\orcid{0000-0003-3905-6346}}, 
\author[B]{\fnms{Alexander} \snm{Giessing}\ead[label=e2]{giessing@nus.edu.sg}\orcid{0000-0002-6917-0652}}
\and
\author[A]{\fnms{Yen-Chi} \snm{Chen}\ead[label=e3]{yenchic@uw.edu}\orcid{0000-0002-4485-306X}}
\address[A]{Department of Statistics,
University of Washington\printead[presep={,\ }]{e1,e3}}

\address[B]{Department of Statistics and Data Science, National University of Singapore\printead[presep={,\ }]{e2}}
\runauthor{Zhang et al.}
\end{aug}

\begin{abstract}
This paper is concerned with inference on the regression function of a high-dimensional linear model when outcomes are missing at random. We propose an estimator that combines a Lasso pilot estimate of the regression function with a bias correction term based on the weighted residuals of the Lasso regression. The weights depend on estimates of the missingness probabilities (propensity scores) and solve a convex optimization program that trades off bias and variance optimally.
Provided that the propensity scores can be pointwise consistently estimated at in-sample data points, our proposed estimator for the regression function is asymptotically normal and semiparametrically efficient among all asymptotically linear estimators. Furthermore, the proposed estimator retains its asymptotic properties even if the propensity scores are estimated by modern machine learning techniques.
We validate the finite-sample performance of the proposed estimator through comparative simulation studies and the real-world problem of inferring the stellar masses of galaxies in the Sloan Digital Sky Survey.
\end{abstract}



\end{frontmatter}
\tableofcontents

\section{Introduction}
\label{sec:Intro}

In this paper, we develop a novel method to conduct statistical inference on the regression function (or conditional mean) of a sparse high-dimensional linear model when outcomes are missing at random. Specifically, let $Y \in \mathbb{R}$ be an outcome variable and $X \in \mathbb{R}^d$ be a high-dimensional covariate vector. The object of interest is the regression function $m_0(x)=\mathrm{E}(Y|X=x)$.

Valid inference on $m_0(x)$ in the presence of high-dimensional covariates and missing outcomes is of practical importance. For one thing, collecting high-dimensional data has become common practice across the board in science and engineering. Examples range from portfolio optimization and cross-sectional home price analysis in finance \citep{fan2011sparse} to the development of biomarker classifiers in biology \citep{baek2009development}. Moreover, some recent works in causal inference lean toward generating high-dimensional covariates in order to control for potential confounders \citep{wyss2022machine}.
For another, semi-automatic processing and storing of vast amounts of unstructured data inevitably entails missingness \citep{huang2016unforeseen}. Such missingness may result from the study dropouts of participants in clinical trials \citep{higgins2008imputation} or noncompliance with the assigned treatments in a survey \citep{frumento2012evaluating}. The emerging field of semi-supervised learning in computer science, where additional data samples with the same distribution are given but without labels, is also a missing-outcome problem \citep{oliver2006semi}.

\begin{figure}
	\centering
	\includegraphics[width=0.995\linewidth]{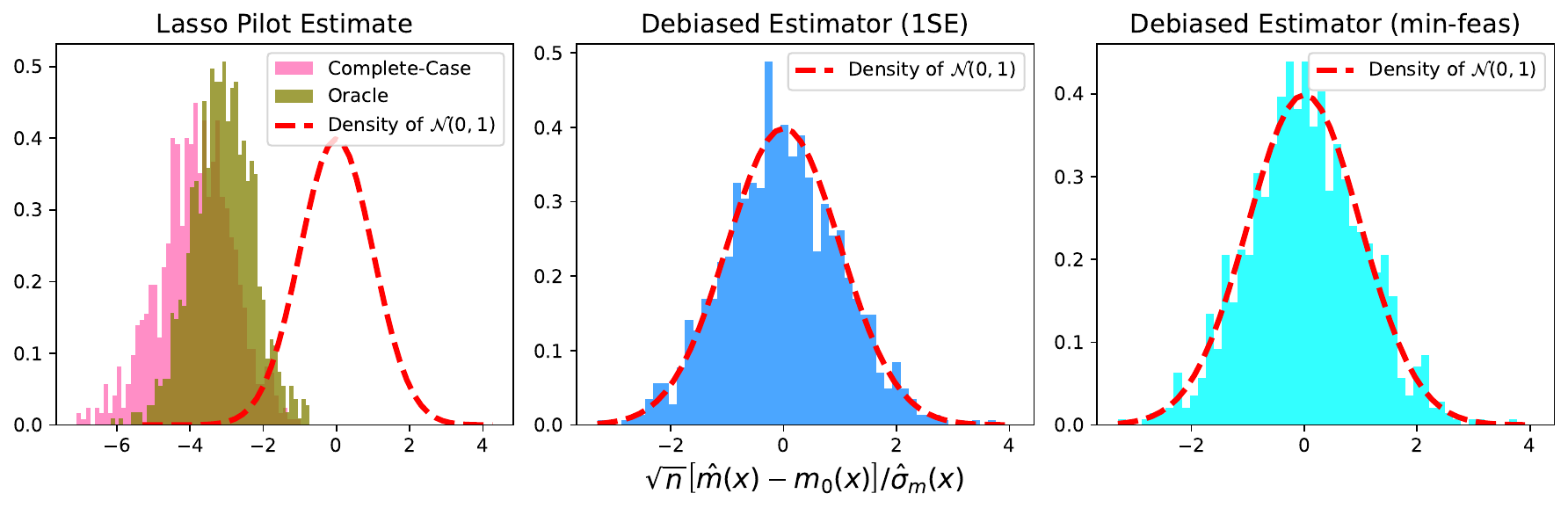}
	\caption{Comparison of our debiased estimators under two different choices of the tuning parameters (``1SE'' and ``min-feas'') with the conventional Lasso estimates based on complete-case or oracle data. Here, we adopt the sparse $\beta_0^{sp}$ and dense $x^{(4)}$ with $d=1000, n=900$, $X\sim\mathcal{N}_d\left(0,\Sigma^{\mathrm{cs}} \right)$, and $\epsilon\sim \mathcal{N}(0,1)$ under the MAR setting \eqref{MAR_correct}; see \autoref{subsec:sim_design} for details.}
	\label{fig:debias_illustration}
\end{figure}

When the outcome variables $Y_i,i=1,...,n$ are fully observed given a data sample (which is known as the oracle data setting), statistical estimation on $m_0(x)$ is tractable with regularization and sparsity constraints under high-dimensional settings \citep{wainwright2019high}. One of the most well-studied approaches is the Lasso \citep{tibshirani1996regression}, which assumes a linear model and imposes $L_1$-regularization on the regression coefficients. When it comes to statistical inference, however, the Lasso solution leads to a biased estimate of $m_0(x)$ even when the linear model assumption is correct \citep{van2014asymptotically,zhang2014confidence}. Missing outcomes further exacerbate the bias of the Lasso solution; see the first panel of \autoref{fig:debias_illustration} for an illustration. While this bias could be partially mitigated via sample-splitting or re-fitting, doing so would reduce the sample size and increase the computational cost, leading to less efficient estimates. Consequently, the goal of this paper is to address one central question: 
\begin{center}
	\emph{``How can we conduct statistically and computationally efficient inference on $m_0(x)$ despite missing outcomes?''} 
\end{center}
This is a challenging question; to answer it, we need to impose some structural assumptions. Throughout the paper, we will assume that the data sample follows a high-dimensional sparse linear regression model with outcomes ``missing at random'' (MAR).
\begin{assumption}[High-dimensional sparse linear regression model with MAR outcomes] 
	\leavevmode 
	\begin{enumerate}[label=(\alph*)]
		\label{assump:basic}
		\item The data sample $\{(Y_i,R_i,X_i)\}_{i=1}^n$ consists of independent and identically distributed (i.i.d.) observations drawn from $(Y, R, X) \in \mathbb{R} \times \{0,1\} \times \mathbb{R}^d$ generated by the linear model
		\begin{equation}
			\label{linear_reg}
			Y = X^T\beta_0 + \epsilon, \quad \quad \mathrm{E}\left(\epsilon \big| X\right) = 0, \quad\quad \mathrm{E}\left(\epsilon^2 \big| X\right) = \sigma_{\epsilon}^2, 
		\end{equation}
		where $\norm{\beta_0}_0 = \sum_{k=1}^d \mathbbm{1}_{\{\beta_{0k} \neq 0\}} = s_{\beta} \ll d$. 
		
		\item The missingness indicator $R$ is conditionally independent of $Y$ given $X$.
	\end{enumerate}
\end{assumption}

Assumption~\ref{assump:basic}(a) is standard in high-dimensional statistics. On the one hand, the sparse linear model can be extended to sparse additive models \citep{ravikumar2009sparse} and partially linear models \citep{muller2015partial,belloni2019valid}. On the other hand, our proposed debiasing method can be generalized to handle heteroscedastic errors; see \autoref{subsec:debias_method}. However, developing theory for these generalizations is beyond the scope of this paper. Assumption~\ref{assump:basic}(b) is common in the missing data literature \citep{tsiatis2007semiparametric,little2019statistical} and is related to the ignorability or unconfoundedness condition in causal inference \citep{imbens2004nonparametric}. Under the MAR assumption, the propensity score $\pi(Y, X) := \mathrm{P}(R=1|Y, X)$ depends only on the covariate vector $X$ so that $\pi(Y, X) \equiv \pi(X) := \mathrm{P}(R=1| X)$, which can thus be estimated from the fully observed data $\{(X_i,R_i)\}_{i=1}^n$ \citep{rosenbaum1983central}. 

Under Assumption~\ref{assump:basic}, our proposed debiasing method directly infers the regression function $m_0(x)=x^T\beta_0$ by combining a Lasso pilot estimate of $\beta_0\in \mathbb{R}^d$ based on the complete-case data with a bias correction term based on the weighted residuals of the Lasso regression. The weights depend on estimates of the propensity scores $\pi(X_i),i=1,...,n$ and are obtained as the solution to a convex debiasing program which trades off bias and variance in a mean-squared-error-optimal way; see \autoref{fig:debias_illustration}(b,c) for a preview.

A crucial difference between the existing works on debiasing the Lasso estimate in the literature \citep{zhang2014confidence,van2014asymptotically,javanmard2014confidence,javanmard2014hypothesis} and our proposal is that we focus our inference on the scalar regression function $m_0(x)=x^T\beta_0$ rather than the high-dimensional regression (coefficient) vector $\beta_0\in \mathbb{R}^d$. By considering $m_0(x)$, we are able to conduct inference on any individual regression coefficient by setting $x$ equal to a standard basis vector in $\mathbb{R}^d$ as well as on joint effects of several regression coefficients. Additionally, inferring $m_0(x)$ itself is more computationally efficient than inferring the regression vector $\beta_0$.

\subsection{Contribution and Outline of the Paper}

\hspace{8pt} {\bf 1. Methodology:} We describe the detailed procedures of our proposed debiasing method and reveal its computational feasibility through the dual formulation. We also provide interpretations of our debiasing method from the perspective of a bias-variance trade-off as well as Neyman near-orthogonality; see \autoref{sec:Method} for details.	

{\bf 2. Asymptotic Theory:} We prove that our proposed debiased estimator is asymptotically unbiased and normally distributed as long as the propensity scores are pointwise consistently estimable (\emph{i.e.}, sample-splitting or cross-fitting are not needed); see \autoref{subsec:asymp_normal}. Moreover, the asymptotic variance of our debiased estimator attains the semiparametric efficiency bound among all asymptotically linear estimators for any given dimension $d$. To establish the asymptotic normality of our debiased estimator, we derive the consistency of the Lasso pilot estimate with complete-case data and the dual solution of our debiasing program as building blocks; see \autoref{subsec:consistency}.

{\bf 3. Simulations and Real-World Applications:} To demonstrate the finite-sample performance of our debiasing method, we compare it with the existing debiasing methods in the literature through extensive Monte Carlo experiments; see \autoref{subsec:sim_design} and \autoref{subsec:sim_results}. We also apply the proposed debiasing method to the problem of inferring the stellar masses of galaxies in the Sloan Digital Sky Survey; see \autoref{sec:real_application}.

{\bf 4. Algorithmic Implementation:} We describe the implementation details of our debiasing method in \autoref{app:imple_debias} and encapsulate them into both the Python package ``\texttt{Debias-Infer}'' and the R package ``\texttt{DebiasInfer}'' \citep{DebiasInfer}. All the code for our experiments is available at \url{https://github.com/zhangyk8/Debias-Infer/tree/main/Paper_Code}. 

\subsection{Related Work}
\label{subsec:related_works}

Statistical inference for the high-dimensional linear model with complete data has been extensively studied by \cite{zhang2014confidence,van2014asymptotically,javanmard2014confidence,javanmard2014hypothesis,javanmard2018debiasing,buhlmann2013statistical,ning2017,athey2018approximate,zhu2018linear,cai2021optimal,bellec2022biasing}, who proposed debiased Lasso estimators to construct confidence intervals and multiple testing procedures for individual regression coefficients or their low-dimensional projections. The main idea of these approaches is to correct the bias of the Lasso by estimating a relaxed inverse of the Gram matrix. This is also known as a one-step update in the classical theory of semiparametric inference \citep{VDV2000}. More recently, \cite{zhang2017simultaneous,sai2020debiasing} develop bootstrap procedures for debiasing the Lasso solution, while \cite{javanmard2020flexible} presents a flexible framework for testing general hypotheses about regression coefficients. \cite{battey2023inference} considers inference in linear models by treating each regression coefficient in turn as the parameter of interest and finding an optimal transformation to orthogonalize the remaining coefficients. Moreover, some of these ideas have been extended to parameter estimation and inference for high-dimensional generalized linear models by \cite{belloni2016post,guo2016tests,xia2020revisit,shi2021statistical,guo2021inference,ma2021global,cai2023statistical}.

The asymptotic normality of the aforementioned debiasing methods requires the regression vector $\beta_0$ to lie in the ultra-sparse regime $s_{\beta}=o\left(\frac{\sqrt{n}}{\log d}\right)$, as does our proposed debiasing method. Under a known covariance matrix and Gaussian design, \cite{javanmard2018debiasing} relaxed the sparsity requirement to $s_{\beta} =o\left(\frac{n}{(\log d)^2}\right)$. In addition, \cite{cai2016confidence} considered constructing confidence intervals for the regression function $m_0(x)=x^T\beta_0$ in the moderate-sparse region $\frac{\sqrt{n}}{\log d} \ll s_{\beta} \lesssim \frac{n}{\log d}$ and established the minimaxity and adaptivity for their confidence intervals when prior knowledge of the sparsity level $s_{\beta}$ is available; see also \cite{nickl2013confidence,cai2018accuracy} for the related discussion.

Statistical estimation of the regression function under the MAR mechanism and its semiparametric efficiency have also been investigated in the literature \citep{robins1994estimation,robins1995semiparametric,robins1995analysis,graham2011efficiency,muller2012efficient}. These works address problems ranging from missing outcomes to unobserved covariates, primarily in low-dimensional settings. In the high-dimensional data setting, \cite{loh2012high} proposed an estimator of the linear regression coefficient based on projected gradient descent in the presence of MAR covariates. \cite{xie2021model} developed a model-averaging method for predicting the in-sample conditional mean vector $\left(\mathrm{E}(Y|X_1),...,\mathrm{E}(Y|X_n)\right)^T\in \mathbb{R}^n$ with high-dimensional covariates and MAR outcomes, although its extension to statistical inference remains unclear. Under the high-dimensional semi-supervised learning setting, \cite{zhang2022high} constructed doubly robust estimators for the unconditional mean $\mathrm{E}(Y)$ and variance $\mathrm{Var}(Y)$, while \cite{hou2023surrogate} proposed a surrogate-assisted approach for inference on $\beta_0$ or $x^T\beta_0$, where $\beta_0$ solves $\mathrm{E}\left[X(Y-g(X^T\beta_0))\right]=0$ for a known link function $g:\mathbb{R}\to \mathbb{R}$.
Recently, \cite{celentano2023challenges} considered semiparametric estimation of the population mean under MAR outcomes in the high-dimensional inconsistency regime and proposed a debiasing remedy. \cite{deng2024optimal} studied minimax estimation rates and proposed corresponding estimators for the best linear approximation coefficient $\beta_0=\argmin_{\beta \in \mathbb{R}^d} \mathrm{E}\left\{\left[\mathrm{E}(Y|X) -X^T\beta\right]^2\right\}$ under the high-dimensional semi-supervised learning setting, although inference properties for their estimators were not established. After the first public version of our paper, \cite{tian2024semi} proposed a general augmented inverse probability weighting (AIPW) method for estimating $\beta_0$ or $\mathrm{E}(Y|Z)=\psi(Z^T\beta_0)$, where $Z$ is a subvector of the observed covariate vector $X$ and $\psi:\mathbb{R}\to \mathbb{R}$ is an inverse link function, under semi-supervised learning settings. \cite{song2024semi} studied high-dimensional linear regression inference with both blockwise missing covariates and missing outcomes, and their proposed estimator reduces to the debiased Lasso estimator of \cite{van2014asymptotically} when the covariates are fully observed. 

The most closely related work to our paper is \cite{chakrabortty2019high}, which also studied high-dimensional estimation and inference with missing outcomes. They proposed a general M-estimation framework for statistical inference based on a Lasso-type debiased and doubly robust (DDR) estimator. Computationally, their DDR estimator requires estimating a $d\times d$ debiasing matrix, whereas our debiasing method is more efficient because it only solves for an $n$-dimensional debiasing vector, or equivalently the $d$-dimensional dual vector, through a convex program. Theoretically, both \cite{chakrabortty2019high} and our work impose the same condition on the rate of convergence of the propensity score estimation in order to establish the asymptotic normality. Nevertheless, we demonstrate through simulation studies in \autoref{sec:experiments} that our debiased estimator remains asymptotically normal even when the propensity scores are estimated nonparametrically at in-sample points, while the DDR estimator in \cite{chakrabortty2019high} may deviate from asymptotic normality in certain scenarios. Furthermore, the DDR estimator in \cite{chakrabortty2019high} asymptotically achieves the coordinatewise semiparametric efficiency bound, which may not be efficient in the worst possible query direction \citep{jankova2018semiparametric}, whereas our debiased estimator is asymptotically efficient whatever the query direction is.

\subsection{Notation}

The general probability measure and expectation with respect to the distribution of $(Y,R,X)$ are denoted by $\mathrm{P}$ and $\mathrm{E}$, respectively. We write $Y\ind R$ when the random variables $Y$ and $R$ are independent. We use $h(x)=O(g(x))$ if there exists an absolute constant $A>0$ such that $|h(x)| \leq A g(x)$ for all sufficiently large $x$. In contrast, $h(x)=o(g(x))$ when $\lim_{x\to\infty} |h(x)|/g(x)=0$. The notation $o_P(1)$ is short for a sequence of random vectors that converges to zero in probability. The expression $O_P(1)$ denotes a sequence that is bounded in probability. 
The norm $\norm{x}_q = \left(\sum_{k=1}^d |x_k|^q\right)^{1/q}$ with $q>0$ stands for the $L_q$-norm in the Euclidean space $\mathbb{R}^d$, though it is no longer a norm when $0<q<1$. In particular, $\norm{x}_{\infty} = \max_{1\leq k \leq d} |x_k|$ and $\norm{x}_0=\sum_{k=1}^d \mathbbm{1}_{\{x_k\neq 0\}}$ indicates the number of nonzero elements in $x\in \mathbb{R}^d$. Furthermore, $\norm{Z}_q =\left(\mathrm{E}|Z|^q\right)^{1/q}$ with $q\geq 1$ is the $L_q$-norm for a random variable $Z$. We define the $\psi_{\alpha}$-Orlicz norm for any random variable $Z$ and $\alpha \in (0,2]$ by
\begin{equation}
	\label{Orlicz}
	\norm{Z}_{\psi_{\alpha}}=\inf\left\{t>0: \mathrm{E}\left[\exp\left(\frac{|Z|^{\alpha}}{t^{\alpha}} \right) \right] \leq 2 \right\}.
\end{equation} 
For $\alpha \in (0,1)$, there is an alternative definition that modifies \eqref{Orlicz} into a legitimate norm which can be used interchangeably with \eqref{Orlicz} up to an absolute constant that only depends on $\alpha$; see page 266 in \cite{van1996weak}. We also use the notation $a_n\lesssim b_n$ or $b_n\gtrsim a_n$ when there exists an absolute constant $A>0$ such that $a_n\leq A b_n$ when $n$ is large. If $a_n\gtrsim b_n$ and $a_n \lesssim b_n$, then $a_n$ and $b_n$ are asymptotically equivalent, denoted by $a_n\asymp b_n$. Finally, we denote the unit sphere in $\mathbb{R}^d$ by $\mathbb{S}^{d-1} = \{x\in \mathbb{R}^d: \norm{x}_2=1\}$ and a ball centered at $x$ with radius $r$ in $\mathbb{R}^d$ by $B_d(x,r)=\{y\in \mathbb{R}^d: \norm{y-x}_2\leq r\}$.


\section{Methodology}
\label{sec:Method}

In this section, we outline our debiasing inference method and discuss a feasible solution to the key debiasing program through its dual form. Furthermore, we motivate our debiasing method from the perspectives of a bias-variance trade-off as well as Neyman near-orthogonality. 

\subsection{Debiasing Inference Procedure}
\label{subsec:debias_method}

Recall that we aim to conduct statistical inference on the regression function $m_0(x) = \mathrm{E}(Y|X=x)=x^T\beta_0$ given a random sample $\{(Y_i,R_i,X_i)\}_{i=1}^n \subset \mathbb{R}\times \{0,1\} \times \mathbb{R}^d$, where $Y_i$ is the outcome variable with missingness indicator $R_i$ and $X_i$ is the fully observed high-dimensional covariate vector for $i=1,...,n$. To this end, we propose the following debiasing procedure.

\begin{itemize}
	\item {\bf Step 1:} Compute the Lasso pilot estimate $\hat{\beta} \in \mathbb{R}^d$ with complete-case data
	\begin{equation}
		\label{lasso_pilot}
		\hat{\beta} = \argmin\limits_{\beta\in \mathbb{R}^d} \left[\frac{1}{n}\sum\limits_{i=1}^n R_i (Y_i-X_i^T \beta)^2+ \lambda\norm{\beta}_1 \right],
	\end{equation}
	where $\lambda>0$ is a regularization parameter.
	
	\item {\bf Step 2:} Obtain consistent estimates $\hat{\pi}_i,i=1,...,n$ of the propensity scores $\pi_i=\pi(X_i)=\mathrm{P}(R_i=1|X_i), i=1,...,n$ by any machine learning method (not necessarily a parametric model) on the data $\{(X_i,R_i)\}_{i=1}^n \subset \mathbb{R}^d \times \{0,1\}$. 
	
	\item {\bf Step 3:} Solve for the debiasing weight vector $\hat{\bm{w}} \equiv \hat{\bm{w}}(x) \in \mathbb{R}^n$ through the following debiasing program with a tuning parameter $\gamma >0$ as:
	\begin{equation}
		\label{debias_prog}
		\min_{\bm{w} \in \mathbb{R}^n} \left\{\sum_{i=1}^n \hat{\pi}_iw_i^2: \norm{x- \frac{1}{\sqrt{n}}\sum_{i=1}^n w_i\cdot \hat{\pi}_i\cdot X_i }_{\infty} \leq \frac{\gamma}{n} \right\}.
	\end{equation}
	
	\item {\bf Step 4:} Define the debiased estimator for $m_0(x)$ as:
	\begin{equation}
		\label{debias_est}
		\hat{m}^{\text{debias}}(x;\hat{\bm{w}}) = x^T \hat{\beta} + \frac{1}{\sqrt{n}} \sum_{i=1}^n \hat{w}_i(x)R_i \left(Y_i-X_i^T \hat{\beta} \right).
	\end{equation}
	
	\item {\bf Step 5:} Construct the asymptotic $(1-\tau)$-level confidence interval for $m_0(x)$ as:
	\begin{equation}
		\label{asym_ci}
		\left[\hat{m}^{\text{debias}}(x;\hat{\bm{w}}) \pm \Phi^{-1}\left(1-\frac{\tau}{2}\right) \cdot \sigma_{\epsilon}\cdot \sqrt{\frac{1}{n}\sum_{i=1}^n \hat{\pi}_i \hat{w}_i(x)^2}\right],
	\end{equation}
	where $\Phi(\cdot)$ denotes the cumulative distribution function (CDF) of $\mathcal{N}(0,1)$. If the noise level $\sigma_{\epsilon}^2$ is unknown, then replace it by any consistent estimator $\hat{\sigma}_{\epsilon}^2$. 
\end{itemize}

{\bf Step 3} above is motivated by the bias-variance trade-off; see \autoref{subsec:bias_var_interp}. In \autoref{subsec:dual_form}, we derive the dual formulation of our debiasing program \eqref{debias_prog} in {\bf Step 3}. The dual is an unconstrained convex program that facilitates both the theoretical analysis of the statistical properties and the practical implementation of our debiasing procedure. Finally, we prove in \autoref{subsec:asymp_normal} that the debiased estimator $\hat{m}^{\text{debias}}(x;\hat{\bm{w}})$ is asymptotically normal and that its asymptotic variance can be estimated consistently by $\hat{\sigma}_{\epsilon}^2\sum_{i=1}^n \hat{\pi}_i \hat{w}_i^2(x)$. This guarantees the asymptotic validity of the confidence interval in {\bf Step 5}.

\begin{remark}[Debiasing program with heteroscedastic errors]
	\label{remark:hetero_error}
		Under heteroscedastic errors satisfying $\mathrm{E}\left(\epsilon_i^2|X_i\right) = \sigma_{\epsilon,i}^2$, our debiasing program \eqref{debias_prog} needs to be modified as follows:
	\begin{equation}
		\label{debias_prog_modified}
		\min_{\bm{w} \in \mathbb{R}^n} \left\{\sum_{i=1}^n \hat{\pi}_i \hat{\sigma}_{\epsilon,i}^2 w_i^2: \norm{x- \frac{1}{\sqrt{n}}\sum_{i=1}^n w_i\cdot \hat{\pi}_i\cdot X_i }_{\infty} \leq \frac{\gamma}{n} \right\},
	\end{equation}
	where $\hat{\sigma}_{\epsilon,i}^2$ is a proxy of $\sigma_{\epsilon,i}^2$ for $i=1,...,n$. For the asymptotic analysis, it suffices to construct $\hat{\sigma}_{\epsilon,i}^2, i=1,...,n$ so that 
	$$\sum_{i=1}^n \hat{\pi}_i \left(\hat{\sigma}_{\epsilon,i}^2 - \sigma_{\epsilon,i}^2\right) \hat{w}_i^2=o_P(1) \quad \text{ or } \quad \frac{\sum_{i=1}^n \hat{\pi}_i \left(\hat{\sigma}_{\epsilon,i}^2 - \sigma_{\epsilon,i}^2\right) \hat{w}_i^2}{\sum_{i=1}^n \hat{\pi}_i \sigma_{\epsilon,i}^2 \hat{w}_i^2}=o_P(1).$$
		One simple choice is the residual-based proxy $\hat{\sigma}_{\epsilon,i}^2 = \left(Y_i-X_i^T\hat{\beta}\right)^2$, where $\hat{\beta}$ is any consistent estimate of $\beta_0$, such as the Lasso pilot estimate from {\bf Step 1}, the refitted Lasso~\citep{belloni2013least}, or the weighted adaptive Lasso \citep{wagener2013adaptive}. Although this residual-based proxy may not be a consistent estimate of $\sigma_{\epsilon,i}^2$ for each $i=1,...,n$, it provides a practical plug-in approximation for the unknown error scale. More in-depth theoretical derivations are left for future work.
\end{remark}

\subsection{Dual Formulation of the Debiasing Program \eqref{debias_prog}}
\label{subsec:dual_form}

By standard results in convex optimization \citep{boyd2004convex}, strong duality refers to the situation in which the primal program and its dual attain the same optimal objective value. Under this notion, we derive the following result.
\begin{proposition}
	\label{prop:dual_debias}
	The dual form of the debiasing program \eqref{debias_prog} is given by
	\begin{equation}
		\label{debias_prog_dual}
		\min_{\ell \in \mathbb{R}^d} \left\{\frac{1}{4n} \sum_{i=1}^n \hat{\pi}_i \left[X_i^T \ell\right]^2 + x^T \ell +\frac{\gamma}{n}\norm{\ell}_1 \right\}.
	\end{equation}
	If strong duality holds, then the solution $\hat{\bm{w}}(x) \in \mathbb{R}^n$ to the primal debiasing program \eqref{debias_prog} and the solution $\hat{\ell}(x) \in \mathbb{R}^d$ to the dual debiasing program \eqref{debias_prog_dual} satisfy the following identity:
	\begin{equation}
		\label{primal_dual_sol2}
		\hat{w}_i(x) = -\frac{1}{2\sqrt{n}} \cdot X_i^T\hat{\ell}(x), \quad i=1,...,n.
	\end{equation}
\end{proposition}

The proof of Proposition~\ref{prop:dual_debias} is in \autoref{subapp:cond_MSE}. This result has two important implications. First, we can select the tuning parameter $\gamma >0$ via cross-validation on the loss function of the dual program. Indeed, since the primal debiasing program \eqref{debias_prog} is a constrained quadratic program, it is difficult to select $\gamma > 0$ based on the primal formulation alone: if $\gamma > 0$ is too small, the program will be infeasible; if $\gamma > 0$ is too large, the estimate will have a non-negligible bias (see \autoref{subsec:bias_var_interp}). Second, if strong duality holds, the debiased estimator \eqref{debias_est} can be expressed via~\eqref{primal_dual_sol2} in terms of the dual solution $\hat{\ell}(x)$ as
\begin{equation}
	\label{debias_est_dual}
	\hat{m}^{\text{debias}}(x;\hat{\bm{w}}) = x^T\hat{\beta} - \frac{1}{2n} \sum_{i=1}^n R_i X_i^T\hat{\ell}(x) \left(Y_i-X_i^T \hat{\beta}\right),
\end{equation}
where the second bias-correction term reduces to a linear combination of the sample average of $n$ random variables. This representation is the key to deriving the asymptotic normality of our debiased estimator; see \autoref{subsec:heuristic} and \autoref{subsec:asymp_normal} for details. Notice that strong duality holds whenever $\gamma > 0$ is sufficiently large; see Lemma~\ref{lem:strong_duality} in~\autoref{subsec:consistency} for more details.

\subsection{Bias-Variance Trade-off Perspective}
\label{subsec:bias_var_interp}

The design of our debiasing program \eqref{debias_prog} is motivated by controlling the conditional mean squared error $\mathrm{E}\left[\left(\sqrt{n}\cdot \hat{m}^{\text{debias}}(x;\hat{\bm{w}}) - \sqrt{n}\cdot m_0(x) \right)^2 \big| \bm{X}\right]$
of our debiased estimator \eqref{debias_est} and can thus be interpreted as solving a (conditional) bias-variance trade-off. To explicate this motivation, we consider the generic debiased estimator $m^{\text{debias}}(x;\bm{w})$ from \eqref{debias_est} with an arbitrary weight vector $\bm{w} =(w_1,...,w_n)^T\in \mathbb{R}^n$ as:
\begin{equation}
	\label{debias_generic}
	m^{\text{debias}}(x;\bm{w}) = x^T\beta + \frac{1}{\sqrt{n}} \sum_{i=1}^n w_i R_i \left(Y_i-X_i^T \beta \right),
\end{equation}
where one would plug in the Lasso pilot estimate $\hat{\beta}$ from \eqref{lasso_pilot} and the solution $\hat{\bm{w}}\equiv \hat{\bm{w}}(x) \in \mathbb{R}^n$ from our debiasing program \eqref{debias_prog} to obtain a concrete debiased estimator \eqref{debias_est} in practice. We decompose the conditional mean squared error of \eqref{debias_generic} into three explicit terms as follows.

\begin{proposition}
	\label{prop:cond_MSE}
	Under Assumption~\ref{assump:basic}, the conditional mean squared error of $\sqrt{n}\,m^{\mathrm{debias}}(x;\bm{w})$ given $\bm{X}=(X_1,...,X_n)^T$ has the following decomposition as:
	\begin{align}
		\label{cond_MSE}
		\begin{split}
			&\mathrm{E}\left[\left(\sqrt{n}\,m^{\mathrm{debias}}(x;\bm{w}) - \sqrt{n}\, m_0(x) \right)^2 \Big| \bm{X} \right] \\
			&= \underbrace{\sigma_{\epsilon}^2\sum_{i=1}^n w_i^2 \pi(X_i)}_{\text{Conditional variance I}} + \underbrace{\left(\beta_0 - \beta\right)^T \left[\sum_{i=1}^n w_i^2\pi(X_i)\left(1-\pi(X_i)\right)X_iX_i^T \right] \left(\beta_0 - \beta\right)}_{\text{Conditional variance II}} \\
			&\quad + \underbrace{\left[\left(\frac{1}{\sqrt{n}}\sum_{i=1}^n w_i \pi(X_i) X_i -x \right)^T \sqrt{n}\left(\beta_0 - \beta\right) \right]^2}_{\text{Conditional bias}}.
		\end{split}
	\end{align}
\end{proposition}

The derivation of Proposition~\ref{prop:cond_MSE} can be found in \autoref{subapp:cond_MSE}. By H{\"o}lder's inequality, we can upper bound the ``Conditional bias'' term as follows:
\begin{align*}
	&\left[\left(\frac{1}{\sqrt{n}}\sum_{i=1}^n w_i \pi(X_i) X_i -x \right)^T \sqrt{n}\left(\beta_0 - \beta\right) \right]^2 \leq \left[\norm{\frac{1}{\sqrt{n}}\sum_{i=1}^n w_i \pi(X_i) X_i -x }_{\infty} \sqrt{n}\norm{\beta_0 - \beta}_1 \right]^2.
\end{align*}
The debiasing program \eqref{debias_prog} is hence designed to optimize the weights $w_i, i=1,...,n$ so that the estimated ``Conditional variance I'' term $\sum_{i=1}^n \hat{\pi}_i w_i^2$ is minimized as the objective function while an upper bound of the estimated ``Conditional bias'' term $\norm{x- \frac{1}{\sqrt{n}} \sum_{i=1}^n w_i\hat{\pi}_i X_i}_{\infty}$ is well-controlled by the constraint. We ignore the ``Conditional variance II'' term because, by Lemma~\ref{lem:negligible_cond_var}, it is asymptotically negligible and dominated by the ``Conditional variance I'' term if $\beta = \hat{\beta}$ (the solution to the Lasso program \eqref{lasso_pilot}) and $\bm{w}=\hat{\bm{w}}(x)$ (the solution to the debiasing program \eqref{debias_prog}). Similar bias-variance trade-offs also appear in \cite{cai2023statistical} for generalized linear models and \cite{giessing2021inference} for quantile regression models. However, the motivation for the procedures in these two papers is ad hoc and not grounded in a rigorous decomposition of the conditional mean squared error. Upon minimizing the (conditional) variance, we expect our debiased estimator to be asymptotically more efficient than other estimators; see our discussion in \autoref{subsec:asymp_normal}.

\subsection{Neyman Near-Orthogonalization Viewpoint}
\label{subsec:neyman_ortho}

Our debiasing method in \autoref{subsec:debias_method} can also be motivated from the perspective of Neyman near-orthogonalization \citep{neyman1959optimal,neyman1979c,chernozhukov2018double}. To this end, we regard the regression function $m\equiv m(x) \in \mathbb{R}$ as the parameter of interest and the regression vector $\beta \in \mathbb{R}^d$ as a high-dimensional nuisance parameter. Under Assumption~\ref{assump:basic}, the true values of these two parameters are $m_0=m_0(x)=x^T\beta_0$ and $\beta_0$, respectively. We define the generic score function by
$$\Xi_x(Y,R,X; m,\beta) = m-x^T\beta - \sqrt{n} \cdot w \cdot R(Y-X^T\beta),$$ 
where $w \in \mathbb{R}$ is a symbolic weight. Given the observed data $\{(Y_i,R_i,X_i)\}_{i=1}^n$, arbitrary $\bm{w} \in \mathbb{R}^n$, and arbitrary $\beta \in \mathbb{R}^d$, the generic debiased estimator $m^{\text{debias}}(x, \bm{w})$ solves the sample-based estimating equation
\begin{equation}
	\label{neyman_score_func}
	\frac{1}{n}\sum_{i=1}^n \Xi_x(Y_i,R_i,X_i; m^{\text{debias}},\beta) =m^{\text{debias}}(x; \bm{w}) -x^T\beta - \frac{1}{\sqrt{n}}\sum_{i=1}^n w_i \cdot R_i\left(Y_i-X_i^T\beta\right) =0.
\end{equation}
Provided that Assumption~\ref{assump:basic} holds and $\bm{w} \in \mathbb{R}^n$ satisfies the box constraint in \eqref{debias_prog}, the Neyman near-orthogonalization condition given $\bm{X}=(X_1,...,X_n)^T \in \mathbb{R}^{n\times d}$ at $(m_0,\beta_0)$ requires that
\begin{align}
	\label{neyman_near_ortho}
	\begin{split}
		&\mathrm{E}\left[\frac{1}{n}\sum_{i=1}^n \Xi_x(Y_i,R_i,X_i; m_0,\beta_0) \bigg| \bm{X}\right] = 0,\\
		& \sup_{\beta \in \mathcal{T}_n} \left|\left\{\frac{\partial}{\partial \beta} \mathrm{E}\left[\frac{1}{n}\sum_{i=1}^n \Xi_x(Y_i,R_i,X_i; m,\beta) \Big| \bm{X}\right] \bigg|_{(m_0,\beta_0)} \right\}^T (\beta-\beta_0)\right| \leq \frac{\delta_n}{\sqrt{n}},  
	\end{split}
\end{align}
where $\mathcal{T}_n$ is a properly shrinking neighborhood of $\beta_0$ and $\delta_n=o(1)$; see Definition 2.2 in \cite{chernozhukov2018double}. Indeed, the first condition obviously holds true, while the second condition is also satisfied because for any $\beta \in \mathcal{T}_n$,
\begin{align*}
	&\left|\left\{\frac{1}{n}\sum_{i=1}^n \frac{\partial}{\partial \beta} \mathrm{E}\left[\Xi_x(Y_i,R_i,X_i; m,\beta) | \bm{X}\right]\big|_{(m_0,\beta_0)} \right\}^T (\beta-\beta_0)\right| \\
	&= \left|\left[x-\frac{1}{\sqrt{n}} \sum_{i=1}^n w_i \cdot \pi(X_i) X_i\right]^T (\beta_0 - \beta) \right| \\
	&\stackrel{\text{(a)}}{``\leq"} \norm{x-\frac{1}{\sqrt{n}} \sum_{i=1}^n w_i \cdot \hat{\pi}_i\cdot X_i}_{\infty} \norm{\beta-\beta_0}_1\\
	&\: \: \stackrel{\text{(b)}}{\leq} \frac{\gamma}{n} \norm{\beta-\beta_0}_1 \leq \frac{\delta_n}{\sqrt{n}},
\end{align*}
where the approximate inequality (a) follows from H{\"o}lder's inequality and substituting the estimated propensity scores $\hat{\pi}_i,i=1,...,n$ while (b) holds since $\bm{w} \in \mathbb{R}^n$ satisfies the box constraint in \eqref{debias_prog}. Therefore, the nuisance realization set $\mathcal{T}_n$ can be chosen as $\left\{\beta \in \mathcal{B}\subset \mathbb{R}^d: \norm{\beta-\beta_0}_1 \leq \frac{\sqrt{n}\delta_n}{\gamma} \right\}$ for some convex set $\mathcal{B}$ containing $\beta_0$. We show in \autoref{thm:lasso_const} that the Lasso pilot estimate $\hat{\beta}$ satisfies $\norm{\hat{\beta} - \beta_0}_1=O_P\left(s_{\beta}\sqrt{\frac{\log d}{n}} \right)$ and thus, our final debiased estimator \eqref{debias_est} satisfies the Neyman near-orthogonality condition \eqref{neyman_near_ortho} with a fine-tuned parameter $\gamma >0$.

By Theorem 3.1 in \cite{chernozhukov2018double}, the (asymptotic) variance of $m^{\text{debias}}(x, \bm{w})$ is $\sigma_{\epsilon}^2\, \mathrm{E}\left[\sum_{i=1}^n w_i^2 R_i^2|\bm{X}\right] = \sigma_{\epsilon}^2 \sum_{i=1}^n w_i^2 \pi(X_i)$. Therefore, our debiasing program \eqref{debias_prog} is designed to yield the debiased estimator with the smallest (estimated) variance among all the estimators that (approximately) satisfy the constraints in \eqref{neyman_near_ortho}. Consequently, we can expect that our debiased estimator is more efficient than most other estimators even without explicitly constructing an efficient score function satisfying the exact Neyman
orthogonality; see Section 2.2.2 in \cite{chernozhukov2018double} with references therein. 
Finally, in light of \eqref{neyman_near_ortho}, our debiasing method de-correlates the Lasso pilot regression from the propensity score estimation and the optimization of weights for the debiased estimator so that we are able to establish the asymptotic normality of our debiased estimator without resorting to sample-splitting.

\section{Asymptotic Theory}
\label{sec:theory}

In this section, we study the consistency and asymptotic normality of our debiasing inference method in \autoref{subsec:debias_method}. The high-level motivation and an overview of the technical results are given in \autoref{subsec:heuristic}, with more detailed analyses in the subsequent subsections.

\subsection{Heuristics and Overview}
\label{subsec:heuristic}

In order to derive the asymptotic normality of our debiased estimator \eqref{debias_est}, we need to analyze the asymptotic behavior of $\sqrt{n}\left[\hat{m}^{\text{debias}}(x;\hat{\bm{w}}) - m_0(x) \right]$. Under Assumption~\ref{assump:basic} and the definition \eqref{debias_est} of $\hat{m}^{\text{debias}}(x;\hat{\bm{w}})$, we deduce that
\begin{align*}
	\sqrt{n}\left[\hat{m}^{\text{debias}}(x;\hat{\bm{w}}) - m_0(x) \right] &= \sum_{i=1}^n \hat{w}_i(x) R_i \epsilon_i + \left[x-\frac{1}{\sqrt{n}} \sum_{i=1}^n \hat{w}_i(x) R_i X_i\right]^T \sqrt{n}\left(\hat{\beta} - \beta_0\right),
\end{align*}
where $\epsilon_i,i=1,...,n$ are i.i.d. noise variables under model \eqref{linear_reg}. From the above expression, it is unclear how one can apply the central limit theorem (CLT) or related asymptotic theories to derive the asymptotic distribution of $\sqrt{n}\left[\hat{m}^{\text{debias}}(x;\hat{\bm{w}}) - m_0(x) \right]$, because the limiting distribution of the weight vector $\hat{\bm{w}}\equiv \hat{\bm{w}}(x)= \left(\hat{w}_1(x),...,\hat{w}_n(x)\right)^T \in \mathbb{R}^n$ is intractable and may not even exist. However, the dual formulation in Proposition~\ref{prop:dual_debias} provides a promising approach to this problem. By the dual representation \eqref{debias_est_dual} of the debiased estimator, we have
\begin{align}
	\label{bias_var_dual}
	\sqrt{n}\left[\hat{m}^{\text{debias}}(x;\hat{\bm{w}}) - m_0(x)\right] 
	&= -\frac{1}{2\sqrt{n}} \sum_{i=1}^n R_i\epsilon_i X_i^T\hat{\ell}(x) + \left[x + \frac{1}{2n}\sum_{i=1}^n R_i X_i^T\hat{\ell}(x) X_i \right]^T \sqrt{n}\left(\beta_0 -\hat{\beta} \right) \nonumber\\
	&= -\frac{1}{2\sqrt{n}} \sum_{i=1}^n R_i\epsilon_i X_i^T\ell_0(x) + \text{``Bias terms''},
\end{align}
where $\ell_0(x) \in \mathbb{R}^d$ is the deterministic solution to the population dual program for which the dual solution $\hat{\ell}(x) \in \mathbb{R}^d$ is consistent; see \eqref{debias_prog_popul_dual} in \autoref{subsec:assump} for its definition. The leading term in \eqref{bias_var_dual} is a sample average of $n$ i.i.d. scalar random variables, which can be handled via the standard CLT, while the ``Bias terms'' in \eqref{bias_var_dual} are of order $o_P(1)$ under mild regularity conditions; see \autoref{thm:asym_normal}.

The rest of this section corroborates the above heuristic arguments. After introducing the notation and necessary assumptions, we derive the following main results:
\begin{itemize}
	\item {\bf Consistency of the Lasso pilot estimate \eqref{lasso_pilot}:} Under a suitable choice of $\lambda > 0$, it holds that $\norm{\hat{\beta}-\beta_0}_2= O_P\left(\frac{1}{\kappa_R^2}\sqrt{\frac{s_{\beta}\log d}{n}}\right)$ when $\log d=o(n)$, where $\kappa_R^2 >0$ is the minimum eigenvalue of the complete-case Gram matrix $\mathrm{E}\left(RXX^T\right)$; see \autoref{thm:lasso_const} for details.
	
	\item {\bf Consistency of the solution to the dual program \eqref{debias_prog_dual}:} Under a suitable choice of $\gamma >0$, we deduce that $\norm{\hat{\ell}(x) -\ell_0(x)}_2 = O_P\left(\frac{1}{\kappa_R^3} \sqrt{\frac{s_{\ell}(x)\log d}{n}} + \left(\frac{1+\kappa_R^2}{\kappa_R^4}\right) r_{\ell} + \frac{r_{\pi}\sqrt{s_{\ell}(x)}}{\kappa_R^4} \right)$, where $\ell_0(x)$ is the solution to the population dual program, $r_{\ell}>0$ is a parameter for the sparse approximation to $\ell_0(x)$, and $r_{\pi}>0$ is the rate of convergence of $\max_{1 \leq i \leq n}\left|\hat{\pi}_i-\pi(X_i) \right|$; see \autoref{thm:dual_consist2} and its related discussions.
	
	\item {\bf Asymptotic normality of the debiased estimator \eqref{debias_est}:} Combining the above consistency results with other mild regularity conditions, we prove that 
	$$\sqrt{n}\left[\hat{m}^{\mathrm{debias}}(x;\hat{\bm{w}}) - m_0(x)\right] \stackrel{d}{\to} \mathcal{N}\left(0,\, \sigma_m^2(x) \right)$$ 
	and discuss the semiparametric efficiency of $\sigma_m^2(x) = \lim_{n\to\infty} \sigma_{\epsilon}^2 \cdot x^T\left[\mathrm{E}\left(R XX^T\right)\right]^{-1}x$; see \autoref{thm:asym_normal} for details.
\end{itemize}

\subsection{Notations and Assumptions}
\label{subsec:assump}

We enumerate the regularity conditions under the high-dimensional sparse linear model in Assumption~\ref{assump:basic} for our subsequent theoretical analysis. 

\begin{assumption}[Sub-Gaussian covariate]
	\label{assump:sub_gaussian}
	The covariate vector $X\in \mathbb{R}^d$ is sub-Gaussian, \emph{i.e.}, $\norm{u^TX}_{\psi_2} \lesssim \left[\mathrm{E}\left(u^TXX^T u\right)\right]^{1/2}$ for any $u\in \mathbb{R}^d$.
\end{assumption}

We introduce the sub-Gaussian condition as Assumption~\ref{assump:sub_gaussian} on the covariate vector $X$ in order to control the exponential tail behavior of $X$ by its second moment and facilitate our proofs. 
Notice that the covariate vector $X \in \mathbb{R}^d$ is not necessarily centered in our definition of a sub-Gaussian covariate vector.
We denote the $s$-sparse maximum eigenvalue of the population Gram matrix $\mathrm{E}\left(XX^T\right) \in \mathbb{R}^{d\times d}$ by
\begin{equation}
	\label{max_eigen_gram}
	\phi_s := \sup_{u\in \mathbb{S}^{d-1}, \norm{u}_0\leq s} \mathrm{E}\left[(X^Tu)^2\right].
\end{equation}
In particular, $\mathrm{E}(X_{ik}^2) \leq \max\limits_{1\leq i \leq n}\max\limits_{1\leq k\leq d} \mathrm{E}(X_{ik}^2) \leq \phi_1$ for each covariate vector $X_i=(X_{i1},...,X_{id})^T \in \mathbb{R}^d$. It is noteworthy that under the high-dimensional setting with $s\ll d$, the $s$-sparse maximum eigenvalue $\phi_s$ is substantially smaller than the unrestricted maximum eigenvalue $\phi_d$ of $\mathrm{E}\left(XX^T\right)$ and can even be independent of the dimension $d$.

\begin{assumption}[Sub-Gaussian noise]
	\label{assump:subgau_error}
	The noise variable $\epsilon\in \mathbb{R}$ is sub-Gaussian, \emph{i.e.}, $\norm{\epsilon}_{\psi_2} \lesssim \left[\mathrm{E}(\epsilon^2) \right]^{1/2} = \sigma_{\epsilon}$ with $\sigma_{\epsilon}^2>0$ being the noise level. 
\end{assumption}

We introduce Assumption~\ref{assump:subgau_error} to simplify the theoretical analysis. In principle, this assumption can be relaxed to a lower-order moment condition; see~\cite{belloni2013least}. We conduct numerical experiments on our debiasing method in which the sub-Gaussian noise assumption is violated in \autoref{subsec:heavy_noises} and demonstrate that our proposed method is robust to various heavy-tailed noise distributions. To establish the consistency of the Lasso pilot estimate \eqref{lasso_pilot} and the dual solution $\hat{\ell}(x) \in \mathbb{R}^d$ to \eqref{debias_prog_dual}, we impose the following eigenvalue lower bound on the complete-case population Gram matrix $\mathrm{E}\left(RXX^T\right) \in \mathbb{R}^{d\times d}$.

\begin{assumption}[Lower eigenvalue bound on $\mathrm{E}\left(RXX^T\right)$]
	\label{assump:gram_lower_bnd}
	The minimum eigenvalue of the complete-case Gram matrix $\mathrm{E}\left(RXX^T\right) \in \mathbb{R}^{d\times d}$ is bounded away from zero, \emph{i.e.}, there exists a constant $\kappa_R>0$ such that 
	$$ \inf_{v\in \mathbb{S}^{d-1}} \mathrm{E}\left[R(X^Tv)^2\right] \geq \kappa_R^2.$$
\end{assumption}

Assumption~\ref{assump:gram_lower_bnd} ensures that the regression vector $\beta_0\in \mathbb{R}^d$ is identifiable. We explicitly introduce the ``constant'' $\kappa_R >0$ to allow for situations in which the minimum eigenvalue of $\mathrm{E}\left(RXX^T\right)$ may decrease (slowly) as $d,n\to\infty$. In particular, $\kappa_R$ serves as a lower bound tracking the smallest eigenvalue of $\mathrm{E}\left(RXX^T\right)$ and will appear explicitly in the rates derived in the subsequent analysis.

When proving the consistency of the Lasso pilot estimate \eqref{lasso_pilot}, however, we allow the common relaxation of Assumption~\ref{assump:gram_lower_bnd} that the minimum eigenvalue of $\mathrm{E}\left(RXX^T\right)$ is only bounded away from zero over the cone $\mathcal{C}_1(S_{\beta},3)$:
\begin{equation}
	\label{res_eigen_cond}
	\inf_{v\in \mathcal{C}_1(S_{\beta},3) \cap \mathbb{S}^{d-1}} \mathrm{E}\left[R(X^Tv)^2\right] \geq \kappa_R^2,
\end{equation}
where $S_{\beta} \subset \{1,...,d\}$ is an index set for nonzero entries of $\beta_0\in \mathbb{R}^d$ while $\mathcal{C}_q(S,\vartheta) = \left\{v\in \mathbb{R}^d: \norm{v_{S^c}}_q \leq \vartheta \norm{v_S}_q\right\}$ is a cone of dominant coordinates (Section 7.2 in \citealt{koltchinskii2011oracle}) for some index set $S \subset \{1,...,d\}$ and the parameter $\vartheta >0$ with respect to the norm $\norm{\cdot}_q$ for $q\geq 1$. Here, $v_S \in \mathbb{R}^{|S|}$ denotes the subvector with elements of $v$ indexed by $S$ and $v_{S^c}$ is defined in an analogous manner. This relaxed assumption \eqref{res_eigen_cond} is known as the restricted eigenvalue condition in the literature; see Section 7.3.1 in \cite{wainwright2019high} and also \cite{bickel2009simultaneous,van2009conditions}. We show by Lemma~\ref{lem:gram_mat_conc} in \autoref{subapp:auxi_res} that under Assumption~\ref{assump:sub_gaussian}, the sample-based complete-case Gram matrix $\frac{1}{n}\sum_{i=1}^n R_iX_iX_i^T$ is close to the population one with high probability, so \eqref{res_eigen_cond} holds interchangeably for these two Gram matrices in the setting of this paper.

As for propensity score estimation, we introduce a mild assumption on the consistency of the estimated propensity scores $\hat{\pi}_i$ for the true propensity scores $\pi_i=\pi(X_i)$ for $i=1,...,n$. 

\begin{assumption}[Stochastic control of the propensity score estimation]
	\label{assump:prop_consistent_rate}
	Given any $n\geq 1$ and $\delta \in (0,1)$, there exists $r_{\pi} \equiv r_{\pi}(n,\delta) >0$ such that $\mathrm{P}\left(\max_{1 \leq i \leq n}\left|\hat{\pi}_i -\pi_i \right| > r_{\pi}\right) <\delta$.
\end{assumption}

Assumption~\ref{assump:prop_consistent_rate} is a non-asymptotic formulation of the stochastic boundedness for the maximum estimation errors of the propensity scores evaluated at the in-sample covariate vectors $X_i,i=1,...,n$. In particular, it follows from Assumption 3.2 in \cite{chakrabortty2019high} by a union bound under mild rate conditions. It can also be replaced by the in-sample $\ell_1$-prediction consistency condition $\mathrm{P}\left(\frac{1}{n}\sum_{i=1}^n \left|\hat{\pi}_i - \pi_i\right| > r_{\pi}\right) < \delta$ without affecting the subsequent results. In applications, the exact rate of consistency $r_{\pi}$ depends on how the propensity scores are estimated, whether by parametric, nonparametric, or semiparametric methods. We derive an explicit rate formula for $r_{\pi}$ under the Lasso-type generalized linear regression estimator \eqref{glm_lasso} in Proposition~\ref{prop:rel_const_glm} of \autoref{app:prop_score_lasso} and discuss the corresponding rates for other semiparametric or nonparametric methods in \autoref{app:prop_est_eg}. 
We demonstrate in \autoref{thm:dual_consist2} and \autoref{thm:asym_normal} below that as long as $r_{\pi}(n,\delta)\to 0$ for some $\delta \to 0$ as $n\to \infty$, the solution to the dual program \eqref{debias_prog_dual} is consistent and consequently, the debiased estimator \eqref{debias_est} is asymptotically normal; see \autoref{subsec:consistency} and \autoref{subsec:asymp_normal} for details. The specific rate at which $r_{\pi} \to 0$ is essentially irrelevant, as shown in \eqref{rate_required}. These results unveil the possibility of combining our debiased estimator with consistent estimates of the propensity scores obtained from any machine learning approach; see \autoref{subsec:nonpar_prop} for numerical experiments.

\begin{remark}
	In contrast to the existing literature on regression models with missing outcomes, incomplete covariates, or potential outcomes \citep{rosenbaum1983central,robins1994estimation,chakrabortty2019high}, we do not require the positivity condition $\pi(x)=\mathrm{P}(R=1|X=x) > \pi_{\min} > 0$ for all $x\in \mathrm{support}(X) \subset \mathbb{R}^d$. As noted by \cite{d2021overlap}, the positivity condition is too stringent in the context of high-dimensional covariates. Instead, we merely assume that the complete-case population Gram matrix $\mathrm{E}\left(RXX^T\right)$ is positive definite as in Assumption~\ref{assump:gram_lower_bnd}, which is valid even when the positivity condition is violated. For example, $\mathrm{E}\left(RXX^T\right)$ will be positive definite if $\mathrm{support}(X)$ has a nonzero Lebesgue measure in $\mathbb{R}^d$ and $\pi(x)$ is positive within a subset of $\mathrm{support}(X)$ with a nonzero Lebesgue measure.
\end{remark}

The population dual program for any $x\in \mathbb{R}^d$ is defined by taking $n\to \infty$ in \eqref{debias_prog_dual} under the true propensity scores as:
\begin{align}
	\label{debias_prog_popul_dual}
	\begin{split}
		\min_{\ell \in \mathbb{R}^d} \left\{ \frac{1}{4}\, \mathrm{E}\left[R \left(X^T\ell\right)^2\right] + x^T \ell  \right\}.
	\end{split}
\end{align}
Under Assumption~\ref{assump:gram_lower_bnd}, the unique solution to the above population dual program \eqref{debias_prog_popul_dual} is given by
\begin{equation}
	\label{popul_dual_exact}
	\ell_0(x) := -2\left[\mathrm{E}\left(R XX^T\right)\right]^{-1} x \in \mathbb{R}^d.
\end{equation}
The population dual solution $\ell_0(x)$ already appears as the ``orthogonalization parameter'' in the debiased machine learning literature~\citep[e.g.][Example 2.1]{chernozhukov2018double} and (up to a scaling factor) as the ``least favorable sub-model'' in~\cite{bellec2022biasing}. However, in both cases, $\ell_0(x)$ does not arise as the solution to a dual program whose corresponding primal program solves a bias-variance trade-off; see \autoref{subsec:bias_var_interp}. We assume that $\ell_0(x)$ is approximately sparse in the following sense; see also Assumption~\ref{assump:sparse_dual}.

\begin{definition}[$r_{\ell}$-approximation to $\ell_0(x)$]
	\label{defn:popul_dual}
	For any $x\in \mathbb{R}^d$, we define an $r_{\ell}$-approximation $\tilde{\ell}(x)$ to the population dual solution $\ell_0(x)$ as:
	\begin{equation}
		\label{popul_dual_approx}
		\tilde{\ell}(x) \equiv \tilde{\ell}(x;r_{\ell}):= \argmin_{\ell(x)\in \mathbb{R}^d} \left\{\norm{\ell(x)}_0: \norm{\ell(x) - \ell_0(x)}_2 \leq r_{\ell}\norm{\ell_0(x)}_2 \right\},
	\end{equation}
	where $r_{\ell} \in \left[0,\frac{1}{2}\right]$ controls the accuracy of approximation.
\end{definition}

\begin{remark}
	\label{remark:res_eigen_ell}
	Under the stronger Assumption~\ref{assump:gram_lower_bnd}, $\mathrm{E}(RXX^T)$ also satisfies the $\mathcal{C}_1(S_{\ell}(x),3)$-restricted eigenvalue condition with some parameter $\kappa_R^2>0$ defined in \eqref{res_eigen_cond}, where $S_{\ell}(x) := \mathrm{support}\left(\tilde{\ell}(x)\right) = \left\{1\leq k\leq d: \tilde{\ell}_k(x) \neq 0\right\}$.
\end{remark}

A larger value of $r_{\ell}$ gives rise to a sparser approximation $\tilde{\ell}(x)$ to the vector $\ell_0(x)$ of interest. Furthermore, the consistency of $\hat{\ell}(x)$ requires the approximation parameter $r_{\ell}$ for $\tilde{\ell}(x)$ to converge to 0 at a certain rate as $n\to \infty$; see Assumption~\ref{assump:sparse_dual} below and the discussion after \autoref{thm:dual_consist2} in \autoref{subsec:consistency}. In other words, as the sample size $n$ increases, the unique population dual solution $\ell_0(x)$ must tend to a sparse vector in order for the consistency result to hold.

\begin{assumption}[Sparsity of $\tilde{\ell}(x)$]
	\label{assump:sparse_dual}
	The $r_{\ell}$-approximation $\tilde{\ell}(x) \in \mathbb{R}^d$ to the population dual solution $\ell_0(x)\in \mathbb{R}^d$ exists and satisfies
	$s_{\ell}(x) := \norm{\tilde{\ell}(x)}_0 \ll \min\{n,d\}$.
\end{assumption}

The existence of $\tilde{\ell}(x)$ and the sparsity quantity $s_{\ell}(x)$ in Assumption~\ref{assump:sparse_dual} reflect the effective sparsity of the population dual solution $\ell_0(x)=-2\left[\mathrm{E}\left(R XX^T\right)\right]^{-1} x$. 
Intuitively, $s_{\ell}(x)$ is small when the query point $x\in \mathbb{R}^d$ interacts effectively with a limited subset of the covariate vector $X\in \mathbb{R}^d$ after accounting for both the missingness pattern in $Y \in \mathbb{R}$ and the dependence structure among covariates in $X$. Such a situation arises, for example, when the inverse complete-case Gram matrix $\left[\mathrm{E}\left(R XX^T\right)\right]^{-1}$ is (approximately) sparse and the query point $x$ itself is concentrated on only a few coordinates. These sparsity structures on $\left[\mathrm{E}\left(R XX^T\right)\right]^{-1}$ occur naturally in many high-dimensional statistical models, including graphical Lasso models \citep{friedman2008sparse}, autoregressive or short-range dependent processes, and settings where the propensity score $\mathrm{P}(R=1|X)$ depends only on a small subset of covariates in $X$. \autoref{app:suff_cond_dual} furnishes several sufficient conditions and motivating examples under which Assumption~\ref{assump:sparse_dual} holds.

\subsection{Consistency Results}
\label{subsec:consistency}

In this section, we present the consistency results for our Lasso pilot estimate \eqref{lasso_pilot} and the solution to the dual program \eqref{debias_prog_dual}. Additionally, we discuss the strong duality theory of our debiasing inference method and validate the asymptotic negligibility of the ``Conditional variance II'' term in Proposition~\ref{prop:cond_MSE} of \autoref{subsec:bias_var_interp}.

\begin{theorem}[Consistency of the Lasso pilot estimate with complete-case data]
	\label{thm:lasso_const}
	Let $\delta \in (0,1)$ and $A_1,A_2>0$ be some absolute constants. Suppose that Assumptions~\ref{assump:basic}, \ref{assump:sub_gaussian}, and \ref{assump:gram_lower_bnd} (more precisely, Eq.\eqref{res_eigen_cond}) hold and $n$ is sufficiently large so that 
	$$\sqrt{\frac{s_{\beta}\log(ed/s_{\beta})}{n}} + \sqrt{\frac{\log(1/\delta)}{n}} < \min\left\{\frac{\kappa_R^2}{A_1 \phi_{s_{\beta}}}, 1 \right\}.$$
	\begin{enumerate}[label=(\alph*)]
		\item If $\lambda> \frac{4}{n}\norm{\sum_{i=1}^n R_iX_i\epsilon_i}_{\infty} >0$ and $d\geq s_{\beta}+2$, then with probability at least $1-\delta$,
		$$\norm{\hat{\beta}-\beta_0}_2 \lesssim \frac{\sqrt{s_{\beta}}\lambda}{\kappa_R^2}.$$ 
		
		\item If, in addition, Assumption~\ref{assump:subgau_error} holds and $\lambda = A_2\sigma_{\epsilon} \sqrt{\frac{\phi_1\log(d/\delta)}{n}}$, then with probability at least $1-\delta$,
		\begin{align*}
			\norm{\hat{\beta}-\beta_0}_2 &\lesssim \frac{ \sigma_{\epsilon}}{\kappa_R^2} \sqrt{\frac{\phi_1 s_{\beta}\log(d/\delta)}{n}}.
		\end{align*}
	\end{enumerate}
\end{theorem}

\begin{remark}
	\label{remark:lasso_const}
	If $\log d=o(n)$ with $d>n$ and $\sigma_{\epsilon},\phi_1>0$ are independent of $n$ and $d$, then we obtain the asymptotic consistency result $\norm{\hat{\beta}-\beta_0}_2 = O_P\left(\frac{1}{\kappa_R^2}\sqrt{\frac{s_{\beta}\log d}{n}}\right)$ by setting $\delta=\frac{1}{d}$.
\end{remark}

The proof of \autoref{thm:lasso_const} is in \autoref{subapp:proof_consistency}. It implies that the Lasso pilot estimate $\hat{\beta}$ in \eqref{lasso_pilot} based on the complete-case data is consistent. This property is a consequence of our MAR assumption and the lower eigenvalue bound condition on $\mathrm{E}(RXX^T)$ (Assumption~\ref{assump:gram_lower_bnd}), because the conditional distribution of $Y$ given $X$ remains unchanged on the complete-case data ($R=1$) and the loss function in \eqref{lasso_pilot} is ``strictly convex'' with the cone in which $\hat{\beta}$ lies. 

Next, we present the consistency result for the dual solution $\hat{\ell}(x)$ to the sparse $r_{\ell}$-approximation $\tilde{\ell}(x)$. Among other things, the rate of consistency for $\hat{\ell}(x)$ depends on the rate of consistency $r_{\pi}=r_{\pi}(n,\delta)$ for estimating propensity scores.

\begin{theorem}[Consistency of the dual solution $\hat{\ell}(x)$]
	\label{thm:dual_consist2}
	Let $\delta \in (0,1)$, $x\in \mathbb{R}^d$ be a fixed covariate vector, and $A_1,A_2,A_3>0$ be some large absolute constants. Suppose that Assumptions~\ref{assump:basic}, \ref{assump:sub_gaussian}, \ref{assump:gram_lower_bnd}, \ref{assump:prop_consistent_rate}, and \ref{assump:sparse_dual} hold
	as well as
	$$\sqrt{\frac{s_{\ell}(x)\log(ed/s_{\ell}(x))}{n}} + \sqrt{\frac{\log(1/\delta)}{n}} + r_{\pi} < \min\left\{\frac{\kappa_R^2}{A_1 \phi_{s_{\ell}(x)}}, 1 \right\} \quad \text{ and } \quad d\geq s_{\ell}(x) + 2.$$
	We denote $r_{\gamma}\equiv r_{\gamma}(n,\delta) = \frac{\sqrt{\phi_1}}{\kappa_R}\norm{x}_2\sqrt{\frac{\log(d/\delta)}{n}} + \frac{\sqrt{\phi_1 \phi_{s_{\ell}(x)}}\norm{x}_2}{\kappa_R^2} \cdot r_{\pi}$. If $\frac{\gamma}{n} = A_2\cdot r_{\gamma}$ and 
	$$r_{\ell} \leq \min\left\{\frac{1}{2},\, \left[\left(\frac{A_2^2 \,\kappa_R^6\, \phi_1}{A_3^2\norm{x}_2 \phi_{s_{\ell}(x)}}\right)^{\frac{1}{4}} \left(\frac{\log(d/\delta)}{n}\right)^{\frac{1}{4}} + \left(\frac{A_2^2 \phi_1 \kappa_R^4}{A_3^2 \phi_{s_{\ell}(x)} \norm{x}_2^2}\right)^{\frac{1}{4}} \sqrt{r_{\pi}}\right]\right\},$$
	then with probability at least $1-\delta$,
	\begin{align*}
		\norm{\hat{\ell}(x) - \tilde{\ell}(x)}_2
		&\lesssim \frac{\norm{x}_2}{\kappa_R^3}\left[ \sqrt{\frac{\phi_1 s_{\ell}(x)\log(d/\delta)}{n}} + \frac{\phi_{s_{\ell}(x)}}{\kappa_R} \cdot r_{\ell} + \frac{\sqrt{\phi_1 \phi_{s_{\ell}(x)} s_{\ell}(x)}}{\kappa_R} \cdot  r_{\pi} \right].
	\end{align*}
\end{theorem}

\begin{remark}
	\label{remark:dual_const}
	If $\log d=o(n)$, $\norm{x}_2=O(1)$, and $\phi_1,\phi_{s_{\ell}(x)}$ are independent of $n$ and $d$, then we obtain the asymptotic consistency result
	$$\norm{\hat{\ell}(x) -\tilde{\ell}(x)}_2 = O_P\left(\frac{1}{\kappa_R^3} \sqrt{\frac{s_{\ell}(x)\log d}{n}} + \frac{r_{\ell}}{\kappa_R^4} + \frac{r_{\pi}\sqrt{s_{\ell}(x)}}{\kappa_R^4} \right)$$
	by setting $\delta = \frac{1}{d}$. This asymptotic rate of convergence comprises two parts. The first part is of the order $O_P\left(\frac{1}{\kappa_R^3}\sqrt{\frac{s_{\ell}(x)\log d}{n}} + \frac{r_{\ell}}{\kappa_R^4}\right)$, which is the oracle rate of convergence for $\norm{\hat{\ell}(x) -\tilde{\ell}(x)}_2$ when the true propensity scores are known; see Lemma~\ref{lem:dual_consist} of \autoref{subapp:proof_consist_dual}. The second part is of the order $O_P\left(\frac{r_{\pi}\sqrt{s_{\ell}(x)} }{\kappa_R^4}\right)$, which is related to the consistency of the propensity score estimation.
\end{remark}

The proof of \autoref{thm:dual_consist2} is in \autoref{subapp:proof_consist_dual}. By the triangle inequality, this result implies that
\begin{align*}
	\norm{\hat{\ell}(x) - \ell_0(x)}_2 &\leq \norm{\hat{\ell}(x) - \tilde{\ell}(x)}_2 + \norm{\tilde{\ell}(x) - \ell_0(x)}_2 \\
	&= O_P\left(\frac{1}{\kappa_R^3} \sqrt{\frac{s_{\ell}(x)\log d}{n}} + \left(\frac{1+\kappa_R^2}{\kappa_R^4}\right) r_{\ell} + \frac{r_{\pi}\sqrt{s_{\ell}(x)}}{\kappa_R^4} \right).
\end{align*}
Thus, the dual solution $\hat{\ell}(x)$ is consistent for the population dual solution $\ell_0(x)$ under the setup in Remark~\ref{remark:dual_const} if $\frac{1}{\kappa_R^3}\sqrt{\frac{s_{\ell}(x)\log d}{n}} \to 0$, $\left(\frac{1+\kappa_R^2}{\kappa_R^4}\right) r_{\ell} \to 0$, and $\frac{r_{\pi}\sqrt{s_{\ell}(x)}}{\kappa_R^4} \to 0$ as $n\to \infty$. In particular, we only need to consistently estimate the propensity scores $\pi_i, i=1,...,n$ because the rate of consistency $r_{\pi}$ depends on the in-sample prediction; see Remark~\ref{remark:benign_overfit} below.

The assumptions in \autoref{thm:dual_consist2} are also sufficient to establish strong duality between the primal debiasing program \eqref{debias_prog} and its dual \eqref{debias_prog_dual} as follows, whose proof is in \autoref{subapp:proofs_dual}.

\begin{lemma}[Sufficient conditions for strong duality]
	\label{lem:strong_duality}
	Let $\delta \in (0,1)$. Under the assumptions of \autoref{thm:dual_consist2}, there exists $\gamma>0$ such that $\frac{\gamma}{n} \asymp r_{\gamma}\equiv r_{\gamma}(n,\delta)$ and strong duality as well as the relation \eqref{primal_dual_sol2} in Proposition~\ref{prop:dual_debias} hold with probability at least $1-\delta$.
\end{lemma}

The consistency results for the Lasso pilot estimate $\hat{\beta}$ and the dual solution $\hat{\ell}(x)$ also imply that under mild regularity conditions, the ``Conditional variance II'' term in Proposition~\ref{prop:cond_MSE} is asymptotically negligible; see Lemma~\ref{lem:negligible_cond_var} and its proof in \autoref{subapp:proofs_dual}. This result completes the motivation of the debiasing program \eqref{debias_prog} from the perspective of a (conditional) bias-variance trade-off in \autoref{subsec:bias_var_interp}.

\begin{lemma}[Negligible conditional variance term]
	\label{lem:negligible_cond_var}
	Suppose that Assumptions~\ref{assump:basic}, \ref{assump:sub_gaussian}, \ref{assump:subgau_error}, \ref{assump:gram_lower_bnd}, \ref{assump:prop_consistent_rate}, and \ref{assump:sparse_dual} hold
	as well as 
	$$\left[s_{\beta} \log(d/s_{\beta}) + s_{\ell}(x) \log(d/s_{\ell}(x))\right]^2 =O\left(\sqrt{n}\right), \, \frac{\norm{x}_2}{\kappa_R^2}\sqrt{\phi_{2s_{\ell}(x)} \phi_{2s_{\beta}}} = O(1), \, \text{ and } \, r_{\ell}=O(1).$$
	If $\norm{\hat{\beta}-\beta_0}_2=o_P(1)$ and $\norm{\hat{\ell}(x) - \tilde{\ell}(x)}_2 =o_P(1)$, then the ``Conditional variance II'' term in Proposition~\ref{prop:cond_MSE} is asymptotically negligible in the sense that
	\begin{align*}
		\left(\beta_0 - \hat{\beta}\right)^T \left[\sum_{i=1}^n \hat{w}_i^2(x)\hat{\pi}_i\left(1-\hat{\pi}_i\right)X_iX_i^T \right] \left(\beta_0 - \hat{\beta}\right) =o_P(1).
	\end{align*} 
\end{lemma}

\subsection{Asymptotic Normality}
\label{subsec:asymp_normal}

In this section, we leverage the consistency results from \autoref{subsec:consistency} to establish the asymptotic normality of the debiased estimator $\hat{m}^{\text{debias}}(x;\hat{\bm{w}})$ in \eqref{debias_est}. As mentioned in \autoref{subsec:heuristic}, the key observation for applying the CLT is the linear expansion \eqref{bias_var_dual} based on the dual solution.

\begin{theorem}[Asymptotic normality of the debiased estimator]
	\label{thm:asym_normal}
	Let $x\in \mathbb{R}^d$ be a fixed query point with $x\neq 0$ and $s_{\max}=\max\left\{s_{\beta},s_{\ell}(x)\right\}$. Suppose that Assumptions~\ref{assump:basic}, \ref{assump:sub_gaussian}, \ref{assump:subgau_error}, \ref{assump:gram_lower_bnd}, \ref{assump:prop_consistent_rate}, and \ref{assump:sparse_dual} hold and $\lambda, \gamma, r_{\ell} >0$ are specified as in \autoref{thm:lasso_const}(b) and \autoref{thm:dual_consist2}, respectively. Furthermore, we assume that $\sigma_{\epsilon}\phi_1\phi_{s_{\max}}^{3/2}\norm{x}_2=O(1)$,
	\begin{align}
		\label{rate_required}
		\frac{(1+\kappa_R^2)s_{\max}\log(nd)}{\kappa_R^4} = o\left(\sqrt{n}\right), \quad \text{ and }  \quad \frac{(1+\kappa_R^4)\sqrt{\log(nd)}}{\kappa_R^6}\left(\sqrt{s_{\max}}\cdot r_{\ell} +s_{\max}r_{\pi}\right)=o(1).
	\end{align}
	Then, when $\sigma_m^2(x) = \lim_{n\to\infty} \sigma_{\epsilon}^2 \cdot x^T\left[\mathrm{E}\left(R XX^T\right)\right]^{-1}x$ exists, we have that
	$$\sqrt{n}\left[\hat{m}^{\mathrm{debias}}(x;\hat{\bm{w}}) - m_0(x)\right] \stackrel{d}{\to} \mathcal{N}\left(0,\, \sigma_m^2(x) \right).$$
\end{theorem}

The proof of \autoref{thm:asym_normal} can be found in \autoref{subapp:proof_asym_normal}, and we make four remarks about this paramount result. 

\begin{remark}[Semiparametric efficiency]
	Given any fixed dimension $d>0$, the asymptotic variance 
	$$\sigma_{m,d}^2(x) = \sigma_{\epsilon}^2 \cdot x^T\left[\mathrm{E}\left(R XX^T\right)\right]^{-1}x = \sigma_{\epsilon}^2\cdot x^T\left\{\mathrm{E}\left[\pi(X) XX^T\right]\right\}^{-1}x$$ 
	for our debiased estimator achieves the semiparametric efficiency bound among all asymptotically linear estimators with the MAR outcome \citep{muller2012efficient}. It also attains the efficiency bound for de-sparsified Lasso in Theorem 3 of \cite{jankova2018semiparametric} under Gaussian covariates and noises. This key property demonstrates the superiority of our debiasing inference method in terms of statistical efficiency.
\end{remark}

\begin{remark}[Requirement on the query point $x \in \mathbb{R}^d$]
	We require $\norm{x}_2=O(1)$ to control higher-order remainder terms and establish asymptotic normality of $\hat{m}^{\text{debias}}(x;\hat{\bm{w}})$. This requirement may seem contradictory to the standard result $\norm{x}_2=O_P\left(\sqrt{d}\right)$ in high-dimensional statistics when $x$ is drawn from a sub-Gaussian distribution; see Theorem 3.1.1 in \cite{vershynin2018high}. However, for those scenarios where $\norm{x}_2$ grows with dimension $d$, we can always apply our debiasing method to the normalized query point $\frac{x}{\norm{x}_2}$. Since the solution $\hat{\bm{w}}\equiv \hat{\bm{w}}\left(\frac{x}{\norm{x}_2}\right)$ to our debiasing program \eqref{debias_prog} is adaptive to the normalization, the resulting debiased estimator $\hat{m}^{\text{debias}}\left(\frac{x}{\norm{x}_2};\hat{\bm{w}} \right)$ and associated asymptotic confidence intervals can be rescaled with $\norm{x}_2$ in order to obtain asymptotically valid confidence intervals for $m_0(x)=x^T\beta_0$.
\end{remark}

\begin{remark}[In-sample consistency of the propensity score estimation]
	\label{remark:benign_overfit}
	Our debiased estimator is asymptotically normal as long as the propensity scores are pointwise consistently estimated at in-sample data points, because the rate requirement \eqref{rate_required} on $r_{\pi}$ relies only on the in-sample prediction. Furthermore, unlike double/debiased machine learning \citep{chernozhukov2018double}, there is no need to leverage sample-splitting or cross-fitting when estimating the nuisance propensity score $\pi(X)=\mathrm{P}(R=1|X)$ in order to guarantee the asymptotic normality. While we derive $r_{\pi} = O_P\left(\frac{\log d}{\sqrt{n}}\right)$ under the Lasso-type generalized linear regression estimator \eqref{glm_lasso} in Proposition~\ref{prop:rel_const_glm} of \autoref{app:prop_score_lasso}, this in-sample prediction rate of consistency is, in general, slower than those rates from more complicated machine learning methods such as deep neural networks \citep{farrell2021deep}, random forests \citep{gao2022towards}, and support vector machines with universal kernels (such as the Gaussian radial basis function) \citep{steinwart2001influence}. The higher learnability of these propensity score estimation methods leads to better performance of our debiasing method; see \autoref{subsec:nonpar_prop} for numerical experiments.
\end{remark}

\begin{remark}[Necessary sparsity condition]
	Compared with the consistency results in \autoref{thm:dual_consist2}, we impose a stricter growth condition $s_{\max}=\max\{s_{\beta}, s_{\ell}(x)\}=o\left(\frac{\sqrt{n}}{\log d} \right)$ on the sparsity level in order to establish the asymptotic normality, provided that other model parameters $\sigma_{\epsilon},\phi_{s_{\max}},\kappa_R$ are independent of $n$ and $d$. The requirement on the sparsity level $s_{\beta}$ is a standard and essentially necessary condition in the high-dimensional debiased inference literature in pursuit of asymptotic normality \citep{van2014asymptotically,javanmard2014confidence,javanmard2018debiasing,cai2023statistical}; see also the discussion in Section 8.6 of \cite{jankova2018semiparametric}. 
	The condition on the sparsity level $s_{\ell}(x)$ of the population dual solution $\ell_0(x)=-2\left[\mathrm{E}\left(RXX^T\right)\right]^{-1} x$ to \eqref{debias_prog_popul_dual} is novel and due to the fact that we establish the asymptotic normality of our debiased estimator for $m_0(x)=x^T\beta_0$ but not the debiased estimator of the regression vector $\beta_0\in \mathbb{R}^d$. Notably, \cite{bellec2022biasing} imposed a slightly weaker sparsity condition on essentially the same object as $\ell_0(x)$ and only required $s_\ell(x) = o\left(\frac{n}{\log d} \right)$. 
\end{remark}

To utilize the asymptotic normality of \autoref{thm:asym_normal} and conduct inference on $m_0(x)$ in practice, we need to estimate the asymptotic variance $\sigma_m^2(x)$ or $\sigma_{m,d}^2(x)$ under the current dimension $d$. Proposition~\ref{prop:var_const_est} below verifies that, when the noise level $\sigma_{\epsilon}^2$ is known, $\sigma_{\epsilon}^2 \sum_{i=1}^n \hat{\pi}_i \hat{w}_i^2(x)$ is a consistent estimator of $\sigma_{m,d}^2(x)$; see \autoref{subapp:proof_var_const_est} for its proof. 

\begin{proposition}[Consistent estimate of the asymptotic variance]
	\label{prop:var_const_est}
	Let $x\in \mathbb{R}^d$ be a fixed query point. Suppose that Assumptions~\ref{assump:basic}, \ref{assump:sub_gaussian}, \ref{assump:gram_lower_bnd}, \ref{assump:prop_consistent_rate}, and \ref{assump:sparse_dual} hold and $\gamma, r_{\ell} >0$ are specified as in \autoref{thm:dual_consist2}. Furthermore, we assume that $\norm{x}_2^2 \phi_{s_{\ell}(x)}^2 = O(1)$, $\frac{(1+\kappa_R^3)}{\kappa_R^5} \sqrt{\frac{s_{\ell}(x)\log(nd)}{n}} = o(1)$, and $\frac{(1+\kappa_R^4)}{\kappa_R^6} \left[r_{\ell} + r_{\pi}\sqrt{s_{\ell}(x)}\right]=o(1)$. Then,
	\begin{align*}
		&\left|\sum_{i=1}^n \hat{\pi}_i \hat{w}_i^2(x) - x^T\left[\mathrm{E}\left(RXX^T\right)\right]^{-1}x \right|=o_P(1).
	\end{align*}
\end{proposition}

This result again motivates the design of our debiasing program as in \autoref{subsec:bias_var_interp} and supports the formulation of our asymptotically valid confidence interval \eqref{asym_ci}. When the noise variance $\sigma_{\epsilon}^2$ is unknown, it can be replaced by any consistent estimator $\hat{\sigma}_{\epsilon}^2$ without affecting asymptotic efficiency. This yields the feasible estimator $\hat{\sigma}_{\epsilon}^2 \sum_{i=1}^n \hat{\pi}_i \hat{w}_i^2(x)$ of the asymptotic variance. There are various approaches available in the literature to construct such an estimator for $\sigma_{\epsilon}^2$ \citep{zhang2010mcp,fan2012variance,dicker2012residual,bayati2013estimating,belloni2013least,reid2016study}. In our simulation studies and real-world application, we use the scaled Lasso \citep{sun2012scaled} with automatically selected regularization parameter $\lambda>0$ to obtain $\hat{\sigma}_{\epsilon}^2$ (and, simultaneously, the Lasso pilot estimate $\hat{\beta}$); see \autoref{app:imple_debias} for details.

\section{Experiments}
\label{sec:experiments}

In this section, we evaluate the empirical performance of our proposed debiasing inference method in \autoref{subsec:debias_method} and compare it with several existing high-dimensional inference methods in the literature through comprehensive simulation studies. 

\subsection{Basic Simulation Design and Methods to Be Compared}
\label{subsec:sim_design}

We generate the i.i.d. data $\{(Y_i,R_i,X_i)\}_{i=1}^n \subset \mathbb{R}\times \{0,1\} \times\mathbb{R}^d$ from the following linear model 
\begin{equation}
	\label{sim_setting}
	Y_i=X_i^T \beta_0 + \epsilon_i \quad \text{ with } \quad X_i\ind \epsilon_i, \quad Y_i\ind R_i|X_i, \quad \text{ and } \quad X_i \sim \mathcal{N}_d(\bm{0},\Sigma),
\end{equation}
where $d=1000$ and $n=900$ unless stated otherwise. In order to investigate the effect of different covariance structures $\Sigma\in \mathbb{R}^{d\times d}$, sparsity patterns of the true regression vector $\beta_0\in \mathbb{R}^d$, and designs of the query point $x \in \mathbb{R}^d$ on the finite-sample performance of the candidate high-dimensional inference methods, we consider all possible combinations of the following simulation designs.

{\bf Choices of the covariance matrix:} We consider two standard designs of $\Sigma=\left(\Sigma_{jk}\right)_{j,k=1}^d$ from the literature:
\begin{enumerate}[label=(\roman*)]
	\item The circulant symmetric matrix $\Sigma^{\mathrm{cs}}$ in \cite{javanmard2014confidence} defined by $\Sigma_{jj}=1$, $\Sigma_{jk}=0.1 \text{ when } j+1\leq k\leq j+5 \text{ or } j+d-5\leq k \leq j+d-1$ with $\Sigma_{jk}=0$ elsewhere for $j\leq k$, and $\Sigma_{jk}=\Sigma_{kj}$;
	
	\item The Toeplitz (or autoregressive) matrix $\Sigma^{\mathrm{ar}}$ in \cite{van2014asymptotically} defined by $\Sigma_{jk}=0.9^{|j-k|}$.
\end{enumerate}

{\bf Designs of the regression vector:} For the true value of $\beta_0 \in \mathbb{R}^d$, we consider three different scenarios as:
\begin{enumerate}[label=(\roman*)]
	\item $\beta_0^{sp}=\left(\underbrace{\sqrt{5},...,\sqrt{5}}_5,0,...,0\right)^T \in \mathbb{R}^d$ is sparse;
	
	\item $\beta_0^{de}\propto \left(1,\frac{1}{\sqrt{2}},...,\frac{1}{\sqrt{d}}\right)^T \in \mathbb{R}^d$ is dense with $\norm{\beta_0^{de}}_2=5$;
	
	\item $\beta_0^{pd}\propto \left(1,\frac{1}{2},...,\frac{1}{d}\right)^T \in \mathbb{R}^d$ is pseudo-dense with $\norm{\beta_0^{pd}}_2=5$. 
\end{enumerate}

{\bf Designs of the query point:} We conduct experiments on the following four different choices of the query point $x\in \mathbb{R}^d$:
\begin{enumerate}[label=(\roman*)]
	\item $x^{(1)}=(1,0,...,0)^T\in \mathbb{R}^d$ to infer the individual effect of the first (significantly nonzero) component of $\beta_0$;
	
	\item $x^{(2)}=\left(1,\frac{1}{2},\frac{1}{4},0,0,0,\frac{1}{2},\frac{1}{8},0,...,0\right)^T\in\mathbb{R}^d$ to infer the joint effects of a few components of $\beta_0$;
	
	\item $x^{(3)}=\left(0,...,0,\underbrace{1}_{100^\text{th}},0,...,0\right)^T\in\mathbb{R}^d$ to infer an inactive (\emph{i.e.}, truly zero or close to zero) component of $\beta_0$;
	
	\item $x^{(4)} = \left(1,\frac{1}{2^2},...,\frac{1}{d^2}\right)^T \in \mathbb{R}^d$ for the circulant symmetric covariance cases, and $x^{(4)} = \left(1,\frac{1}{2},...,\frac{1}{d}\right)^T \in \mathbb{R}^d$ for the Toeplitz covariance cases. The purpose of choosing a relatively dense query point $x^{(4)}$ is to study the joint effect of all the components of $\beta_0$;
	
	\item $x^{(5)} = \left(\frac{1}{\sqrt{d}},\frac{1}{\sqrt{d}},...,\frac{1}{\sqrt{d}}\right)^T \in \mathbb{R}^d$, which is used to examine how a fully dense query point and the corresponding population dual solution $\ell_0(x)$ affect the inference results.
\end{enumerate}

{\bf MAR mechanism:} After $Y_i,i=1,...,n$ are sampled according to \eqref{sim_setting}, we define the missingness indicators $R_i,i=1,...,n$ for $Y_i,i=1,...,n$ through the MAR mechanism by 
\begin{equation}
	\label{MAR_correct}
	\mathrm{P}(R_i=1|X_i) = \frac{1}{1+\exp\left(-1+X_{i7}-X_{i8}\right)} \quad \text{ for } \quad i=1,...,n.
\end{equation}
In general, the above MAR mechanism yields around 28\% of missingness for the outcome variables $Y_i,i=1,...,n$, making the complete-case data even more high-dimensional. Additional simulation results under a simpler MCAR setting are deferred to \autoref{subapp:mcar_res}.

{\bf Noise distribution:} We first generate the noise variables $\epsilon_i,i=1,...,n$ independently from $\mathcal{N}(0,1)$. Other noise distributions that violate the sub-Gaussian noise condition (Assumption~\ref{assump:subgau_error}) are considered in \autoref{subsec:heavy_noises}. 

{\bf Methods to be compared:} To demonstrate the statistical efficiency of our proposed debiasing method ({\bf ``Debias''}), we compare it with several existing inference methods in the literature. The detailed implementation of our proposed method can be found in \autoref{app:imple_debias}, where we also explain the three different criteria ``min-CV'', ``1SE'', and ``min-feas'' for selecting the tuning parameter $\gamma>0$ via cross-validation. For a fair comparison, we implement the first group of comparative methods below on (i) the complete-case (CC) data $(X_i,Y_i)\in \mathbb{R}^d \times \mathbb{R}$ with $R_i=1$ for $i=1,...,n$, (ii) the inverse probability weighted (IPW) data $\left(\frac{X_i}{\sqrt{\hat{\pi}_i}}, \frac{Y_i}{\sqrt{\hat{\pi}_i}}\right) \in \mathbb{R}^d\times \mathbb{R}$ with $R_i=1$ for $i=1,...,n$, and (iii) the oracle fully observed data $(X_i,Y_i),i=1,...,n$. Here, $\hat{\pi}_i,i=1,...,n$ are the propensity scores estimated by the default Lasso-type logistic regression as described in \autoref{app:prop_score_lasso}.
\begin{itemize}
	\item {\bf ``DL-Jav''} is the debiased Lasso proposed by \cite{javanmard2014confidence} that solves $d$ convex programs for an unbiased estimator for $\beta_0$ and a consistent estimate for the asymptotic covariance matrix. The method is implemented by the source code \texttt{sslasso} provided in their paper. We select the tuning parameters as in their Section 5.1.  
	
	\item {\bf ``DL-vdG''} is the debiased Lasso proposed by \cite{van2014asymptotically} that solves $d$ nodewise regressions for the debiased estimator of $\beta_0$ and the estimate of its asymptotic covariance matrix. This method is implemented via the function \texttt{lasso.proj} in the R package \texttt{hdi} \citep{dezeure2015high}.
	
	\item {\bf ``R-Proj''} is the ridge projection method proposed by \cite{buhlmann2013statistical} that produces a bias-corrected estimator of $\beta_0$ via a ridge regression estimator. This method is implemented by the function \texttt{ridge.proj} in the R package \texttt{hdi}. We initialize the estimation of $\beta_0$ with an estimate obtained from the scaled Lasso. Since it is difficult to obtain a consistent estimate of the asymptotic covariance matrix from the function \texttt{ridge.proj}, we only run ``R-Proj'' when the query point is $x^{(1)}$ or $x^{(3)}$. This means that we only conduct inference on single coordinates of $\beta_0$ under these scenarios.
	
	\item {\bf ``Refit''} first runs the scaled Lasso to obtain a pilot estimate $\hat{\beta}$ and then computes the final least-squares estimate based on covariates in the support set of $\hat{\beta}$ \citep{belloni2013least}.
\end{itemize}
Additionally, we also implement several competing high-dimensional inference methods that handle the missing-outcome scenarios:
\begin{itemize}
	\item {\bf ``DDR''} is an $\ell_1$-regularized debiased and doubly robust estimator of $\beta_0$ based on a high-dimensional adaptation of the traditional doubly robust construction when the outcome $Y$ is possibly MAR \citep{chakrabortty2019high}.
	
	\item {\bf ``SAS''} is a surrogate-assisted imputation-based inference approach for $\beta_0$ with a high-dimensional predictor $X$ in the semi-supervised learning setting (\emph{e.g.}, $Y$ is MCAR) \citep{hou2023surrogate}. For a fair comparison, we apply their method to our simulation setting with no surrogate outcomes.
	
	\item {\bf ``SSL-AIPW''} is an AIPW method for estimating $\beta_0$ in the semi-supervised learning (SSL) setting with MAR outcomes \citep{tian2024semi}.
\end{itemize}

Confidence intervals for ``DL-Jav'', ``DL-vdG'', ``R-Proj'', ``DDR'', ``SAS'', and ``SSL-AIPW'' are constructed according to their asymptotic normality results, while the confidence intervals from ``Refit'' are based on the assumption that the selected model by the Lasso pilot estimate $\hat{\beta}$ contains the true support set.

\subsection{Simulation Results Under the Gaussian Noise Distribution}
\label{subsec:sim_results}

We evaluate the performance of our proposed debiasing method and the existing inference methods stated above via the average absolute bias $\left|\hat{m}^{\mathrm{debias}}(x) - m_0(x) \right|$, the average coverage of confidence intervals with the 95\% nominal coverage probability, and the average length of confidence intervals over 1000 Monte Carlo experiments for each simulation scenario stated in \autoref{subsec:sim_design}. 

\begin{figure}[!ht]
	\captionsetup[subfigure]{justification=centering}
	\begin{subfigure}[t]{0.49\linewidth}
		\centering
		\includegraphics[width=1\linewidth]{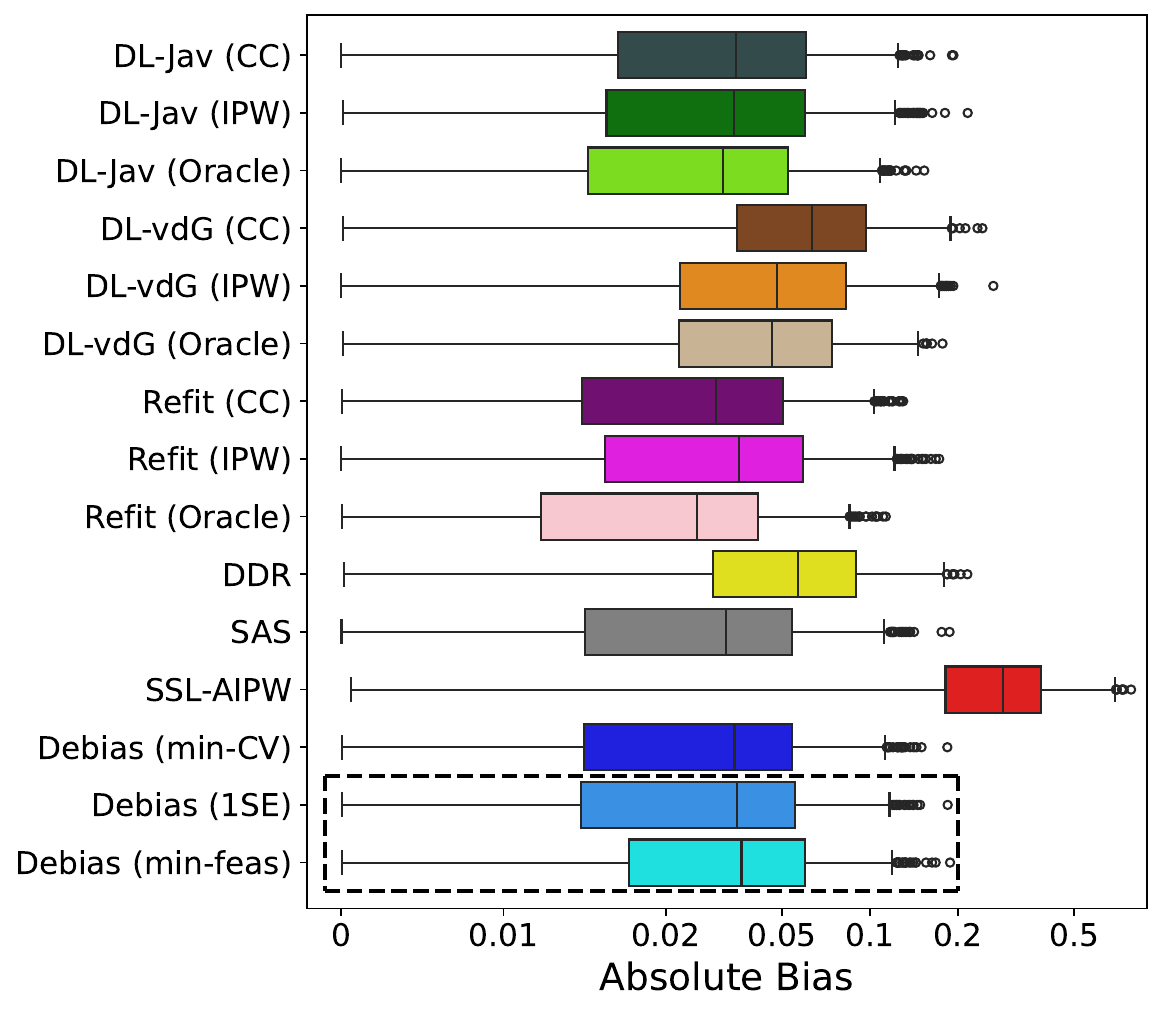}
	\end{subfigure}
	\hfil
	\begin{subfigure}[t]{0.49\linewidth}
		\centering		\includegraphics[width=1\linewidth]{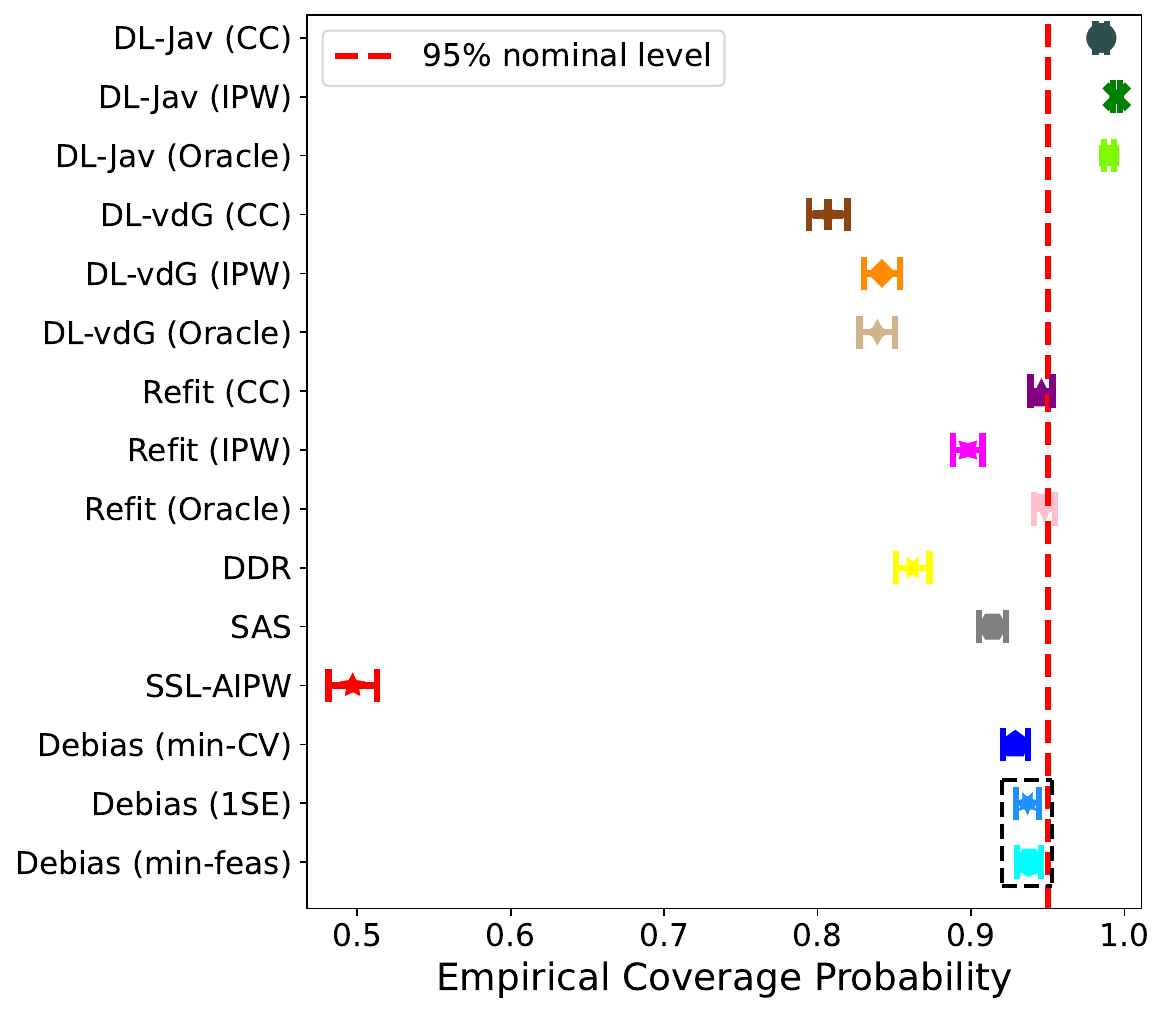}
	\end{subfigure}
	\begin{subfigure}[c]{0.485\linewidth}
		\centering
		\includegraphics[width=1\linewidth]{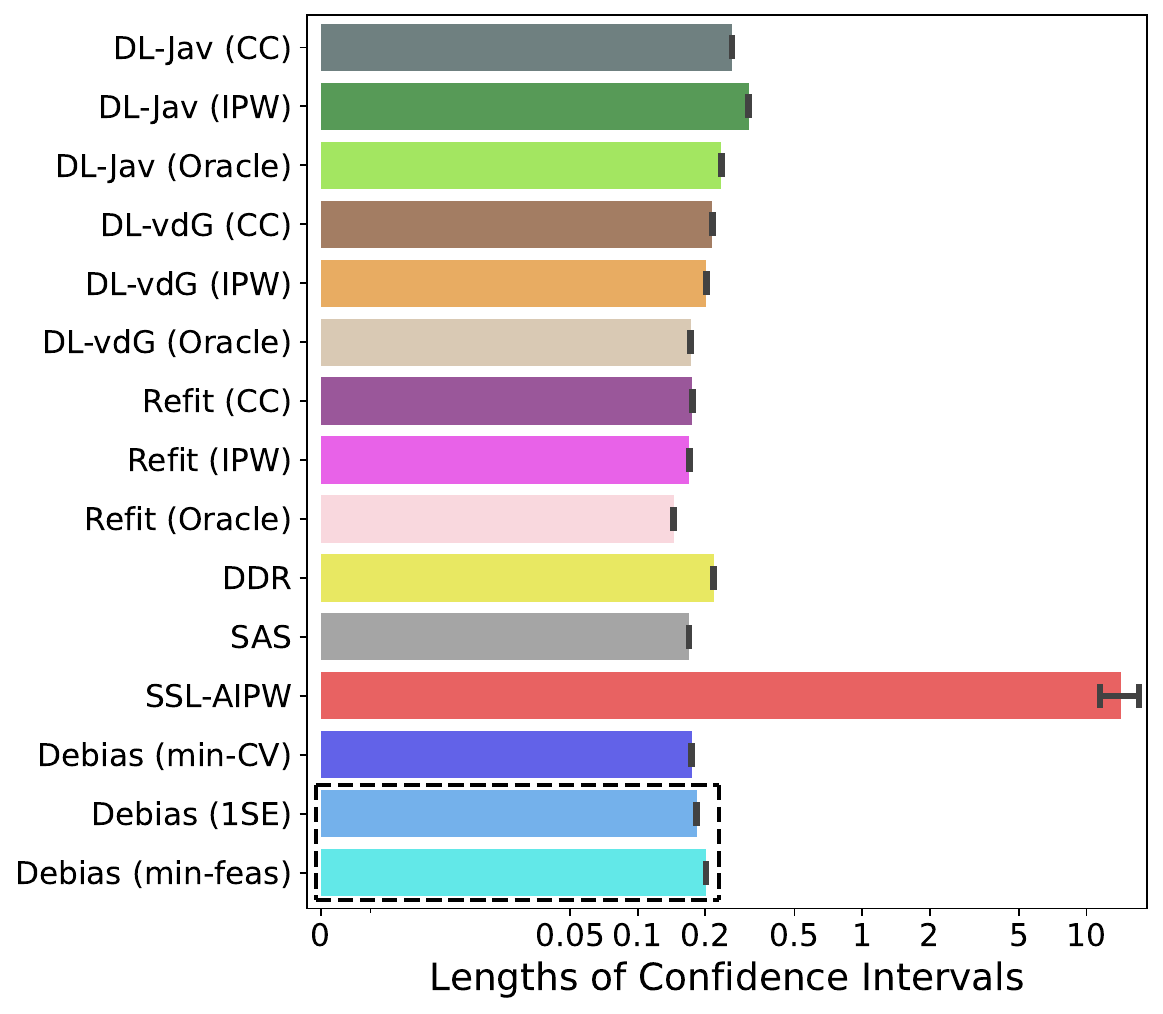}
	\end{subfigure}
	\hspace{1mm}
	\begin{subfigure}[c]{0.492\linewidth}
		\centering		\includegraphics[width=1\linewidth]{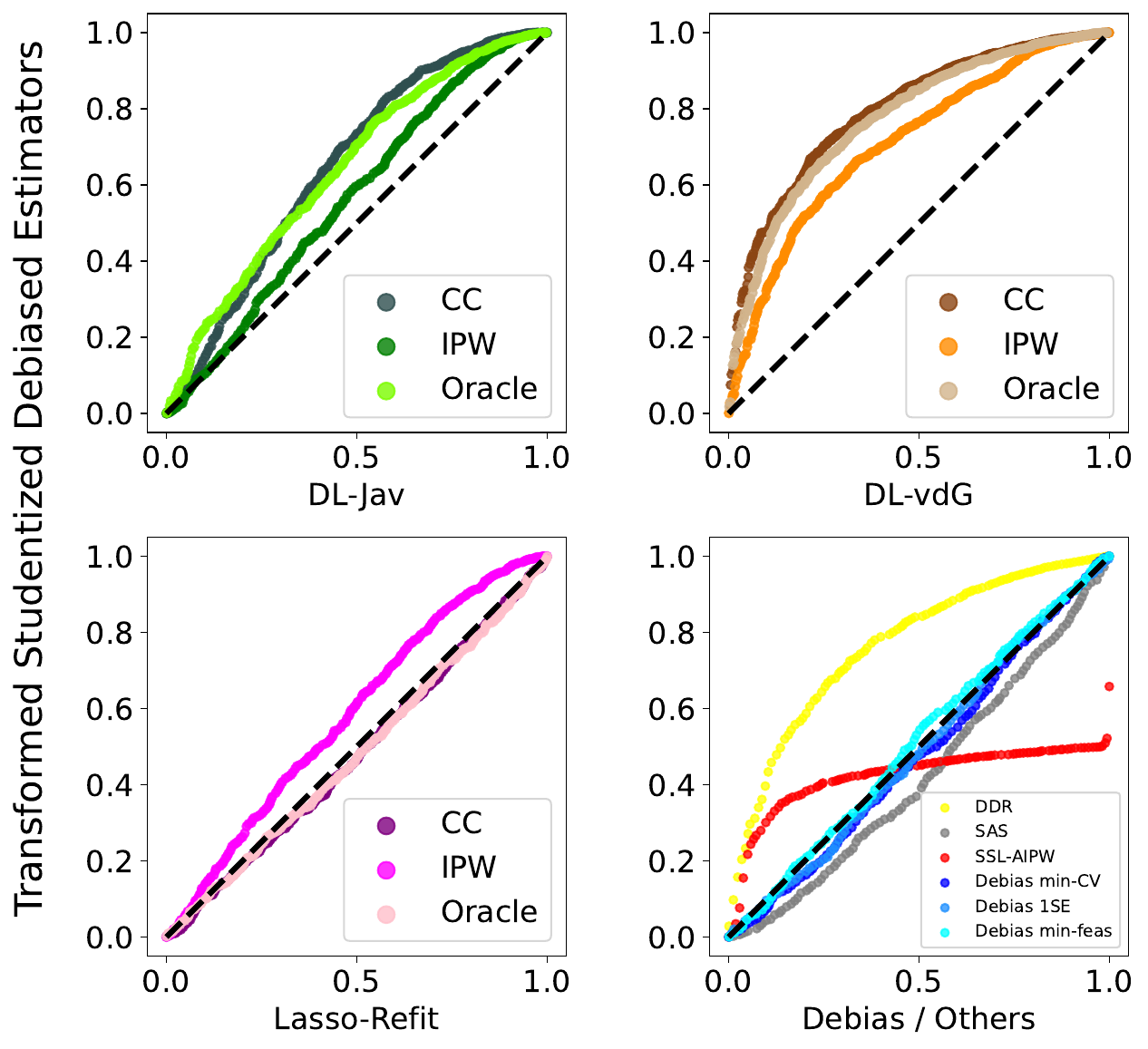}
	\end{subfigure}
	\caption{Simulation results with sparse $\beta_0^{sp}$ and sparse $x^{(2)}$ for the circulant symmetric covariance matrix $\Sigma^{\mathrm{cs}}$ under the Gaussian noise $\mathcal{N}(0,1)$ and MAR setting \eqref{MAR_correct}. {\bf Top Left:} Boxplots of absolute bias. {\bf Top Right:} Coverage probabilities with standard error bars. {\bf Bottom Left:} Average lengths of confidence intervals with standard error bars. {\bf Bottom Right:} Four comparative ``QQ-plots'' of the transformed studentized debiased estimators $\Phi\left(\sqrt{n} \cdot\hat{\sigma}_n^{-1}(x) \left[\hat{m}(x) -m_0(x)\right] \right)$ obtained from different debiasing methods, where $\hat{\sigma}_n(x)^2$ is the estimated (asymptotic) variance of $\hat{m}(x)$ by each method and $\Phi(\cdot)$ is the CDF of $\mathcal{N}(0,1)$. We also highlight the recommended rules ``1SE'' and ``min-feas'' for our proposed debiasing method via dashed rectangles in the first three panels.}
	\label{fig:cirsym_MAR_gauss}
\end{figure}

We provide comparative plots in terms of the three evaluation metrics and normality assessments for a selected simulation setting in \autoref{fig:cirsym_MAR_gauss}, while deferring the full simulation results for both the circulant symmetric and Toeplitz (or autoregressive) covariance designs to \autoref{fig:AR_MAR_gauss} and \autoref{table:cirsym_MAR_res}-\ref{table:AR_MAR_res} in \autoref{subapp:full_sim_res}, respectively.
In summary, our proposed debiasing method achieves bias reduction comparable to that of the debiased Lasso ``DL-Jav'' and the Lasso refitting ``Refit'', while yielding relatively shorter confidence intervals and more accurate coverage probabilities. Although the ridge projection method ``R-Proj'' sometimes produces better coverage probabilities than other methods, it consistently produces larger biases and more conservative confidence intervals. Among the competing methods for missing outcomes, ``DDR'' and ``SSL-AIPW'' fail to maintain asymptotic normality under our simulation settings, with ``SSL-AIPW'' exhibiting particularly large absolute bias. Our debiasing method performs comparably to ``SAS'', but provides more accurate coverage for the resulting confidence intervals. Finally, under the fully dense query point $x^{(5)}$, only ``DL-Jav'', ``DL-vdG'', ``DDR'', and our proposed debiasing method achieve approximately the desired coverage probabilities, suggesting that our method remains robust even when the sparsity assumption on the population dual solution $\ell_0(x)$ is violated.

The superiority of our debiasing method is more evident when the query point $x\in \mathbb{R}^d$ has many nonzero entries. The debiased estimator from our method is asymptotically semiparametrically efficient for arbitrary $x$ in \autoref{thm:asym_normal}, and its asymptotic variance can be naturally estimated as in Proposition~\ref{prop:var_const_est}. Hence, the asymptotic normal approximation for our debiased estimator may be more accurate in finite samples than those of the competing methods. All these simulation results corroborate the theoretical properties of our debiasing method derived in \autoref{sec:theory} and demonstrate its effectiveness in conducting valid statistical inference with missing outcomes.

\subsection{Simulation Results Under Heavy-Tailed Noise Distributions}
\label{subsec:heavy_noises}

The consistency and asymptotic normality results for our debiasing inference method are derived under the sub-Gaussian noise condition (Assumption~\ref{assump:subgau_error}), and our previous simulation results in \autoref{subsec:sim_results} are accordingly based on normally distributed noise. We demonstrate through additional simulations that the asymptotic normality of our debiased estimator \eqref{debias_est} remains valid even when the distribution of the noise $\epsilon$ in model \eqref{linear_reg} is not sub-Gaussian. To this end, we follow the simulation designs in \autoref{subsec:sim_design} but explore two heavy-tailed noise distributions:
\begin{enumerate}[label=(\roman*)]
	\item $\epsilon$ follows a Laplace distribution with mean 0 and scale parameter $\frac{1}{\sqrt{2}}$. Note that $\epsilon$ becomes sub-exponentially distributed and has variance $\mathrm{Var}(\epsilon)=1$. 
	
	\item $\epsilon$ follows a mean-zero $t$-distribution with 2 degrees of freedom. In this case, $\epsilon$ is not even sub-exponential and has infinite variance. 
\end{enumerate}
We present selected plots and asymptotic normality assessments in \autoref{fig:cirsym_MAR_laperr} and \autoref{fig:AR_MAR_terr}, while again deferring the full comparative simulation results for both the circulant symmetric and Toeplitz covariance designs to \autoref{table:cirsym_MAR_laperr_res}-\ref{table:cirsym_MAR_terr_res} and \autoref{table:AR_MAR_laperr_res}-\ref{table:AR_MAR_terr_res} in \autoref{subapp:full_sim_res}, respectively. In short, the debiased estimator appears to be approximately normally distributed in finite samples even when the noise distribution is heavy-tailed, while other competing methods except for ``SAS'' fail to maintain their asymptotic normality. Other conclusions from these simulation results are similar to those described in \autoref{subsec:sim_results}.


\begin{figure}[!ht]
	\captionsetup[subfigure]{justification=centering}
	\begin{subfigure}[t]{0.49\linewidth}
		\centering
		\includegraphics[width=1\linewidth]{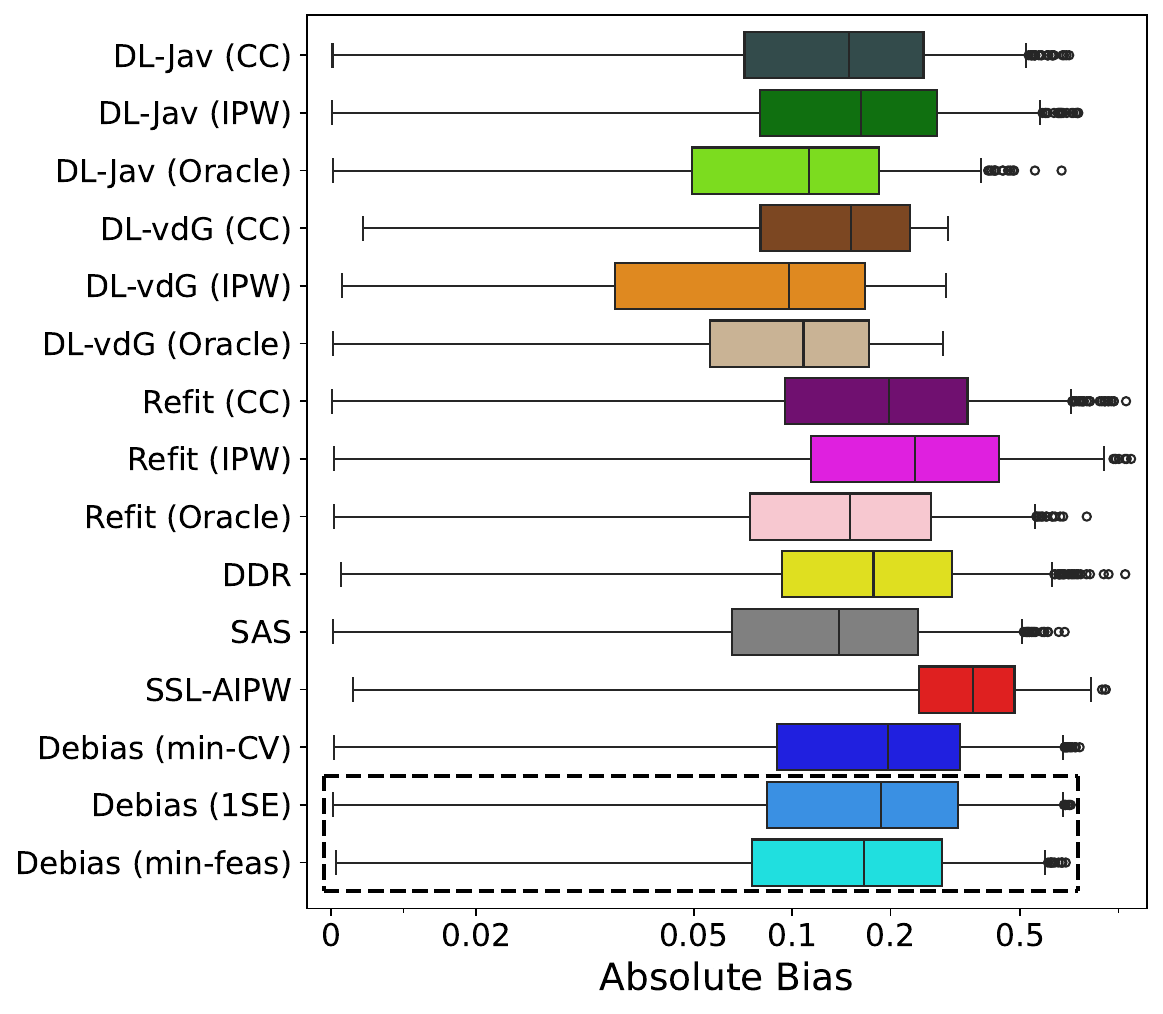}
	\end{subfigure}
	\hfil
	\begin{subfigure}[t]{0.49\linewidth}
		\centering		\includegraphics[width=1\linewidth]{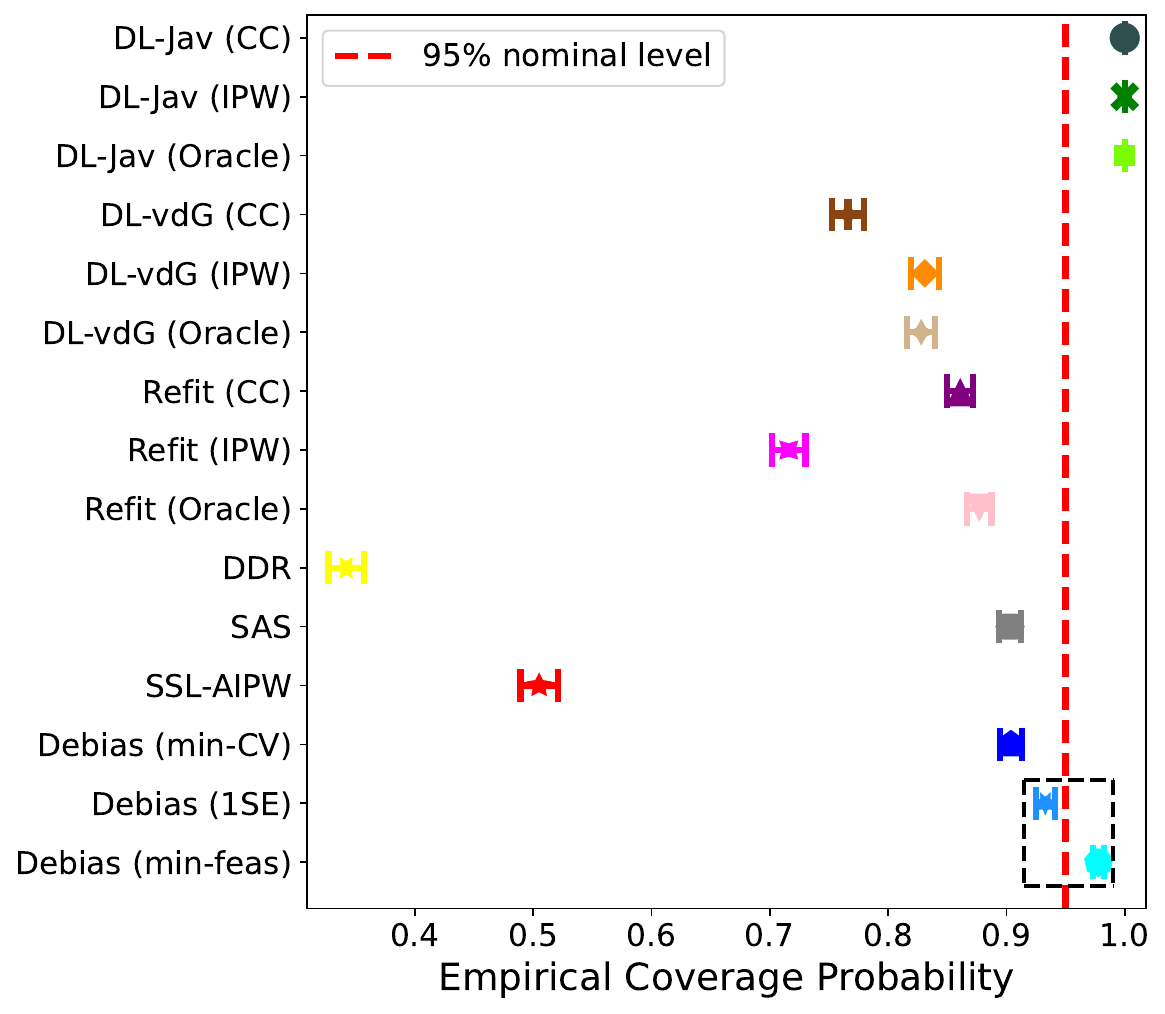}
	\end{subfigure}
	\begin{subfigure}[c]{0.485\linewidth}
		\centering
		\includegraphics[width=1\linewidth]{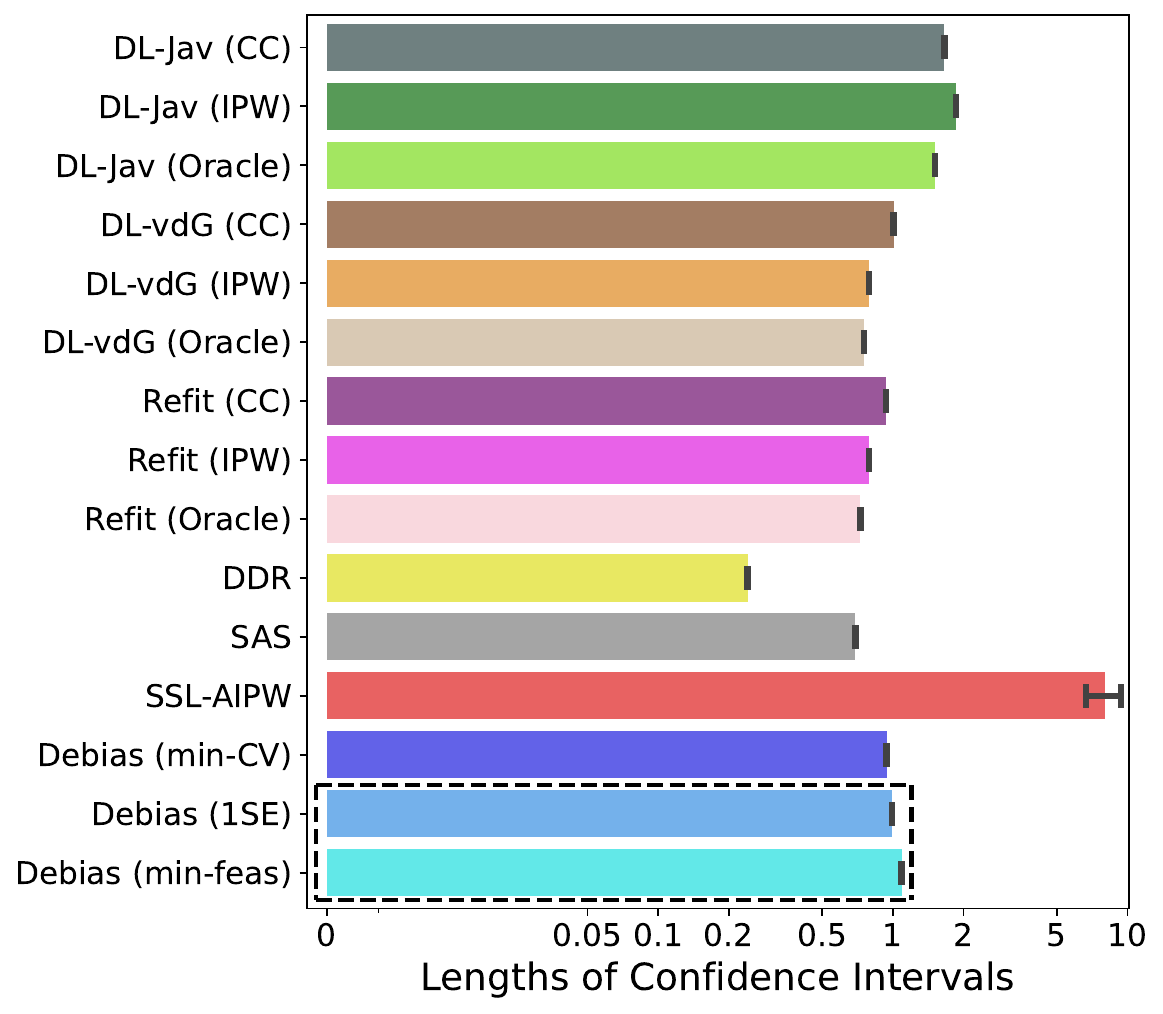}
	\end{subfigure}
	\hspace{1mm}
	\begin{subfigure}[c]{0.492\linewidth}
		\centering		\includegraphics[width=1\linewidth]{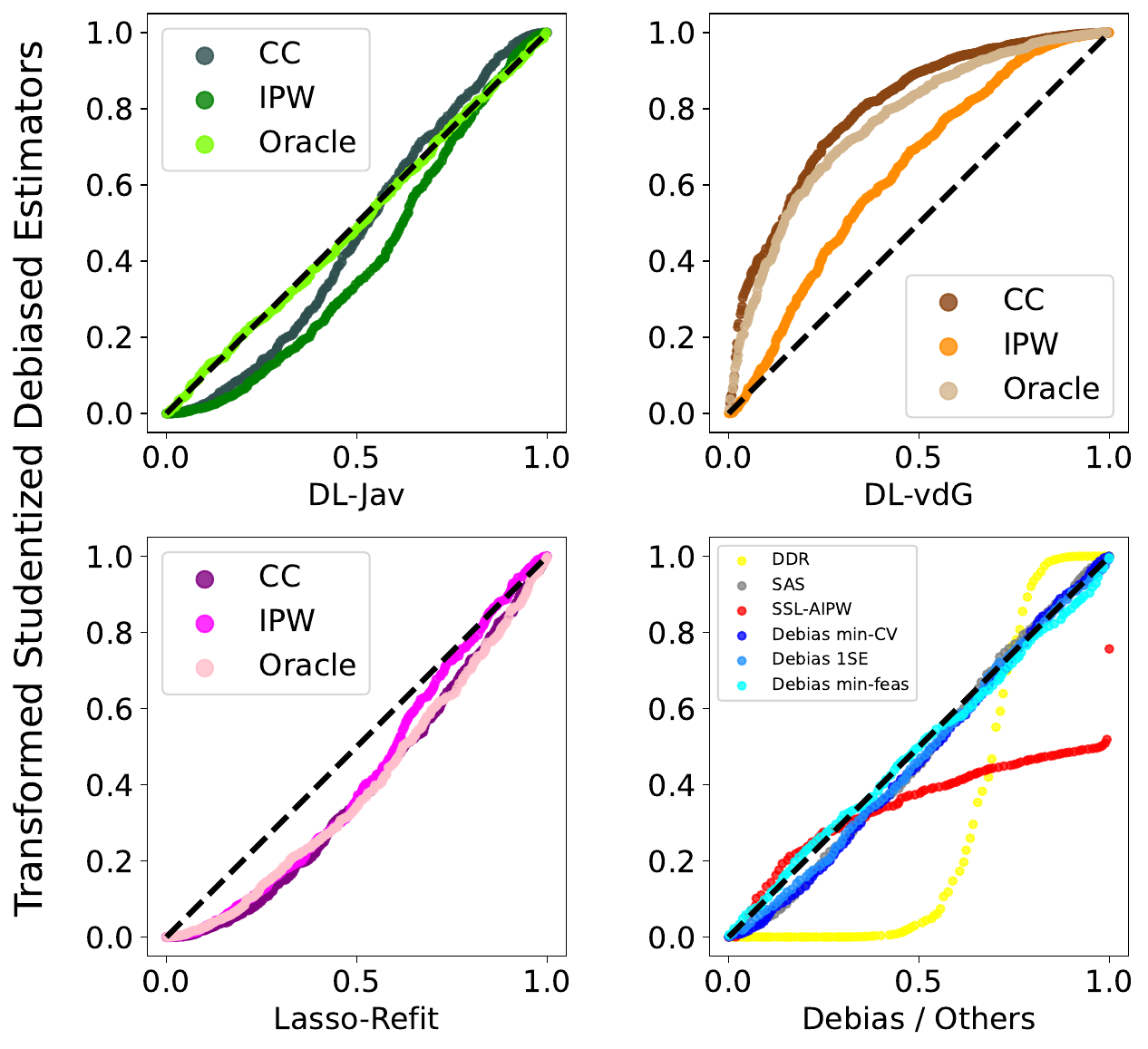}
	\end{subfigure}
	\caption{Simulation results with dense $\beta_0^{de}$ and sparse $x^{(2)}$ for the circulant symmetric covariance matrix $\Sigma^{\mathrm{cs}}$ under $\mathrm{Laplace}\left(0, \frac{1}{\sqrt{2}} \right)$-distributed noise and the MAR setting \eqref{MAR_correct}. {\bf Top Left:} Boxplots of the absolute bias. {\bf Top Right:} Coverage probabilities with standard error bars. {\bf Bottom Left:} Average lengths of confidence intervals with standard error bars. {\bf Bottom Right:} Four comparative ``QQ-plots'' of the transformed studentized debiased estimators $\Phi\left(\sqrt{n} \cdot\hat{\sigma}_n^{-1}(x) \left[\hat{m}(x) -m_0(x)\right] \right)$ obtained from different debiasing methods, where $\hat{\sigma}_n(x)^2$ is the estimated (asymptotic) variance of $\hat{m}(x)$ by each method and $\Phi(\cdot)$ is the CDF of $\mathcal{N}(0,1)$.}
	\label{fig:cirsym_MAR_laperr}
\end{figure}



\begin{figure}[!ht]
	\captionsetup[subfigure]{justification=centering}
	\begin{subfigure}[t]{0.49\linewidth}
		\centering
		\includegraphics[width=1\linewidth]{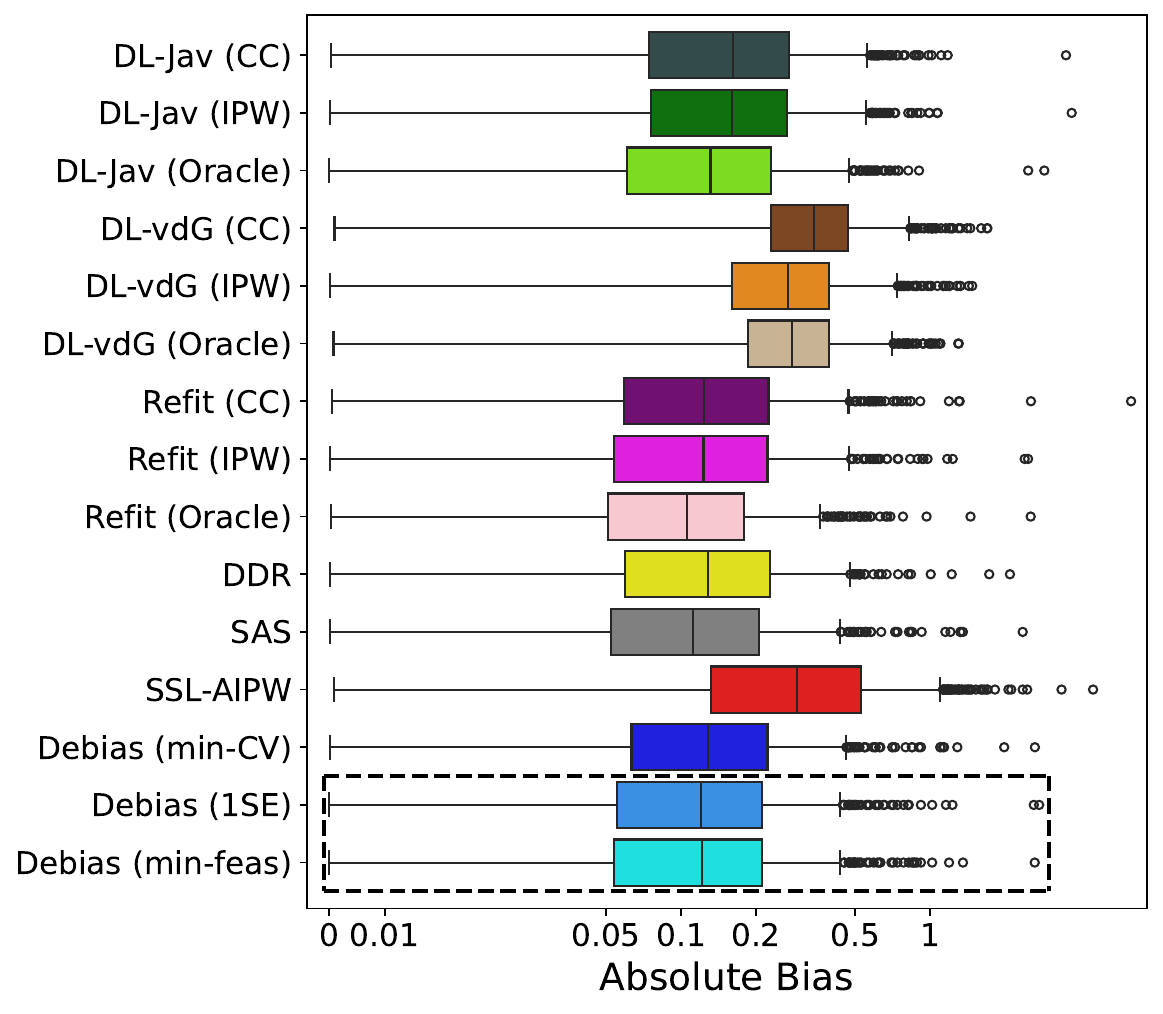}
	\end{subfigure}
	\hfil
	\begin{subfigure}[t]{0.49\linewidth}
		\centering		\includegraphics[width=1\linewidth]{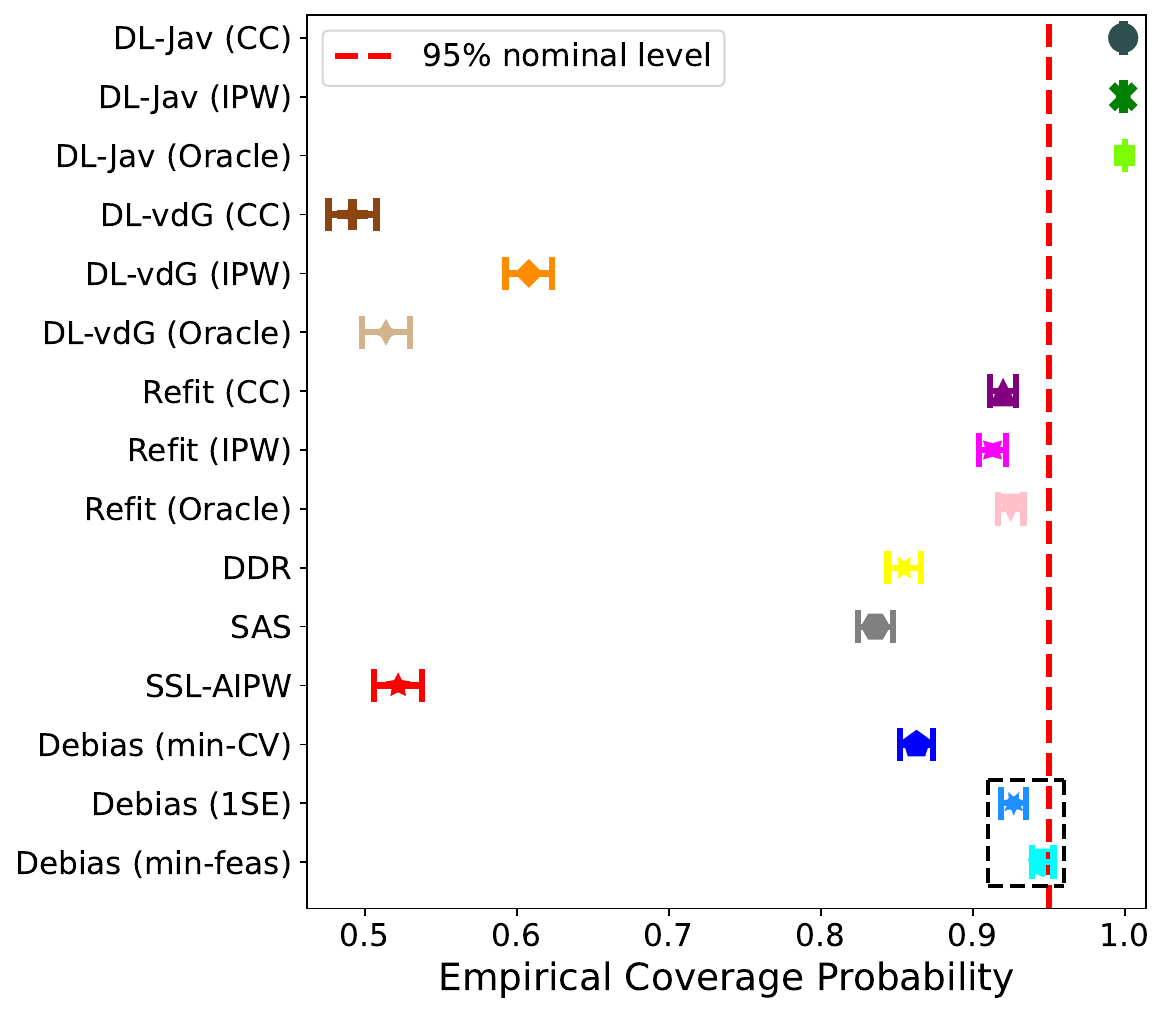}
	\end{subfigure}
	\begin{subfigure}[c]{0.485\linewidth}
		\centering
		\includegraphics[width=1\linewidth]{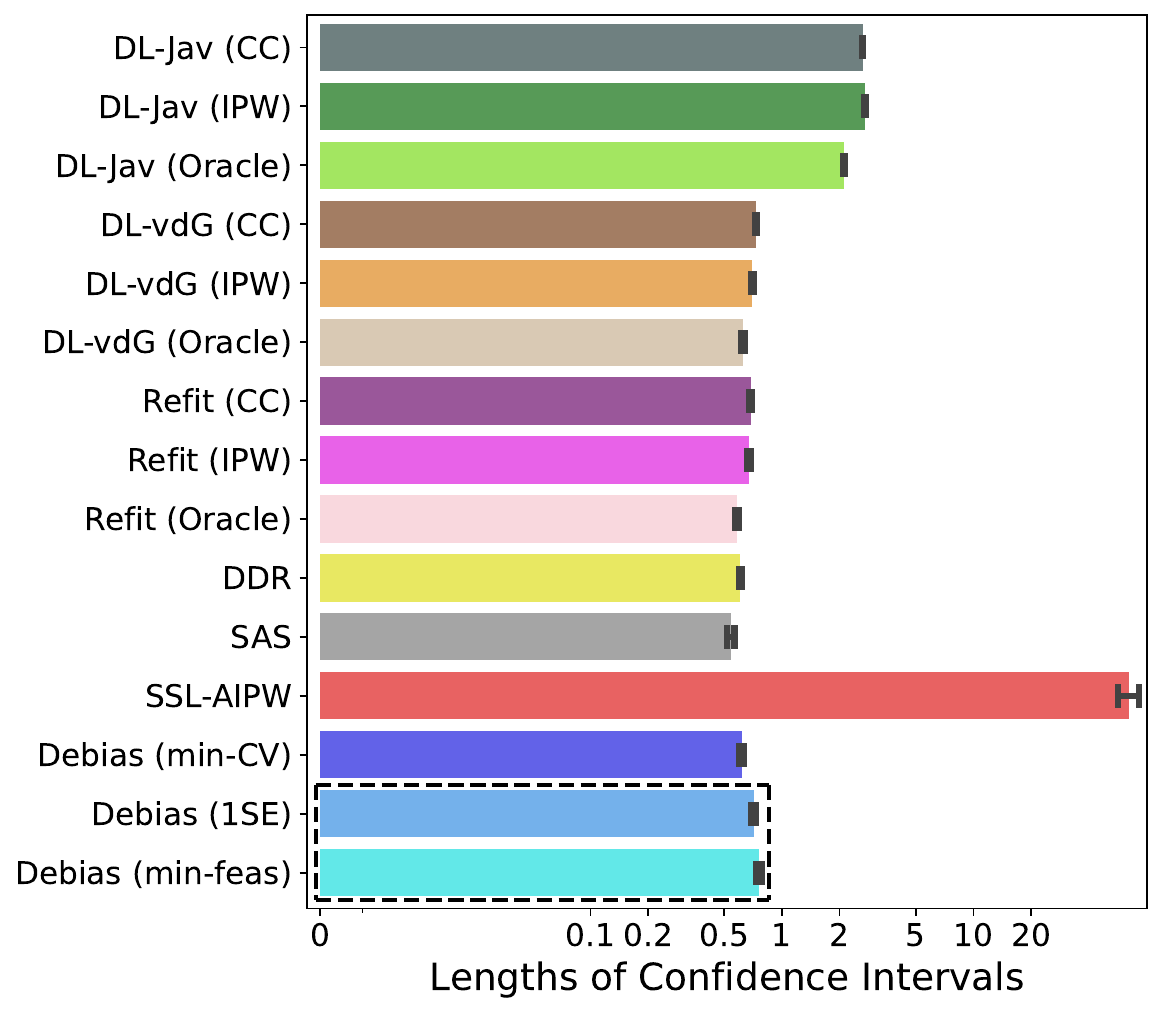}
	\end{subfigure}
	\hspace{1mm}
	\begin{subfigure}[c]{0.492\linewidth}
		\centering		\includegraphics[width=1\linewidth]{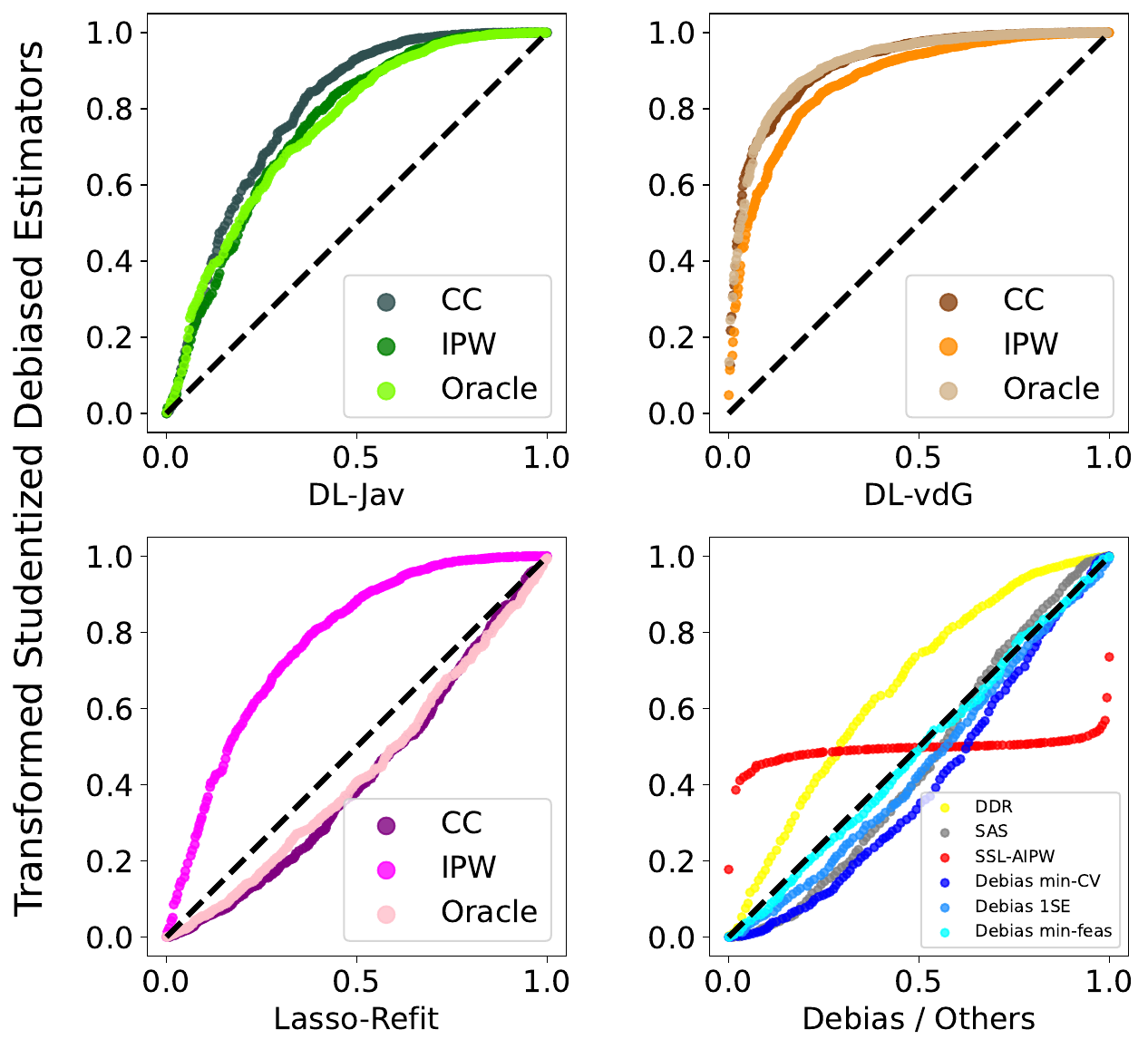}
	\end{subfigure}
	\caption{Simulation results with pseudo-dense $\beta_0^{pd}$ and dense $x^{(4)}$ for the Toeplitz (or autoregressive) covariance matrix $\Sigma^{\mathrm{ar}}$ under $t_2$-distributed noise and the MAR setting \eqref{MAR_correct}. {\bf Top Left:} Boxplots of the absolute bias. {\bf Top Right:} Coverage probabilities with standard error bars. {\bf Bottom Left:} Average lengths of confidence intervals with standard error bars. {\bf Bottom Right:} Four comparative ``QQ-plots'' of the transformed studentized debiased estimators $\Phi\left(\sqrt{n} \cdot\hat{\sigma}_n^{-1}(x) \left[\hat{m}(x) -m_0(x)\right] \right)$ obtained from different debiasing methods, where $\hat{\sigma}_n(x)^2$ is the estimated (asymptotic) variance of $\hat{m}(x)$ by each method and $\Phi(\cdot)$ is the CDF of $\mathcal{N}(0,1)$.}
	\label{fig:AR_MAR_terr}
\end{figure}


\subsection{Simulation Results for the Proposed Debiasing Method With Nonparametrically Estimated Propensity Scores}
\label{subsec:nonpar_prop}

As shown in \autoref{thm:dual_consist2} and \autoref{thm:asym_normal}, the consistency and asymptotic normality of the proposed debiasing method do not require any parametric assumption on the propensity score $\mathrm{P}(R=1|X)$ under the MAR condition. Therefore, we now conduct additional simulation studies for our debiasing method in which the propensity scores $\pi(X_i),i=1,...,n$ are estimated by the following nonlinear/nonparametric machine learning methods based on the observed data $\{(X_i,R_i)\}_{i=1}^n$. All of these methods are implemented in the \texttt{scikit-learn} package \citep{scikit-learn} in Python.

{\bf Naive Bayes (``NB''):} We adapt the Gaussian naive Bayes method \citep{zhang2004} to model the propensity scores through $\pi(X_i)=\mathrm{P}(R_i=1|X_{i1},...,X_{id}) = \frac{\mathrm{P}(R_i=1) \cdot \prod_{k=1}^d \mathrm{P}(X_{ik}|R_i=1)}{\mathrm{P}(X_{i1},...,X_{id})}$ for any $i=1,...,n$, where each class-conditional density $\mathrm{P}(X_{ik}|R_i=r)$ with $r\in\{0,1\}$ is modeled by a Gaussian density whose parameters are estimated by maximum likelihood.

{\bf Random Forest (``RF''):} We implement the random forest method \citep{breiman2001random} with 100 trees, bootstrap samples, and the Gini impurity to measure the quality of a split.

{\bf Support Vector Machine (``SVM''):} We fit a support vector machine \citep{chen2005tutorial} with the Gaussian radial basis function to classify the missingness indicator $R_i$ based on the observed covariate vector $X_i\in \mathbb{R}^d$ for $i=1,...,n$. The outputs of the trained SVM are regarded as preliminary estimates of the propensity scores.

{\bf Neural Network (``NN''):} We train a neural network model with two hidden layers of size $80\times 50$ and use the rectified linear unit function $h(x)=\max\{x,0\}$ as the activation function. The learning rate $\eta_{\text{NN}}$ is initially set to $0.001$ as long as the training loss keeps decreasing and is adaptively changed to $\eta_{\text{NN}}/5$ when two consecutive epochs fail to decrease the training loss.

For the above methods, we also consider calibrated versions obtained through Platt scaling \citep{platt1999probabilistic}, which fits an extra logistic regression on the estimated propensity scores $\hat{\pi}_i,i=1,...,n$. However, since the errors in the estimated propensity scores and the final performance of the proposed debiased estimator worsen after this calibration, we choose not to report them here.

The covariate vectors $X_i \in \mathbb{R}^d,i=1,...,n$ are again sampled independently from $\mathcal{N}_d(\bm{0}, \Sigma^{\mathrm{cs}})$ with $d=1000$ and $n=900$, and the noise variables $\epsilon_i,i=1,...,n$ are generated independently from $\mathcal{N}(0,1)$ as in \autoref{subsec:sim_design}. To increase the complexity of estimating the propensity scores, we generate the missingness indicators $R_i,i=1,...,n$ for the outcome variables $Y_i,i=1,...,n$ through a different MAR mechanism from the one in \autoref{subsec:sim_design} as:
\begin{equation}
	\label{MAR_mis}
	\mathrm{P}(R_i=1|X_i) = \Phi\left(-4+\sum_{k=1}^K Z_{ik} \right),
\end{equation}
where $\Phi(\cdot)$ is the CDF of $\mathcal{N}(0,1)$ and the vector $(Z_{i1},...,Z_{iK})$ contains all polynomial combinations of the first eight components $X_{i1},...,X_{i8}$ of the covariate vector $X_i$ with degrees less than or equal to two (\emph{i.e.}, including the linear, quadratic, and one-way interaction terms).

The simulation results for our debiasing method under the oracle propensity scores (``Oracle''), the Lasso-type logistic regression \eqref{lr_lasso} (``LR''), and the aforementioned nonparametric methods for propensity score estimation are shown in \autoref{fig:Cir_Nonpar_Prop_Score2}. We provide more comprehensive results with the evaluation metrics from \autoref{subsec:sim_results} and an additional metric, the average mean absolute error (``Avg-MAE''), for the estimated propensity scores in \autoref{fig:Cir_Nonpar_Prop_Score1}, \autoref{table:CirSym_MAR_Non_Prop_Score}, as well as additional ``QQ-plots'' in \autoref{fig:Cir_Nonpar_Prop_qqplot} of \autoref{subapp:full_sim_res}.

\begin{figure}[hbtp!]
	\captionsetup[subfigure]{justification=centering}
	\begin{subfigure}[t]{0.49\linewidth}
		\centering
		\includegraphics[width=1\linewidth]{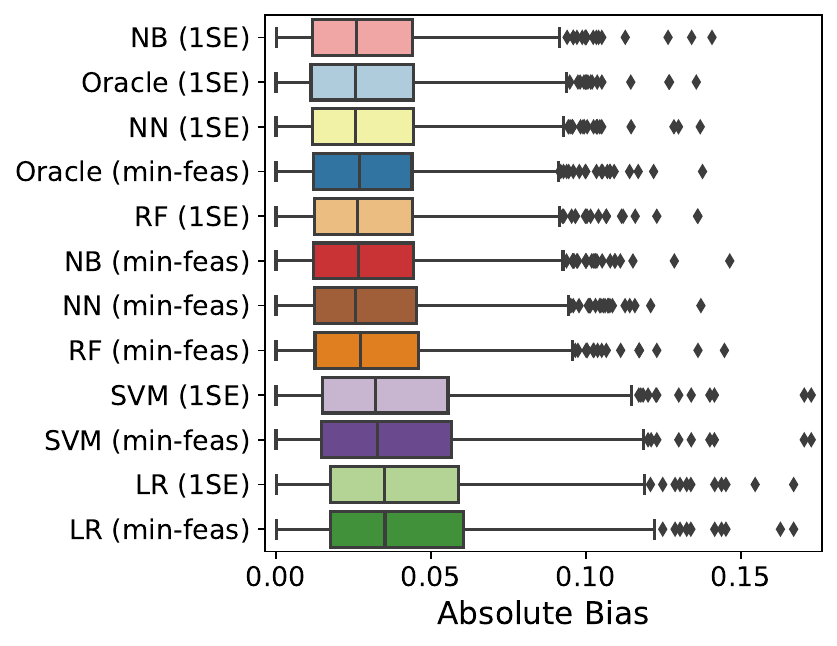}
	\end{subfigure}
	\hfil
	\begin{subfigure}[t]{0.49\linewidth}
		\centering		\includegraphics[width=1\linewidth]{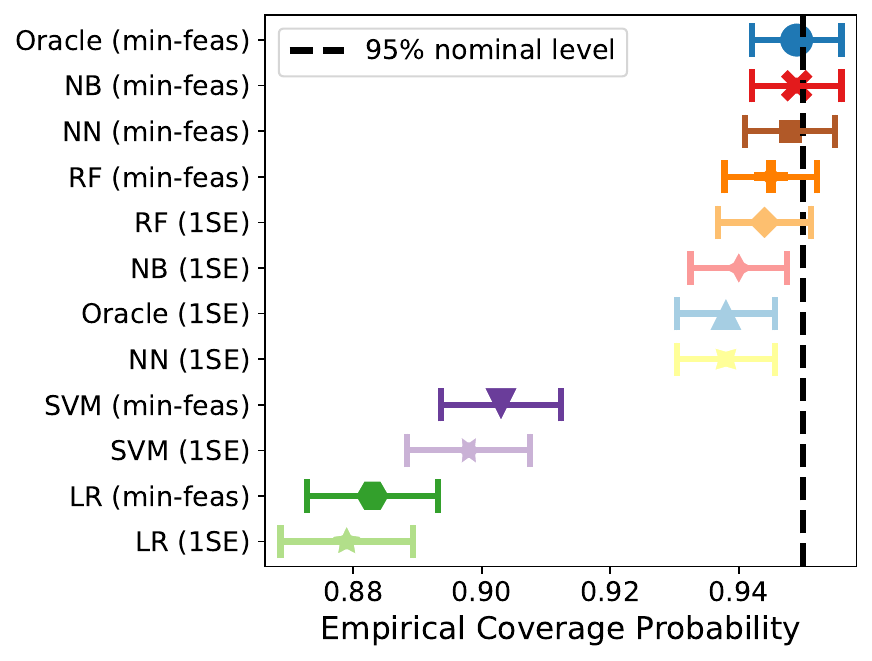}
	\end{subfigure}
	\begin{subfigure}[c]{0.49\linewidth}
		\centering
		\includegraphics[width=1\linewidth]{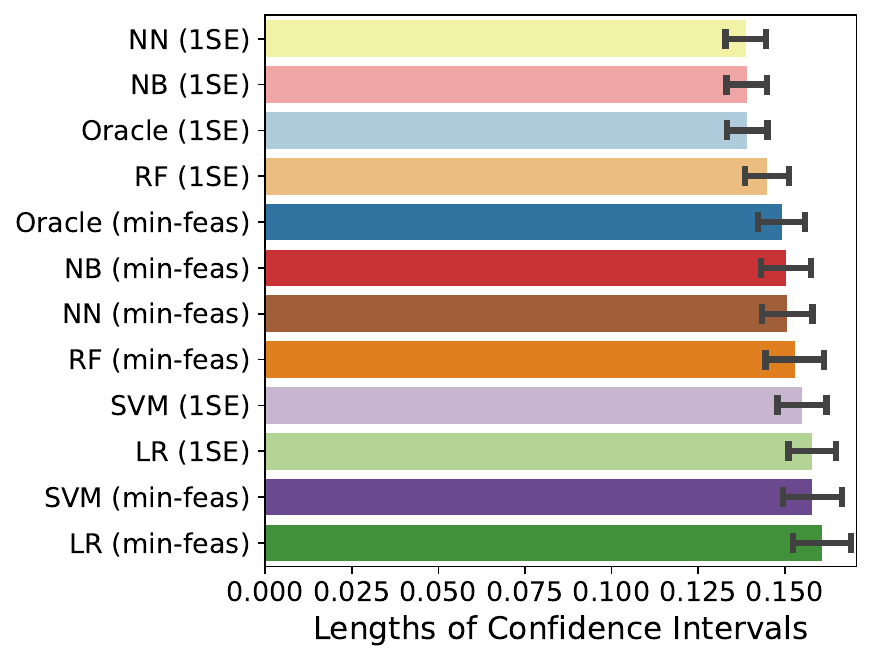}
	\end{subfigure}
	\hfil
	\begin{subfigure}[c]{0.49\linewidth}
		\centering		\includegraphics[width=1\linewidth]{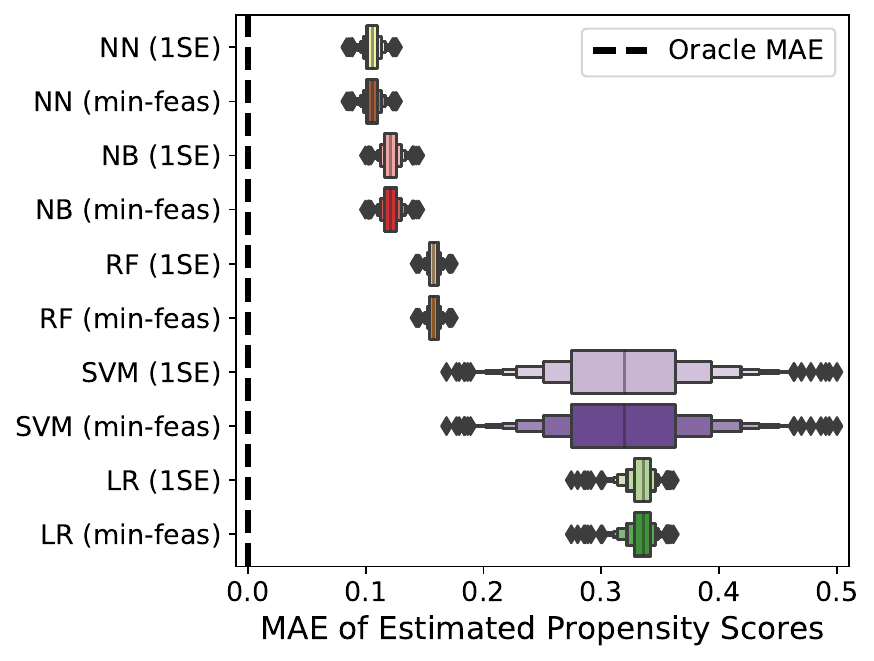}
	\end{subfigure}
	\begin{subfigure}[c]{0.995\linewidth}
		\centering		\includegraphics[width=1\linewidth]{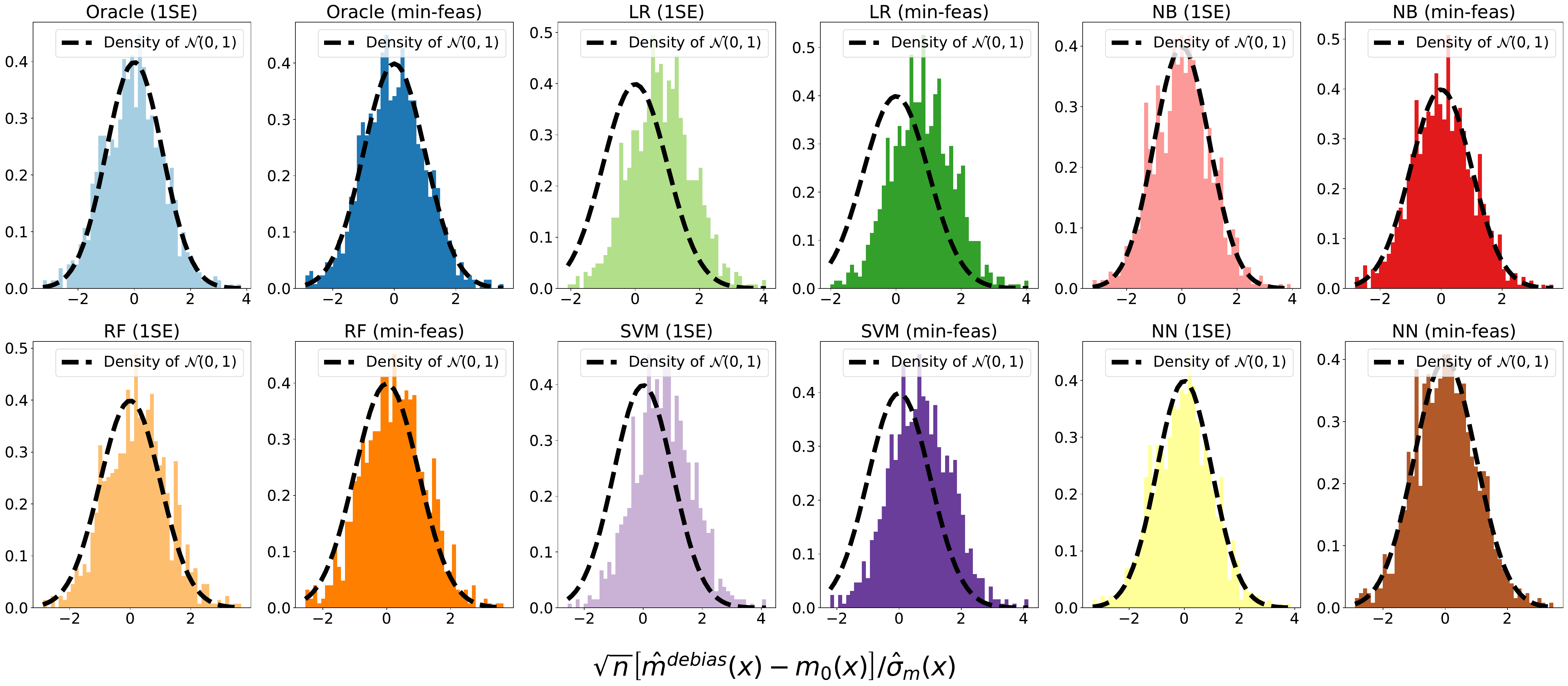}
	\end{subfigure}
	\caption{Simulation results for our debiasing method with sparse $\beta_0^{sp}$ and (weakly) dense $x^{(4)}$ when the propensity scores are estimated by various machine learning methods under the new MAR setting \eqref{MAR_mis}. {\bf Top Left:} Boxplots of the absolute bias. {\bf Top Right:} Coverage probabilities with standard error bars. {\bf Middle Left:} Average lengths of confidence intervals with standard deviation bars. {\bf Middle Right:} Letter-valued plots of the MAE of estimated propensity scores. {\bf Bottom Two Rows:} Histograms of the studentized debiased estimators under different methods. We also order the performance metrics according to their means in the first four panels.}
	\label{fig:Cir_Nonpar_Prop_Score2}
\end{figure}

The main conclusion from this simulation study is that the performance of our debiasing method can be further improved when the propensity scores are better estimated. Under the MAR mechanism \eqref{MAR_mis} and Gaussian design for covariate vectors $X_i,i=1,...,n$, neural network and Gaussian naive Bayes models outperform other machine learning methods in estimating the propensity scores and generally lead to subsequent debiased estimators with lower biases, better coverage probabilities, and shorter confidence intervals. Conversely, the Lasso-type logistic regression is misspecified and consequently degrades the performance of the resulting debiased estimator. Another interesting phenomenon from these simulation results is that our debiasing method sometimes performs better when using the estimated propensity scores $\hat{\pi}_i,i=1,...,n$ than when using the oracle propensity scores $\pi_i=\pi(X_i),i=1,...,n$. Such a paradox has been analyzed in propensity score matching \citep{rosenbaum1987model}, for the IPW estimator in causal inference \citep{robins1992estimating,su2023estimated}, and in other general parameter estimation problems in the presence of a nuisance parameter \citep{henmi2004paradox,hitomi2008puzzling,lok2021estimating}. A rigorous study of this paradox for our debiasing method is beyond the scope of this paper.

\section{Real-World Application}
\label{sec:real_application}

We showcase a real-world application of our proposed debiasing method to stellar mass inference for selected galaxies in the Sloan Digital Sky Survey, Fourth Phase and Data Release 16 (SDSS-IV DR16; \citealt{ahumada202016th})\footnote{See \url{https://www.sdss4.org/dr16}.}. 

\subsection{Background and Study Design}
\label{subsec:real_app_design}

SDSS-IV DR16 contains three main surveys: the Extended Baryon Oscillation Spectroscopic Survey (eBOSS; \citealt{eBOSS2016}), Mapping Nearby Galaxies at APO (MaNGA; \citealt{MaNGA2014}), and APO
Galactic Evolution Experiment 2 (APOGEE-2; \citealt{APOGEE2017}), among which the eBOSS survey and its predecessor BOSS cover a broad range of cosmological scales and observe most galaxies in the Universe. While a couple of value-added catalogs for the estimated stellar masses of observed galaxies based on the BOSS spectra are available \citep{blanton2005new,conroy2009propagation,chang2015stellar}, we focus on the eBOSS Firefly value-added catalog \citep{Comparat2017}
\footnote{See \url{https://www.sdss4.org/dr17/spectro/galaxy_firefly/}.} 
for our stellar mass inference study given its timeliness and completeness. In more detail, this catalog fits combinations of single-burst stellar population models to spectroscopic data that have a positive definite redshift and are classified as galaxies by SDSS-IV DR16, following an iterative best-fitting process controlled by the Bayesian Information Criterion to produce stellar masses and other stellar properties of the galaxies \citep{wilkinson2017firefly}. A fraction of the stellar mass estimates in the Firefly catalog are missing because of the limited use of the observational run in SDSS-IV DR16 dedicated to galaxy targets, data contamination, and the misclassification of galaxies as stars \citep{Comparat2017}. Given that we incorporate various spectroscopic and photometric properties of the observed galaxies into the covariate vector for our stellar mass inference task, it is reasonable to assume that the estimated stellar masses are missing at random.

Since nearly all existing catalogs yield only point estimates of stellar mass, we intend to leverage our debiasing inference method to answer two scientific questions:
\begin{itemize}
	\item {\it Question 1:} How can we produce uncertainty measures (or confidence intervals) for the estimated stellar mass of a newly observed galaxy based on its observed spectroscopic and photometric properties?
	
	\item {\it Question 2:} Is it statistically significant that the stellar mass is negatively correlated with its distance to the nearby cosmic filament structures?
\end{itemize}

To this end, we obtain the spectroscopic and photometric properties of $n=1185$ observed galaxies within a thin redshift slice $0.4\sim 0.4005$ from the SDSS-IV DR16 database as basic covariates. Such a thin redshift slice helps control redshift distortions caused by the Kaiser effect (\emph{i.e.}, galaxies falling toward the galaxy cluster; \citealt{kaiser1987clustering}) and the ``Finger-of-God'' effect (\emph{i.e.}, the elongation of galaxy distributions along the line-of-sight direction; \citealt{jackson1972critique}). Among the galaxies in the selected redshift slice, 30.2\% of their stellar masses are missing in the Firefly catalog under the ``Chabrier'' initial stellar mass function \citep{chabrier2003galactic} and the ``MILES'' stellar population model \citep{maraston2011stellar}. In addition, the covariate vector for each galaxy consists of right ascension, declination, ugriz bands and their measurement errors, model and cmodel magnitudes for ugriz bandpasses, galactic extinction corrections for ugriz bandpasses, median signal-to-noise per pixel within each ugriz bandpass, original and best-fit template spectra projected onto ugriz filters as well as their inverse variances, sky flux in each of the ugriz imaging filters, and signal-to-noise squared for spectrographs \#1 and \#2 at g=20.20, r=20.20, and i=20.20 for SDSS spectrograph spectra. To capture nonlinear patterns in stellar mass estimation and handle extreme values, we generate additional covariates by applying the logarithmic transformations $x\mapsto \text{sign}(x)\cdot \log(1+|x|)$ to ugriz bands, galactic extinction corrections, and flux features. We also remove covariates whose absolute correlation coefficients exceed 0.95 before generating univariate B-spline basis covariates of polynomial order 3 with 40 knots. We adopt B-spline basis covariates instead of the usual polynomial combinations because covariates generated by univariate B-spline bases tend to be less linearly correlated and can facilitate the subsequent inference task. To address \emph{Question 2} above, we compute the angular diameter distances from the galaxies to the two-dimensional spherical cosmic filaments constructed by the directional subspace constrained mean shift algorithm in \cite{zhang2023linear,zhang2022sconce} using SDSS-IV DR16 data.\footnote{This cosmic web catalog can be downloaded at \url{https://doi.org/10.5281/zenodo.6244866}.} We also control for the confounding effects of nearby galaxy clusters by including the angular diameter distances from galaxies to the estimated local modes and intersections of filaments as extra covariates. The final dimension of the covariate vector for each galaxy in the selected redshift slice is $d=1409$, and we take the natural logarithm of the Firefly stellar mass of each galaxy as the outcome variable.

\subsection{Results}

\begin{figure}[t]
	\captionsetup[subfigure]{justification=centering}
	\begin{subfigure}[t]{0.32\linewidth}
		\centering
		\includegraphics[width=1\linewidth]{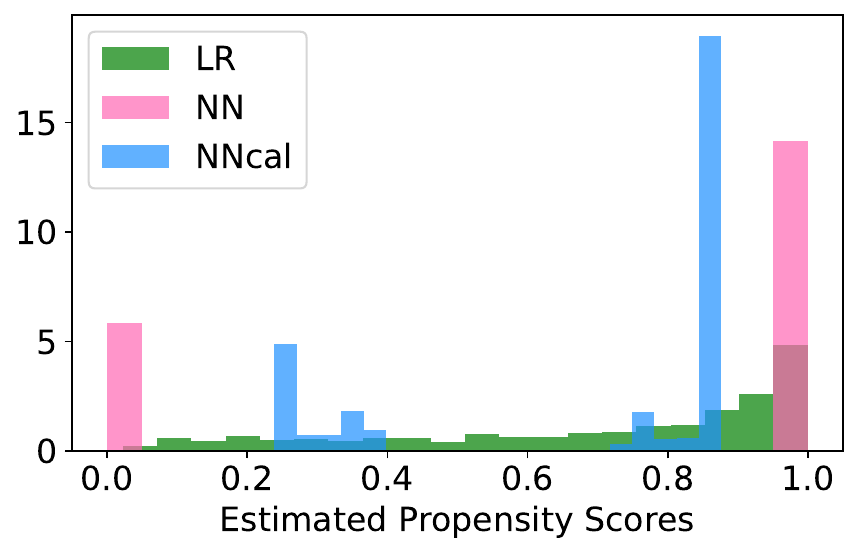}
	\end{subfigure}
	\hfil
	\begin{subfigure}[t]{0.32\linewidth}
		\centering		\includegraphics[width=1\linewidth]{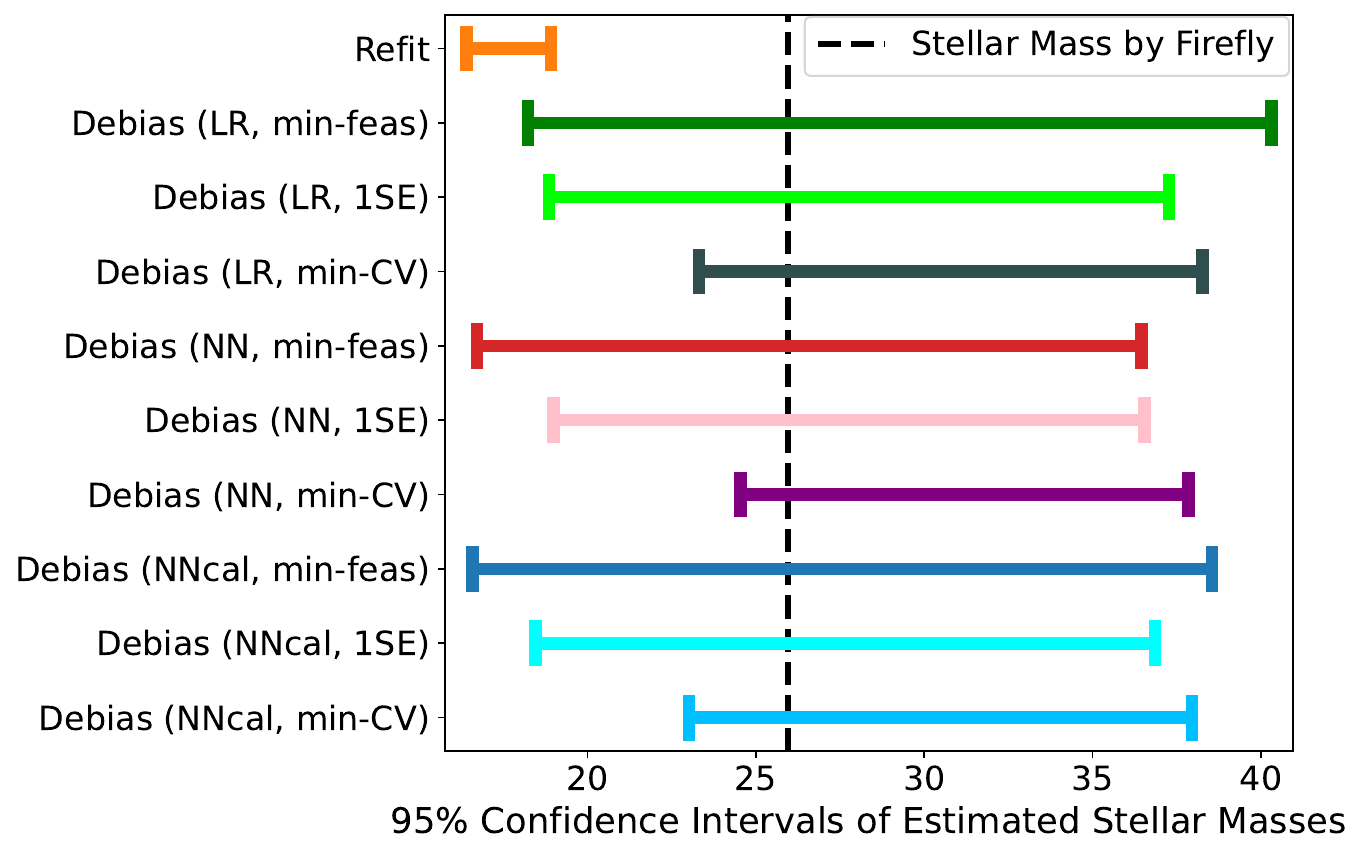}
	\end{subfigure}
	\hfil
	\begin{subfigure}[t]{0.32\linewidth}
		\centering
		\includegraphics[width=1\linewidth]{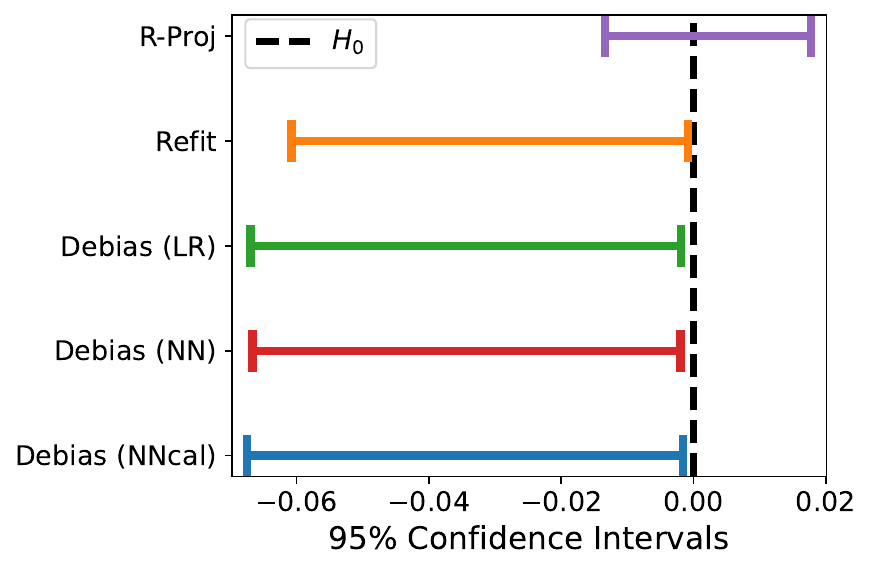}
	\end{subfigure}
	\caption{Stellar mass inference with our debiasing and related methods. {\bf Left Panel:} Estimated propensity scores by ``LR'', ``NN'', and ``NNcal''. {\bf Middle Panel:} 95\% confidence intervals by different methods for the estimated stellar mass of the new galaxy, addressing \emph{Question 1}. {\bf Right Panel:} 95\% confidence intervals by different methods for the estimated regression coefficient associated with the distance to cosmic filaments, addressing \emph{Question 2}.}
	\label{fig:stellar_mass_res}
\end{figure}

When approaching \emph{Question 1}, we randomly select a galaxy in the nearby redshift slice $0.4005\sim 0.401$ as a new observation and construct its associated covariate vector as the query point $x^{(Q1)}$ following the procedures in \autoref{subsec:real_app_design}. As for \emph{Question 2}, we set the query point as $x^{(Q2)}=(0,0,1,0,...,0)^T\in \mathbb{R}^d$, where the only nonzero entry corresponds to the covariate of the angular diameter distance to cosmic filaments. We estimate the propensity scores (\emph{i.e.}, the probabilities of having a non-missing stellar mass) for the galaxies in the redshift slice by the Lasso-type logistic regression (``LR'') in \autoref{app:prop_score_lasso} and the neural network (``NN'') model described in \autoref{subsec:nonpar_prop}. Nevertheless, as the left panel of \autoref{fig:stellar_mass_res} shows, the estimated propensity scores from the neural network model are too extreme, \emph{i.e.}, close to 0 or 1, so we also consider calibrating them through Platt scaling \citep{platt1999probabilistic}, denoted by ``NNcal''.

\autoref{fig:stellar_mass_res} displays the inference results from our proposed debiasing method and the classical approaches in \autoref{subsec:sim_design} applied to the complete-case data. Due to the sparse structure of the design matrix created by B-spline bases, ``DL-Jav'' and ``DL-vdG'' fail to produce reliable inference results. For \emph{Question 1}, the 95\% confidence intervals obtained from our debiasing method cover the stellar mass of the new galaxy in the Firefly catalog regardless of how we estimate the propensity scores or select the tuning parameter, while the Lasso refitting method underestimates the stellar mass by a statistically significant amount. For \emph{Question 2}, the 95\% confidence intervals produced by our debiasing and Lasso refitting methods for the regression coefficient of the distance to cosmic filaments are both below 0, aligning with the existing findings that galaxies are less massive when they are farther away from cosmic filaments \citep{alpaslan2016galaxy,kraljic2018galaxy,malavasi2022relative}. By contrast, the 95\% confidence interval obtained from the ridge projection method contains zero and cannot support a statistically significant conclusion about the relationship between the stellar masses of galaxies and their distances to cosmic filaments. These results demonstrate the usefulness of our proposed debiasing method in practice.

\section{Discussion}
\label{sec:conclusion}

In this paper, we propose a novel debiasing method for conducting valid inference on high-dimensional linear models with MAR outcomes. We establish the asymptotic normality of the debiasing method and illuminate its statistical and computational efficiencies through the dual formulation of the key debiasing program. Simulation studies and real-world applications of our proposed method demonstrate its finite-sample superiority over the existing high-dimensional inference methods in the literature. There are several potential applications and extensions that can further advance the impact of our debiasing method.

{\bf 1. Applications in causal inference:} As discussed by \cite{ding2018causal,chakrabortty2019high}, the missing outcome setting in this paper incorporates the classical setup in causal inference under the potential outcome framework \citep{rubin1974estimating} as a special case. Specifically, the observable data in causal inference problems are of the form $\left(\mathbb{Y},T,X \right) \in \mathbb{R}\times \{0,1\}\times \mathbb{R}^d$, where $T\in \{0,1\}$ denotes a binary treatment assignment indicator and $\mathbb{Y}= T\cdot Y(1) + (1-T)\cdot Y(0)$, with $Y(0),Y(1)$ being the potential outcomes of $Y$ within the control group $T=0$ and the treatment group $T=1$, respectively. Given that at most one of the potential outcomes $Y(0),Y(1)$ for each subject is observed, the potential outcome framework can be related to our problem setting by taking $(Y,R)=(Y(1), T)$ in the treatment group or $(Y,R)=(Y(0),1-T)$ in the control group. Therefore, our debiasing inference method can be applied to high-dimensional causal inference problems.

First, our debiasing method can be used to conduct valid statistical inference on the regression function (or the conditional mean outcome) of the treatment group if we focus on the observational data $\{(Y_i, R_i,X_i)\}_{i=1}^n = \{(Y(1)_i, T_i,X_i)\}_{i=1}^n$, or of the control group if the data of interest are $\{(Y_i, R_i,X_i)\}_{i=1}^n = \{(Y(0)_i, 1-T_i,X_i)\}_{i=1}^n$. Similar to regression adjustment in causal inference \citep{freedman2008regression,lin2013agnostic,negi2021revisiting}, we utilize all covariate vectors $X_i\in \mathbb{R}^d,i=1,...,n$ in both treatment and control groups to address the above inference problem, which could provide additional efficiency gains.

Second, under Assumption~\ref{assump:basic}, our proposed debiasing program \eqref{debias_prog} can also be extended to conduct statistical inference on the linear conditional average treatment effect (CATE) $\mathrm{E}[Y(1) - Y(0) | X]$ with no unmeasured confounding. To this end, borrowing an idea from~\cite{giessing2021inference}, the debiasing program \eqref{debias_prog} can be modified as follows:
\begin{align*}
	&\argmin_{\bm{w}_{(0)},\bm{w}_{(1)} \in \mathbb{R}^n}\sum_{i=1}^n \left[ \hat{\pi}_i w_{i(1)}^2 + (1-\hat{\pi}_i) w_{i(0)}^2\right]\\
	&\text{ subject to } \norm{x- \frac{1}{\sqrt{n}}\sum_{i=1}^n w_{i(1)}\cdot \hat{\pi}_i\cdot X_i }_{\infty} \leq \frac{\gamma_1}{n} \; \text{ and } \; \norm{x- \frac{1}{\sqrt{n}}\sum_{i=1}^n w_{i(0)}\left(1- \hat{\pi}_i\right) X_i }_{\infty} \leq \frac{\gamma_2}{n},
\end{align*}
where $\gamma_1,\gamma_2 >0$ are the tuning parameters and $\hat{\pi}_i,i=1,...,n$ are the estimated treatment assignment probabilities (\emph{i.e.}, propensity scores). The debiased estimator \eqref{debias_est} will become
\begin{align*}
	&\hat{m}^{\text{debias}}(x;\hat{\bm{w}}_{(1)},\hat{\bm{w}}_{(0)})\\ &= x^T \left(\hat{\beta}_{(1)} - \hat{\beta}_{(0)} \right) + \frac{1}{\sqrt{n}} \sum_{i=1}^n \left[\hat{w}_{i(1)} \cdot T_i \left(\mathbb{Y}_i-X_i^T \hat{\beta}_{(1)} \right) - \hat{w}_{i(0)} \cdot \left(1-T_i\right) \left(\mathbb{Y}_i-X_i^T \hat{\beta}_{(0)} \right) \right].
\end{align*}
However, this new estimator may no longer be asymptotically semiparametrically efficient, and we leave a thorough theoretical investigation for future work. 

{\bf 2. Model mis-specification and robustness:} Theoretical guarantees of our debiasing method are validated under the linear model \eqref{linear_reg} and sub-Gaussianity assumptions. While we have conducted additional simulations in \autoref{subsec:heavy_noises} to certify the robustness of our debiasing method against the violations of the sub-Gaussian noise assumption, it is of research interest to investigate the performance and theoretical implications of our debiasing method when the assumed linear model is misspecified. Given that our analysis of the inference with random design is unconditional on the covariate vector $X \in \mathbb{R}^d$, the modification technique in \cite{buhlmann2015high} might provide some insights. Additionally, it would be worthwhile to extend the proposed debiasing method to more general inferential targets, beyond the conditional mean, generalized linear models, and conditional quantiles.

{\bf 3. Missing covariates:} Extending our debiasing method to address issues related to incomplete covariates remains an open problem. Notably, the problem is more challenging in high-dimensional settings because there are numerous covariates, resulting in as many as $2^d$ potential missingness patterns.
While there have been some efforts to leverage the MCAR assumption \citep{wang2019rate} to tackle this problem, it is noteworthy that MCAR is a strong assumption that may not hold in many real-world scenarios.
A possible remedy is to use the multiple imputation method \citep{carpenter2023multiple} with a well-designed debiasing program similar to \cite{xue2021semi,song2024semi}, but it is unclear how to design a reliable imputation model that is compatible with the downstream debiasing program.

\vspace{10mm}
\appendix

\section{Implementation of the Proposed Debiasing Method}
\label{app:imple_debias}

Our proposed debiasing method consists of three main estimation steps: (i) computing the Lasso pilot estimate \eqref{lasso_pilot}, (ii) estimating the propensity scores $\pi_i=\pi(X_i)=\mathrm{P}(R_i=1|X_i),i=1,...,n$, and (iii) solving the debiasing programs \eqref{debias_prog} and \eqref{debias_prog_dual}. We explain the implementation details and strategies of selecting two tuning parameters $\lambda>0$ in \eqref{lasso_pilot} and $\gamma >0$ in \eqref{debias_prog} as follows.

{\bf 1. Lasso pilot estimate:} We use the complete-case (or equivalently, observed) data $(X_i,Y_i) \in \mathbb{R}^d\times \mathbb{R}$ with $R_i=1$ for $i=1,...,n$ to construct the Lasso pilot estimate $\hat{\beta}$. To select the regularization parameter $\lambda>0$ in a data adaptive way, we adopt the scaled Lasso \citep{sun2012scaled} with its universal regularization parameter $\lambda_0=\sqrt{\frac{2\log d}{n}}$ as the initialization. Specifically, it provides an iterative algorithm to obtain the Lasso estimate $\hat{\beta}$ and a consistent estimator $\hat{\sigma}_{\epsilon}$ of the noise level $\sigma_{\epsilon}$ from the following jointly convex optimization problem:
$$\left(\hat{\beta}(\tilde{\lambda}), \hat{\sigma}_{\epsilon}(\tilde{\lambda})\right) = \argmin_{\beta\in \mathbb{R}^d, \sigma_{\epsilon} >0} \left[\frac{1}{2n\sigma_{\epsilon}} \sum_{i=1}^n R_i\left(Y_i - X_i^T\beta\right)^2 + \frac{\sigma_{\epsilon}}{2}+\tilde{\lambda}\norm{\beta}_1\right],$$
where the regularization parameter $\tilde{\lambda} >0$ is updated iteratively as well. We also utilize this estimated noise level $\hat{\sigma}_{\epsilon}$ to construct our asymptotic $1-\tau$ confidence interval \eqref{asym_ci}.

{\bf 2. Propensity score estimation:} When comparing our debiasing method with other existing high-dimensional inference methods, we utilize the Lasso-type logistic regression \eqref{lr_lasso} in \autoref{app:prop_score_lasso} to estimate the propensity scores $\pi_i=\mathrm{P}(R_i=1|X_i), i=1,...,n$, in which the regularization parameter $\zeta$ is selected from 40 logarithmically spaced values within the interval $\left[0.1\sqrt{\frac{\log d}{n}}, 300\sqrt{\frac{\log d}{n}}\right]$ via 5-fold cross-validation. However, as demonstrated in \autoref{subsec:nonpar_prop} through simulation studies, our debiasing method can perform even better when the propensity scores are estimated more accurately by other nonparametric methods, reflecting the theoretical results in \autoref{sec:theory}.

{\bf 3. Solving the debiasing program:} 
When our primal debiasing program is strictly feasible, strong duality holds by Slater's condition (Theorem 3.2.6 in \citealt{Renegar2001Amvo}), and the debiasing weights $\hat{w}_i(x), i=1,...,n$ can be estimated from either the primal program \eqref{debias_prog} or the dual program \eqref{debias_prog_dual} by the identity \eqref{primal_dual_sol2} in Proposition~\ref{prop:dual_debias}. We implement our debiasing method in both Python and R. In particular, we solve the primal program \eqref{debias_prog} using the Python package ``\texttt{CVXPY}'' \citep{diamond2016cvxpy,agrawal2018rewriting} or the R package ``\texttt{CVXR}'' \citep{cvxr2020}. For the dual program \eqref{debias_prog_dual}, we formulate and implement the classical coordinate descent algorithm \citep{wright2015coordinate} as:
\begin{equation}
	\label{coord_desc}
		\left[\hat{\ell}(x)\right]_j \gets \frac{\mathcal{S}_{\frac{\gamma}{n}}\left(-\frac{1}{2n} \sum_{i=1}^n \hat{\pi}_i\cdot \left(\sum_{k\neq j} X_{ij} X_{ik} \left[\hat{\ell}(x)\right]_k\right) -x_j\right)}{\frac{1}{2n} \sum_{i=1}^n \hat{\pi}_i X_{ij}^2} \quad \text{ for } \quad j=1,...,d,
\end{equation}
where $\mathcal{S}_{\frac{\gamma}{n}}(u) = \mathrm{sign}(u)\cdot \left(|u|-\frac{\gamma}{n}\right)_+$ is the soft-thresholding operator, and the above procedure iterates cyclically from $j=1$ to $d$ until convergence. 

The tuning parameter $\gamma >0$ of our debiasing program is selected from 41 equally spaced values between 0 and $n\norm{x}_{\infty}$ via 5-fold cross-validation according to the dual problem \eqref{debias_prog_dual}, in which $x\in \mathbb{R}^d$ is the current query point. Here, the upper bound for the candidate values of $\gamma>0$ is set to $n\norm{x}_{\infty}$ because the optimal solution $\hat{w}(x)$ for the primal problem \eqref{debias_prog} becomes 0 when $\frac{\gamma}{n} \geq \norm{x}_{\infty}$. There are three different criteria for selecting the final value of $\gamma>0$ from the cross-validation results. 
\begin{enumerate}[label=(\roman*)]
	\item The first criterion selects the value of $\gamma>0$ that minimizes the average cross-validated risk of the dual debiasing program \eqref{debias_prog_dual}, which is known as the ``min-CV'' rule.
	
	\item The second criterion selects a value of $\gamma>0$ that is smaller than the risk-minimizing choice and whose average cross-validated risk is at least one standard error (calculated across the five folds) away from the minimum; see Section 3.4.3 in \cite{BreiFrieStonOlsh84} and \cite{chen2021one} for details. We call this criterion the ``1SE'' rule. Its rationale is to improve inference by sacrificing some statistical efficiency (\emph{i.e.}, potentially increasing the estimated asymptotic variance) in exchange for less bias. 
	
	\item The third criterion chooses the smallest value of $\gamma>0$ for which the primal program is feasible for all five folds of the data. We call this criterion the ``min-feas'' rule, which leads to the most conservative confidence interval for our debiasing method.
\end{enumerate}
We recommend choosing the final value of $\gamma$ between the values selected by the ``min-feas'' and ``1SE'' rules, depending on how conservative the subsequent inference needs to be; recall from \autoref{subsec:sim_results} for our simulation results.

\section{Sufficient Conditions for Assumption~\ref{assump:sparse_dual}}
\label{app:suff_cond_dual}

We present several sufficient conditions in Lemma~\ref{lem:suff_sparse_dual} under which Assumption~\ref{assump:sparse_dual} on the sparse approximation to the population dual solution $\ell_0(x) \in \mathbb{R}^d$ in Definition~\ref{defn:popul_dual} holds and illuminate these sufficient conditions through some motivating examples.

\begin{lemma}[Sufficient conditions for the sparse approximation to $\ell_0(x)$]
	\label{lem:suff_sparse_dual}
	Under Assumption~\ref{assump:gram_lower_bnd}, we let $\left[\mathrm{E}(RXX^T) \right]^{-1}:=M=(\bm{m}_1,...,\bm{m}_d) \in \mathbb{R}^{d\times d}$, where $\bm{m}_k,k=1,...,d$ are columns of the matrix $M\in \mathbb{R}^{d\times d}$. We also denote by $M_{S,T} \in \mathbb{R}^{|S|\times |T|}$ the submatrix of $M$ indexed by $S,T \subset \{1,...,d\}$ and let $\sigma_{\min}(M_{S,T})$ be its smallest singular value.
	\begin{enumerate}[label=(\alph*)]
		\item If each column of $M$ has at most $m\geq 1$ nonzero entries and the query vector $x\in \mathbb{R}^d$ has at most $s_x\geq 1$ nonzero entries, then Assumption~\ref{assump:sparse_dual} holds with $s_{\ell}(x)=m\cdot s_x$ and $r_{\ell}=0$ for all $n\geq 1$.
		
		\item Suppose that each column $\bm{m}_k$ of $M$ lies in $\mathcal{C}_1(J_k,\vartheta)$ for some $J_k\subset \{1,...,d\}$ and $\vartheta\in [0,\infty)$. If the query vector $x\in \mathbb{R}^d$ has its support set as $\mathrm{support}(x)=S_x$ of size at most $s_x>0$ and an index set $J \subset \{1,...,d\}$ satisfies $\sigma_{\min}(M_{J,S_x}) >0$, then Assumption~\ref{assump:sparse_dual} holds with $s_{\ell}(x) = |J|\cdot R_{\ell}(n)$ and $r_{\ell}=O\left(\frac{(1+\vartheta) \cdot K_{J,x}}{\sqrt{R_{\ell}(n)}}\right)$ for some $R_{\ell}(n) \to \infty$ as $n\to \infty$, where $K_{J,x} = \max_{k\in S_x} \frac{\sqrt{s_x} \norm{M_{J_k,k}}_1}{\sigma_{\min}(M_{J,S_x})}$. 
		
		\item Suppose that each column $\bm{m}_k$ of $M$ lies in $\mathcal{C}_1(J_k,\vartheta)$ for some $J_k\subset \{1,...,d\}$ and $\vartheta \in [0,\infty)$. If the query vector $x\in \mathbb{R}^d$ lies in $\mathcal{C}_1(J,\vartheta_x)$ with an index set $J \subset \{1,...,d\}$ satisfying $\kappa_{\min}(J,\vartheta_x) >0$, then Assumption~\ref{assump:sparse_dual} holds with $s_{\ell}(x) = |J|\cdot R_{\ell}(n)$ and $r_{\ell}=O\left(\frac{(1+\vartheta) K_{J,x}}{\sqrt{R_{\ell}(n)}}\right)$ for some $R_{\ell}(n) \to \infty$ as $n\to \infty$, where $K_{J,x} = (1+\vartheta_x)\max\limits_{1\leq k \leq d}\frac{\sqrt{|J|}\cdot\norm{M_{J_k,k}}_1}{\kappa_{\min}(J,\vartheta_x)}$ and $\kappa_{\min}(J,\vartheta_x)=\inf\limits_{u\in \mathcal{C}_1(J,\vartheta_x)} \frac{\sqrt{|J|}\cdot \norm{M_Ju}_2}{\norm{u_J}_1}$ is the compatibility constant of the submatrix $M_J$.
	\end{enumerate}
\end{lemma}

\begin{proof}[Proof of Lemma~\ref{lem:suff_sparse_dual}]
	The proof of Lemma~\ref{lem:suff_sparse_dual} is inspired and modified from Lemma 4 in \cite{giessing2021inference}.
	
	\noindent (a) Let $S_x=\mathrm{support}(x) \subset \{1,...,d\}$ for the fixed query vector $x\in \mathbb{R}^d$. Given that $\max_{1\leq k \leq d} \norm{M_k}_0 \leq m$ and $\ell_0(x) = -2\left[\mathrm{E}\left(RXX^T\right)\right]^{-1}x \in \mathbb{R}^d$, we compute that
	\begin{align*}
		\norm{\ell_0(x)}_0 = \norm{Mx}_0=\norm{\sum_{k=1}^d \bm{m}_kx_k}_0 = \norm{\sum_{k\in S_x} \bm{m}_k x_k}_0 \leq \sum_{k\in S_x} \norm{\bm{m}_k}_0 \leq m\cdot s_x. 
	\end{align*}
	Thus, $\ell_0(x)$ by itself is sparse and we can take $\tilde{\ell}(x)=\ell_0(x)$ in Assumption~\ref{assump:sparse_dual} with $s_{\ell}(x)=m\cdot s_x$ and $r_{\ell}=0$ for all $n\geq 1$.\\
	
	\noindent (b) Since each column of $M$ lies in $\mathcal{C}_1(J_k,\vartheta)$ for some $J_k\subset \{1,...,d\}$ and $\vartheta\in [0,\infty)$, we know that $\norm{M_{J_k^c,k}}_1 \leq \vartheta \norm{M_{J_k,k}}_1$ for all $k\in S_x$. Moreover, $\sigma_{\min}(M_{J,S_x}) = \inf_{\mathrm{support}(u)=S_x} \frac{\norm{M_J u}_2}{\norm{u}_2} >0$ by our assumption. We therefore have that
	\begin{align*}
		\norm{\left[\ell_0(x)\right]_{J^c}}_1 &=2\norm{\left[Mx\right]_{J^c}}_1\\
		&= 2\norm{\sum_{k=1}^d M_{J^c,k} \cdot x_k}_1\\
		&\leq 2\sum_{k\in S_x} \norm{M_{J^c,k}}_1 |x_k|\\
		&\leq 2 \sum_{k\in S_x} \left(\norm{M_{J^c\cap J_k^c,k}}_1 + \norm{M_{J^c\cap J_k,k}}_1 \right) |x_k|\\
		&\leq 2 \sum_{k\in S_x} \left(\norm{M_{J_k^c,k}}_1 + \norm{M_{J_k,k}}_1 \right) |x_k|\\
		&\leq 2(1+\vartheta) \sum_{k\in S_x} \norm{M_{J_k,k}}_1 |x_k|\\
		&\leq 2(1+\vartheta) \sqrt{s_x} \cdot\max_{k\in S_x} \norm{M_{J_k,k}}_1 \norm{x}_2\\
		&\leq  2(1+\vartheta) \sqrt{s_x} \cdot\max_{k\in S_x} \norm{M_{J_k,k}}_1 \norm{(Mx)_J}_2 \sup_{\mathrm{support}(u)=S_x} \frac{\norm{u}_2}{\norm{(M u)_J}_2}\\
		&\leq \frac{2(1+\vartheta) \sqrt{s_x}\max_{k\in S_x} \norm{M_{J_k,k}}_1}{\sigma_{\min}(M_{J,S_x})} \cdot \norm{(Mx)_J}_1\\
		&\equiv (1+\vartheta) K_{J,x}  \norm{\left[\ell_0(x)\right]_{J}}_1,
	\end{align*}
	showing that $\ell_0(x) \in \mathcal{C}_1\left(J, (1+\vartheta) K_{J,x}\right)$. By Theorem 2.5 in \cite{foucart2013mathematical}, there exists a $s_{\ell}(x)$-sparse approximation $\tilde{\ell}(x)\in \mathbb{R}^d$ to the population dual solution $\ell_0(x)$ such that
	$$\norm{\tilde{\ell}(x) -\ell_0(x)}_2 \leq \frac{1}{2\sqrt{s_{\ell}(x)}} \norm{\ell_0(x)}_1.$$
	Combining with $\ell_0(x) \in \mathcal{C}_1\left(J, (1+\vartheta) K_{J,x}\right)$, we have that
	\begin{align*}
		\norm{\tilde{\ell}(x) -\ell_0(x)}_2 &\leq \frac{1}{2\sqrt{s_{\ell}(x)}}\left[ \norm{\left[\ell_0(x)\right]_J}_1 + \norm{\left[\ell_0(x)\right]_{J^c}}_1 \right] \\
		&\leq \frac{1+(1+\vartheta) K_{J,x}}{2\sqrt{s_{\ell}(x)}}\cdot \norm{\left[\ell_0(x)\right]_J}_1\\
		&\leq \frac{1+(1+\vartheta) K_{J,x}}{2} \sqrt{\frac{|J|}{s_{\ell}(x)}} \cdot\norm{\ell_0(x)}_2.
	\end{align*}
	Thus, when $s_{\ell}(x)=|J|\cdot R_{\ell}(n)$, it follows that
	$$\norm{\tilde{\ell}(x) -\ell_0(x)}_2 \leq \frac{1+(1+\vartheta) K_{J,x}}{2\sqrt{R_{\ell}(n)}} \cdot \norm{\ell_0(x)}_2,$$
	so we take $r_{\ell}=O\left(\frac{(1+\vartheta) K_{J,x}}{\sqrt{R_{\ell}(n)}} \right)$ to fulfill Assumption~\ref{assump:sparse_dual}.\\
	
	\noindent (c) Since $x\in \mathcal{C}_1(J,\vartheta_x)$, it follows that $\norm{x_{J^c}}_1 \leq \vartheta_x \norm{x_J}_1$. Now, we compute that
	\begin{align*}
		\norm{\left[\ell_0(x)\right]_{J^c}}_1 &=2\norm{\left(M x\right)_{J^c}}_1 \\
		&= 2\norm{\sum_{k=1}^d M_{J^c,k} x_k}_1\\
		&\leq 2\max_{1\leq k \leq d} \norm{M_{J^c,k}}_1 \norm{x}_1\\
		&\leq 2(1+\vartheta_x) \max_{1\leq k \leq d} \norm{M_{J^c,k}}_1 \norm{x_J}_1\\
		&\leq 2(1+\vartheta_x) \max_{1\leq k \leq d} \norm{M_{J^c,k}}_1 \norm{(Mx)_J}_2 \sup_{u\in \mathcal{C}_1(J,\vartheta_x)} \frac{\norm{u_J}_1}{\norm{(Mu)_J}_2}\\
		&\leq 2(1+\vartheta) (1+\vartheta_x) \max_{1\leq k \leq d} \norm{M_{J_k,k}}_1 \norm{(Mx)_J}_2 \sup_{u\in \mathcal{C}_1(J,\vartheta_x)} \frac{\norm{u_J}_1}{\norm{M_J u}_2}\\
		&= 2(1+\vartheta) (1+\vartheta_x) \max_{1\leq k \leq d} \frac{\sqrt{|J|}\cdot \norm{M_{J_k,k}}_1}{\kappa_{\min}(J,\vartheta_x)} \cdot \norm{\left(Mx\right)_J}_1\\
		&\equiv (1+\vartheta)K_{J,x} \norm{\left[\ell_0(x)\right]_J}_1,
	\end{align*}
	showing that $\ell_0(x) \in \mathcal{C}_1\left(J,(1+\vartheta)K_{J,x}\right)$. Therefore, as in the proof of Part (b), we conclude that Assumption~\ref{assump:sparse_dual} holds with $s_{\ell}(x)=|J|\cdot R_{\ell}(n)$ and $r_{\ell} =O\left(\frac{(1+\vartheta) K_{J,x}}{\sqrt{R_{\ell}(n)}}\right)$.
\end{proof}

\begin{remark}
	The constants $\sigma_{\min}(M_{J,S_x})$ and $\kappa_{\min}(J,\vartheta_x)$ in (b) and (c) of Lemma~\ref{lem:suff_sparse_dual} are well-known in the statistical literature. \cite{sun2012scaled,sun2013sparse} introduced the eigenvalue version of $\sigma_{\min}(M_{J,S_x}) >0$ to study least-squares and precision matrix estimation after the scaled Lasso selection. The compatibility constant $\kappa_{\min}(J,\vartheta_x)$ appears even more commonly in the literature (see Section 6.2.2 in \citealt{buhlmann2011statistics}), given that it is also related to our restricted eigenvalue condition \eqref{res_eigen_cond} as
	$$ \kappa_{\min}(J,\vartheta_x)=\inf\limits_{u\in \mathcal{C}_1(J,\vartheta_x)} \frac{\sqrt{|J|} \cdot \norm{M_Ju}_2}{\norm{u_J}_1} \geq \inf\limits_{u\in \mathcal{C}_1(J,\vartheta_x)} \frac{\norm{M_Ju}_2}{\norm{u}_2}.$$
\end{remark}

Now, we illustrate these sufficient conditions in Lemma~\ref{lem:suff_sparse_dual} through the linear regression model \eqref{linear_reg} under both MCAR and MAR data settings and connect the sparsity structure of $\left[\mathrm{E}(RXX^T) \right]^{-1} \in \mathbb{R}^{d\times d}$ with the well-studied covariance matrix $\Sigma \in \mathbb{R}^{d\times d}$ of the covariate vector $X\in \mathbb{R}^d$ in the statistical literature.

\begin{example}[Linear regression model under the MCAR mechanism]
	\label{example:linear_reg_mcar}
	Recall the high-dimensional sparse linear regression model \eqref{linear_reg} in Assumption~\ref{assump:basic} as:
	$$Y=X^T\beta +\epsilon \quad \text{ with } \quad \mathrm{E}(\epsilon|X)=0 \quad \text{ and }\quad \mathrm{E}\left(\epsilon^2|X\right)=\sigma_{\epsilon}^2.$$	
	Furthermore, we suppose, without loss of generality, that the covariate vector $X \in \mathbb{R}^d$ is centered in this example, \emph{i.e.}, $\mathrm{E}(X)=0$. In order to study the sparsity structure of $$\ell_0(x)=-2\left[\mathrm{E}(RXX^T) \right]^{-1} x :=-2Mx,$$ 
	it suffices to exploit the (weak) sparsity of $M= \left[\mathrm{E}(RXX^T) \right]^{-1} \in \mathbb{R}^{d\times d}$ by Lemma~\ref{lem:suff_sparse_dual} when the query point $x\in \mathbb{R}^d$ is sparse or lies in a cone of dominant coordinates $\mathcal{C}_1\left(J,\vartheta_x \right)$. In particular, under the MCAR setting, we know that $R\ind (Y,X)$, and the propensity score satisfies $\pi(X,Y)=\mathrm{P}(R=1|X,Y) = \mathrm{P}(R=1)$ so that 
	$$M = \left[\mathrm{E}(RXX^T) \right]^{-1}= \frac{1}{\mathrm{P}(R=1)}\cdot \Sigma^{-1}$$ 
	will be sparse if and only if the precision matrix $\Sigma^{-1}$ is sparse. Examples of sparse precision matrices include the autoregressive model of order $q\geq 1$ (whose precision matrix is banded; see, \emph{e.g.}, Section 9.2 in \citealt{lindsey2004statistical}) and solutions to the graphical Lasso model \citep{friedman2008sparse} or other $L_1$-minimization programs \citep{cai2016estimating}. In addition, as the sufficient conditions in Lemma~\ref{lem:suff_sparse_dual}(b,c) only require each column of $\Sigma^{-1}$ to lie in $\mathcal{C}_1(J,\vartheta)$ for some $J\subset \{1,...,d\}$, any moving-average process of order $q\geq 1$ (whose precision matrix is Toeplitz; recall $\Sigma^{\mathrm{ar}}$ in \autoref{subsec:sim_design}) and other short-range dependent processes \citep{rosenblatt2015short} will have their precision matrices fulfilling the requirements as well. 
\end{example}

\begin{example}[Linear regression model under the MAR mechanism]
	\label{example:linear_reg_mar}
	Under the same linear model setup in Example~\ref{example:linear_reg_mcar}, we consider the MAR setting $R\ind Y|X$, in which the propensity score becomes $\pi(X) = \mathrm{P}(R=1|X)$. We consider several scenarios that relate the sparsity structures of $M= \left[\mathrm{E}(RXX^T) \right]^{-1} \in \mathbb{R}^{d\times d}$ to the covariance matrix $\Sigma=\mathrm{E}(XX^T) \in \mathbb{R}^{d\times d}$.
	\begin{enumerate}
		\item Suppose that the propensity score $\pi(X)$ depends on at most $1\leq s_{\pi} \ll d$ covariates in $X\in \mathbb{R}^d$, \emph{i.e.}, without loss of generality,
		\begin{equation}
			\label{sparse_prop_score}
			\pi(x)=\mathrm{P}\left(R=1|X=(x_1,...,x_d)^T\right) = \pi(x_1,...,x_{s_{\pi}}).
		\end{equation}
		We further assume that the first few covariates $X^{(1)},...,X^{(d_1)}$ with $d_1 \geq s_{\pi}$ are independent of the rest covariates $X^{(d_1+1)},...,X^{(d_1+d_2)}$ in the covariate vector $X=\left(X^{(1)},...,X^{(d)}\right)^T \in \mathbb{R}^d$ with $d=d_1+d_2$. Then, the covariance matrix $\Sigma$ is block diagonal as $\Sigma = \begin{pmatrix}
			\Sigma_1 & \bm{0}\\
			\bm{0} & \Sigma_2
		\end{pmatrix}$, where $\Sigma_1\in \mathbb{R}^{d_1\times d_1}$ and $\Sigma_2\in \mathbb{R}^{d_2\times d_2}$. Under the sparsity assumption \eqref{sparse_prop_score} on the propensity score $\pi(X)$, we know that
		$$\mathrm{E}(RXX^T)=\mathrm{E}\left(\pi(X)XX^T\right)= \begin{pmatrix}
			\tilde{\Sigma}_1 & \bm{0}\\
			\bm{0} & \tilde{\Sigma}_2
		\end{pmatrix} \quad \text{ and } \quad M=\left[\mathrm{E}(RXX^T)\right]^{-1} = \begin{pmatrix}
			\tilde{\Sigma}_1^{-1} & \bm{0}\\
			\bm{0} & \tilde{\Sigma}_2^{-1}
		\end{pmatrix}$$
	for some matrices $\tilde{\Sigma}_1^{-1}$ and $\Sigma_2^{-1}$. Thus, the sufficient condition (a) in Lemma~\ref{lem:suff_sparse_dual} holds.
		
		\item More generally, the complete-case Gram matrix $\mathrm{E}(RXX^T)= \begin{pmatrix}
			\tilde{\Sigma}_{11} & \tilde{\Sigma}_{12}\\
			\tilde{\Sigma}_{21} & \tilde{\Sigma}_{22}
		\end{pmatrix}$ has its inverse as:
		\begin{align*}
			M &=\left[\mathrm{E}(RXX^T) \right]^{-1} \\
			&= \begin{pmatrix}
				\tilde{\Sigma}_{11}^{-1} + \tilde{\Sigma}_{11}^{-1} \tilde{\Sigma}_{12} \left(\tilde{\Sigma}_{22} - \tilde{\Sigma}_{21} \tilde{\Sigma}_{11}^{-1} \tilde{\Sigma}_{12} \right)^{-1} \tilde{\Sigma}_{21} \tilde{\Sigma}_{11}^{-1}  &  -\tilde{\Sigma}_{11}^{-1} \tilde{\Sigma}_{12} \left(\tilde{\Sigma}_{22} - \tilde{\Sigma}_{21} \tilde{\Sigma}_{11}^{-1} \tilde{\Sigma}_{12} \right)^{-1}\\
				-\left(\tilde{\Sigma}_{22} - \tilde{\Sigma}_{21} \tilde{\Sigma}_{11}^{-1} \tilde{\Sigma}_{12} \right)^{-1} \tilde{\Sigma}_{21} \tilde{\Sigma}_{11}^{-1} & \left(\tilde{\Sigma}_{22} - \tilde{\Sigma}_{21} \tilde{\Sigma}_{11}^{-1} \tilde{\Sigma}_{12} \right)^{-1}
			\end{pmatrix}.
		\end{align*}
			If the entries of $\tilde{\Sigma}_{21} = \tilde{\Sigma}_{12}^T$ are small in absolute value, as can occur when $\Sigma=\mathrm{E}(XX^T)$ is diagonally dominant (such as for banded, Toeplitz, or other sparse covariance matrices; recall $\Sigma^{\mathrm{cs}}$ and $\Sigma^{\mathrm{ar}}$ in \autoref{subsec:sim_design}) together with \eqref{sparse_prop_score}, then the off-diagonal block matrices of $M=\left[\mathrm{E}(RXX^T) \right]^{-1}$ given by
		$$-\left(\tilde{\Sigma}_{22} - \tilde{\Sigma}_{21} \tilde{\Sigma}_{11}^{-1} \tilde{\Sigma}_{12} \right)^{-1} \tilde{\Sigma}_{21} \tilde{\Sigma}_{11}^{-1} = -\left[\tilde{\Sigma}_{11}^{-1} \tilde{\Sigma}_{12} \left(\tilde{\Sigma}_{22} - \tilde{\Sigma}_{21} \tilde{\Sigma}_{11}^{-1} \tilde{\Sigma}_{12} \right)^{-1} \right]^T$$
		are also small in their absolute values compared with the diagonal block matrices of $M$ and therefore, the sufficient conditions (b) and (c) in Lemma~\ref{lem:suff_sparse_dual} apply.
	\end{enumerate}
\end{example}

\section{Generalized Linear Model with Lasso Regularization for Propensity Score Estimation}
\label{app:prop_score_lasso}

Recall that the consistency and asymptotic normality of our proposed debiasing method hold as long as the propensity scores are consistently estimated at the sample covariate vectors $X_i,i=1,...,n$ under the mild rate condition \eqref{rate_required}. While it is not necessary to estimate the propensity scores via parametric models for our debiasing method (see \autoref{subsec:nonpar_prop} for our simulation results), the generalized linear model serves as a baseline approach and can approximate many complex nonparametric models by including sufficiently many well-designed covariates. As a result, the covariate vector is often high-dimensional with $d\gg n$ in practice, and we consider estimating the propensity score $\pi(x)=\mathrm{P}(R=1|X=x)$ under the MAR assumption (Assumption~\ref{assump:basic}) with the data $\{(X_i,R_i)\}_{i=1}^n \subset \mathbb{R}^d \times \{0,1\}$ by the Lasso-type generalized linear regression defined as:
\begin{align}
	\label{glm_lasso}
	\begin{split}
		\hat{\theta}
		&=\argmin_{\theta \in \mathbb{R}^d} \left\{-\frac{1}{n} \sum_{i=1}^n R_i \log\left[\frac{\varphi(X_i^T\theta)}{1-\varphi(X_i^T\theta)} \right] -\frac{1}{n}\sum_{i=1}^n \log\left[1-\varphi(X_i^T \theta) \right]  + \zeta\norm{\theta}_1\right\}\\
		&:= \argmin_{\theta \in \mathbb{R}^d} \left[\mathcal{L}_{\varphi}(\theta) + \zeta\norm{\theta}_1 \right],
	\end{split}
\end{align}
where $\mathcal{L}_{\varphi}:\mathbb{R}^d \to \mathbb{R}$ is the negative log-likelihood function associated with a known link function $\varphi:\mathbb{R}\to (0,1)$ and $\zeta >0$ is a regularization parameter. In particular, when the logistic link function $\varphi_{\mathrm{log}}(u) = \frac{1}{1+e^{-u}}$ is used, the regression problem \eqref{glm_lasso} reduces to the following Lasso-type logistic regression as:
\begin{equation}
	\label{lr_lasso}
	\hat{\theta} = \argmin_{\theta \in \mathbb{R}^d} \left\{-\frac{1}{n}\sum_{i=1}^n \left[R_i X_i^T \theta - \log\left(1+e^{X_i^T \theta} \right) \right] + \zeta\norm{\theta}_1 \right\}.
\end{equation}

With the estimate $\hat{\theta}$ from the Lasso-type generalized linear regression \eqref{glm_lasso}, it is also feasible to derive an explicit formula for $r_{\pi}=r_{\pi}(n,\delta)$ in Assumption~\ref{assump:prop_consistent_rate} for the estimated propensity scores $\hat{\pi}_i,i=1,...,n$ under the following regularity conditions.

\begin{assumption}[Sparse generalized linear propensity score model]
	\label{assump:glm_prop_score}
	Under Assumption~\ref{assump:basic}, the missingness variable $R\in \{0,1\}$ follows the model
	\begin{equation}
		\label{prop_score_model}
		R|X \sim \mathrm{Bernoulli}\left(\varphi(X^T \theta_0) \right),
	\end{equation}
	where $\varphi: \mathbb{R} \to (0,1)$ is a known link function and $\theta_0\in \mathbb{R}^d$ is a high-dimensional sparse vector with support $S_{\theta}$ and $|S_{\theta}|=\norm{\theta_0}_0=s_{\theta} \ll d$. 
\end{assumption}

\begin{assumption}[$\mathcal{C}_1(S_{\theta},3)$-restricted eigenvalue condition on $\mathrm{E}(RXX^T)$]
	\label{assump:res_eigen_theta}
	The $\mathcal{C}_1(S_{\theta},3)$-restricted minimum eigenvalue of the complete-case population Gram matrix $\mathrm{E}\left(RXX^T \right) \in \mathbb{R}^{d\times d}$ is bounded away from zero, \emph{i.e.}, there exists an absolute constant $\rho>0$ such that
	$$\mathrm{E}\left[R(X^Tv)^2\right] > \rho \norm{v}_2^2 \quad \text{ for any } \quad v\in \mathcal{C}_1(S_{\theta},3),$$
	where $S_{\theta}$ is the support of the $s_{\theta}$-sparse vector $\theta_0 \in \mathbb{R}^d$.
\end{assumption}

\begin{assumption}[Regular link function]
	\label{assump:link_func}
	The link function $\varphi$ satisfies the following regularity conditions.
	\begin{enumerate}[label=(\alph*)]
		\item $\varphi:\mathbb{R} \to [0,1]$ is twice differentiable, nondecreasing, and $L_{\varphi}$-Lipschitz on $\mathbb{R}$ as well as concave on $\mathbb{R}^+$ for some absolute constant $L_{\varphi}>0$; also, for any $u\in \mathbb{R}$, it holds that $\varphi(u)+\varphi(-u)=1$.
		\item There exists an absolute constant $A_{\varphi}>0$ such that $\frac{\varphi'(u)}{1-\varphi(u)} \leq A_{\varphi}$ for all $u\geq 0$.
		
		\item For $\mathcal{L}_{\varphi}(\theta)$ defined in \eqref{glm_lasso}, there exists an absolute constant $A_H>1$ such that the Hessian matrix $\nabla\nabla \mathcal{L}_{\varphi}(\theta)$ can be expressed as $\nabla\nabla \mathcal{L}_{\varphi}(\theta) = \frac{1}{n} \sum_{i=1}^n H(\theta; R_i,X_i) X_iX_i^T$
		for some $H(\theta;R_i,X_i) > 0$ satisfying 
		$$\left|\log H(\theta + b; R_i,X_i) -\log H(\theta;R_i,X_i) \right| \leq A_H\left(|X_i^T\theta|^2 + |X_i^T b|^2 + |X_i^T b| \right)$$ 
		for any $\theta, b\in \mathbb{R}^d$ and $i=1,...,n$.
	\end{enumerate}
\end{assumption}

The sparse generalized linear regression setup in Assumption~\ref{assump:glm_prop_score} is commonly imposed in the statistical literature to handle high-dimensional data \citep{van2008high,huang2012estimation,cai2023statistical}. Assumption~\ref{assump:res_eigen_theta} is the restricted eigenvalue condition over the cone $\mathcal{C}_1(S_{\theta},3)$; recall \eqref{res_eigen_cond}. In Assumption~\ref{assump:link_func}, the constants $A_{\varphi},A_H>0$ depend only on the link function $\varphi$. A similar condition is also imposed in \cite{cai2023statistical} to establish the consistency of $\hat{\theta}$ in \eqref{glm_lasso} and derive the asymptotic normality of their debiased estimator for $\hat{\theta}$. One can verify by direct calculation and Example 8 in \cite{huang2012estimation} that the logistic link function $\varphi_{\mathrm{log}}(u) = \frac{1}{1+e^{-u}}$ satisfies Assumption~\ref{assump:link_func}. Furthermore, Section 3 of the supplement in \cite{cai2023statistical} also verifies Assumption~\ref{assump:link_func} for the probit link function $\varphi_{\mathrm{probit}}(u)=\frac{1}{\sqrt{2\pi}}\int_{-\infty}^u e^{-\frac{t^2}{2}} dt$, the class of generalized logistic functions, as well as the CDFs of Student's $t_{\nu}$-distributions with any $\nu \in \mathbb{N}$ in their latent variable model (Example 3 in their supplement).

\begin{remark}
	\label{link_func_bound}
	Assumption~\ref{assump:link_func}(b) implies that $\max\left\{\frac{\varphi'(u)}{\varphi(u)}, \frac{\varphi'(u)}{1-\varphi(u)} \right\} \leq A_{\varphi}$
	for the same constant $A_{\varphi}>0$. This is because if $u\geq 0$, then by Assumption~\ref{assump:link_func}(a), $\varphi(u) \geq \frac{1}{2}$ and consequently, 
	$$\frac{\varphi'(u)}{\varphi(u)} \leq \frac{\varphi'(u)}{1-\varphi(u)}\leq A_{\varphi}.$$
	On the other hand, if $u<0$, then Assumption~\ref{assump:link_func}(a) also shows that $\varphi(u) \leq \frac{1}{2}$ and therefore,
	$$\frac{\varphi'(u)}{1-\varphi(u)} \leq \frac{\varphi'(u)}{\varphi(u)} = \frac{\varphi'(-u)}{1-\varphi(-u)} \leq A_{\varphi},$$
	where we use the condition that $\varphi(x)+\varphi(-x)=1$ and $\varphi'(x)-\varphi'(-x)=0$ in Assumption~\ref{assump:link_func}(a).
	
	As for Assumption~\ref{assump:link_func}(c), some algebra shows that
	\begin{align*}
	H(\theta;R_i,X_i) &= -\frac{R_i\left[\varphi''(X_i^T\theta) \varphi(X_i^T\theta) - \left(\varphi'(X_i^T\theta)\right)^2\right]}{\left[\varphi(X_i^T\theta)\right]^2} \\
	&\quad + \frac{(1-R_i)\left[\varphi''(X_i^T\theta) \left(1-\varphi(X_i^T\theta)\right) + \left(\varphi'(X_i^T\theta) \right)^2\right]}{\left[1-\varphi(X_i^T\theta)\right]^2}.
	\end{align*}
	Hence, a sufficient condition for $H(\theta;R_i,X_i)>0$ is that 
	\begin{equation}
		\label{link_func_convex_cond}
		\left[\varphi'(X_i^T\theta)\right]^2 > \max\Big\{\varphi''(X_i^T\theta) \varphi(X_i^T \theta), \varphi''(X_i^T\theta)\left[\varphi(X_i^T\theta)-1\right] \Big\},
	\end{equation}
	which is also equivalent to $\left[\varphi'(X_i^T\theta)\right]^2 > \varphi''(X_i^T\theta) \varphi(X_i^T\theta)$ for all $X_i^T\theta < 0$. This condition can be easily verified for the logistic link function $\varphi_{\mathrm{log}}(u) = \frac{1}{1+e^{-u}}$ and the probit link function $\varphi_{\mathrm{probit}}(u)=\frac{1}{\sqrt{2\pi}}\int_{-\infty}^u e^{-\frac{t^2}{2}} dt$.
\end{remark}

Under the above regularity conditions, we derive the explicit formula for $r_{\pi}=r_{\pi}(n,\delta)$ as follows.

\begin{proposition}
	\label{prop:rel_const_glm}
	Let $\delta \in (0,1)$ and $A_1,A_2,A_3,A_4>0$ be some absolute constants. Suppose that Assumptions~\ref{assump:sub_gaussian}, \ref{assump:glm_prop_score}, \ref{assump:res_eigen_theta}, and \ref{assump:link_func} hold as well as
	$$d \geq s_{\theta} + 2, \quad A_1\phi_{s_{\theta}} \log\left(\frac{\phi_{s_{\theta}}}{\rho}\right) \sqrt{\frac{\phi_1 s_{\theta}\log d}{n}} \leq \frac{\rho}{8}, \; \text{ and }\;   A_2\exp\left[-\frac{n\rho^2}{\phi_{s_{\theta}}^2 \log^2\left(\frac{\phi_{s_{\theta}}}{\rho}\right)} - \phi_1 s_{\theta}\log d\right] \leq \delta.$$
	If $\zeta =A_3 A_{\varphi} \sqrt{\frac{\phi_1\log(d/\delta)}{n}}$, then the rate of consistency $r_{\pi}=r_{\pi}(n,\delta)$ in Assumption~\ref{assump:prop_consistent_rate} can be taken as:
	$$r_{\pi}(n,\delta) = \frac{A_4 L_{\varphi} A_{\varphi}}{\rho} \sqrt{\frac{\phi_1 \phi_{s_{\theta}} s_{\theta} \log(d/\delta)}{n}\left[\log\left(\frac{n}{\delta}\right) + s_{\theta} \log\left(\frac{e d}{s_{\theta}} \right) \right]}.$$
\end{proposition}

The proof of Proposition~\ref{prop:rel_const_glm} is in \autoref{subapp:proof_prop_score}, which relies on the consistency of $\norm{\hat{\theta}-\theta_0}_2$ obtained from the Lasso-type generalized linear regression \eqref{glm_lasso} as in Lemma~\ref{lem:prop_score_const}. It suggests that the rate of convergence for $r_{\pi}$ is of the order $O_P\left(\frac{s_{\theta}\log d}{\sqrt{n}}\right)$, provided that all other model parameters $L_{\varphi}, A_{\varphi}, \rho, \phi_1,\phi_{s_{\theta}} >0$ are independent of $d$ and $n$ with $\delta = \frac{1}{d}$ and $d>n$.

\subsection{Other Examples for Propensity Score Estimation}
\label{app:prop_est_eg}

$\bullet$ \textbf{Single-index model with unknown link function.} As a preliminary step for more complex high-dimensional semiparametric models, we first consider a simple univariate semiparametric model 
\begin{align*}
	\pi(z) = \mathrm{E}(R_i|Z_i = z) = F_0(z), \quad \text{ and } \quad \pi_i=\pi(Z_i), \quad i = 1, \ldots, n,
\end{align*}
where $F_0 \in \Lambda = \left\{ F : \mathbb{R} \rightarrow (0,1) : \int_{-\infty}^{\infty} \big[F''(x)\big]^2 dx < \infty \right\}$ and $z=g_0(x) \in \mathbb{R}$ for some transformation function $g_0:\mathbb{R}^d\to \mathbb{R}$. Assume tentatively that $g_0$ is given so that $Z_i=g_0(X_i), i=1,...,n$.
The unknown function $F_0$ can be estimated by solving the penalized maximum likelihood estimator
\begin{align*}
	\hat{F} = \argmax_{ F \in \Lambda}  \left[\frac{1}{n} \sum_{i=1}^n  \log \big[ R_iF(Z_i) + (1 - R_i) \big(1- F(Z_i)\big) \big] - \lambda^2\int_{-\infty}^\infty  \big[F''(x)\big]^2 dx \right],
\end{align*}
and the propensity score estimates are $\hat{\pi}_i = \hat{F}(Z_i)$, $i=1, \ldots, n$.

If $\lambda \asymp n^{-\frac{2}{5}}$ and $F_0$ is bounded away from zero and one and its first derivative $F_0'$ is bounded away from zero over a compact set, we have
\begin{align*}
	\max_{1\leq i\leq n}\left|\hat{\pi}_i -\pi_i \right| \leq \sup_{z \in \mathbb{R}} |\hat{F}(z) - F_0(z)| = O_P\left(n^{-\frac{3}{10}}\right).
\end{align*}
For detailed derivations and precise regularity conditions, see Lemma 11.14 on Page 235 in~\cite{geer2000empirical}.
Hence, Assumption~\ref{assump:prop_consistent_rate} holds with $r_\pi = O\left(n^{-\frac{3}{10}} \right)$, and the rate condition \eqref{rate_required} is satisfied if $s_{\max} \sqrt{\log d} = o\left(n^{\frac{3}{10}} \right)$ for those high-dimensional settings.

$\bullet$ \textbf{High-dimensional series estimator.} Consider the following generalization of the example from \eqref{prop_score_model}:
\begin{align*}
	\pi(x) = \mathrm{E}(R_i|X_i = x) =  g\big(\alpha_0^T \Psi(x)\big), \quad \text{ and } \quad \pi_i=\pi(X_i), \quad i = 1, \ldots, n,  
\end{align*}
where $g : \mathbb{R} \rightarrow [0,1]$ is a known $L_g$-Lipschitz continuous link function, $\Psi(x) = \left\{\psi_k(x) \right\}_{k=1}^K$ is a set of $K$ known basis functions, possibly high-dimensional with $K \gg n$, $\alpha_0 \in \mathbb{R}^K$ is an unknown sparse parameter vector, and $x \in \mathcal{X}$. Under this model, Section 5.1 in \cite{chakrabortty2019high} suggests estimating the propensity scores at the sample points $X_1,...,X_n$ via
\begin{align*}
	\hat{\pi}_i = g\big(\hat{\alpha}^T\Psi(X_i)\big), \quad i= 1, \ldots, n,
\end{align*}
where $\hat{\alpha}$ denotes an $\ell_1$-norm consistent estimate of $\alpha_0$. 

If $\mathrm{P} \left(\|\hat{\alpha} - \alpha_0\|_1 > a_n \right) \leq q_n$ for some $a_n, q_n= o(1)$ and $a_n \geq 0$, $q_n \in [0,1]$, then
\begin{align*}
	\max_{1\leq i\leq n}\left|\hat{\pi}_i -\pi_i \right| \leq \sup_{x \in \mathcal{X}} \left| g\big(\hat{\alpha}^T\Psi(x)\big) - g\big(\alpha_0^T\Psi(x)\big) \right| = O_P(a_n).
\end{align*}
For detailed derivations and precise regularity conditions, see Theorem B.1 in~\cite{chakrabortty2019high}.
Furthermore, according to Appendix B.1 in \cite{chakrabortty2019high}, it is reasonable to assume that $a_n \propto s_\pi \sqrt{\frac{\log K}{n}} $ and $q_n = o\left(\frac{1}{nK} \right)$, where $s_\pi$ is the (small) number of basis functions of the series representation of $\pi_i$ (\emph{i.e.}, sparsity of $\alpha_0$). Hence, Assumption~\ref{assump:prop_consistent_rate} holds with $r_\pi = O\left(s_\pi \sqrt{\frac{\log K}{n}}\right)$ and the rate condition \eqref{rate_required} is satisfied if $s_{\max}^2 s_\pi^2 (\log d) (\log K) = o(n)$.

$\bullet$  \textbf{High-dimensional single-index kernel smoother.} Consider the following semiparametric single-index model
\begin{align*}
	\pi(x) = \mathrm{E}(R_i|X_i = x) =  g_0(\gamma_0^Tx), \quad \text{ and } \quad \pi_i=\pi(X_i), \quad i = 1, \ldots, n,  
\end{align*}  
where $g_0:\mathbb{R} \rightarrow [0,1]$ is some unknown link function and $\gamma_0 \in \mathbb{R}^d$ is an unknown direction parameter (identifiable only up to scalar multiples).

We consider estimating the $\pi_i$'s via a one-dimensional Nadaraya--Watson-type kernel smoother over the estimated scores $\{\hat{\gamma}^TX_i\}_{i=1}^n$ as:
\begin{align*}
	\hat{\pi}_i \equiv \hat{\pi}(\hat{\gamma}^TX_i) \equiv \frac{\sum_{j =1}^n R_j\cdot K\left(\frac{\hat{\gamma}^TX_j - \hat{\gamma}^TX_i}{h} \right)}{\sum_{j=1}^n K\left(\frac{\hat{\gamma}^TX_j - \hat{\gamma}^TX_i}{h} \right)}, \quad i=1, \ldots, n,
\end{align*}
where $K : \mathbb{R} \rightarrow [0,\infty)$ is a (symmetric) kernel function, $h > 0$ is a smoothing bandwidth parameter, and $\hat{\gamma}$ is an estimate of $\gamma_0$. The estimate $\hat{\gamma}$ can be computed via any standard method in the literature on signal recovery for single index models (possibly under additional assumptions such as elliptically distributed covariates). More details can be found in Section 5.2 in~\cite{chakrabortty2019high} and references therein.

If $\mathrm{P} \left(\|\hat{\gamma} - \tilde{\gamma}_0\|_1 > a_n \right) \leq q_n$ for some $a_n, q_n= o(1)$ for some direction $\tilde{\gamma}_0 \propto \gamma_0$, and $a_n \geq 0$, $q_n \in [0,1]$, then for all $t > 0$, with probability at least $1 - 27 (nd)^{-t^2} - 6 q_n$,
\begin{align*}
	\left|\hat{\pi}_i - \pi_i \right| \lesssim t \: \sqrt{\frac{\log(nd)}{nh}} +  h^2 + a_n + \frac{a_n^2}{h^2} + \frac{\log(nd)}{nh}.
\end{align*}
For detailed derivations and precise regularity conditions, see Theorem B.3 in~\cite{chakrabortty2019high}.
Furthermore, according to Section 5.2 in \cite{chakrabortty2019high}, there exist estimates $\hat{\gamma}$ such that $a_n \propto s_\pi \sqrt{\frac{\log d}{n}} $ and $q_n = o\left(\frac{1}{n d}\right)$, where $s_\pi$ is the sparsity of the direction $\tilde{\gamma}_0 \propto \gamma_0$. If $\log(nd) = o(nh)$, then Assumption 3.2 in \cite{chakrabortty2019high} holds with $v_{n,\pi} = O\left(\sqrt{\frac{\log(nd)}{nh}}\right)$, $w_{n, \pi} = O\left(h^2 + s_\pi \sqrt{\frac{\log d}{n}} + \frac{s_\pi^2\log d}{nh^2} + \frac{\log(nd)}{nh}\right)$, and $q_{n, \pi} = o\left(\frac{1}{nd}\right)$, which in turn implies our Assumption~\ref{assump:prop_consistent_rate}.   
The rate condition \eqref{rate_required} is therefore satisfied if $\big(v_{n,\pi} \sqrt{\log n} + w_{n,\pi} \big) s_{\max} \sqrt{\log d} = o(1)$. 

$\bullet$  \textbf{Partially linear single-index model.} Consider the following semiparametric model 
\begin{align*}
	\pi(x) = \mathrm{E}(R_i|X_i = x) =  F\big(u^T\alpha_0 + \gamma_0(v) \big), \quad x = (u,v), \quad \text{ and } \quad \pi_i=\pi(X_i),\quad i = 1, \ldots, n,
\end{align*}
where the link function $F: \mathbb{R} \rightarrow (0,1)$ is known and differentiable, $\alpha_0 \in \mathbb{R}^{d_\pi}$, and $\gamma_0 \in \mathcal{G}$ with
\begin{align*}
	\mathcal{G} = \left\{ \gamma: [0,1] \rightarrow \mathbb{R},  \int \big[\gamma^{(m)}(v) \big]^2 dv < \infty \right\},
\end{align*}
for some integer $m \geq 1$.
To estimate $g_0(u,v) = u^T\alpha_0 + \gamma_0(v)$, we use the penalized maximum likelihood estimator $\hat{g}(u,v) = u^T\hat{\alpha} + \hat{\gamma}(v)$ defined as
\begin{align*}
	\hat{g} = \argmax_{ g \in \Lambda}  \left[\frac{1}{n} \sum_{i=1}^n  \log \big[ R_iF(g(X_i)) + (1 - R_i) \big(1- F(g(X_i))\big) \big] - \lambda^2\int_{-\infty}^\infty  \big[\gamma^{(m)}(v)\big]^2 dv \right],
\end{align*}
where $\Lambda = \{g : \mathbb{R}^{d_\pi} \times [0,1] \rightarrow \mathbb{R}, g(u,v) = u^T\beta + \gamma(v), \beta \in \mathbb{R}^{d_\pi}, \gamma \in \mathcal{G}\}$.

If $\lambda \asymp n^{-\frac{m}{2m+1}}$,  $F \circ g$ is bounded away from zero and one, $\sup_{x \in \mathbb{R}} |F'(x)| < \infty$, and the dimension $d_\pi$ of the linear component is fixed, then 
\begin{align*}
	\frac{1}{n}\sum_{i=1}^n \big[F(\hat{g}(X_i)) - F(g_0(X_i))\big]^2 = O_P\big(n^{-\frac{2m}{2m+1}} \big).
\end{align*}
For detailed derivations and precise regularity conditions, see Lemma 11.4 on Page 206 in~\cite{geer2000empirical}.
Since $\frac{1}{n}\sum_{i=1}^n\left|\hat{\pi}_i -\pi_i \right| \leq  \left(\frac{1}{n}\sum_{i=1}^n\left|\hat{\pi}_i -\pi_i \right|^2\right)^{1/2}$, the alternative condition to Assumption~\ref{assump:prop_consistent_rate}, $\mathrm{P}\left(\frac{1}{n}\sum_{i=1}^n \left|\hat{\pi}_i - \pi_i\right| > r_{\pi}\right) < \delta$, holds with $r_\pi = O\left(n^{-\frac{m}{2m+1}}\right)$.

Notice that a more careful analysis of the metric entropy numbers will allow one to extend this result to high-dimensional partially linear single-index models, provided that the high-dimensional regression vector is sparse. Similarly, it is possible to generalize this model to allow for a multivariate nuisance function $\gamma_0 : \mathbb{R}^{p} \rightarrow \mathbb{R}$; see Page 19 in~\cite{geer2000empirical}.

$\bullet$ \textbf{High-dimensional generalized sparse additive model.} 
Consider the following high-dimensional generalized sparse additive model as:
\begin{align*}
	\pi(x) = \mathrm{E}(R_i|X_i=x) = g^{-1} \left(\sum_{j=1}^d f_j^*\left(x_j\right) \right), \quad \text{ and } \quad \pi_i=\pi(X_i) \quad i = 1, \ldots, n,
\end{align*}
where $g:\mathbb{R}\to \mathbb{R}$ is a known link function whose inverse is $L_g$-Lipschitz (\emph{e.g.}, $g(u)=\log\left(\frac{u}{1-u}\right)$), $x_j$ is the $j$th entry of vector $x \in \mathbb{R}^d$, and $f_1^*, \ldots, f_d^* \in \mathcal{F}$ for some suitable function class $\mathcal{F}$. For concreteness, we assume that
\begin{align*}
	\mathcal{F} = \left\{ f: [0,1] \rightarrow \mathbb{R}, \int \big| f^{(m)}(x) \big|^r dx < \infty\right\},
\end{align*}
for some $r \geq 1$ and $1 \leq m \leq 2$. Furthermore, we assume that there exists an absolute constant $M > 0 $ such that $\sum_{j=1}^d \bm{1}(f_j^*\neq 0)\leq M$; see Section 4.1 in \cite{haris2022generalized} for details and definitions.

The unknown functions $f_j^*,j=1,...,d$ can be estimated by solving the doubly penalized M-estimation problem
\begin{equation*}
	\label{gsam}
	\left(\hat{f}_1,...,\hat{f}_d\right)=\argmin_{f_1,...,f_d\in \mathcal{F}}\left[-\frac{1}{n}\sum_{i=1}^n L\left(R_i,\sum_{j=1}^d f_j(X_{ij})\right) +\lambda^2 \sum_{j=1}^d P_{st}(f_j) + \frac{\lambda}{n}\sum_{j=1}^d \sum_{i=1}^n f_j^2(X_{ij})\right],
\end{equation*}
where $L(\cdot,\cdot)$ is a loss function (\emph{e.g.}, cross-entropy loss), $\lambda>0$ is a regularization parameter, and $P_{st}$ is a structure-inducing penalty; see Section 2.1 in \cite{haris2022generalized} for detailed examples. 
Define the estimator of $\pi_i = \pi(X_i)$ by
\begin{align*}
	\hat{\pi}_i = g^{-1}\left(\sum_{j=1}^d \hat{f}_j(X_{ij})\right), \quad i = 1, \ldots, n.
\end{align*}
If $\lambda \asymp n^{-\frac{1}{2}}$ and under the above regularity conditions on $\mathcal{F}$, then, with probability at least $1- \delta$,
\begin{align*}
	\frac{1}{n}\sum_{i=1}^n \left| \hat{\pi}_i - \pi(X_i) \right| \overset{\text{(i)}}{\leq} L_g \left(\frac{1}{n}\sum_{i=1}^n \left[\sum_{j=1}^d \left\{\hat{f}_j(X_{ij}) - f_j^*(X_{ij})\right\}\right]^2 \right)^{1/2} \lesssim M \left(\frac{\log (d/\delta)}{n} \right)^{1/4},
\end{align*}
where (i) follows from the $L_g$-Lipschitz continuity of $g^{-1}$ and from H{\"o}lder's inequality between the empirical $\ell_1$- and $\ell_2$-norms. For detailed derivations and precise regularity conditions, see Proposition 3 in \cite{tan2019doubly}.
Hence, the alternative condition to Assumption~\ref{assump:prop_consistent_rate} as $\mathrm{P}\left(\frac{1}{n}\sum_{i=1}^n \left|\hat{\pi}_i - \pi_i\right| > r_{\pi}\right) < \delta$ holds with $r_\pi = O\left( (\log d)^{\frac{1}{4}} n^{-\frac{1}{4}}\right)$.

\section{Additional Simulation Results}
\label{app:add_simulation}

We have carried out simulations in \autoref{sec:experiments} of the main paper to compare our proposed debiasing inference method with standard methods in the literature under various data-generating processes and MAR settings. Here, we provide the complete simulation results mentioned in the main paper. Additionally, we conduct comparative experiments under the simpler MCAR setting for completeness and to illustrate Example~\ref{example:linear_reg_mcar}.

\subsection{Full Comparative Simulation Results}
\label{subapp:full_sim_res}

Here is an outline for all the comparative simulation results in this section.
\begin{itemize}
	\item \autoref{fig:AR_MAR_gauss} displays comparative plots of the three evaluation metrics and a normality assessment under the Toeplitz (or autoregressive) covariance design with Gaussian $\mathcal{N}(0,1)$ noise for a selected simulation setting.
	
	\item \autoref{table:cirsym_MAR_res} presents the simulation results under the circulant symmetric covariance matrix $\Sigma^{\mathrm{cs}}$, Gaussian noise $\mathcal{N}(0,1)$, and the MAR setting \eqref{MAR_correct} in \autoref{subsec:sim_design}.
	
	\item \autoref{table:AR_MAR_res} presents the simulation results under the Toeplitz (or autoregressive) covariance matrix $\Sigma^{\mathrm{ar}}$, Gaussian noise $\mathcal{N}(0,1)$, and the MAR setting \eqref{MAR_correct}.
	
	\item \autoref{table:cirsym_MAR_laperr_res} presents the simulation results under the circulant symmetric covariance matrix $\Sigma^{\mathrm{cs}}$, $\mathrm{Laplace}\left(0, \frac{1}{\sqrt{2}}\right)$-distributed noise, and the MAR setting \eqref{MAR_correct}.
	
	\item \autoref{table:cirsym_MAR_terr_res} presents the simulation results under the circulant symmetric covariance matrix $\Sigma^{\mathrm{cs}}$, $t_2$-distributed noise, and the MAR setting \eqref{MAR_correct}.
	
	\item \autoref{table:AR_MAR_laperr_res} presents the simulation results under the Toeplitz (or autoregressive) covariance matrix $\Sigma^{\mathrm{ar}}$, $\mathrm{Laplace}\left(0, \frac{1}{\sqrt{2}}\right)$-distributed noise, and the MAR setting \eqref{MAR_correct}.
	
	\item \autoref{table:AR_MAR_terr_res} presents the simulation results under the Toeplitz (or autoregressive) covariance matrix $\Sigma^{\mathrm{ar}}$, $t_2$-distributed noise, and the MAR setting \eqref{MAR_correct}.
	
	\item \autoref{fig:Cir_Nonpar_Prop_Score1} displays the simulation results for our debiasing method with dense $\beta_0^{de}$ and sparse $x^{(2)}$ when the propensity scores are estimated by various machine learning methods under the new MAR setting \eqref{MAR_mis}.
	
	\item \autoref{table:CirSym_MAR_Non_Prop_Score} presents the simulation results for our proposed debiasing method when the propensity scores are estimated by various machine learning methods under the circulant symmetric covariance matrix $\Sigma^{\mathrm{cs}}$, Gaussian noise $\mathcal{N}(0,1)$, and the new MAR setting \eqref{MAR_mis} in \autoref{subsec:nonpar_prop}.
	
	\item \autoref{fig:Cir_Nonpar_Prop_qqplot} supplements the ``QQ-plots'' corresponding to the asymptotic normality plots for our proposed debiasing method in \autoref{fig:Cir_Nonpar_Prop_Score1} and \autoref{fig:Cir_Nonpar_Prop_Score2} when the propensity scores are estimated by various machine learning methods under the circulant symmetric covariance matrix $\Sigma^{\mathrm{cs}}$, Gaussian noise $\mathcal{N}(0,1)$, and the new MAR setting \eqref{MAR_mis}. 
\end{itemize}


\begin{figure}[!ht]
	\captionsetup[subfigure]{justification=centering}
	\begin{subfigure}[t]{0.49\linewidth}
		\centering
		\includegraphics[width=1\linewidth]{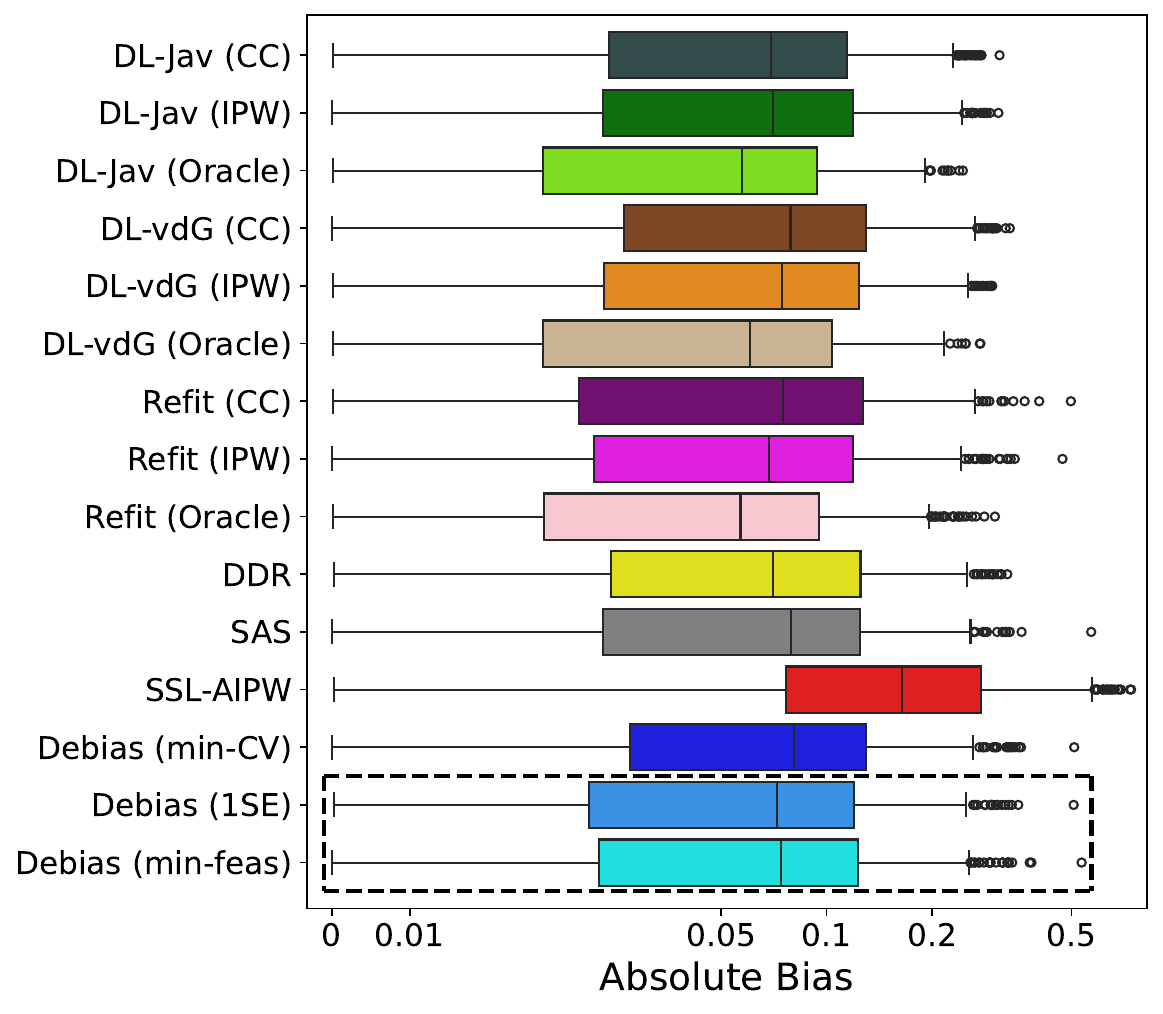}
	\end{subfigure}
	\hfil
	\begin{subfigure}[t]{0.49\linewidth}
		\centering		\includegraphics[width=1\linewidth]{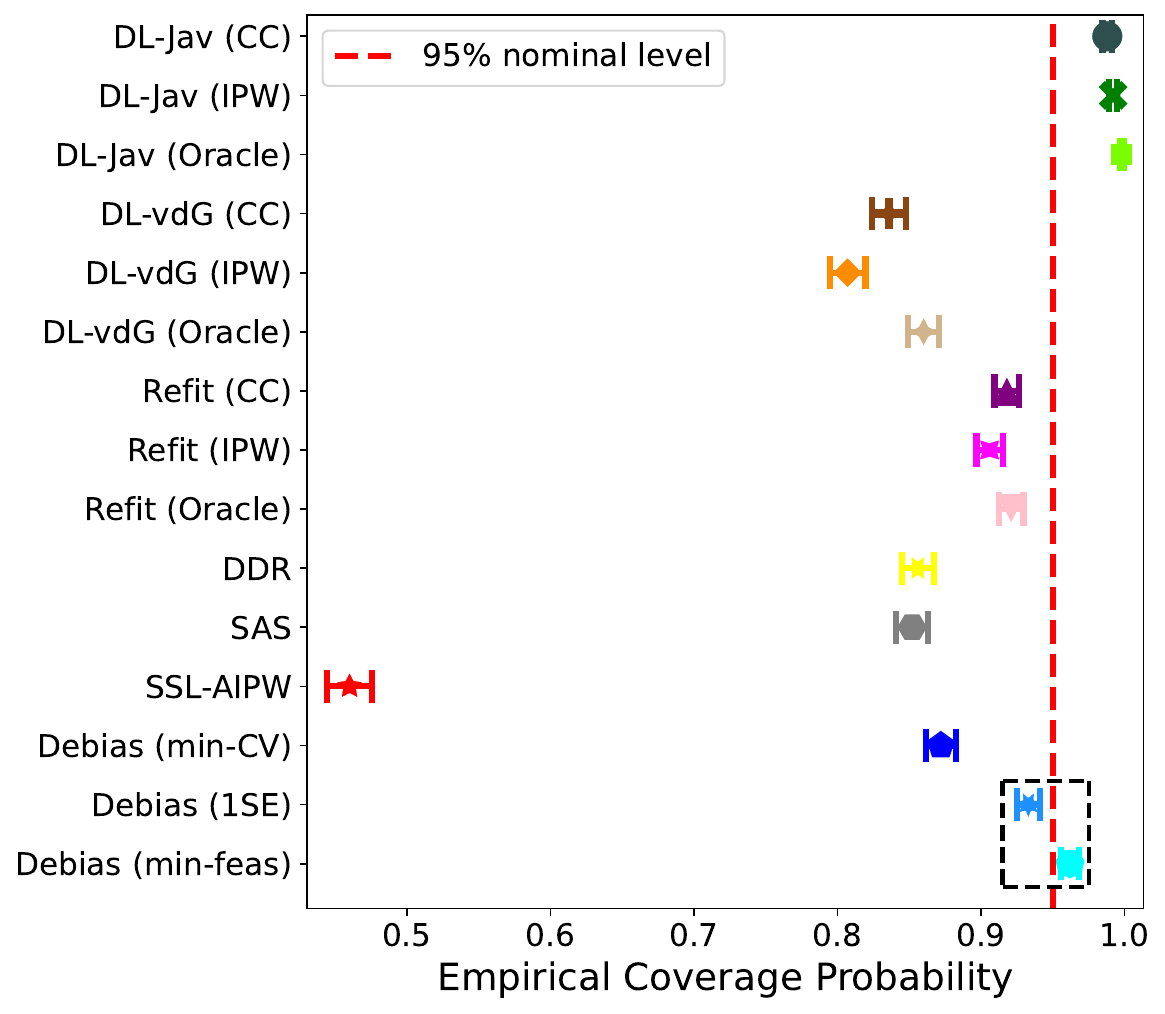}
	\end{subfigure}
	\begin{subfigure}[c]{0.485\linewidth}
		\centering
		\includegraphics[width=1\linewidth]{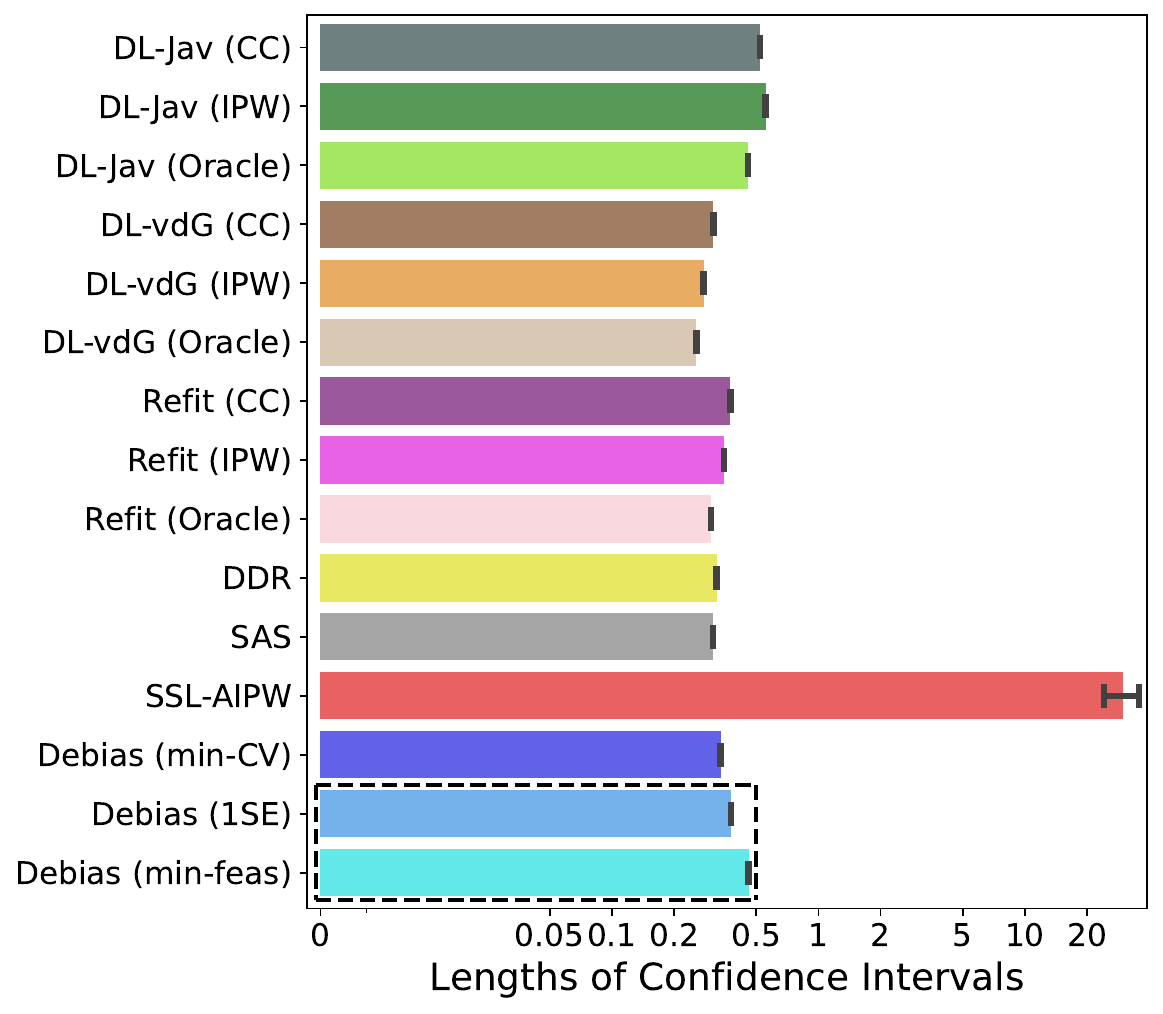}
	\end{subfigure}
	\hfil
	\begin{subfigure}[c]{0.495\linewidth}
		\centering		\includegraphics[width=1\linewidth]{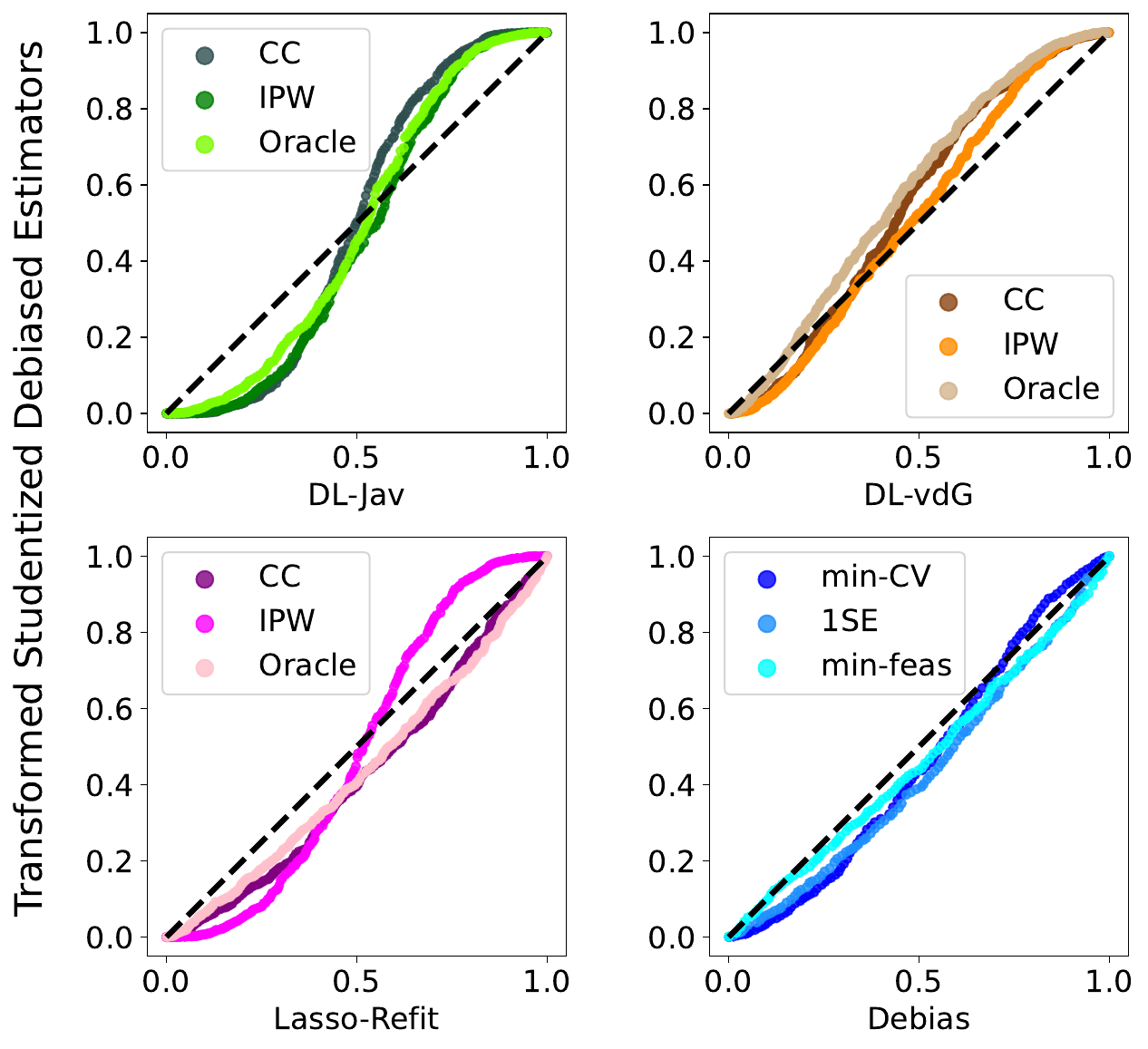}
	\end{subfigure}
	\caption{Simulation results with pseudo-dense $\beta_0^{pd}$ and sparse $x^{(2)}$ for the Toeplitz (or autoregressive) covariance matrix $\Sigma^{\mathrm{ar}}$ under Gaussian noise $\mathcal{N}(0,1)$ and the MAR setting \eqref{MAR_correct} in \autoref{subsec:sim_design}. {\bf Top Left:} Boxplots of the absolute bias. {\bf Top Right:} Coverage probabilities with standard error bars. {\bf Bottom Left:} Average lengths of confidence intervals with standard error bars. {\bf Bottom Right:} Four comparative ``QQ-plots'' of the transformed studentized debiased estimators $\Phi\left(\sqrt{n} \cdot\hat{\sigma}_n^{-1}(x) \left[\hat{m}(x) -m_0(x)\right] \right)$ obtained from different debiasing methods, where $\hat{\sigma}_n(x)^2$ is the estimated (asymptotic) variance of $\hat{m}(x)$ by each method and $\Phi(\cdot)$ is the CDF of $\mathcal{N}(0,1)$. We also highlight the recommended rules ``1SE'' and ``min-feas'' for our proposed debiasing method via dashed rectangles in the first three panels.}
	\label{fig:AR_MAR_gauss}
\end{figure}

\begin{figure}[h]
	\captionsetup[subfigure]{justification=centering}
	\begin{subfigure}[c]{0.995\linewidth}
		\centering		\includegraphics[width=1\linewidth]{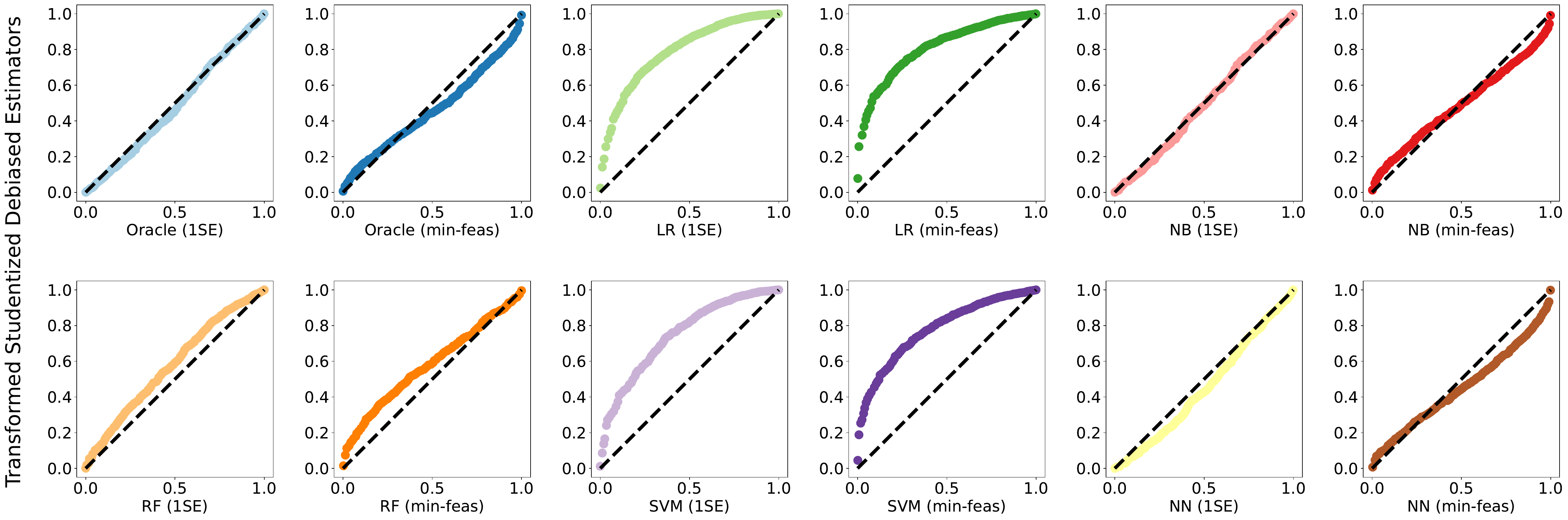}
		\caption{Simulation results with dense $\beta_0^{de}$ and sparse $x^{(2)}$.}
	\end{subfigure}
	\begin{subfigure}[c]{0.995\linewidth}
		\centering		\includegraphics[width=1\linewidth]{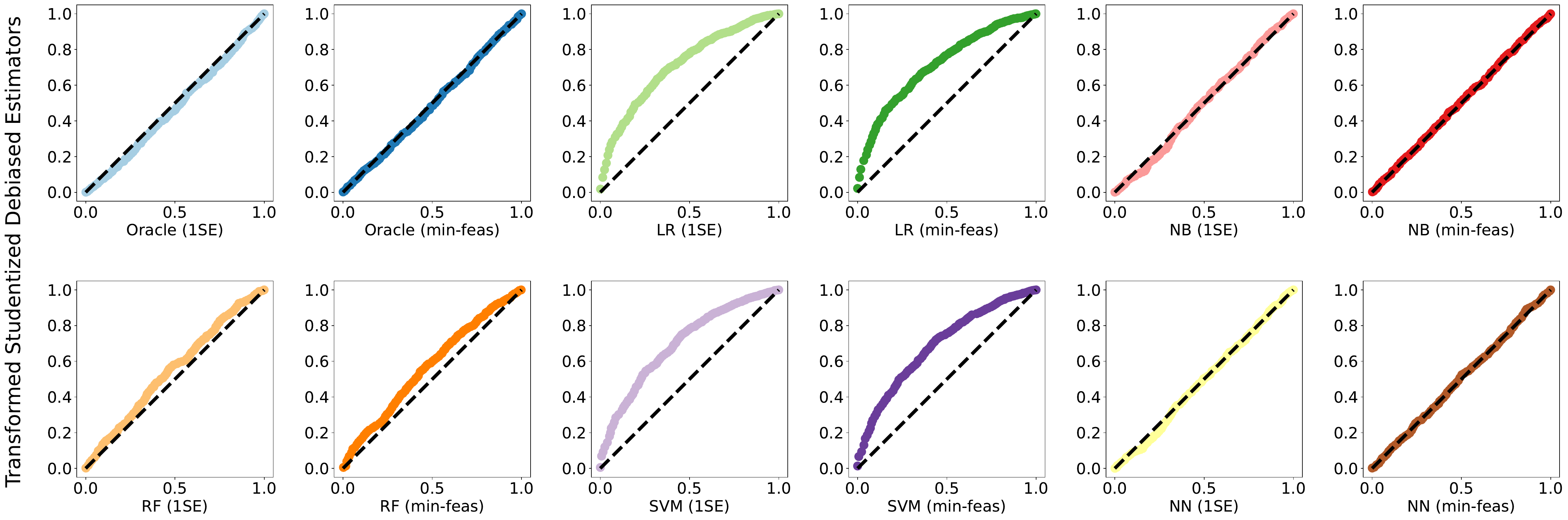}
		\caption{Simulation results with sparse $\beta_0^{sp}$ and dense $x^{(4)}$.}
	\end{subfigure}
	\caption{``QQ-plots'' of the transformed studentized debiased estimators $\Phi\left(\sqrt{n} \cdot\hat{\sigma}_n^{-1}(x) \left[\hat{m}(x) -m_0(x)\right] \right)$ obtained from our debiasing method when the propensity scores are estimated by various machine learning methods under the new MAR setting \eqref{MAR_mis}, where $\Phi(\cdot)$ is the CDF of $\mathcal{N}(0,1)$.}
	\label{fig:Cir_Nonpar_Prop_qqplot}
\end{figure}

\begin{table}[hbtp!]
\begin{center}
\resizebox{1.01\textwidth}{!}{	
	\begin{tabular}{ccccccccccccccccccc}
		\hline \hline 
		& \begin{tabular}{@{}c@{}} DL-Jav \\ (CC)\end{tabular} & \begin{tabular}{@{}c@{}} DL-Jav \\ (IPW)\end{tabular} & \begin{tabular}{@{}c@{}} DL-Jav \\ (Oracle)\end{tabular} & \begin{tabular}{@{}c@{}} DL-vdG \\ (CC)\end{tabular} & \begin{tabular}{@{}c@{}} DL-vdG \\ (IPW)\end{tabular} & \begin{tabular}{@{}c@{}} DL-vdG \\ (Oracle)\end{tabular} & \begin{tabular}{@{}c@{}} R-Proj \\ (CC)\end{tabular} & \begin{tabular}{@{}c@{}} R-Proj \\ (IPW)\end{tabular} & \begin{tabular}{@{}c@{}} R-Proj \\ (Oracle)\end{tabular} & \begin{tabular}{@{}c@{}} Refit \\ (CC)\end{tabular} & \begin{tabular}{@{}c@{}} Refit \\ (IPW)\end{tabular} & \begin{tabular}{@{}c@{}} Refit \\ (Oracle)\end{tabular} & \begin{tabular}{@{}c@{}} DDR \end{tabular} & \begin{tabular}{@{}c@{}} SAS \end{tabular} & \begin{tabular}{@{}c@{}} SSL-AIPW \end{tabular} & \begin{tabular}{@{}c@{}} \textbf{Debias} \\ (min-CV)\end{tabular} & \begin{tabular}{@{}c@{}} \textbf{Debias} \\ (1SE)\end{tabular} & \begin{tabular}{@{}c@{}} \textbf{Debias} \\ (min-feas)\end{tabular} \\ \cline{2-19}
		& \multicolumn{18}{c}{\bf $d=1000$, $n = 900$, sparse $\beta_0^{sp}$, and $x^{(1)}$ so that $m(x^{(1)})=x^{(1)T}\beta_0^{sp} = \beta_{0,1}^{(a)}$} \\ \cline{2-19}
		Avg-Bias & 0.033 & 0.035 & 0.028 & 0.040 & 0.038 & 0.032 & 0.249 & 0.198 & 0.315 & 0.032 & 0.036 & 0.027 & 0.039 & 0.038 & 0.188 & 0.042 & 0.039 & 0.038 \\
		Avg-Cov & 0.962 & 0.962 & 0.967 & 0.899 & 0.882 & 0.904 & 0.962 & 0.958 & 0.940 & 0.954 & 0.906 & 0.953 & 0.910 & 0.899 & 0.482 & 0.853 & 0.898 & 0.914 \\
		Avg-Len & 0.176 & 0.184 & 0.145 & 0.161 & 0.152 & 0.133 & 0.369 & 0.305 & 0.638 & 0.162 & 0.156 & 0.133 & 0.166 & 0.154 & 10.887 & 0.151 & 0.161 & 0.168 \\
		\cline{2-19}
		
		& \multicolumn{18}{c}{\bf $d=1000$, $n = 900$, dense $\beta_0^{de}$, and $x^{(1)}$ so that $m(x^{(1)})=x^{(1)T}\beta_0^{de} = \beta_{0,1}^{(b)}$} \\ \cline{2-19}
		Avg-Bias & 0.146 & 0.154 & 0.108 & 0.181 & 0.158 & 0.136 & 0.098 & 0.104 & 0.207 & 0.184 & 0.193 & 0.143 & 0.161 & 0.152 & 0.222 & 0.224 & 0.191 & 0.172 \\
		Avg-Cov & 0.995 & 0.992 & 0.996 & 0.891 & 0.855 & 0.913 & 1.00 & 0.998 & 0.996 & 0.922 & 0.838 & 0.919 & 0.357 & 0.906 & 0.497 & 0.860 & 0.931 & 0.969 \\
		Avg-Len & 0.981 & 0.998 & 0.844 & 0.757 & 0.593 & 0.583 & 1.253 & 1.042 & 1.784 & 0.791 & 0.665 & 0.610 & 0.185 & 0.627 & 7.805 & 0.819 & 0.871 & 0.908 \\
		\cline{2-19}
		
		& \multicolumn{18}{c}{\bf $d=1000$, $n = 900$, pseudo-dense $\beta_0^{pd}$, and $x^{(1)}$ so that $m(x^{(1)})=x^{(1)T}\beta_0^{pd} = \beta_{0,1}^{(c)}$} \\ \cline{2-19}
		Avg-Bias & 0.044 & 0.049 & 0.036 & 0.054 & 0.049 & 0.041 & 0.068 & 0.057 & 0.088 & 0.049 & 0.050 & 0.037 & 0.048 & 0.052 & 0.379 & 0.065 & 0.059 & 0.055 \\
		Avg-Cov & 0.991 & 0.989 & 0.996 & 0.891 & 0.868 & 0.908 & 0.987 & 0.983 & 0.985 & 0.930 & 0.897 & 0.938 & 0.898 & 0.914 & 0.466 & 0.857 & 0.913 & 0.949 \\
		Avg-Len & 0.284 & 0.292 & 0.237 & 0.223 & 0.190 & 0.173 & 0.404 & 0.351 & 0.527 & 0.222 & 0.200 & 0.172 & 0.197 & 0.220 & 7.729 & 0.241 & 0.256 & 0.267 \\
		\hline
		
		& \multicolumn{18}{c}{\bf $d=1000$, $n = 900$, sparse $\beta_0^{sp}$, and $x^{(2)}$ so that $m(x^{(2)})=x^{(2)T}\beta_0^{sp}$} \\ \cline{2-19}
		Avg-Bias & 0.042 & 0.042 & 0.037 & 0.069 & 0.056 & 0.051 & NA & NA & NA & 0.035 & 0.042 & 0.030 & 0.062 & 0.038 & 0.286 & 0.039 & 0.039 & 0.042 \\
		Avg-Cov & 0.985 & 0.995 & 0.990 & 0.807 & 0.842 & 0.839 & NA & NA & NA & 0.946 & 0.898 & 0.948 & 0.862 & 0.914 & 0.497 & 0.929 & 0.937 & 0.938 \\
		Avg-Len & 0.262 & 0.312 & 0.236 & 0.215 & 0.202 & 0.172 & NA & NA & NA & 0.175 & 0.170 & 0.144 & 0.218 & 0.169 & 14.284 & 0.174 & 0.183 & 0.201 \\
		\cline{2-19}
		
		& \multicolumn{18}{c}{\bf $d=1000$, $n = 900$, dense $\beta_0^{de}$, and $x^{(2)}$ so that $m(x^{(2)})=x^{(2)T}\beta_0^{de}$} \\ \cline{2-19}
		Avg-Bias & 0.171 & 0.186 & 0.130 & 0.344 & 0.215 & 0.226 & NA & NA & NA & 0.252 & 0.297 & 0.184 & 0.211 & 0.167 & 0.368 & 0.217 & 0.207 & 0.183 \\
		Avg-Cov & 1.00 & 1.00 & 1.00 & 0.775 & 0.851 & 0.821 & NA & NA & NA & 0.844 & 0.704 & 0.887 & 0.344 & 0.902 & 0.477 & 0.921 & 0.940 & 0.990 \\
		Avg-Len & 1.659 & 1.860 & 1.514 & 1.008 & 0.789 & 0.755 & NA & NA & NA & 0.931 & 0.788 & 0.726 & 0.241 & 0.690 & 9.101 & 0.939 & 0.989 & 1.089 \\
		\cline{2-19}
		
		& \multicolumn{18}{c}{\bf $d=1000$, $n = 900$, pseudo-dense $\beta_0^{pd}$, and $x^{(2)}$ so that $m(x^{(2)})=x^{(2)T}\beta_0^{pd}$} \\ \cline{2-19}
		Avg-Bias & 0.054 & 0.056 & 0.046 & 0.103 & 0.073 & 0.069 & NA & NA & NA & 0.061 & 0.067 & 0.046 & 0.063 & 0.057 & 0.506 & 0.067 & 0.065 & 0.060 \\
		Avg-Cov & 1.00 & 1.00 & 1.00 & 0.765 & 0.826 & 0.806 & NA & NA & NA & 0.935 & 0.842 & 0.928 & 0.904 & 0.900 & 0.480 & 0.904 & 0.927 & 0.965 \\
		Avg-Len & 0.467 & 0.533 & 0.414 & 0.297 & 0.253 & 0.224 & NA & NA & NA & 0.266 & 0.239 & 0.205 & 0.258 & 0.243 & 10.243 & 0.277 & 0.291 & 0.321 \\
		\hline
		
		& \multicolumn{18}{c}{\bf $d=1000$, $n = 900$, sparse $\beta_0^{sp}$, and $x^{(3)}$ so that $m(x^{(3)})=x^{(3)T}\beta_0^{sp} = \beta_{0,100}^{(a)}$} \\ \cline{2-19}
		Avg-Bias & 0.025 & 0.026 & 0.022 & 0.032 & 0.032 & 0.027 & 0.740 & 0.556 & 0.739 & 0.000 & 0.001 & 0.000 & 0.034 & 0.032 & 0.101 & 0.031 & 0.033 & 0.035 \\
		Avg-Cov & 0.991 & 0.992 & 0.990 & 0.952 & 0.945 & 0.950 & 0.960 & 0.963 & 0.937 & 1.00 & 0.992 & 1.00 & 0.947 & 0.941 & 0.491 & 0.949 & 0.944 & 0.944 \\
		Avg-Len & 0.176 & 0.185 & 0.145 & 0.162 & 0.152 & 0.133 & 6.590 & 1.483 & 1.177 & 0.000 & 0.001 & 0.000 & 0.163 & 0.151 & 12.879 & 0.151 & 0.161 & 0.168 \\
		\cline{2-19}
		
		& \multicolumn{18}{c}{\bf $d=1000$, $n = 900$, dense $\beta_0^{de}$, and $x^{(3)}$ so that $m(x^{(3)})=x^{(3)T}\beta_0^{de} = \beta_{0,100}^{(b)}$} \\ \cline{2-19}
		Avg-Bias & 0.102 & 0.102 & 0.090 & 0.169 & 0.132 & 0.128 & 0.093 & 0.096 & 0.084 & 0.186 & 0.193 & 0.191 & 0.139 & 0.143 & 0.109 & 0.192 & 0.174 & 0.161 \\
		Avg-Cov & 0.999 & 1.00 & 0.999 & 0.922 & 0.909 & 0.929 & 1.00 & 0.997 & 0.997 & 0.000 & 0.050 & 0.004 & 0.364 & 0.917 & 0.516 & 0.901 & 0.948 & 0.972 \\
		Avg-Len & 0.982 & 0.999 & 0.844 & 0.758 & 0.594 & 0.583 & 1.256 & 1.073 & 1.728 & 0.006 & 0.057 & 0.021 & 0.184 & 0.623 & 6.645 & 0.819 & 0.870 & 0.910 \\
		\cline{2-19}
		
		& \multicolumn{18}{c}{\bf $d=1000$, $n = 900$, pseudo-dense $\beta_0^{pd}$, and $x^{(3)}$ so that $m(x^{(3)})=x^{(3)T}\beta_0^{pd} = \beta_{0,100}^{(c)}$} \\ \cline{2-19}
		Avg-Bias & 0.031 & 0.031 & 0.028 & 0.047 & 0.042 & 0.038 & 0.122 & 0.561 & 0.316 & 0.039 & 0.042 & 0.040 & 0.041 & 0.049 & 0.084 & 0.053 & 0.052 & 0.051 \\
		Avg-Cov & 0.996 & 0.996 & 0.999 & 0.947 & 0.925 & 0.942 & 0.996 & 0.987 & 0.987 & 0.000 & 0.019 & 0.000 & 0.938 & 0.920 & 0.487 & 0.929 & 0.948 & 0.966 \\
		Avg-Len & 0.285 & 0.293 & 0.237 & 0.223 & 0.190 & 0.173 & 1.567 & 1.857 & 28.765 & 0.000 & 0.008 & 0.002 & 0.195 & 0.217 & 8.732 & 0.241 & 0.256 & 0.268 \\
		\hline
		
		& \multicolumn{18}{c}{\bf $d=1000$, $n = 900$, sparse $\beta_0^{sp}$, and $x^{(4)}$ so that $m(x^{(4)})=x^{(4)T}\beta_0^{sp}$} \\ \cline{2-19}
		Avg-Bias & 0.043 & 0.042 & 0.037 & 0.078 & 0.060 & 0.056 & NA & NA & NA & 0.036 & 0.043 & 0.030 & 0.068 & 0.058 & 0.388 & 0.052 & 0.045 & 0.041 \\
		Avg-Cov & 1.00 & 1.00 & 1.00 & 0.800 & 0.864 & 0.827 & NA & NA & NA & 0.949 & 0.880 & 0.945 & 0.868 & 0.721 & 0.511 & 0.788 & 0.881 & 0.934 \\
		Avg-Len & 0.541 & 0.630 & 0.449 & 0.232 & 0.221 & 0.183 & NA & NA & NA & 0.175 & 0.169 & 0.144 & 0.235 & 0.160 & 15.971 & 0.162 & 0.171 & 0.186 \\
		\cline{2-19}
		
		& \multicolumn{18}{c}{\bf $d=1000$, $n = 900$, dense $\beta_0^{de}$, and $x^{(4)}$ so that $m(x^{(4)})=x^{(4)T}\beta_0^{de}$} \\ \cline{2-19}
		Avg-Bias & 0.165 & 0.168 & 0.145 & 0.892 & 0.436 & 0.590 & NA & NA & NA & 0.800 & 0.799 & 0.620 & 0.276 & 0.629 & 0.623 & 1.259 & 1.062 & 0.762 \\
		Avg-Cov & 1.00 & 1.00 & 1.00 & 0.067 & 0.501 & 0.156 & NA & NA & NA & 0.082 & 0.040 & 0.074 & 0.265 & 0.055 & 0.477 & 0.000 & 0.009 & 0.117 \\
		Avg-Len & 3.992 & 4.368 & 3.449 & 1.088 & 0.861 & 0.806 & NA & NA & NA & 0.863 & 0.730 & 0.669 & 0.263 & 0.646 & 12.970 & 0.878 & 0.925 & 1.004 \\
		\cline{2-19}
		
		& \multicolumn{18}{c}{\bf $d=1000$, $n = 900$, pseudo-dense $\beta_0^{pd}$, and $x^{(4)}$ so that $m(x^{(4)})=x^{(4)T}\beta_0^{pd}$} \\ \cline{2-19}
		Avg-Bias & 0.078 & 0.063 & 0.082 & 0.250 & 0.156 & 0.170 & NA & NA & NA & 0.150 & 0.149 & 0.110 & 0.099 & 0.268 & 0.582 & 0.364 & 0.303 & 0.216 \\
		Avg-Cov & 1.00 & 1.00 & 1.00 & 0.106 & 0.393 & 0.170 & NA & NA & NA & 0.334 & 0.286 & 0.398 & 0.785 & 0.009 & 0.462 & 0.000 & 0.017 & 0.166 \\
		Avg-Len & 1.090 & 1.211 & 0.907 & 0.320 & 0.276 & 0.239 & NA & NA & NA & 0.244 & 0.220 & 0.189 & 0.291 & 0.228 & 11.660 & 0.259 & 0.273 & 0.296 \\
		\hline
		
		& \multicolumn{18}{c}{\bf $d=1000$, $n = 900$, sparse $\beta_0^{sp}$, and $x^{(5)}$ so that $m(x^{(5)})=x^{(5)T}\beta_0^{sp}$} \\ \cline{2-19}
		Avg-Bias & 0.036 & 0.036 & 0.031 & 0.043 & 0.043 & 0.032 & NA & NA & NA & 0.003 & 0.008 & 0.002 & 0.040 & 0.016 & 0.076 & 0.018 & 0.017 & 0.020 \\
		Avg-Cov & 1.00 & 1.00 & 1.00 & 0.936 & 0.937 & 0.939 & NA & NA & NA & 0.889 & 0.534 & 0.870 & 0.944 & 0.000 & 0.491 & 0.000 & 0.000 & 0.961 \\
		Avg-Len & 1.768 & 2.028 & 1.367 & 0.203 & 0.197 & 0.154 & NA & NA & NA & 0.010 & 0.015 & 0.008 & 0.190 & 0.000 & 7.220 & 0.000 & 0.006 & 0.104 \\
		\hline\hline
\end{tabular}}
\caption{Simulation results under the circulant symmetric covariance matrix $\Sigma^{\mathrm{cs}}$, Gaussian noise $\mathcal{N}(0,1)$, and the MAR setting \eqref{MAR_correct}. 
	Nearly all the standard errors for the above average quantities are less than 0.01 except for those associated with ``DL-vdG'', ``R-Proj'', or ``SSL-AIPW'' and are thus omitted.}
\label{table:cirsym_MAR_res}
\end{center}
\end{table}

\begin{table}[hbtp!]
\begin{center}
\resizebox{1.01\textwidth}{!}{	
	\begin{tabular}{ccccccccccccccccccc}
		\hline \hline 
		& \begin{tabular}{@{}c@{}} DL-Jav \\ (CC)\end{tabular} & \begin{tabular}{@{}c@{}} DL-Jav \\ (IPW)\end{tabular} & \begin{tabular}{@{}c@{}} DL-Jav \\ (Oracle)\end{tabular} & \begin{tabular}{@{}c@{}} DL-vdG \\ (CC)\end{tabular} & \begin{tabular}{@{}c@{}} DL-vdG \\ (IPW)\end{tabular} & \begin{tabular}{@{}c@{}} DL-vdG \\ (Oracle)\end{tabular} & \begin{tabular}{@{}c@{}} R-Proj \\ (CC)\end{tabular} & \begin{tabular}{@{}c@{}} R-Proj \\ (IPW)\end{tabular} & \begin{tabular}{@{}c@{}} R-Proj \\ (Oracle)\end{tabular} & \begin{tabular}{@{}c@{}} Refit \\ (CC)\end{tabular} & \begin{tabular}{@{}c@{}} Refit \\ (IPW)\end{tabular} & \begin{tabular}{@{}c@{}} Refit \\ (Oracle)\end{tabular} & DDR & SAS & \begin{tabular}{@{}c@{}} SSL \\ AIPW\end{tabular} & \begin{tabular}{@{}c@{}} \textbf{Debias} \\ (min-CV)\end{tabular} & \begin{tabular}{@{}c@{}} \textbf{Debias} \\ (1SE)\end{tabular} & \begin{tabular}{@{}c@{}} \textbf{Debias} \\ (min-feas)\end{tabular} \\ \cline{2-19}
		& \multicolumn{18}{c}{\bf $d=1000$, $n = 900$, sparse $\beta_0^{sp}$, and $x^{(1)}$ so that $m(x^{(1)})=x^{(1)T}\beta_0^{sp} = \beta_{0,1}^{(a)}$} \\ \cline{2-19}
		Avg-Bias & 0.073 & 0.073 & 0.061 & 0.073 & 0.073 & 0.061 & 0.579 & 0.481 & 0.608 & 0.073 & 0.074 & 0.060 & 0.073 & 0.074 & 0.216 & 0.074 & 0.075 & 0.089 \\
		Avg-Cov & 0.871 & 0.870 & 0.894 & 0.900 & 0.893 & 0.917 & 0.953 & 0.950 & 0.950 & 0.946 & 0.944 & 0.956 & 0.910 & 0.919 & 0.496 & 0.889 & 0.913 & 0.939 \\
		Avg-Len & 0.279 & 0.281 & 0.247 & 0.308 & 0.305 & 0.266 & 1.196 & 1.539 & 1.917 & 0.354 & 0.352 & 0.301 & 0.306 & 0.322 & 40.120 & 0.296 & 0.326 & 0.413 \\
		\cline{2-19}
		& \multicolumn{18}{c}{\bf $d=1000$, $n = 900$, dense $\beta_0^{de}$, and $x^{(1)}$ so that $m(x^{(1)})=x^{(1)T}\beta_0^{de} = \beta_{0,1}^{(b)}$} \\ \cline{2-19}
		Avg-Bias & 0.274 & 0.276 & 0.152 & 0.267 & 0.218 & 0.153 & 0.297 & 0.219 & 0.372 & 0.428 & 0.257 & 0.178 & 0.254 & 0.225 & 0.172 & 0.468 & 0.405 & 0.358 \\
		Avg-Cov & 0.757 & 0.762 & 0.927 & 0.571 & 0.427 & 0.559 & 0.943 & 0.800 & 0.945 & 0.781 & 0.748 & 0.824 & 0.630 & 0.957 & 0.486 & 0.845 & 0.915 & 0.987 \\
		Avg-Len & 0.828 & 0.842 & 0.671 & 0.528 & 0.306 & 0.301 & 1.008 & 0.617 & 0.833 & 1.157 & 0.758 & 0.593 & 0.576 & 1.190 & 16.077 & 1.651 & 1.819 & 2.306 \\
		\cline{2-19}
		& \multicolumn{18}{c}{\bf $d=1000$, $n = 900$, pseudo-dense $\beta_0^{pd}$, and $x^{(1)}$ so that $m(x^{(1)})=x^{(1)T}\beta_0^{pd} = \beta_{0,1}^{(c)}$} \\ \cline{2-19}
		Avg-Bias & 0.091 & 0.092 & 0.071 & 0.093 & 0.091 & 0.073 & 0.118 & 0.181 & 0.265 & 0.098 & 0.095 & 0.075 & 0.095 & 0.092 & 0.226 & 0.103 & 0.096 & 0.101 \\
		Avg-Cov & 0.870 & 0.863 & 0.906 & 0.847 & 0.799 & 0.863 & 0.948 & 0.935 & 0.925 & 0.912 & 0.903 & 0.921 & 0.852 & 0.953 & 0.499 & 0.874 & 0.926 & 0.961 \\
		Avg-Len & 0.340 & 0.341 & 0.300 & 0.329 & 0.295 & 0.275 & 0.611 & 0.593 & 1.490 & 0.416 & 0.388 & 0.335 & 0.339 & 0.457 & 28.389 & 0.390 & 0.429 & 0.544 \\
		\hline
		& \multicolumn{18}{c}{\bf $d=1000$, $n = 900$, sparse $\beta_0^{sp}$, and $x^{(2)}$ so that $m(x^{(2)})=x^{(2)T}\beta_0^{sp}$} \\ \cline{2-19}
		Avg-Bias & 0.057 & 0.057 & 0.047 & 0.068 & 0.065 & 0.055 & NA & NA & NA & 0.048 & 0.049 & 0.041 & 0.066 & 0.053 & 0.188 & 0.056 & 0.060 & 0.070 \\
		Avg-Cov & 0.982 & 0.990 & 0.994 & 0.910 & 0.913 & 0.927 & NA & NA & NA & 0.948 & 0.948 & 0.952 & 0.921 & 0.908 & 0.493 & 0.934 & 0.942 & 0.956 \\
		Avg-Len & 0.353 & 0.386 & 0.319 & 0.290 & 0.287 & 0.248 & NA & NA & NA & 0.237 & 0.237 & 0.202 & 0.290 & 0.217 & 58.341 & 0.255 & 0.287 & 0.348 \\
		\cline{2-19}
		& \multicolumn{18}{c}{\bf $d=1000$, $n = 900$, dense $\beta_0^{de}$, and $x^{(2)}$ so that $m(x^{(2)})=x^{(2)T}\beta_0^{de}$} \\ \cline{2-19}
		Avg-Bias & 0.235 & 0.241 & 0.126 & 0.239 & 0.199 & 0.137 & NA & NA & NA & 0.376 & 0.244 & 0.177 & 0.223 & 0.214 & 0.154 & 0.386 & 0.336 & 0.310 \\
		Avg-Cov & 0.987 & 0.987 & 0.999 & 0.577 & 0.429 & 0.577 & NA & NA & NA & 0.698 & 0.693 & 0.750 & 0.668 & 0.873 & 0.492 & 0.852 & 0.939 & 0.981 \\
		Avg-Len & 1.367 & 1.457 & 1.071 & 0.497 & 0.288 & 0.280 & NA & NA & NA & 0.996 & 0.650 & 0.519 & 0.546 & 0.802 & 17.531 & 1.426 & 1.603 & 1.943 \\
		\cline{2-19}
		& \multicolumn{18}{c}{\bf $d=1000$, $n = 900$, pseudo-dense $\beta_0^{pd}$, and $x^{(2)}$ so that $m(x^{(2)})=x^{(2)T}\beta_0^{pd}$} \\ \cline{2-19}
		Avg-Bias & 0.082 & 0.083 & 0.065 & 0.090 & 0.085 & 0.070 & NA & NA & NA & 0.086 & 0.083 & 0.067 & 0.085 & 0.088 & 0.193 & 0.091 & 0.085 & 0.087 \\
		Avg-Cov & 0.988 & 0.992 & 0.998 & 0.836 & 0.807 & 0.860 & NA & NA & NA & 0.918 & 0.906 & 0.921 & 0.856 & 0.852 & 0.460 & 0.872 & 0.933 & 0.962 \\
		Avg-Len & 0.521 & 0.555 & 0.456 & 0.310 & 0.278 & 0.256 & NA & NA & NA & 0.375 & 0.349 & 0.301 & 0.321 & 0.309 & 29.895 & 0.336 & 0.378 & 0.459 \\
		\hline
		& \multicolumn{18}{c}{\bf $d=1000$, $n = 900$, sparse $\beta_0^{sp}$, and $x^{(3)}$ so that $m(x^{(3)})=x^{(3)T}\beta_0^{sp} = \beta_{0,100}^{(a)}$} \\ \cline{2-19}
		Avg-Bias & 0.057 & 0.057 & 0.054 & 0.076 & 0.076 & 0.068 & 0.597 & 0.926 & 0.733 & 0.000 & 0.000 & 0.000 & 0.078 & 0.096 & 0.264 & 0.080 & 0.091 & 0.101 \\
		Avg-Cov & 0.984 & 0.986 & 0.978 & 0.948 & 0.943 & 0.939 & 0.971 & 0.970 & 0.958 & 1.000 & 0.999 & 1.000 & 0.950 & 0.953 & 0.479 & 0.947 & 0.947 & 0.943 \\
		Avg-Len & 0.348 & 0.350 & 0.312 & 0.367 & 0.364 & 0.322 & 1.667 & 11.133 & 8.075 & 0.000 & 0.000 & 0.000 & 0.385 & 0.463 & 38.258 & 0.373 & 0.428 & 0.480 \\
		\cline{2-19}
		& \multicolumn{18}{c}{\bf $d=1000$, $n = 900$, dense $\beta_0^{de}$, and $x^{(3)}$ so that $m(x^{(3)})=x^{(3)T}\beta_0^{de} = \beta_{0,100}^{(b)}$} \\ \cline{2-19}
		Avg-Bias & 0.173 & 0.173 & 0.130 & 0.224 & 0.190 & 0.167 & 0.309 & 0.239 & 0.279 & 0.260 & 0.223 & 0.200 & 0.231 & 0.304 & 0.186 & 0.399 & 0.372 & 0.390 \\
		Avg-Cov & 0.978 & 0.981 & 0.990 & 0.782 & 0.529 & 0.604 & 0.989 & 0.888 & 0.954 & 0.251 & 0.349 & 0.391 & 0.827 & 0.976 & 0.507 & 0.960 & 0.986 & 0.992 \\
		Avg-Len & 1.055 & 1.072 & 0.856 & 0.629 & 0.365 & 0.364 & 1.304 & 0.792 & 1.097 & 0.411 & 0.376 & 0.331 & 0.724 & 1.718 & 40.498 & 2.084 & 2.392 & 2.681 \\
		\cline{2-19}
		& \multicolumn{18}{c}{\bf $d=1000$, $n = 900$, pseudo-dense $\beta_0^{pd}$, and $x^{(3)}$ so that $m(x^{(3)})=x^{(3)T}\beta_0^{pd} = \beta_{0,100}^{(c)}$} \\ \cline{2-19}
		Avg-Bias & 0.072 & 0.072 & 0.067 & 0.094 & 0.088 & 0.082 & 0.421 & 0.244 & 0.394 & 0.059 & 0.058 & 0.056 & 0.096 & 0.128 & 0.238 & 0.103 & 0.107 & 0.115 \\
		Avg-Cov & 0.973 & 0.973 & 0.969 & 0.892 & 0.884 & 0.890 & 0.962 & 0.946 & 0.963 & 0.302 & 0.341 & 0.377 & 0.923 & 0.968 & 0.523 & 0.943 & 0.960 & 0.969 \\
		Avg-Len & 0.431 & 0.432 & 0.381 & 0.392 & 0.352 & 0.333 & 2.658 & 0.774 & 1.953 & 0.136 & 0.146 & 0.143 & 0.427 & 0.660 & 57.514 & 0.492 & 0.564 & 0.633 \\
		\hline
		& \multicolumn{18}{c}{\bf $d=1000$, $n = 900$, sparse $\beta_0^{sp}$, and $x^{(4)}$ so that $m(x^{(4)})=x^{(4)T}\beta_0^{sp}$} \\ \cline{2-19}
		Avg-Bias & 0.050 & 0.050 & 0.040 & 0.069 & 0.063 & 0.054 & NA & NA & NA & 0.044 & 0.044 & 0.037 & 0.068 & 0.044 & 0.134 & 0.044 & 0.045 & 0.047 \\
		Avg-Cov & 1.000 & 1.000 & 1.000 & 0.850 & 0.890 & 0.873 & NA & NA & NA & 0.955 & 0.952 & 0.960 & 0.854 & 0.855 & 0.494 & 0.906 & 0.932 & 0.941 \\
		Avg-Len & 0.718 & 0.750 & 0.568 & 0.244 & 0.242 & 0.203 & NA & NA & NA & 0.214 & 0.213 & 0.182 & 0.236 & 0.158 & 31.767 & 0.180 & 0.208 & 0.222 \\
		\cline{2-19}
		& \multicolumn{18}{c}{\bf $d=1000$, $n = 900$, dense $\beta_0^{de}$, and $x^{(4)}$ so that $m(x^{(4)})=x^{(4)T}\beta_0^{de}$} \\ \cline{2-19}
		Avg-Bias & 0.166 & 0.170 & 0.089 & 0.167 & 0.126 & 0.090 & NA & NA & NA & 0.343 & 0.218 & 0.159 & 0.159 & 0.156 & 0.119 & 0.408 & 0.295 & 0.254 \\
		Avg-Cov & 1.000 & 1.000 & 1.000 & 0.682 & 0.555 & 0.684 & NA & NA & NA & 0.626 & 0.583 & 0.613 & 0.742 & 0.851 & 0.487 & 0.664 & 0.866 & 0.945 \\
		Avg-Len & 3.585 & 3.681 & 2.527 & 0.417 & 0.242 & 0.229 & NA & NA & NA & 0.728 & 0.460 & 0.359 & 0.447 & 0.582 & 12.082 & 1.004 & 1.161 & 1.238 \\
		\cline{2-19}
		& \multicolumn{18}{c}{\bf $d=1000$, $n = 900$, pseudo-dense $\beta_0^{pd}$, and $x^{(4)}$ so that $m(x^{(4)})=x^{(4)T}\beta_0^{pd}$} \\ \cline{2-19}
		Avg-Bias & 0.068 & 0.068 & 0.054 & 0.138 & 0.105 & 0.103 & NA & NA & NA & 0.064 & 0.061 & 0.050 & 0.114 & 0.063 & 0.146 & 0.075 & 0.062 & 0.059 \\
		Avg-Cov & 1.000 & 1.000 & 1.000 & 0.458 & 0.579 & 0.528 & NA & NA & NA & 0.887 & 0.880 & 0.886 & 0.622 & 0.852 & 0.482 & 0.793 & 0.906 & 0.941 \\
		Avg-Len & 1.289 & 1.321 & 1.020 & 0.260 & 0.234 & 0.210 & NA & NA & NA & 0.252 & 0.235 & 0.203 & 0.264 & 0.224 & 20.412 & 0.237 & 0.274 & 0.292 \\
		
		\hline
		& \multicolumn{18}{c}{\bf $d=1000$, $n = 900$, sparse $\beta_0^{sp}$, and $x^{(5)}$ so that $m(x^{(5)})=x^{(5)T}\beta_0^{sp}$} \\ \cline{2-19}
		Avg-Bias & 0.027 & 0.027 & 0.020 & 0.037 & 0.035 & 0.029 & NA & NA & NA & 0.001 & 0.003 & 0.001 & 0.033 & 0.005 & 0.023 & 0.006 & 0.007 & 0.009 \\
		Avg-Cov & 1.000 & 1.000 & 1.000 & 0.946 & 0.955 & 0.951 & NA & NA & NA & 0.892 & 0.630 & 0.889 & 0.950 & 0.000 & 0.519 & 0.896 & 0.903 & 0.960 \\
		Avg-Len & 2.545 & 2.584 & 1.879 & 0.183 & 0.182 & 0.142 & NA & NA & NA & 0.006 & 0.007 & 0.005 & 0.165 & 0.000 & 7.839 & 0.030 & 0.034 & 0.044 \\
		
		\hline\hline
\end{tabular}}
\caption{Simulation results under the Toeplitz (or autoregressive) covariance matrix $\Sigma^{\mathrm{ar}}$, Gaussian noise $\mathcal{N}(0,1)$, and the MAR setting \eqref{MAR_correct}. Nearly all the standard errors for the above average quantities are less than 0.01 except for those associated with ``R-Proj'' or ``SSL-AIPW'' and are thus omitted.}
\label{table:AR_MAR_res}
\end{center}
\end{table}

\begin{table}[hbtp!]
\begin{center}
\resizebox{1.01\textwidth}{!}{	
	\begin{tabular}{ccccccccccccccccccc}
		\hline \hline 
		& \begin{tabular}{@{}c@{}} DL-Jav \\ (CC)\end{tabular} & \begin{tabular}{@{}c@{}} DL-Jav \\ (IPW)\end{tabular} & \begin{tabular}{@{}c@{}} DL-Jav \\ (Oracle)\end{tabular} & \begin{tabular}{@{}c@{}} DL-vdG \\ (CC)\end{tabular} & \begin{tabular}{@{}c@{}} DL-vdG \\ (IPW)\end{tabular} & \begin{tabular}{@{}c@{}} DL-vdG \\ (Oracle)\end{tabular} & \begin{tabular}{@{}c@{}} R-Proj \\ (CC)\end{tabular} & \begin{tabular}{@{}c@{}} R-Proj \\ (IPW)\end{tabular} & \begin{tabular}{@{}c@{}} R-Proj \\ (Oracle)\end{tabular} & \begin{tabular}{@{}c@{}} Refit \\ (CC)\end{tabular} & \begin{tabular}{@{}c@{}} Refit \\ (IPW)\end{tabular} & \begin{tabular}{@{}c@{}} Refit \\ (Oracle)\end{tabular} & \begin{tabular}{@{}c@{}} DDR \end{tabular} & \begin{tabular}{@{}c@{}} SAS \end{tabular} & \begin{tabular}{@{}c@{}} SSL-AIPW \end{tabular} & \begin{tabular}{@{}c@{}} \textbf{Debias} \\ (min-CV)\end{tabular} & \begin{tabular}{@{}c@{}} \textbf{Debias} \\ (1SE)\end{tabular} & \begin{tabular}{@{}c@{}} \textbf{Debias} \\ (min-feas)\end{tabular} \\ \cline{2-19}
		& \multicolumn{18}{c}{\bf $d=1000$, $n = 900$, sparse $\beta_0^{sp}$, and $x^{(1)}$ so that $m(x^{(1)})=x^{(1)T}\beta_0^{sp} = \beta_{0,1}^{(a)}$} \\ \cline{2-19}
		Avg-Bias & 0.033 & 0.036 & 0.028 & 0.039 & 0.038 & 0.031 & 0.070 & 0.063 & 0.492 & 0.034 & 0.039 & 0.028 & 0.039 & 0.039 & 0.189 & 0.043 & 0.040 & 0.039 \\
		Avg-Cov & 0.969 & 0.961 & 0.964 & 0.905 & 0.881 & 0.924 & 0.969 & 0.969 & 0.928 & 0.945 & 0.898 & 0.948 & 0.911 & 0.895 & 0.485 & 0.859 & 0.901 & 0.911 \\
		Avg-Len & 0.176 & 0.185 & 0.145 & 0.162 & 0.152 & 0.133 & 0.334 & 0.329 & 0.926 & 0.162 & 0.156 & 0.133 & 0.166 & 0.154 & 6.343 & 0.151 & 0.161 & 0.167 \\
		\cline{2-19}
		
		& \multicolumn{18}{c}{\bf $d=1000$, $n = 900$, dense $\beta_0^{de}$, and $x^{(1)}$ so that $m(x^{(1)})=x^{(1)T}\beta_0^{de} = \beta_{0,1}^{(b)}$} \\ \cline{2-19}
		Avg-Bias & 0.148 & 0.159 & 0.108 & 0.183 & 0.163 & 0.133 & 0.093 & 0.207 & 0.088 & 0.178 & 0.193 & 0.143 & 0.160 & 0.148 & 0.223 & 0.227 & 0.198 & 0.180 \\
		Avg-Cov & 0.994 & 0.991 & 0.997 & 0.892 & 0.851 & 0.917 & 0.998 & 0.999 & 0.996 & 0.918 & 0.821 & 0.909 & 0.372 & 0.914 & 0.511 & 0.846 & 0.917 & 0.958 \\
		Avg-Len & 0.982 & 0.998 & 0.844 & 0.758 & 0.595 & 0.582 & 1.250 & 1.304 & 1.729 & 0.792 & 0.666 & 0.610 & 0.185 & 0.628 & 6.032 & 0.821 & 0.873 & 0.908 \\
		\cline{2-19}
		
		& \multicolumn{18}{c}{\bf $d=1000$, $n = 900$, pseudo-dense $\beta_0^{pd}$, and $x^{(1)}$ so that $m(x^{(1)})=x^{(1)T}\beta_0^{pd} = \beta_{0,1}^{(c)}$} \\ \cline{2-19}
		Avg-Bias & 0.045 & 0.049 & 0.035 & 0.053 & 0.050 & 0.040 & 0.056 & 0.078 & 0.083 & 0.050 & 0.053 & 0.038 & 0.049 & 0.053 & 0.378 & 0.063 & 0.058 & 0.055 \\
		Avg-Cov & 0.988 & 0.981 & 0.993 & 0.897 & 0.859 & 0.906 & 0.992 & 0.981 & 0.980 & 0.926 & 0.858 & 0.928 & 0.890 & 0.906 & 0.472 & 0.868 & 0.918 & 0.945 \\
		Avg-Len & 0.285 & 0.293 & 0.236 & 0.223 & 0.190 & 0.173 & 0.387 & 0.362 & 0.543 & 0.223 & 0.199 & 0.171 & 0.197 & 0.219 & 7.248 & 0.241 & 0.256 & 0.267 \\
		\hline
		
		& \multicolumn{18}{c}{\bf $d=1000$, $n = 900$, sparse $\beta_0^{sp}$, and $x^{(2)}$ so that $m(x^{(2)})=x^{(2)T}\beta_0^{sp}$} \\ \cline{2-19}
		Avg-Bias & 0.041 & 0.041 & 0.036 & 0.068 & 0.055 & 0.050 & NA & NA & NA & 0.036 & 0.043 & 0.029 & 0.060 & 0.039 & 0.288 & 0.041 & 0.041 & 0.045 \\
		Avg-Cov & 0.983 & 0.995 & 0.991 & 0.827 & 0.867 & 0.854 & NA & NA & NA & 0.945 & 0.888 & 0.944 & 0.875 & 0.912 & 0.491 & 0.910 & 0.922 & 0.930 \\
		Avg-Len & 0.262 & 0.312 & 0.236 & 0.215 & 0.202 & 0.172 & NA & NA & NA & 0.175 & 0.170 & 0.144 & 0.216 & 0.170 & 10.974 & 0.173 & 0.183 & 0.201 \\
		\cline{2-19}
		
		& \multicolumn{18}{c}{\bf $d=1000$, $n = 900$, dense $\beta_0^{de}$, and $x^{(2)}$ so that $m(x^{(2)})=x^{(2)T}\beta_0^{de}$} \\ \cline{2-19}
		Avg-Bias & 0.179 & 0.195 & 0.130 & 0.352 & 0.230 & 0.225 & NA & NA & NA & 0.244 & 0.284 & 0.184 & 0.216 & 0.168 & 0.364 & 0.226 & 0.217 & 0.194 \\
		Avg-Cov & 1.00 & 1.00 & 1.00 & 0.766 & 0.831 & 0.828 & NA & NA & NA & 0.861 & 0.716 & 0.877 & 0.342 & 0.903 & 0.505 & 0.904 & 0.933 & 0.978 \\
		Avg-Len & 1.659 & 1.860 & 1.513 & 1.008 & 0.790 & 0.754 & NA & NA & NA & 0.935 & 0.791 & 0.727 & 0.240 & 0.692 & 8.013 & 0.941 & 0.992 & 1.091 \\
		\cline{2-19}
		
		& \multicolumn{18}{c}{\bf $d=1000$, $n = 900$, pseudo-dense $\beta_0^{pd}$, and $x^{(2)}$ so that $m(x^{(2)})=x^{(2)T}\beta_0^{pd}$} \\ \cline{2-19}
		Avg-Bias & 0.054 & 0.056 & 0.045 & 0.102 & 0.072 & 0.068 & NA & NA & NA & 0.061 & 0.070 & 0.047 & 0.063 & 0.059 & 0.505 & 0.070 & 0.068 & 0.065 \\
		Avg-Cov & 0.998 & 0.999 & 0.999 & 0.765 & 0.835 & 0.815 & NA & NA & NA & 0.912 & 0.828 & 0.922 & 0.907 & 0.900 & 0.490 & 0.890 & 0.909 & 0.956 \\
		Avg-Len & 0.467 & 0.532 & 0.413 & 0.297 & 0.252 & 0.224 & NA & NA & NA & 0.267 & 0.238 & 0.205 & 0.258 & 0.244 & 9.533 & 0.277 & 0.291 & 0.321 \\
		\hline
		
		& \multicolumn{18}{c}{\bf $d=1000$, $n = 900$, sparse $\beta_0^{sp}$, and $x^{(3)}$ so that $m(x^{(3)})=x^{(3)T}\beta_0^{sp} = \beta_{0,100}^{(a)}$} \\ \cline{2-19}
		Avg-Bias & 0.025 & 0.025 & 0.021 & 0.032 & 0.032 & 0.026 & 0.430 & 1.097 & 0.407 & 0.000 & 0.001 & 0.000 & 0.033 & 0.030 & 0.100 & 0.033 & 0.034 & 0.036 \\
		Avg-Cov & 0.988 & 0.993 & 0.986 & 0.958 & 0.941 & 0.948 & 0.960 & 0.958 & 0.930 & 1.00 & 0.994 & 1.00 & 0.953 & 0.963 & 0.502 & 0.952 & 0.951 & 0.948 \\
		Avg-Len & 0.176 & 0.185 & 0.145 & 0.162 & 0.152 & 0.133 & 1.138 & 14.143 & 0.509 & 0.000 & 0.001 & 0.000 & 0.163 & 0.151 & 10.364 & 0.151 & 0.161 & 0.168 \\
		\cline{2-19}
		
		& \multicolumn{18}{c}{\bf $d=1000$, $n = 900$, dense $\beta_0^{de}$, and $x^{(3)}$ so that $m(x^{(3)})=x^{(3)T}\beta_0^{de} = \beta_{0,100}^{(b)}$} \\ \cline{2-19}
		Avg-Bias & 0.105 & 0.105 & 0.087 & 0.170 & 0.134 & 0.127 & 0.134 & 0.108 & 0.226 & 0.191 & 0.192 & 0.190 & 0.137 & 0.140 & 0.111 & 0.200 & 0.180 & 0.168 \\
		Avg-Cov & 0.996 & 0.997 & 0.999 & 0.924 & 0.910 & 0.932 & 0.998 & 0.999 & 0.995 & 0.001 & 0.044 & 0.005 & 0.352 & 0.923 & 0.512 & 0.900 & 0.935 & 0.964 \\
		Avg-Len & 0.983 & 0.998 & 0.844 & 0.758 & 0.595 & 0.583 & 1.638 & 52.948 & 2.199 & 0.017 & 0.051 & 0.020 & 0.184 & 0.625 & 6.294 & 0.820 & 0.873 & 0.910 \\
		\cline{2-19}
		
		& \multicolumn{18}{c}{\bf $d=1000$, $n = 900$, pseudo-dense $\beta_0^{pd}$, and $x^{(3)}$ so that $m(x^{(3)})=x^{(3)T}\beta_0^{pd} = \beta_{0,100}^{(c)}$} \\ \cline{2-19}
		Avg-Bias & 0.033 & 0.032 & 0.027 & 0.049 & 0.042 & 0.037 & 0.433 & 0.159 & 0.350 & 0.040 & 0.041 & 0.040 & 0.040 & 0.046 & 0.085 & 0.055 & 0.053 & 0.052 \\
		Avg-Cov & 0.997 & 0.996 & 0.997 & 0.930 & 0.918 & 0.932 & 0.998 & 0.989 & 0.984 & 0.001 & 0.015 & 0.000 & 0.932 & 0.938 & 0.511 & 0.915 & 0.934 & 0.953 \\
		Avg-Len & 0.285 & 0.292 & 0.236 & 0.223 & 0.190 & 0.173 & 5.491 & 1.625 & 0.909 & 0.001 & 0.006 & 0.002 & 0.195 & 0.218 & 7.371 & 0.241 & 0.256 & 0.267 \\
		\hline
		
		& \multicolumn{18}{c}{\bf $d=1000$, $n = 900$, sparse $\beta_0^{sp}$, and $x^{(4)}$ so that $m(x^{(4)})=x^{(4)T}\beta_0^{sp}$} \\ \cline{2-19}
		Avg-Bias & 0.034 & 0.035 & 0.029 & 0.047 & 0.042 & 0.036 & NA & NA & NA & 0.034 & 0.039 & 0.027 & 0.044 & 0.036 & 0.260 & 0.036 & 0.036 & 0.036 \\
		Avg-Cov & 0.976 & 0.989 & 0.983 & 0.880 & 0.902 & 0.898 & NA & NA & NA & 0.942 & 0.889 & 0.946 & 0.905 & 0.904 & 0.483 & 0.897 & 0.927 & 0.933 \\
		Avg-Len & 0.197 & 0.234 & 0.180 & 0.174 & 0.164 & 0.141 & NA & NA & NA & 0.161 & 0.156 & 0.133 & 0.178 & 0.153 & 7.641 & 0.152 & 0.161 & 0.168 \\
		\cline{2-19}
		
		& \multicolumn{18}{c}{\bf $d=1000$, $n = 900$, dense $\beta_0^{de}$, and $x^{(4)}$ so that $m(x^{(4)})=x^{(4)T}\beta_0^{de}$} \\ \cline{2-19}
		Avg-Bias & 0.148 & 0.161 & 0.109 & 0.230 & 0.171 & 0.157 & NA & NA & NA & 0.184 & 0.220 & 0.148 & 0.168 & 0.140 & 0.287 & 0.211 & 0.190 & 0.177 \\
		Avg-Cov & 1.00 & 1.00 & 1.00 & 0.852 & 0.870 & 0.883 & NA & NA & NA & 0.897 & 0.768 & 0.894 & 0.373 & 0.908 & 0.497 & 0.887 & 0.937 & 0.963 \\
		Avg-Len & 1.214 & 1.355 & 1.122 & 0.815 & 0.642 & 0.620 & NA & NA & NA & 0.791 & 0.665 & 0.609 & 0.197 & 0.625 & 6.076 & 0.823 & 0.876 & 0.911 \\
		\cline{2-19}
		
		& \multicolumn{18}{c}{\bf $d=1000$, $n = 900$, pseudo-dense $\beta_0^{pd}$, and $x^{(4)}$ so that $m(x^{(4)})=x^{(4)T}\beta_0^{pd}$} \\ \cline{2-19}
		Avg-Bias & 0.044 & 0.048 & 0.036 & 0.067 & 0.053 & 0.047 & NA & NA & NA & 0.050 & 0.056 & 0.039 & 0.050 & 0.052 & 0.438 & 0.067 & 0.061 & 0.058 \\
		Avg-Cov & 0.996 & 0.997 & 0.998 & 0.849 & 0.870 & 0.882 & NA & NA & NA & 0.925 & 0.832 & 0.921 & 0.895 & 0.907 & 0.468 & 0.854 & 0.924 & 0.941 \\
		Avg-Len & 0.344 & 0.391 & 0.309 & 0.240 & 0.205 & 0.184 & NA & NA & NA & 0.222 & 0.199 & 0.171 & 0.211 & 0.219 & 7.349 & 0.242 & 0.257 & 0.268 \\
		\hline
		
		& \multicolumn{18}{c}{\bf $d=1000$, $n = 900$, sparse $\beta_0^{sp}$, and $x^{(5)}$ so that $m(x^{(5)})=x^{(5)T}\beta_0^{sp}$} \\ \cline{2-19}
		Avg-Bias & 0.036 & 0.035 & 0.030 & 0.042 & 0.042 & 0.031 & NA & NA & NA & 0.002 & 0.007 & 0.002 & 0.039 & 0.016 & 0.073 & 0.018 & 0.017 & 0.021 \\
		Avg-Cov & 1.00 & 1.00 & 1.00 & 0.942 & 0.936 & 0.950 & NA & NA & NA & 0.890 & 0.557 & 0.880 & 0.948 & 0.000 & 0.485 & 0.000 & 0.000 & 0.950 \\
		Avg-Len & 1.761 & 2.029 & 1.367 & 0.203 & 0.197 & 0.154 & NA & NA & NA & 0.010 & 0.015 & 0.008 & 0.190 & 0.000 & 6.109 & 0.000 & 0.006 & 0.103 \\
		\hline\hline
\end{tabular}}
\caption{Simulation results under the circulant symmetric covariance matrix $\Sigma^{\mathrm{cs}}$, $\mathrm{Laplace}\left(0, \frac{1}{\sqrt{2}} \right)$-distributed noise, and the MAR setting \eqref{MAR_correct}. Nearly all the standard errors for the above average quantities are less than 0.01 except for those associated with ``R-Proj'' or ``SSL-AIPW'' and are thus omitted.}
\label{table:cirsym_MAR_laperr_res}
\end{center}
\end{table}

\begin{table}[hbtp!]
\begin{center}
\resizebox{\textwidth}{!}{	
	\begin{tabular}{ccccccccccccccccccc}
		\hline \hline 
		& \begin{tabular}{@{}c@{}} DL-Jav \\ (CC)\end{tabular} & \begin{tabular}{@{}c@{}} DL-Jav \\ (IPW)\end{tabular} & \begin{tabular}{@{}c@{}} DL-Jav \\ (Oracle)\end{tabular} & \begin{tabular}{@{}c@{}} DL-vdG \\ (CC)\end{tabular} & \begin{tabular}{@{}c@{}} DL-vdG \\ (IPW)\end{tabular} & \begin{tabular}{@{}c@{}} DL-vdG \\ (Oracle)\end{tabular} & \begin{tabular}{@{}c@{}} R-Proj \\ (CC)\end{tabular} & \begin{tabular}{@{}c@{}} R-Proj \\ (IPW)\end{tabular} & \begin{tabular}{@{}c@{}} R-Proj \\ (Oracle)\end{tabular} & \begin{tabular}{@{}c@{}} Refit \\ (CC)\end{tabular} & \begin{tabular}{@{}c@{}} Refit \\ (IPW)\end{tabular} & \begin{tabular}{@{}c@{}} Refit \\ (Oracle)\end{tabular} & \begin{tabular}{@{}c@{}} DDR \end{tabular} & \begin{tabular}{@{}c@{}} SAS \end{tabular} & \begin{tabular}{@{}c@{}} SSL-AIPW \end{tabular} & \begin{tabular}{@{}c@{}} \textbf{Debias} \\ (min-CV)\end{tabular} & \begin{tabular}{@{}c@{}} \textbf{Debias} \\ (1SE)\end{tabular} & \begin{tabular}{@{}c@{}} \textbf{Debias} \\ (min-feas)\end{tabular} \\ \cline{2-19}
		& \multicolumn{18}{c}{\bf $d=1000$, $n = 900$, sparse $\beta_0^{sp}$, and $x^{(1)}$ so that $m(x^{(1)})=x^{(1)T}\beta_0^{sp} = \beta_{0,1}^{(a)}$} \\ \cline{2-19}
		Avg-Bias & 0.099 & 0.103 & 0.086 & 0.118 & 0.112 & 0.098 & 0.189 & 0.198 & 0.452 & 0.098 & 0.110 & 0.085 & 0.135 & 0.108 & 0.329 & 0.125 & 0.117 & 0.110 \\
		Avg-Cov & 0.948 & 0.957 & 0.962 & 0.888 & 0.890 & 0.902 & 0.948 & 0.945 & 0.923 & 0.957 & 0.911 & 0.947 & 0.546 & 0.916 & 0.518 & 0.868 & 0.919 & 0.939 \\
		Avg-Len & 0.516 & 0.529 & 0.439 & 0.474 & 0.444 & 0.403 & 0.858 & 1.144 & 1.629 & 0.480 & 0.462 & 0.403 & 0.239 & 0.465 & 27.737 & 0.471 & 0.501 & 0.527 \\
		\cline{2-19}
		
		& \multicolumn{18}{c}{\bf $d=1000$, $n = 900$, dense $\beta_0^{de}$, and $x^{(1)}$ so that $m(x^{(1)})=x^{(1)T}\beta_0^{de} = \beta_{0,1}^{(b)}$} \\ \cline{2-19}
		Avg-Bias & 0.183 & 0.188 & 0.142 & 0.220 & 0.198 & 0.170 & 0.177 & 0.177 & 0.471 & 0.218 & 0.228 & 0.176 & 0.233 & 0.185 & 0.312 & 0.259 & 0.227 & 0.206 \\
		Avg-Cov & 0.990 & 0.990 & 0.995 & 0.889 & 0.862 & 0.917 & 0.999 & 0.993 & 0.993 & 0.924 & 0.851 & 0.913 & 0.313 & 0.917 & 0.483 & 0.874 & 0.930 & 0.957 \\
		Avg-Len & 1.194 & 1.205 & 1.049 & 0.910 & 0.753 & 0.726 & 1.487 & 1.301 & 2.202 & 0.935 & 0.814 & 0.741 & 0.233 & 0.787 & 19.773 & 0.965 & 1.026 & 1.075 \\
		\cline{2-19}
		
		& \multicolumn{18}{c}{\bf $d=1000$, $n = 900$, pseudo-dense $\beta_0^{pd}$, and $x^{(1)}$ so that $m(x^{(1)})=x^{(1)T}\beta_0^{pd} = \beta_{0,1}^{(c)}$} \\ \cline{2-19}
		Avg-Bias & 0.110 & 0.114 & 0.091 & 0.129 & 0.120 & 0.103 & 0.168 & 0.174 & 0.298 & 0.117 & 0.123 & 0.098 & 0.141 & 0.119 & 0.446 & 0.140 & 0.129 & 0.121 \\
		Avg-Cov & 0.982 & 0.978 & 0.983 & 0.889 & 0.884 & 0.904 & 0.970 & 0.960 & 0.960 & 0.934 & 0.892 & 0.944 & 0.575 & 0.912 & 0.511 & 0.865 & 0.918 & 0.945 \\
		Avg-Len & 0.627 & 0.640 & 0.537 & 0.522 & 0.474 & 0.435 & 0.905 & 0.873 & 1.366 & 0.527 & 0.499 & 0.434 & 0.267 & 0.502 & 21.064 & 0.537 & 0.572 & 0.601 \\
		\hline
		
		& \multicolumn{18}{c}{\bf $d=1000$, $n = 900$, sparse $\beta_0^{sp}$, and $x^{(2)}$ so that $m(x^{(2)})=x^{(2)T}\beta_0^{sp}$} \\ \cline{2-19}
		Avg-Bias & 0.128 & 0.126 & 0.110 & 0.206 & 0.169 & 0.152 & NA & NA & NA & 0.106 & 0.118 & 0.092 & 0.167 & 0.121 & 0.429 & 0.121 & 0.122 & 0.128 \\
		Avg-Cov & 0.983 & 0.996 & 0.993 & 0.805 & 0.851 & 0.847 & NA & NA & NA & 0.953 & 0.895 & 0.953 & 0.579 & 0.920 & 0.488 & 0.928 & 0.935 & 0.943 \\
		Avg-Len & 0.769 & 0.893 & 0.714 & 0.631 & 0.590 & 0.522 & NA & NA & NA & 0.520 & 0.503 & 0.437 & 0.308 & 0.509 & 34.504 & 0.540 & 0.570 & 0.626 \\
		\cline{2-19}
		
		& \multicolumn{18}{c}{\bf $d=1000$, $n = 900$, dense $\beta_0^{de}$, and $x^{(2)}$ so that $m(x^{(2)})=x^{(2)T}\beta_0^{de}$} \\ \cline{2-19}
		Avg-Bias & 0.224 & 0.234 & 0.178 & 0.437 & 0.293 & 0.293 & NA & NA & NA & 0.309 & 0.342 & 0.228 & 0.305 & 0.210 & 0.450 & 0.257 & 0.249 & 0.233 \\
		Avg-Cov & 0.999 & 1.00 & 1.00 & 0.755 & 0.820 & 0.813 & NA & NA & NA & 0.838 & 0.756 & 0.887 & 0.339 & 0.894 & 0.506 & 0.913 & 0.943 & 0.968 \\
		Avg-Len & 1.969 & 2.201 & 1.851 & 1.211 & 1.002 & 0.940 & NA & NA & NA & 1.087 & 0.955 & 0.876 & 0.302 & 0.865 & 27.198 & 1.107 & 1.167 & 1.283 \\
		\cline{2-19}
		
		& \multicolumn{18}{c}{\bf $d=1000$, $n = 900$, pseudo-dense $\beta_0^{pd}$, and $x^{(2)}$ so that $m(x^{(2)})=x^{(2)T}\beta_0^{pd}$} \\ \cline{2-19}
		Avg-Bias & 0.148 & 0.147 & 0.125 & 0.265 & 0.202 & 0.185 & NA & NA & NA & 0.146 & 0.158 & 0.123 & 0.184 & 0.139 & 0.582 & 0.143 & 0.141 & 0.144 \\
		Avg-Cov & 0.996 & 1.00 & 0.997 & 0.713 & 0.794 & 0.785 & NA & NA & NA & 0.921 & 0.862 & 0.931 & 0.585 & 0.892 & 0.512 & 0.903 & 0.931 & 0.951 \\
		Avg-Len & 0.983 & 1.122 & 0.909 & 0.695 & 0.630 & 0.564 & NA & NA & NA & 0.624 & 0.592 & 0.516 & 0.345 & 0.552 & 33.809 & 0.617 & 0.651 & 0.715 \\
		\hline
		
		& \multicolumn{18}{c}{\bf $d=1000$, $n = 900$, sparse $\beta_0^{sp}$, and $x^{(3)}$ so that $m(x^{(3)})=x^{(3)T}\beta_0^{sp} = \beta_{0,100}^{(a)}$} \\ \cline{2-19}
		Avg-Bias & 0.075 & 0.074 & 0.067 & 0.094 & 0.093 & 0.082 & 0.810 & 0.566 & 0.628 & 0.000 & 0.003 & 0.000 & 0.081 & 0.099 & 0.297 & 0.097 & 0.102 & 0.106 \\
		Avg-Cov & 0.986 & 0.985 & 0.988 & 0.943 & 0.943 & 0.949 & 0.959 & 0.957 & 0.940 & 1.00 & 0.993 & 1.00 & 0.778 & 0.951 & 0.528 & 0.955 & 0.953 & 0.957 \\
		Avg-Len & 0.516 & 0.529 & 0.440 & 0.473 & 0.443 & 0.403 & 5.314 & 2.157 & 2.126 & 0.000 & 0.004 & 0.000 & 0.236 & 0.467 & 25.547 & 0.471 & 0.502 & 0.522 \\
		\cline{2-19}
		
		& \multicolumn{18}{c}{\bf $d=1000$, $n = 900$, dense $\beta_0^{de}$, and $x^{(3)}$ so that $m(x^{(3)})=x^{(3)T}\beta_0^{de} = \beta_{0,100}^{(b)}$} \\ \cline{2-19}
		Avg-Bias & 0.133 & 0.130 & 0.115 & 0.207 & 0.173 & 0.159 & 0.174 & 0.196 & 0.396 & 0.185 & 0.195 & 0.188 & 0.188 & 0.181 & 0.233 & 0.227 & 0.211 & 0.204 \\
		Avg-Cov & 0.998 & 0.997 & 0.996 & 0.928 & 0.912 & 0.927 & 0.999 & 0.996 & 0.995 & 0.000 & 0.028 & 0.001 & 0.316 & 0.920 & 0.509 & 0.908 & 0.951 & 0.972 \\
		Avg-Len & 1.193 & 1.203 & 1.049 & 0.909 & 0.752 & 0.726 & 1.757 & 2.160 & 2.170 & 0.004 & 0.045 & 0.011 & 0.231 & 0.788 & 18.631 & 0.966 & 1.028 & 1.070 \\
		\cline{2-19}
		
		& \multicolumn{18}{c}{\bf $d=1000$, $n = 900$, pseudo-dense $\beta_0^{pd}$, and $x^{(3)}$ so that $m(x^{(3)})=x^{(3)T}\beta_0^{pd} = \beta_{0,100}^{(c)}$} \\ \cline{2-19}
		Avg-Bias & 0.082 & 0.080 & 0.071 & 0.109 & 0.103 & 0.089 & 0.271 & 0.523 & 0.269 & 0.040 & 0.042 & 0.040 & 0.089 & 0.109 & 0.276 & 0.114 & 0.117 & 0.120 \\
		Avg-Cov & 0.997 & 0.996 & 0.995 & 0.944 & 0.936 & 0.947 & 0.974 & 0.969 & 0.963 & 0.000 & 0.000 & 0.000 & 0.805 & 0.933 & 0.487 & 0.943 & 0.954 & 0.963 \\
		Avg-Len & 0.627 & 0.639 & 0.538 & 0.522 & 0.473 & 0.436 & 7.329 & 6.314 & 1.360 & 0.001 & 0.004 & 0.001 & 0.263 & 0.505 & 27.951 & 0.538 & 0.573 & 0.596 \\
		\hline
		
		& \multicolumn{18}{c}{\bf $d=1000$, $n = 900$, sparse $\beta_0^{sp}$, and $x^{(4)}$ so that $m(x^{(4)})=x^{(4)T}\beta_0^{sp}$} \\ \cline{2-19}
		Avg-Bias & 0.127 & 0.124 & 0.114 & 0.231 & 0.179 & 0.173 & NA & NA & NA & 0.108 & 0.125 & 0.095 & 0.173 & 0.132 & 0.487 & 0.167 & 0.143 & 0.132 \\
		Avg-Cov & 1.00 & 1.00 & 1.00 & 0.805 & 0.871 & 0.846 & NA & NA & NA & 0.951 & 0.883 & 0.952 & 0.588 & 0.846 & 0.510 & 0.774 & 0.854 & 0.908 \\
		Avg-Len & 1.585 & 1.803 & 1.360 & 0.680 & 0.644 & 0.556 & NA & NA & NA & 0.520 & 0.503 & 0.436 & 0.329 & 0.477 & 32.541 & 0.505 & 0.533 & 0.577 \\
		\cline{2-19}
		
		& \multicolumn{18}{c}{\bf $d=1000$, $n = 900$, dense $\beta_0^{de}$, and $x^{(4)}$ so that $m(x^{(4)})=x^{(4)T}\beta_0^{de}$} \\ \cline{2-19}
		Avg-Bias & 0.263 & 0.234 & 0.274 & 1.053 & 0.599 & 0.738 & NA & NA & NA & 0.887 & 0.861 & 0.708 & 0.363 & 0.641 & 0.666 & 1.335 & 1.120 & 0.803 \\
		Avg-Cov & 1.00 & 1.00 & 1.00 & 0.078 & 0.447 & 0.147 & NA & NA & NA & 0.109 & 0.090 & 0.103 & 0.269 & 0.174 & 0.498 & 0.008 & 0.032 & 0.192 \\
		Avg-Len & 4.620 & 5.043 & 4.105 & 1.306 & 1.093 & 1.003 & NA & NA & NA & 1.018 & 0.891 & 0.810 & 0.327 & 0.809 & 26.129 & 1.035 & 1.091 & 1.182 \\
		\cline{2-19}
		
		& \multicolumn{18}{c}{\bf $d=1000$, $n = 900$, pseudo-dense $\beta_0^{pd}$, and $x^{(4)}$ so that $m(x^{(4)})=x^{(4)T}\beta_0^{pd}$} \\ \cline{2-19}
		Avg-Bias & 0.214 & 0.179 & 0.224 & 0.519 & 0.370 & 0.381 & NA & NA & NA & 0.284 & 0.277 & 0.234 & 0.195 & 0.312 & 0.639 & 0.563 & 0.456 & 0.326 \\
		Avg-Cov & 1.00 & 1.00 & 1.00 & 0.182 & 0.431 & 0.265 & NA & NA & NA & 0.522 & 0.485 & 0.568 & 0.583 & 0.341 & 0.488 & 0.035 & 0.149 & 0.462 \\
		Avg-Len & 2.169 & 2.416 & 1.875 & 0.749 & 0.687 & 0.602 & NA & NA & NA & 0.574 & 0.546 & 0.474 & 0.374 & 0.516 & 28.206 & 0.576 & 0.608 & 0.658 \\
		\hline
		
		& \multicolumn{18}{c}{\bf $d=1000$, $n = 900$, sparse $\beta_0^{sp}$, and $x^{(5)}$ so that $m(x^{(5)})=x^{(5)T}\beta_0^{sp}$} \\ \cline{2-19}
		Avg-Bias & 0.108 & 0.105 & 0.092 & 0.129 & 0.125 & 0.098 & NA & NA & NA & 0.008 & 0.020 & 0.007 & 0.087 & 0.031 & 0.210 & 0.053 & 0.051 & 0.060 \\
		Avg-Cov & 1.00 & 1.00 & 1.00 & 0.948 & 0.941 & 0.941 & NA & NA & NA & 0.880 & 0.585 & 0.885 & 0.825 & 0.000 & 0.461 & 0.000 & 0.000 & 0.961 \\
		Avg-Len & 5.186 & 5.798 & 4.133 & 0.596 & 0.573 & 0.468 & NA & NA & NA & 0.029 & 0.042 & 0.024 & 0.279 & 0.000 & 19.908 & 0.000 & 0.018 & 0.308 \\
		\hline\hline
\end{tabular}}
\caption{Simulation results under the circulant symmetric covariance matrix $\Sigma^{\mathrm{cs}}$, $t_2$-distributed noise, and the MAR setting \eqref{MAR_correct}. Nearly all the standard errors for the above average quantities are less than 0.01 except for those associated with ``R-Proj'' and ``SSL-AIPW'' and are thus omitted.}
\label{table:cirsym_MAR_terr_res}
\end{center}
\end{table}

\begin{table}[hbtp!]
\begin{center}
\resizebox{\textwidth}{!}{	
	\begin{tabular}{ccccccccccccccccccc}
		\hline \hline 
		& \begin{tabular}{@{}c@{}} DL-Jav \\ (CC)\end{tabular}  & \begin{tabular}{@{}c@{}} DL-Jav \\ (IPW)\end{tabular} & \begin{tabular}{@{}c@{}} DL-Jav \\ (Oracle)\end{tabular} 
		& \begin{tabular}{@{}c@{}} DL-vdG \\ (CC)\end{tabular}  & \begin{tabular}{@{}c@{}} DL-vdG \\ (IPW)\end{tabular} & \begin{tabular}{@{}c@{}} DL-vdG \\ (Oracle)\end{tabular} 
		& \begin{tabular}{@{}c@{}} R-Proj \\ (CC)\end{tabular}  & \begin{tabular}{@{}c@{}} R-Proj \\ (IPW)\end{tabular}  & \begin{tabular}{@{}c@{}} R-Proj \\ (Oracle)\end{tabular} 
		& \begin{tabular}{@{}c@{}} Refit \\ (CC)\end{tabular} & \begin{tabular}{@{}c@{}} Refit \\ (IPW)\end{tabular} & \begin{tabular}{@{}c@{}} Refit \\ (Oracle)\end{tabular} 
		& DDR & SAS & SSL-AIPW 
		& \begin{tabular}{@{}c@{}} \textbf{Debias} \\ (min-CV)\end{tabular} & \begin{tabular}{@{}c@{}} \textbf{Debias} \\ (1SE)\end{tabular} & \begin{tabular}{@{}c@{}} \textbf{Debias} \\ (min-feas)\end{tabular}\\ \cline{2-19}
		
		& \multicolumn{18}{c}{\bf $d=1000$, $n = 900$, sparse $\beta_0^{sp}$, and $x^{(1)}$ so that $m(x^{(1)})=x^{(1)T}\beta_0^{sp} = \beta_{0,1}^{(a)}$} \\ \cline{2-19}
		Avg-Bias & 0.071 & 0.072 & 0.060 & 0.070 & 0.071 & 0.060 & 0.545 & 0.738 & 0.250 & 0.073 & 0.074 & 0.061 & 0.078 & 0.076 & 0.215 & 0.074 & 0.076 & 0.088 \\
		Avg-Cov & 0.887 & 0.888 & 0.907 & 0.915 & 0.913 & 0.933 & 0.939 & 0.946 & 0.944 & 0.945 & 0.940 & 0.947 & 0.886 & 0.902 & 0.490 & 0.881 & 0.914 & 0.942 \\
		Avg-Len & 0.279 & 0.281 & 0.247 & 0.308 & 0.305 & 0.266 & 0.967 & 1.735 & 2.550 & 0.353 & 0.351 & 0.300 & 0.307 & 0.319 & 91.508 & 0.295 & 0.325 & 0.412 \\
		\cline{2-19}
		
		& \multicolumn{18}{c}{\bf $d=1000$, $n = 900$, dense $\beta_0^{de}$, and $x^{(1)}$ so that $m(x^{(1)})=x^{(1)T}\beta_0^{de} = \beta_{0,1}^{(b)}$} \\ \cline{2-19}
		Avg-Bias & 0.287 & 0.290 & 0.157 & 0.278 & 0.227 & 0.160 & 0.205 & 0.236 & 0.377 & 0.411 & 0.246 & 0.176 & 0.249 & 0.217 & 0.164 & 0.457 & 0.394 & 0.363 \\
		Avg-Cov & 0.747 & 0.746 & 0.909 & 0.557 & 0.426 & 0.552 & 0.937 & 0.774 & 0.953 & 0.781 & 0.782 & 0.815 & 0.642 & 0.967 & 0.464 & 0.845 & 0.924 & 0.982 \\
		Avg-Len & 0.828 & 0.842 & 0.672 & 0.528 & 0.305 & 0.301 & 0.996 & 0.625 & 1.368 & 1.166 & 0.761 & 0.593 & 0.574 & 1.189 & 52.275 & 1.648 & 1.816 & 2.302 \\
		\cline{2-19}
		
		& \multicolumn{18}{c}{\bf $d=1000$, $n = 900$, pseudo-dense $\beta_0^{pd}$, and $x^{(1)}$ so that $m(x^{(1)})=x^{(1)T}\beta_0^{pd} = \beta_{0,1}^{(c)}$} \\ \cline{2-19}
		Avg-Bias & 0.094 & 0.095 & 0.074 & 0.098 & 0.095 & 0.076 & 0.318 & 0.127 & 0.251 & 0.097 & 0.095 & 0.080 & 0.102 & 0.093 & 0.218 & 0.102 & 0.096 & 0.103 \\
		Avg-Cov & 0.857 & 0.856 & 0.893 & 0.826 & 0.794 & 0.851 & 0.940 & 0.927 & 0.931 & 0.917 & 0.895 & 0.915 & 0.823 & 0.946 & 0.511 & 0.866 & 0.911 & 0.956 \\
		Avg-Len & 0.341 & 0.342 & 0.300 & 0.329 & 0.296 & 0.276 & 0.765 & 0.569 & 0.879 & 0.415 & 0.387 & 0.334 & 0.340 & 0.456 & 27.226 & 0.389 & 0.429 & 0.543 \\
		\hline
		
		& \multicolumn{18}{c}{\bf $d=1000$, $n = 900$, sparse $\beta_0^{sp}$, and $x^{(2)}$ so that $m(x^{(2)})=x^{(2)T}\beta_0^{sp}$} \\ \cline{2-19}
		Avg-Bias & 0.057 & 0.058 & 0.048 & 0.067 & 0.066 & 0.056 & NA & NA & NA & 0.050 & 0.051 & 0.043 & 0.067 & 0.055 & 0.197 & 0.059 & 0.064 & 0.075 \\
		Avg-Cov & 0.987 & 0.992 & 0.991 & 0.922 & 0.919 & 0.935 & NA & NA & NA & 0.942 & 0.930 & 0.936 & 0.925 & 0.888 & 0.509 & 0.916 & 0.929 & 0.937 \\
		Avg-Len & 0.353 & 0.387 & 0.319 & 0.290 & 0.287 & 0.248 & NA & NA & NA & 0.237 & 0.237 & 0.202 & 0.290 & 0.216 & 23.525 & 0.255 & 0.286 & 0.347 \\
		\cline{2-19}
		
		& \multicolumn{18}{c}{\bf $d=1000$, $n = 900$, dense $\beta_0^{de}$, and $x^{(2)}$ so that $m(x^{(2)})=x^{(2)T}\beta_0^{de}$} \\ \cline{2-19}
		Avg-Bias & 0.233 & 0.240 & 0.128 & 0.237 & 0.197 & 0.138 & NA & NA & NA & 0.377 & 0.247 & 0.179 & 0.229 & 0.214 & 0.154 & 0.390 & 0.341 & 0.328 \\
		Avg-Cov & 0.987 & 0.994 & 1.00 & 0.584 & 0.433 & 0.587 & NA & NA & NA & 0.683 & 0.688 & 0.757 & 0.648 & 0.866 & 0.478 & 0.866 & 0.941 & 0.981 \\
		Avg-Len & 1.366 & 1.456 & 1.072 & 0.497 & 0.287 & 0.281 & NA & NA & NA & 0.999 & 0.652 & 0.519 & 0.545 & 0.802 & 19.405 & 1.423 & 1.599 & 1.939 \\
		\cline{2-19}
		
		& \multicolumn{18}{c}{\bf $d=1000$, $n = 900$, pseudo-dense $\beta_0^{pd}$, and $x^{(2)}$ so that $m(x^{(2)})=x^{(2)T}\beta_0^{pd}$} \\ \cline{2-19}
		Avg-Bias & 0.081 & 0.082 & 0.066 & 0.089 & 0.085 & 0.071 & NA & NA & NA & 0.088 & 0.085 & 0.071 & 0.090 & 0.087 & 0.196 & 0.090 & 0.086 & 0.092 \\
		Avg-Cov & 0.991 & 0.995 & 0.994 & 0.830 & 0.804 & 0.849 & NA & NA & NA & 0.909 & 0.894 & 0.899 & 0.850 & 0.824 & 0.495 & 0.857 & 0.918 & 0.949 \\
		Avg-Len & 0.522 & 0.556 & 0.456 & 0.310 & 0.278 & 0.257 & NA & NA & NA & 0.373 & 0.348 & 0.301 & 0.322 & 0.309 & 30.022 & 0.336 & 0.377 & 0.458 \\
		\hline
		
		& \multicolumn{18}{c}{\bf $d=1000$, $n = 900$, sparse $\beta_0^{sp}$, and $x^{(3)}$ so that $m(x^{(3)})=x^{(3)T}\beta_0^{sp} = \beta_{0,100}^{(a)}$} \\ \cline{2-19}
		Avg-Bias & 0.056 & 0.055 & 0.053 & 0.074 & 0.073 & 0.066 & 0.700 & 1.163 & 0.562 & 0.000 & 0.000 & 0.000 & 0.077 & 0.090 & 0.268 & 0.074 & 0.085 & 0.097 \\
		Avg-Cov & 0.986 & 0.984 & 0.980 & 0.957 & 0.948 & 0.947 & 0.952 & 0.948 & 0.948 & 1.00 & 0.999 & 1.00 & 0.952 & 0.961 & 0.503 & 0.953 & 0.951 & 0.953 \\
		Avg-Len & 0.347 & 0.350 & 0.312 & 0.367 & 0.363 & 0.322 & 3.483 & 15.279 & 1.778 & 0.000 & 0.000 & 0.000 & 0.384 & 0.461 & 140.571 & 0.373 & 0.427 & 0.479 \\
		\cline{2-19}
		
		& \multicolumn{18}{c}{\bf $d=1000$, $n = 900$, dense $\beta_0^{de}$, and $x^{(3)}$ so that $m(x^{(3)})=x^{(3)T}\beta_0^{de} = \beta_{0,100}^{(b)}$} \\ \cline{2-19}
		Avg-Bias & 0.171 & 0.171 & 0.125 & 0.220 & 0.186 & 0.165 & 0.175 & 0.333 & 0.218 & 0.250 & 0.215 & 0.198 & 0.245 & 0.304 & 0.188 & 0.378 & 0.352 & 0.371 \\
		Avg-Cov & 0.978 & 0.980 & 0.985 & 0.777 & 0.561 & 0.621 & 0.993 & 0.893 & 0.975 & 0.237 & 0.352 & 0.371 & 0.800 & 0.975 & 0.490 & 0.954 & 0.986 & 0.997 \\
		Avg-Len & 1.054 & 1.071 & 0.857 & 0.629 & 0.363 & 0.364 & 1.287 & 0.803 & 1.125 & 0.384 & 0.370 & 0.329 & 0.723 & 1.716 & 73.685 & 2.081 & 2.386 & 2.676 \\
		\cline{2-19}
		
		& \multicolumn{18}{c}{\bf $d=1000$, $n = 900$, pseudo-dense $\beta_0^{pd}$, and $x^{(3)}$ so that $m(x^{(3)})=x^{(3)T}\beta_0^{pd} = \beta_{0,100}^{(c)}$} \\ \cline{2-19}
		Avg-Bias & 0.069 & 0.068 & 0.064 & 0.090 & 0.085 & 0.077 & 0.129 & 0.212 & 0.345 & 0.056 & 0.055 & 0.053 & 0.101 & 0.118 & 0.242 & 0.095 & 0.098 & 0.108 \\
		Avg-Cov & 0.974 & 0.977 & 0.973 & 0.911 & 0.904 & 0.904 & 0.984 & 0.966 & 0.959 & 0.302 & 0.330 & 0.363 & 0.912 & 0.973 & 0.482 & 0.954 & 0.970 & 0.977 \\
		Avg-Len & 0.431 & 0.432 & 0.382 & 0.392 & 0.352 & 0.333 & 0.793 & 1.870 & 1.138 & 0.134 & 0.142 & 0.135 & 0.426 & 0.659 & 39.055 & 0.491 & 0.563 & 0.631 \\
		\hline
		
		& \multicolumn{18}{c}{\bf $d=1000$, $n = 900$, sparse $\beta_0^{sp}$, and $x^{(4)}$ so that $m(x^{(4)})=x^{(4)T}\beta_0^{sp}$} \\ \cline{2-19}
		Avg-Bias & 0.047 & 0.048 & 0.040 & 0.065 & 0.059 & 0.052 & NA & NA & NA & 0.045 & 0.045 & 0.038 & 0.065 & 0.046 & 0.138 & 0.046 & 0.048 & 0.049 \\
		Avg-Cov & 1.00 & 1.00 & 1.00 & 0.876 & 0.909 & 0.893 & NA & NA & NA & 0.934 & 0.929 & 0.948 & 0.868 & 0.837 & 0.504 & 0.891 & 0.922 & 0.932 \\
		Avg-Len & 0.717 & 0.752 & 0.570 & 0.244 & 0.242 & 0.203 & NA & NA & NA & 0.214 & 0.213 & 0.182 & 0.237 & 0.157 & 60.803 & 0.179 & 0.207 & 0.221 \\
		\cline{2-19}
		
		& \multicolumn{18}{c}{\bf $d=1000$, $n = 900$, dense $\beta_0^{de}$, and $x^{(4)}$ so that $m(x^{(4)})=x^{(4)T}\beta_0^{de}$} \\ \cline{2-19}
		Avg-Bias & 0.174 & 0.176 & 0.094 & 0.169 & 0.129 & 0.095 & NA & NA & NA & 0.340 & 0.214 & 0.159 & 0.156 & 0.153 & 0.116 & 0.403 & 0.296 & 0.257 \\
		Avg-Cov & 1.00 & 1.00 & 1.00 & 0.673 & 0.548 & 0.651 & NA & NA & NA & 0.632 & 0.575 & 0.614 & 0.738 & 0.870 & 0.490 & 0.647 & 0.872 & 0.943 \\
		Avg-Len & 3.580 & 3.675 & 2.529 & 0.417 & 0.242 & 0.230 & NA & NA & NA & 0.731 & 0.462 & 0.359 & 0.446 & 0.582 & 36.161 & 1.002 & 1.158 & 1.236 \\
		\cline{2-19}
		
		& \multicolumn{18}{c}{\bf $d=1000$, $n = 900$, pseudo-dense $\beta_0^{pd}$, and $x^{(4)}$ so that $m(x^{(4)})=x^{(4)T}\beta_0^{pd}$} \\ \cline{2-19}
		Avg-Bias & 0.066 & 0.066 & 0.054 & 0.133 & 0.100 & 0.101 & NA & NA & NA & 0.064 & 0.061 & 0.052 & 0.111 & 0.062 & 0.144 & 0.074 & 0.063 & 0.060 \\
		Avg-Cov & 1.00 & 1.00 & 1.00 & 0.509 & 0.628 & 0.538 & NA & NA & NA & 0.873 & 0.867 & 0.883 & 0.637 & 0.844 & 0.506 & 0.799 & 0.904 & 0.938 \\
		Avg-Len & 1.290 & 1.324 & 1.020 & 0.260 & 0.234 & 0.210 & NA & NA & NA & 0.251 & 0.234 & 0.203 & 0.265 & 0.224 & 18.848 & 0.236 & 0.273 & 0.292 \\
		\hline
		
		& \multicolumn{18}{c}{\bf $d=1000$, $n = 900$, sparse $\beta_0^{sp}$, and $x^{(5)}$ so that $m(x^{(5)})=x^{(5)T}\beta_0^{sp}$} \\ \cline{2-19}
		Avg-Bias & 0.028 & 0.028 & 0.021 & 0.038 & 0.036 & 0.029 & NA & NA & NA & 0.002 & 0.003 & 0.001 & 0.034 & 0.005 & 0.024 & 0.006 & 0.007 & 0.009 \\
		Avg-Cov & 1.00 & 1.00 & 1.00 & 0.945 & 0.962 & 0.954 & NA & NA & NA & 0.859 & 0.653 & 0.885 & 0.947 & 0.000 & 0.502 & 0.891 & 0.895 & 0.953 \\
		Avg-Len & 2.542 & 2.592 & 1.889 & 0.183 & 0.182 & 0.142 & NA & NA & NA & 0.006 & 0.007 & 0.005 & 0.165 & 0.000 & 25.538 & 0.030 & 0.034 & 0.044 \\
		
		\hline\hline
\end{tabular}}
\caption{Simulation results under the Toeplitz (or autoregressive) covariance matrix $\Sigma^{\mathrm{ar}}$, $\mathrm{Laplace}\left(0, \frac{1}{\sqrt{2}} \right)$-distributed noise, and the MAR setting \eqref{MAR_correct}. Nearly all the standard errors for the above average quantities are less than 0.01 except for those associated with ``R-Proj'' or ``SSL-AIPW'' and are thus omitted.}
\label{table:AR_MAR_laperr_res}
\end{center}
\end{table}

\begin{table}[hbtp!]
\begin{center}
\resizebox{\textwidth}{!}{	
	\begin{tabular}{ccccccccccccccccccc}
		\hline \hline 
		& \begin{tabular}{@{}c@{}} DL-Jav \\ (CC)\end{tabular}  & \begin{tabular}{@{}c@{}} DL-Jav \\ (IPW)\end{tabular} & \begin{tabular}{@{}c@{}} DL-Jav \\ (Oracle)\end{tabular} 
		& \begin{tabular}{@{}c@{}} DL-vdG \\ (CC)\end{tabular}  & \begin{tabular}{@{}c@{}} DL-vdG \\ (IPW)\end{tabular} & \begin{tabular}{@{}c@{}} DL-vdG \\ (Oracle)\end{tabular} 
		& \begin{tabular}{@{}c@{}} R-Proj \\ (CC)\end{tabular}  & \begin{tabular}{@{}c@{}} R-Proj \\ (IPW)\end{tabular}  & \begin{tabular}{@{}c@{}} R-Proj \\ (Oracle)\end{tabular} 
		& \begin{tabular}{@{}c@{}} Refit \\ (CC)\end{tabular} & \begin{tabular}{@{}c@{}} Refit \\ (IPW)\end{tabular} & \begin{tabular}{@{}c@{}} Refit \\ (Oracle)\end{tabular} 
		& DDR & SAS & SSL-AIPW
		& \begin{tabular}{@{}c@{}} \textbf{Debias} \\ (min-CV)\end{tabular} & \begin{tabular}{@{}c@{}} \textbf{Debias} \\ (1SE)\end{tabular} & \begin{tabular}{@{}c@{}} \textbf{Debias} \\ (min-feas)\end{tabular} \\ \cline{2-19}
		& \multicolumn{18}{c}{\bf $d=1000$, $n = 900$, sparse $\beta_0^{sp}$, and $x^{(1)}$ so that $m(x^{(1)})=x^{(1)T}\beta_0^{sp} = \beta_{0,1}^{(a)}$} \\ \cline{2-19}
		Avg-Bias & 0.213 & 0.211 & 0.186 & 0.218 & 0.218 & 0.191 & 0.497 & 0.569 & 0.939 & 0.228 & 0.227 & 0.195 & 0.233 & 0.244 & 0.636 & 0.233 & 0.235 & 0.267 \\
		Avg-Cov & 0.879 & 0.886 & 0.907 & 0.910 & 0.917 & 0.920 & 0.944 & 0.939 & 0.956 & 0.939 & 0.933 & 0.949 & 0.839 & 0.912 & 0.517 & 0.897 & 0.911 & 0.937 \\
		Avg-Len & 0.816 & 0.820 & 0.749 & 0.903 & 0.894 & 0.808 & 2.284 & 2.156 & 3.056 & 1.068 & 1.060 & 0.918 & 0.796 & 1.052 & 74.236 & 0.926 & 1.020 & 1.294 \\
		\cline{2-19}
		
		& \multicolumn{18}{c}{\bf $d=1000$, $n = 900$, dense $\beta_0^{de}$, and $x^{(1)}$ so that $m(x^{(1)})=x^{(1)T}\beta_0^{de} = \beta_{0,1}^{(b)}$} \\ \cline{2-19}
		Avg-Bias & 0.382 & 0.380 & 0.263 & 0.423 & 0.373 & 0.291 & 0.409 & 0.369 & 0.391 & 0.556 & 0.399 & 0.328 & 0.349 & 0.336 & 0.520 & 0.557 & 0.488 & 0.485 \\
		Avg-Cov & 0.802 & 0.814 & 0.884 & 0.678 & 0.614 & 0.729 & 0.950 & 0.887 & 0.961 & 0.779 & 0.815 & 0.843 & 0.660 & 0.937 & 0.505 & 0.861 & 0.944 & 0.988 \\
		Avg-Len & 1.234 & 1.241 & 1.049 & 1.097 & 0.842 & 0.845 & 2.024 & 1.589 & 2.093 & 1.519 & 1.277 & 1.103 & 0.847 & 1.583 & 68.001 & 2.110 & 2.324 & 2.948 \\
		\cline{2-19}
		
		& \multicolumn{18}{c}{\bf $d=1000$, $n = 900$, pseudo-dense $\beta_0^{pd}$, and $x^{(1)}$ so that $m(x^{(1)})=x^{(1)T}\beta_0^{pd} = \beta_{0,1}^{(c)}$} \\ \cline{2-19}
		Avg-Bias & 0.230 & 0.229 & 0.196 & 0.239 & 0.236 & 0.204 & 0.326 & 0.428 & 0.395 & 0.260 & 0.256 & 0.216 & 0.246 & 0.255 & 0.620 & 0.262 & 0.258 & 0.286 \\
		Avg-Cov & 0.874 & 0.874 & 0.895 & 0.888 & 0.874 & 0.895 & 0.950 & 0.942 & 0.956 & 0.929 & 0.925 & 0.938 & 0.808 & 0.915 & 0.517 & 0.882 & 0.921 & 0.947 \\
		Avg-Len & 0.873 & 0.875 & 0.793 & 0.925 & 0.884 & 0.822 & 1.648 & 1.723 & 2.003 & 1.129 & 1.097 & 0.955 & 0.795 & 1.110 & 104.853 & 1.014 & 1.116 & 1.416 \\
		\hline
		
		& \multicolumn{18}{c}{\bf $d=1000$, $n = 900$, sparse $\beta_0^{sp}$, and $x^{(2)}$ so that $m(x^{(2)})=x^{(2)T}\beta_0^{sp}$} \\ \cline{2-19}
		Avg-Bias & 0.166 & 0.167 & 0.150 & 0.200 & 0.194 & 0.174 & NA & NA & NA & 0.162 & 0.159 & 0.133 & 0.180 & 0.182 & 0.599 & 0.187 & 0.199 & 0.227 \\
		Avg-Cov & 0.985 & 0.993 & 0.994 & 0.918 & 0.928 & 0.938 & NA & NA & NA & 0.942 & 0.939 & 0.953 & 0.917 & 0.889 & 0.495 & 0.919 & 0.941 & 0.953 \\
		Avg-Len & 1.031 & 1.126 & 0.963 & 0.851 & 0.841 & 0.752 & NA & NA & NA & 0.723 & 0.722 & 0.621 & 0.743 & 0.716 & 147.268 & 0.799 & 0.901 & 1.093 \\
		\cline{2-19}
		
		& \multicolumn{18}{c}{\bf $d=1000$, $n = 900$, dense $\beta_0^{de}$, and $x^{(2)}$ so that $m(x^{(2)})=x^{(2)T}\beta_0^{de}$} \\ \cline{2-19}
		Avg-Bias & 0.320 & 0.321 & 0.219 & 0.351 & 0.317 & 0.252 & NA & NA & NA & 0.474 & 0.366 & 0.294 & 0.296 & 0.289 & 0.493 & 0.458 & 0.409 & 0.407 \\
		Avg-Cov & 0.982 & 0.986 & 0.996 & 0.740 & 0.664 & 0.778 & NA & NA & NA & 0.729 & 0.742 & 0.790 & 0.717 & 0.863 & 0.472 & 0.888 & 0.958 & 0.988 \\
		Avg-Len & 1.894 & 2.012 & 1.562 & 1.034 & 0.792 & 0.787 & NA & NA & NA & 1.345 & 1.104 & 0.964 & 0.797 & 1.073 & 105.389 & 1.821 & 2.050 & 2.487 \\
		\cline{2-19}
		
		& \multicolumn{18}{c}{\bf $d=1000$, $n = 900$, pseudo-dense $\beta_0^{pd}$, and $x^{(2)}$ so that $m(x^{(2)})=x^{(2)T}\beta_0^{pd}$} \\ \cline{2-19}
		Avg-Bias & 0.199 & 0.200 & 0.174 & 0.225 & 0.218 & 0.195 & NA & NA & NA & 0.227 & 0.219 & 0.187 & 0.208 & 0.214 & 0.579 & 0.226 & 0.226 & 0.245 \\
		Avg-Cov & 0.983 & 0.992 & 0.988 & 0.880 & 0.876 & 0.896 & NA & NA & NA & 0.925 & 0.915 & 0.936 & 0.860 & 0.829 & 0.486 & 0.897 & 0.936 & 0.960 \\
		Avg-Len & 1.190 & 1.284 & 1.082 & 0.872 & 0.831 & 0.765 & NA & NA & NA & 0.970 & 0.948 & 0.836 & 0.744 & 0.755 & 189.113 & 0.874 & 0.986 & 1.196 \\
		\hline
		
		& \multicolumn{18}{c}{\bf $d=1000$, $n = 900$, sparse $\beta_0^{sp}$, and $x^{(3)}$ so that $m(x^{(3)})=x^{(3)T}\beta_0^{sp} = \beta_{0,100}^{(a)}$} \\ \cline{2-19}
		Avg-Bias & 0.167 & 0.169 & 0.156 & 0.224 & 0.224 & 0.194 & 0.818 & 0.814 & 0.862 & 0.000 & 0.001 & 0.000 & 0.212 & 0.353 & 0.914 & 0.247 & 0.282 & 0.318 \\
		Avg-Cov & 0.985 & 0.982 & 0.982 & 0.943 & 0.936 & 0.943 & 0.957 & 0.953 & 0.955 & 1.00 & 0.997 & 1.00 & 0.946 & 0.940 & 0.500 & 0.942 & 0.945 & 0.940 \\
		Avg-Len & 1.017 & 1.021 & 0.943 & 1.076 & 1.065 & 0.977 & 3.865 & 3.571 & 3.892 & 0.000 & 0.001 & 0.000 & 0.987 & 1.507 & 183.936 & 1.170 & 1.344 & 1.506 \\
		\cline{2-19}
		
		& \multicolumn{18}{c}{\bf $d=1000$, $n = 900$, dense $\beta_0^{de}$, and $x^{(3)}$ so that $m(x^{(3)})=x^{(3)T}\beta_0^{de} = \beta_{0,100}^{(b)}$} \\ \cline{2-19}
		Avg-Bias & 0.246 & 0.247 & 0.215 & 0.345 & 0.301 & 0.277 & 0.413 & 0.560 & 0.513 & 0.268 & 0.253 & 0.240 & 0.306 & 0.484 & 0.732 & 0.483 & 0.465 & 0.507 \\
		Avg-Cov & 0.976 & 0.979 & 0.978 & 0.867 & 0.811 & 0.863 & 0.989 & 0.949 & 0.960 & 0.239 & 0.325 & 0.327 & 0.843 & 0.963 & 0.495 & 0.968 & 0.992 & 0.996 \\
		Avg-Len & 1.563 & 1.570 & 1.332 & 1.307 & 1.003 & 1.022 & 2.625 & 2.187 & 2.728 & 0.481 & 0.512 & 0.454 & 1.058 & 2.279 & 81.276 & 2.665 & 3.058 & 3.428 \\
		\cline{2-19}
		
		& \multicolumn{18}{c}{\bf $d=1000$, $n = 900$, pseudo-dense $\beta_0^{pd}$, and $x^{(3)}$ so that $m(x^{(3)})=x^{(3)T}\beta_0^{pd} = \beta_{0,100}^{(c)}$} \\ \cline{2-19}
		Avg-Bias & 0.175 & 0.176 & 0.165 & 0.238 & 0.229 & 0.209 & 0.406 & 0.495 & 0.521 & 0.071 & 0.070 & 0.066 & 0.223 & 0.368 & 0.880 & 0.259 & 0.285 & 0.319 \\
		Avg-Cov & 0.982 & 0.978 & 0.982 & 0.931 & 0.928 & 0.934 & 0.965 & 0.959 & 0.948 & 0.139 & 0.181 & 0.171 & 0.931 & 0.936 & 0.516 & 0.950 & 0.957 & 0.960 \\
		Avg-Len & 1.094 & 1.096 & 1.002 & 1.103 & 1.052 & 0.994 & 2.628 & 2.263 & 2.614 & 0.140 & 0.170 & 0.144 & 0.988 & 1.593 & 143.525 & 1.280 & 1.470 & 1.647 \\
		\hline
		
		& \multicolumn{18}{c}{\bf $d=1000$, $n = 900$, sparse $\beta_0^{sp}$, and $x^{(4)}$ so that $m(x^{(4)})=x^{(4)T}\beta_0^{sp}$} \\ \cline{2-19}
		Avg-Bias & 0.147 & 0.147 & 0.126 & 0.204 & 0.185 & 0.169 & NA & NA & NA & 0.143 & 0.138 & 0.119 & 0.143 & 0.143 & 0.394 & 0.141 & 0.144 & 0.145 \\
		Avg-Cov & 1.00 & 1.00 & 1.00 & 0.873 & 0.903 & 0.869 & NA & NA & NA & 0.945 & 0.946 & 0.950 & 0.921 & 0.830 & 0.495 & 0.894 & 0.933 & 0.942 \\
		Avg-Len & 2.091 & 2.176 & 1.710 & 0.715 & 0.708 & 0.616 & NA & NA & NA & 0.650 & 0.650 & 0.561 & 0.607 & 0.513 & 386.621 & 0.563 & 0.650 & 0.695 \\
		\cline{2-19}
		
		& \multicolumn{18}{c}{\bf $d=1000$, $n = 900$, dense $\beta_0^{de}$, and $x^{(4)}$ so that $m(x^{(4)})=x^{(4)T}\beta_0^{de}$} \\ \cline{2-19}
		Avg-Bias & 0.235 & 0.233 & 0.167 & 0.385 & 0.272 & 0.284 & NA & NA & NA & 0.420 & 0.301 & 0.235 & 0.215 & 0.207 & 0.327 & 0.448 & 0.339 & 0.302 \\
		Avg-Cov & 1.00 & 1.00 & 1.00 & 0.631 & 0.686 & 0.623 & NA & NA & NA & 0.689 & 0.716 & 0.763 & 0.763 & 0.846 & 0.494 & 0.733 & 0.904 & 0.967 \\
		Avg-Len & 4.691 & 4.776 & 3.413 & 0.868 & 0.667 & 0.644 & NA & NA & NA & 0.985 & 0.797 & 0.684 & 0.650 & 0.772 & 41.916 & 1.283 & 1.483 & 1.582 \\
		\cline{2-19}
		
		& \multicolumn{18}{c}{\bf $d=1000$, $n = 900$, pseudo-dense $\beta_0^{pd}$, and $x^{(4)}$ so that $m(x^{(4)})=x^{(4)T}\beta_0^{pd}$} \\ \cline{2-19}
		Avg-Bias & 0.202 & 0.200 & 0.168 & 0.380 & 0.310 & 0.314 & NA & NA & NA & 0.166 & 0.159 & 0.133 & 0.164 & 0.151 & 0.385 & 0.168 & 0.160 & 0.157 \\
		Avg-Cov & 0.999 & 0.999 & 1.00 & 0.492 & 0.608 & 0.514 & NA & NA & NA & 0.920 & 0.913 & 0.925 & 0.855 & 0.836 & 0.522 & 0.863 & 0.927 & 0.946 \\
		Avg-Len & 2.638 & 2.722 & 2.106 & 0.732 & 0.700 & 0.626 & NA & NA & NA & 0.686 & 0.670 & 0.583 & 0.608 & 0.541 & 64.975 & 0.616 & 0.711 & 0.760 \\
		\hline
		
		& \multicolumn{18}{c}{\bf $d=1000$, $n = 900$, sparse $\beta_0^{sp}$, and $x^{(5)}$ so that $m(x^{(5)})=x^{(5)T}\beta_0^{sp}$} \\ \cline{2-19}
		Avg-Bias & 0.079 & 0.079 & 0.064 & 0.108 & 0.102 & 0.091 & NA & NA & NA & 0.005 & 0.009 & 0.004 & 0.042 & 0.017 & 0.071 & 0.020 & 0.022 & 0.027 \\
		Avg-Cov & 1.00 & 1.00 & 1.00 & 0.949 & 0.963 & 0.943 & NA & NA & NA & 0.866 & 0.657 & 0.881 & 1.00 & 0.000 & 0.515 & 0.908 & 0.913 & 0.957 \\
		Avg-Len & 7.400 & 7.480 & 5.642 & 0.537 & 0.533 & 0.430 & NA & NA & NA & 0.017 & 0.021 & 0.015 & 0.435 & 0.000 & 30.483 & 0.096 & 0.107 & 0.139 \\
		\hline\hline
\end{tabular}}
\caption{Simulation results under the Toeplitz (or autoregressive) covariance matrix $\Sigma^{\mathrm{ar}}$, $t_2$-distributed noise, and the MAR setting in \autoref{subsec:sim_design}. Nearly all the standard errors for the above average quantities are less than 0.01 except for those associated with ``R-Proj'' or ``SSL-AIPW'' and are thus omitted.}
\label{table:AR_MAR_terr_res}
\end{center}
\end{table}

\begin{figure}[hbtp!]
	\captionsetup[subfigure]{justification=centering}
	\begin{subfigure}[t]{0.49\linewidth}
		\centering
		\includegraphics[width=1\linewidth]{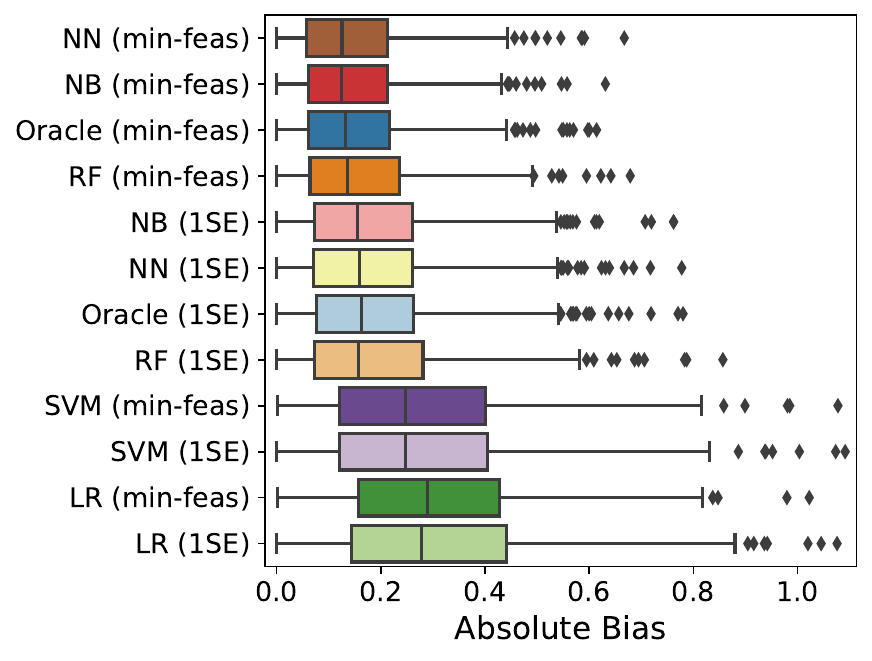}
	\end{subfigure}
	\hfil
	\begin{subfigure}[t]{0.49\linewidth}
		\centering		\includegraphics[width=1\linewidth]{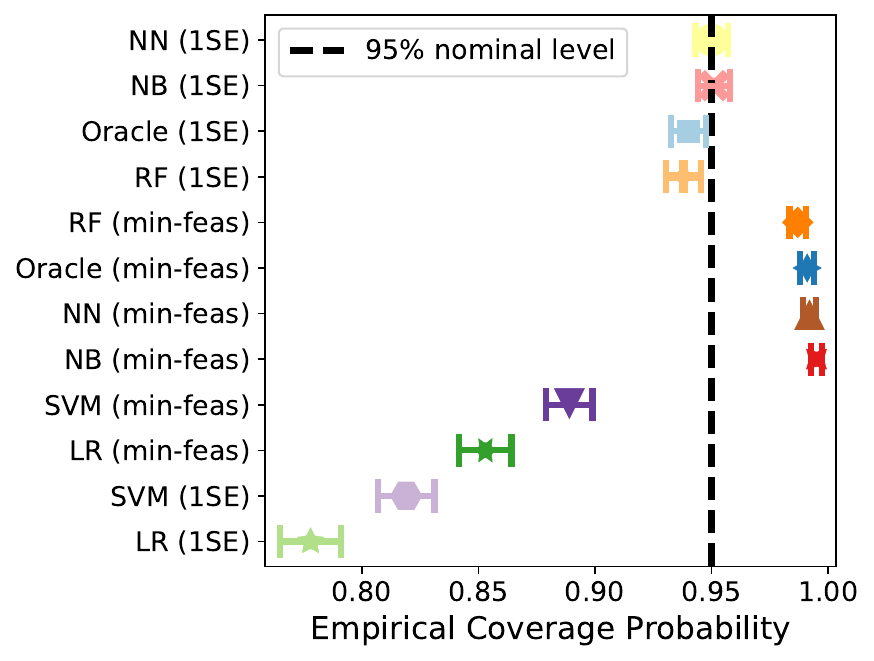}
	\end{subfigure}
	\begin{subfigure}[c]{0.49\linewidth}
		\centering
		\includegraphics[width=1\linewidth]{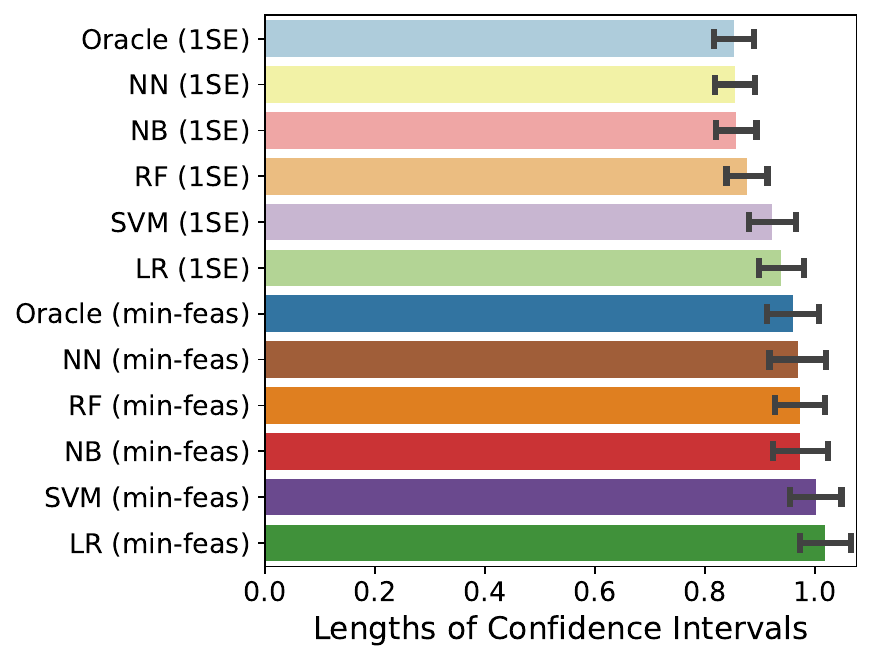}
	\end{subfigure}
	\hfil
	\begin{subfigure}[c]{0.49\linewidth}
		\centering		\includegraphics[width=1\linewidth]{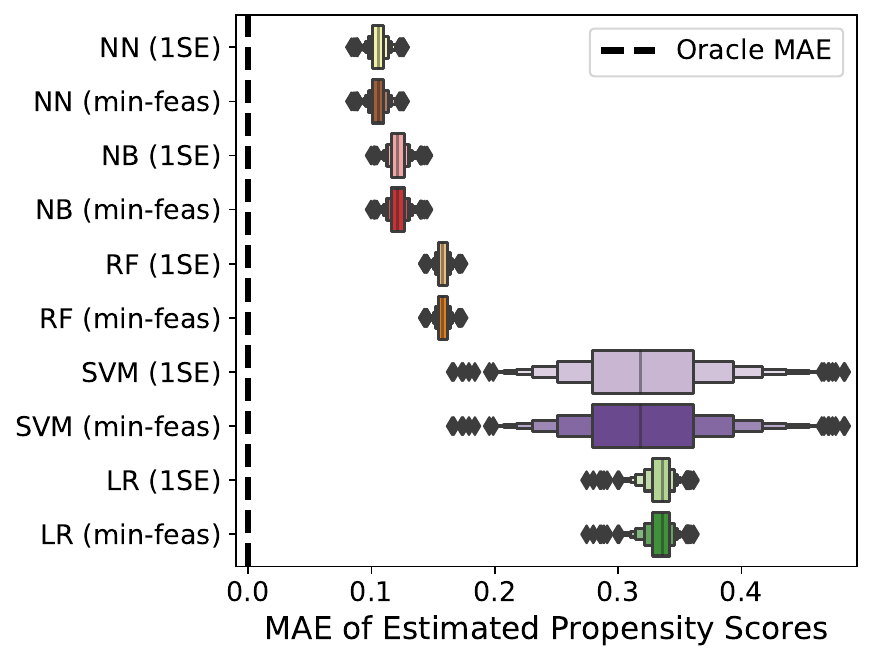}
	\end{subfigure}
	\begin{subfigure}[c]{0.995\linewidth}
		\centering		\includegraphics[width=1\linewidth]{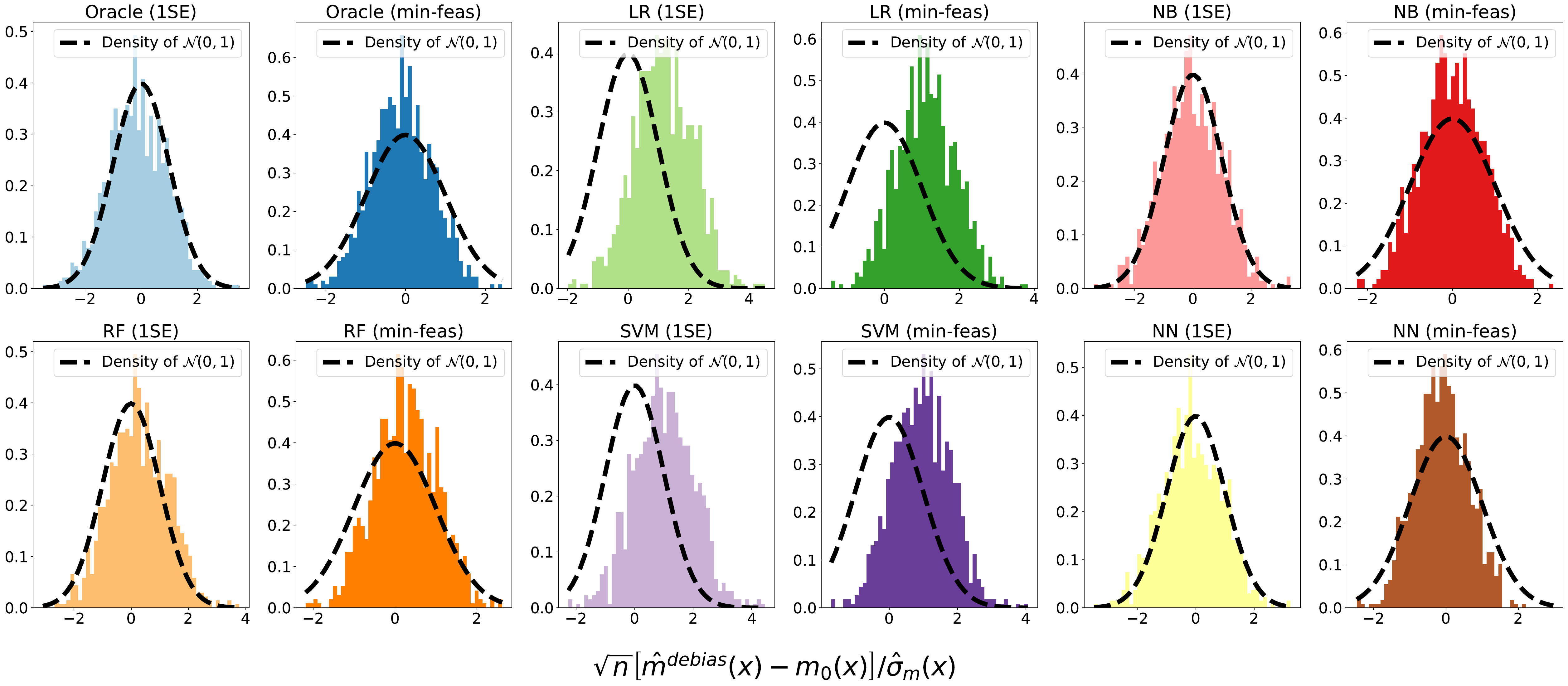}
	\end{subfigure}
	\caption{Simulation results for our debiasing method with dense $\beta_0^{de}$ and sparse $x^{(2)}$ when the propensity scores are estimated by various machine learning methods under the new MAR setting \eqref{MAR_mis}. {\bf Top Left:} Boxplots of the absolute bias. {\bf Top Right:} Coverage probabilities with standard error bars. {\bf Middle Left:} Average lengths of confidence intervals with standard deviation bars. {\bf Middle Right:} Letter-valued plots of the MAE of estimated propensity scores. {\bf Bottom Two Rows:} Histograms of the studentized debiased estimators under different methods. We also order the performance metrics according to their means in the first four panels.}
	\label{fig:Cir_Nonpar_Prop_Score1}
\end{figure}

\begin{table}[hbtp!]
	\begin{center}
		\resizebox{\textwidth}{!}{	
			\begin{tabular}{ccccccccccccc}
				\hline \hline 
				& \begin{tabular}{@{}c@{}} Oracle \\ (1SE)\end{tabular}  & \begin{tabular}{@{}c@{}} Oracle \\ (min-feas)\end{tabular} 
				& \begin{tabular}{@{}c@{}} LR \\ (1SE)\end{tabular} & \begin{tabular}{@{}c@{}} LR \\ (min-feas)\end{tabular}  
				& \begin{tabular}{@{}c@{}} NB \\ (1SE)\end{tabular}  & \begin{tabular}{@{}c@{}} NB \\ (min-feas)\end{tabular} 
				& \begin{tabular}{@{}c@{}} RF \\ (1SE)\end{tabular} & \begin{tabular}{@{}c@{}} RF \\ (min-feas)\end{tabular} 
				& \begin{tabular}{@{}c@{}} SVM \\ (1SE)\end{tabular} 
				& \begin{tabular}{@{}c@{}} SVM \\ (min-feas)\end{tabular} 
				& \begin{tabular}{@{}c@{}} NN \\ (1SE)\end{tabular} & \begin{tabular}{@{}c@{}} NN \\ (min-feas)\end{tabular} \\ \cline{1-13}
				{\bf Avg-MAE} & 0 & 0 & 0.334 & 0.334 & 0.121 & 0.121 & 0.158 & 0.158 & 0.321 & 0.321 & {\bf 0.106} & {\bf 0.106} \\
				\cline{1-13}
				& \multicolumn{12}{c}{\bf $d=1000$, $n = 900$, sparse $\beta_0^{sp}$, and $x^{(1)}$ so that $m(x^{(1)})=x^{(1)T}\beta_0^{sp} = \beta_{0,1}^{(a)}$} \\ \cline{2-13}
				Avg-Bias & 0.034 & 0.033 & 0.048 & 0.047 & 0.035 & 0.034 & 0.037 & 0.036 & 0.045 & 0.044 & 0.034 & {\bf 0.033} \\
				Avg-Cov & 0.902 & 0.928 & 0.815 & 0.824 & 0.899 & 0.924 & 0.885 & 0.918 & 0.828 & 0.846 & 0.901 & {\bf 0.931} \\
				Avg-Len & 0.140 & 0.150 & 0.157 & 0.160 & 0.139 & 0.151 & 0.145 & 0.153 & 0.154 & 0.158 & {\bf 0.134} & 0.139 \\
				\cline{2-13}
				
				& \multicolumn{12}{c}{\bf $d=1000$, $n = 900$, dense $\beta_0^{de}$, and $x^{(1)}$ so that $m(x^{(1)})=x^{(1)T}\beta_0^{de} = \beta_{0,1}^{(b)}$} \\ \cline{2-13}
				Avg-Bias & 0.172 & 0.144 & 0.234 & 0.226 & 0.174 & 0.142 & 0.181 & 0.157 & 0.221 & 0.211 & 0.173 & {\bf 0.138} \\
				Avg-Cov & 0.904 & 0.974 & 0.871 & 0.887 & 0.899 & 0.977 & 0.904 & {\bf 0.961} & 0.875 & 0.899 & 0.902 & 0.978 \\
				Avg-Len & 0.747 & 0.800 & 0.842 & 0.858 & 0.746 & 0.807 & 0.776 & 0.820 & 0.826 & 0.843 & {\bf 0.744} & 0.810 \\
				\cline{2-13}
				
				& \multicolumn{12}{c}{\bf $d=1000$, $n = 900$, pseudo-dense $\beta_0^{pd}$, and $x^{(1)}$ so that $m(x^{(1)})=x^{(1)T}\beta_0^{pd} = \beta_{0,1}^{(c)}$} \\ \cline{2-13}	
				Avg-Bias & 0.050 & 0.045 & 0.068 & 0.066 & 0.051 & 0.045 & 0.053 & 0.049 & 0.064 & 0.062 & 0.050 & {\bf 0.044} \\
				Avg-Cov & 0.913 & 0.966 & 0.861 & 0.879 & 0.909 & 0.968 & 0.911 & {\bf 0.959} & 0.873 & 0.889 & 0.912 & 0.973 \\
				Avg-Len & 0.218 & 0.233 & 0.245 & 0.250 & {\bf 0.217} & 0.235 & 0.226 & 0.239 & 0.241 & 0.246 & {\bf 0.217} & 0.236 \\
				\hline
				
				& \multicolumn{12}{c}{\bf $d=1000$, $n = 900$, sparse $\beta_0^{sp}$, and $x^{(2)}$ so that $m(x^{(2)})=x^{(2)T}\beta_0^{sp}$} \\ \cline{2-13}
				Avg-Bias & 0.035 & 0.037 & 0.050 & 0.053 & {\bf 0.035} & 0.038 & 0.036 & 0.038 & 0.047 & 0.050 & 0.036 & 0.038 \\
				Avg-Cov & 0.924 & 0.945 & 0.845 & 0.850 & 0.924 & 0.940 & 0.937 & {\bf 0.945} & 0.864 & 0.867 & 0.923 & 0.940 \\
				Avg-Len & 0.159 & 0.179 & 0.176 & 0.191 & {\bf 0.160} & 0.182 & 0.164 & 0.182 & 0.172 & 0.187 & {\bf 0.160} & 0.181 \\
				\cline{2-13}
				
				& \multicolumn{12}{c}{\bf $d=1000$, $n = 900$, dense $\beta_0^{de}$, and $x^{(2)}$ so that $m(x^{(2)})=x^{(2)T}\beta_0^{de}$} \\ \cline{2-13}
				Avg-Bias & 0.184 & 0.150 & 0.305 & 0.303 & 0.180 & 0.147 & 0.188 & 0.160 & 0.276 & 0.271 & 0.181 & {\bf 0.146} \\
				Avg-Cov & 0.940 & 0.991 & 0.778 & 0.853 & 0.951 & 0.995 & 0.938 & 0.987 & 0.819 & 0.889 & {\bf 0.950} & 0.992 \\
				Avg-Len & 0.853 & 0.960 & 0.939 & 1.019 & 0.857 & 0.974 & 0.877 & 0.973 & 0.923 & 1.002 & {\bf 0.855} & 0.969 \\
				\cline{2-13}
				
				& \multicolumn{12}{c}{\bf $d=1000$, $n = 900$, pseudo-dense $\beta_0^{pd}$, and $x^{(2)}$ so that $m(x^{(2)})=x^{(2)T}\beta_0^{pd}$} \\ \cline{2-13}
				Avg-Bias & 0.056 & 0.051 & 0.085 & 0.088 & 0.054 & {\bf 0.051} & 0.055 & 0.052 & 0.078 & 0.079 & 0.055 & {\bf 0.051} \\
				Avg-Cov & 0.922 & 0.972 & 0.813 & 0.837 & 0.923 & 0.976 & 0.931 & {\bf 0.967} & 0.836 & 0.866 & 0.920 & 0.968 \\
				Avg-Len & 0.249 & 0.280 & 0.274 & 0.297 & 0.250 & 0.284 & 0.256 & 0.284 & 0.269 & 0.292 & {\bf 0.249} & 0.282 \\
				\hline
				
				& \multicolumn{12}{c}{\bf $d=1000$, $n = 900$, sparse $\beta_0^{sp}$, and $x^{(3)}$ so that $m(x^{(3)})=x^{(3)T}\beta_0^{sp} = \beta_{0,100}^{(a)}$} \\ \cline{2-13}
				Avg-Bias & 0.0307 & 0.033 & 0.033 & 0.034 & {\bf 0.031} & 0.033 & {\bf 0.031} & 0.032 & 0.032 & 0.033 & {\bf 0.031} & 0.034 \\
				Avg-Cov & 0.946 & 0.939 & 0.939 & 0.936 & {\bf 0.948} & 0.938 & 0.945 & 0.944 & 0.941 & 0.938 & 0.946 & 0.938 \\
				Avg-Len & 0.1485 & 0.161 & 0.158 & 0.161 & {\bf 0.149} & 0.162 & 0.150 & 0.160 & 0.155 & 0.159 & {\bf 0.149} & 0.164 \\
				\cline{2-13}
				
				& \multicolumn{12}{c}{\bf $d=1000$, $n = 900$, dense $\beta_0^{de}$, and $x^{(3)}$ so that $m(x^{(3)})=x^{(3)T}\beta_0^{de} = \beta_{0,100}^{(b)}$} \\ \cline{2-13}
				Avg-Bias & 0.164 & 0.139 & 0.162 & 0.155 & 0.163 & 0.137 & 0.162 & 0.138 & 0.155 & 0.149 & 0.162 & {\bf 0.132} \\
				Avg-Cov & 0.945 & 0.991 & 0.967 & 0.976 & 0.947 & 0.991 & 0.947 & 0.989 & 0.968 & 0.978 & {\bf 0.950} & 0.993 \\
				Avg-Len & 0.794 & 0.859 & 0.843 & 0.862 & {\bf 0.795} & 0.867 & 0.804 & 0.858 & 0.829 & 0.849 & 0.796 & 0.875 \\
				\cline{2-13}
				
				& \multicolumn{12}{c}{\bf $d=1000$, $n = 900$, pseudo-dense $\beta_0^{pd}$, and $x^{(3)}$ so that $m(x^{(3)})=x^{(3)T}\beta_0^{pd} = \beta_{0,100}^{(c)}$} \\ \cline{2-13}
				Avg-Bias & 0.047 & 0.045 & 0.049 & 0.048 & 0.047 & 0.046 & 0.047 & {\bf 0.045} & 0.048 & 0.047 & 0.047 & {\bf 0.045} \\
				Avg-Cov & 0.956 & 0.973 & 0.966 & 0.967 & {\bf 0.957} & 0.974 & 0.963 & 0.974 & 0.966 & 0.967 & 0.959 & 0.975 \\
				Avg-Len & 0.232 & 0.250 & 0.246 & 0.251 & {\bf 0.232} & 0.253 & 0.234 & 0.250 & 0.242 & 0.247 & {\bf 0.232} & 0.255 \\
				\hline
				
				& \multicolumn{12}{c}{\bf $d=1000$, $n = 900$, sparse $\beta_0^{sp}$, and $x^{(4)}$ so that $m(x^{(4)})=x^{(4)T}\beta_0^{sp}$} \\ \cline{2-13}
				Avg-Bias & 0.030 & 0.031 & 0.041 & 0.041 & {\bf 0.030} & 0.031 & 0.031 & 0.032 & 0.038 & 0.038 & {\bf 0.030} & 0.031 \\
				Avg-Cov & 0.938 & 0.949 & 0.879 & 0.883 & 0.940 & {\bf 0.949} & 0.944 & 0.945 & 0.898 & 0.903 & 0.938 & 0.948 \\
				Avg-Len & 0.1392 & 0.149 & 0.158 & 0.161 & {\bf 0.139} & 0.150 & 0.145 & 0.153 & 0.155 & 0.158 & {\bf 0.139} & 0.151 \\
				\cline{2-13}
				
				& \multicolumn{12}{c}{\bf $d=1000$, $n = 900$, dense $\beta_0^{de}$, and $x^{(4)}$ so that $m(x^{(4)})=x^{(4)T}\beta_0^{de}$} \\ \cline{2-13}
				Avg-Bias & 0.166 & 0.145 & 0.193 & 0.192 & 0.165 & {\bf 0.139} & 0.157 & 0.141 & 0.180 & 0.177 & 0.166 & {\bf 0.139} \\
				Avg-Cov & 0.934 & 0.972 & 0.927 & 0.932 & 0.933 & 0.976 & {\bf 0.945} & 0.973 & 0.936 & {\bf 0.945} & 0.935 & 0.979 \\
				Avg-Len & 0.745 & 0.798 & 0.845 & 0.860 & 0.744 & 0.804 & 0.775 & 0.818 & 0.829 & 0.845 & {\bf 0.742} & 0.806 \\
				\cline{2-13}
				
				& \multicolumn{12}{c}{\bf $d=1000$, $n = 900$, pseudo-dense $\beta_0^{pd}$, and $x^{(4)}$ so that $m(x^{(4)})=x^{(4)T}\beta_0^{pd}$} \\ \cline{2-13}
				Avg-Bias & 0.050 & 0.046 & 0.055 & 0.056 & 0.050 & {\bf 0.045} & 0.047 & {\bf 0.045} & 0.053 & 0.052 & 0.050 & {\bf 0.045} \\
				Avg-Cov & 0.915 & 0.960 & 0.931 & 0.932 & 0.918 & 0.970 & {\bf 0.943} & 0.971 & 0.932 & 0.937 & 0.917 & 0.967 \\
				Avg-Len & 0.217 & 0.232 & 0.246 & 0.251 & 0.217 & 0.234 & 0.226 & 0.238 & 0.242 & 0.246 & {\bf 0.216} & 0.235 \\
				\hline\hline
		\end{tabular}}
		\caption{Simulation results for our proposed debiasing method when the propensity scores are estimated by various machine learning methods under the circulant symmetric covariance matrix $\Sigma^{\mathrm{cs}}$, Gaussian noise $\mathcal{N}(0,1)$, and the new MAR setting \eqref{MAR_mis}. Note that the ``Avg-MAE'' for propensity score estimation is independent of the choices of $\beta_0$ and $x$, so we place it in the top row. We also highlight the best results in terms of the evaluation metrics for each simulation design, excluding those with the oracle propensity scores, in boldface.}
		\label{table:CirSym_MAR_Non_Prop_Score}
	\end{center}
\end{table}

\subsection{Simulations Under the MCAR Setting}
\label{subapp:mcar_res}

We follow the comparative simulation designs in \autoref{subsec:sim_design} but generate the missingness indicators $R_i, i=1,...,n$, under an MCAR mechanism with $\mathrm{P}(R_i=1)=0.7$. Under MCAR, the inverses of both the circulant symmetric matrix $\Sigma^{\mathrm{cs}}$ and the Toeplitz matrix $\Sigma^{\mathrm{ar}}$ satisfy our sufficient conditions for the sparse approximation of the population dual solution $\ell_0(x)$ in Lemma~\ref{lem:suff_sparse_dual}, as illustrated in Example~\ref{example:linear_reg_mcar}. The simulation results under Gaussian $\mathcal{N}(0,1)$ noise are shown in \autoref{table:cirsym_MCAR_res} and \autoref{table:AR_MCAR_res}, and the conclusions are similar to the main findings from the MAR simulation studies in \autoref{subsec:sim_results}. We omit results for other noise distributions because they exhibit patterns similar to those under Gaussian noise, and we have demonstrated in \autoref{subsec:heavy_noises} demonstrates that our proposed debiasing method is robust to other heavy-tailed noise distributions.

\begin{table}[hbtp!]
\begin{center}
\resizebox{\textwidth}{!}{	
	\begin{tabular}{ccccccccccccccccccc}
		\hline \hline 
		& \begin{tabular}{@{}c@{}} DL-Jav \\ (CC)\end{tabular} & \begin{tabular}{@{}c@{}} DL-Jav \\ (IPW)\end{tabular} & \begin{tabular}{@{}c@{}} DL-Jav \\ (Oracle)\end{tabular} 
		& \begin{tabular}{@{}c@{}} DL-vdG \\ (CC)\end{tabular} & \begin{tabular}{@{}c@{}} DL-vdG \\ (IPW)\end{tabular} & \begin{tabular}{@{}c@{}} DL-vdG \\ (Oracle)\end{tabular} 
		& \begin{tabular}{@{}c@{}} R-Proj \\ (CC)\end{tabular} & \begin{tabular}{@{}c@{}} R-Proj \\ (IPW)\end{tabular} & \begin{tabular}{@{}c@{}} R-Proj \\ (Oracle)\end{tabular} 
		& \begin{tabular}{@{}c@{}} Refit \\ (CC)\end{tabular} & \begin{tabular}{@{}c@{}} Refit \\ (IPW)\end{tabular} & \begin{tabular}{@{}c@{}} Refit \\ (Oracle)\end{tabular} 
		& \begin{tabular}{@{}c@{}} DDR \end{tabular} & \begin{tabular}{@{}c@{}} SAS \end{tabular} & \begin{tabular}{@{}c@{}} SSL-AIPW \end{tabular}
		& \begin{tabular}{@{}c@{}} \textbf{Debias} \\ (min-CV)\end{tabular} & \begin{tabular}{@{}c@{}} \textbf{Debias} \\ (1SE)\end{tabular} & \begin{tabular}{@{}c@{}} \textbf{Debias} \\ (min-feas)\end{tabular} \\ \cline{2-19}
		& \multicolumn{18}{c}{\bf $d=1000$, $n = 900$, sparse $\beta_0^{sp}$, and $x^{(1)}$ so that $m(x^{(1)})=x^{(1)T}\beta_0^{sp} = \beta_{0,1}^{(a)}$} \\ \cline{2-19}
		Avg-Bias & 0.033 & 0.033 & 0.028 & 0.039 & 0.036 & 0.032 & 0.104 & 0.283 & 0.315 & 0.031 & 0.032 & 0.027 & 0.040 & 0.039 & 0.183 & 0.042 & 0.039 & 0.039 \\
		Avg-Cov & 0.957 & 0.957 & 0.967 & 0.890 & 0.904 & 0.904 & 0.969 & 0.955 & 0.940 & 0.957 & 0.952 & 0.953 & 0.895 & 0.879 & 0.513 & 0.848 & 0.905 & 0.912 \\
		Avg-Len & 0.173 & 0.173 & 0.145 & 0.159 & 0.155 & 0.133 & 0.379 & 0.376 & 0.638 & 0.160 & 0.158 & 0.133 & 0.164 & 0.151 & 9.534 & 0.149 & 0.161 & 0.164 \\
		\cline{2-19}
		
		& \multicolumn{18}{c}{\bf $d=1000$, $n = 900$, dense $\beta_0^{de}$, and $x^{(1)}$ so that $m(x^{(1)})=x^{(1)T}\beta_0^{de} = \beta_{0,1}^{(b)}$} \\ \cline{2-19}
		Avg-Bias & 0.146 & 0.146 & 0.108 & 0.178 & 0.157 & 0.136 & 0.096 & 0.202 & 0.207 & 0.181 & 0.177 & 0.143 & 0.147 & 0.152 & 0.213 & 0.228 & 0.186 & 0.180 \\
		Avg-Cov & 0.995 & 0.995 & 0.996 & 0.891 & 0.881 & 0.913 & 0.998 & 0.997 & 0.996 & 0.923 & 0.866 & 0.919 & 0.375 & 0.895 & 0.474 & 0.852 & 0.941 & 0.949 \\
		Avg-Len & 0.970 & 0.971 & 0.844 & 0.742 & 0.616 & 0.583 & 1.255 & 2.215 & 1.784 & 0.776 & 0.680 & 0.610 & 0.175 & 0.611 & 4.591 & 0.808 & 0.872 & 0.888 \\
		\cline{2-19}
		
		& \multicolumn{18}{c}{\bf $d=1000$, $n = 900$, pseudo-dense $\beta_0^{pd}$, and $x^{(1)}$ so that $m(x^{(1)})=x^{(1)T}\beta_0^{pd} = \beta_{0,1}^{(c)}$} \\ \cline{2-19}
		Avg-Bias & 0.044 & 0.044 & 0.036 & 0.053 & 0.047 & 0.041 & 0.059 & 0.104 & 0.088 & 0.047 & 0.045 & 0.037 & 0.047 & 0.053 & 0.369 & 0.063 & 0.055 & 0.054 \\
		Avg-Cov & 0.986 & 0.987 & 0.996 & 0.890 & 0.889 & 0.908 & 0.989 & 0.985 & 0.985 & 0.940 & 0.937 & 0.938 & 0.893 & 0.890 & 0.472 & 0.856 & 0.937 & 0.942 \\
		Avg-Len & 0.280 & 0.281 & 0.237 & 0.218 & 0.195 & 0.173 & 0.407 & 0.347 & 0.527 & 0.218 & 0.201 & 0.172 & 0.193 & 0.215 & 6.362 & 0.237 & 0.256 & 0.260 \\
		\hline
		
		& \multicolumn{18}{c}{\bf $d=1000$, $n = 900$, sparse $\beta_0^{sp}$, and $x^{(2)}$ so that $m(x^{(2)})=x^{(2)T}\beta_0^{sp}$} \\ \cline{2-19}
		Avg-Bias & 0.042 & 0.042 & 0.037 & 0.066 & 0.056 & 0.051 & NA & NA & NA & 0.034 & 0.035 & 0.030 & 0.065 & 0.038 & 0.283 & 0.037 & 0.038 & 0.041 \\
		Avg-Cov & 0.986 & 0.991 & 0.990 & 0.812 & 0.860 & 0.839 & NA & NA & NA & 0.953 & 0.943 & 0.948 & 0.830 & 0.914 & 0.530 & 0.924 & 0.939 & 0.950 \\
		Avg-Len & 0.258 & 0.286 & 0.236 & 0.209 & 0.204 & 0.172 & NA & NA & NA & 0.173 & 0.171 & 0.144 & 0.215 & 0.167 & 10.778 & 0.171 & 0.180 & 0.196 \\
		\cline{2-19}
		
		& \multicolumn{18}{c}{\bf $d=1000$, $n = 900$, dense $\beta_0^{de}$, and $x^{(2)}$ so that $m(x^{(2)})=x^{(2)T}\beta_0^{de}$} \\ \cline{2-19}
		Avg-Bias & 0.168 & 0.168 & 0.130 & 0.327 & 0.218 & 0.226 & NA & NA & NA & 0.239 & 0.249 & 0.184 & 0.187 & 0.170 & 0.347 & 0.223 & 0.214 & 0.193 \\
		Avg-Cov & 1.00 & 1.00 & 1.00 & 0.774 & 0.862 & 0.821 & NA & NA & NA & 0.852 & 0.800 & 0.887 & 0.393 & 0.886 & 0.443 & 0.905 & 0.932 & 0.967 \\
		Avg-Len & 1.640 & 1.759 & 1.514 & 0.977 & 0.811 & 0.755 & NA & NA & NA & 0.913 & 0.807 & 0.726 & 0.228 & 0.674 & 6.303 & 0.922 & 0.975 & 1.057 \\
		\cline{2-19}
		
		& \multicolumn{18}{c}{\bf $d=1000$, $n = 900$, pseudo-dense $\beta_0^{pd}$, and $x^{(2)}$ so that $m(x^{(2)})=x^{(2)T}\beta_0^{pd}$} \\ \cline{2-19}
		Avg-Bias & 0.052 & 0.052 & 0.046 & 0.095 & 0.072 & 0.069 & NA & NA & NA & 0.059 & 0.059 & 0.046 & 0.062 & 0.058 & 0.489 & 0.066 & 0.064 & 0.061 \\
		Avg-Cov & 0.999 & 1.00 & 1.00 & 0.792 & 0.864 & 0.806 & NA & NA & NA & 0.927 & 0.903 & 0.928 & 0.900 & 0.895 & 0.475 & 0.901 & 0.926 & 0.962 \\
		Avg-Len & 0.461 & 0.498 & 0.414 & 0.287 & 0.256 & 0.224 & NA & NA & NA & 0.261 & 0.240 & 0.205 & 0.253 & 0.238 & 9.531 & 0.270 & 0.286 & 0.310 \\
		\hline
		
		& \multicolumn{18}{c}{\bf $d=1000$, $n = 900$, sparse $\beta_0^{sp}$, and $x^{(3)}$ so that $m(x^{(3)})=x^{(3)T}\beta_0^{sp}=\beta_{0,100}^{(a)}$} \\ \cline{2-19}
		Avg-Bias & 0.026 & 0.026 & 0.022 & 0.032 & 0.031 & 0.027 & 0.320 & 1.289 & 0.739 & 0.00 & 0.00 & 0.00 & 0.034 & 0.031 & 0.101 & 0.032 & 0.034 & 0.035 \\
		Avg-Cov & 0.990 & 0.990 & 0.990 & 0.942 & 0.941 & 0.950 & 0.972 & 0.963 & 0.937 & 1.00 & 0.999 & 1.00 & 0.941 & 0.944 & 0.503 & 0.934 & 0.944 & 0.944 \\
		Avg-Len & 0.173 & 0.173 & 0.145 & 0.159 & 0.155 & 0.133 & 2.341 & 5.584 & 1.177 & 0.00 & 0.00 & 0.00 & 0.161 & 0.149 & 11.692 & 0.149 & 0.161 & 0.165 \\
		\cline{2-19}
		
		& \multicolumn{18}{c}{\bf $d=1000$, $n = 900$, dense $\beta_0^{de}$, and $x^{(3)}$ so that $m(x^{(3)})=x^{(3)T}\beta_0^{de}=\beta_{0,100}^{(b)}$} \\ \cline{2-19}
		Avg-Bias & 0.102 & 0.102 & 0.090 & 0.165 & 0.134 & 0.128 & 0.122 & 0.196 & 0.084 & 0.186 & 0.196 & 0.191 & 0.128 & 0.143 & 0.101 & 0.200 & 0.175 & 0.170 \\
		Avg-Cov & 0.999 & 0.999 & 0.999 & 0.923 & 0.921 & 0.929 & 0.998 & 0.997 & 0.997 & 0.00 & 0.040 & 0.004 & 0.364 & 0.906 & 0.512 & 0.897 & 0.957 & 0.965 \\
		Avg-Len & 0.971 & 0.972 & 0.844 & 0.742 & 0.616 & 0.583 & 1.343 & 2.516 & 1.728 & 0.008 & 0.057 & 0.021 & 0.175 & 0.608 & 3.763 & 0.807 & 0.870 & 0.891 \\
		\cline{2-19}
		
		& \multicolumn{18}{c}{\bf $d=1000$, $n = 900$, pseudo-dense $\beta_0^{pd}$, and $x^{(3)}$ so that $m(x^{(3)})=x^{(3)T}\beta_0^{pd}=\beta_{0,100}^{(c)}$} \\ \cline{2-19}
		Avg-Bias & 0.032 & 0.032 & 0.028 & 0.046 & 0.041 & 0.038 & 0.333 & 0.369 & 0.316 & 0.040 & 0.042 & 0.040 & 0.041 & 0.048 & 0.096 & 0.054 & 0.052 & 0.051 \\
		Avg-Cov & 0.998 & 0.998 & 0.999 & 0.934 & 0.923 & 0.942 & 0.991 & 0.990 & 0.987 & 0.00 & 0.015 & 0.00 & 0.936 & 0.913 & 0.483 & 0.923 & 0.947 & 0.952 \\
		Avg-Len & 0.281 & 0.281 & 0.237 & 0.218 & 0.195 & 0.173 & 2.153 & 1.071 & 28.765 & 0.001 & 0.008 & 0.002 & 0.192 & 0.213 & 7.919 & 0.237 & 0.255 & 0.261 \\
		\hline
		
		& \multicolumn{18}{c}{\bf $d=1000$, $n = 900$, sparse $\beta_0^{sp}$, and $x^{(4)}$ so that $m(x^{(4)})=x^{(4)T}\beta_0^{sp}$} \\ \cline{2-19}
		Avg-Bias & 0.042 & 0.042 & 0.037 & 0.075 & 0.061 & 0.056 & NA & NA & NA & 0.035 & 0.036 & 0.030 & 0.073 & 0.056 & 0.379 & 0.050 & 0.041 & 0.038 \\
		Avg-Cov & 1.00 & 1.00 & 1.00 & 0.799 & 0.866 & 0.827 & NA & NA & NA & 0.942 & 0.932 & 0.945 & 0.825 & 0.729 & 0.514 & 0.813 & 0.900 & 0.938 \\
		Avg-Len & 0.530 & 0.560 & 0.449 & 0.227 & 0.221 & 0.183 & NA & NA & NA & 0.173 & 0.171 & 0.144 & 0.232 & 0.157 & 11.657 & 0.160 & 0.171 & 0.181 \\
		\cline{2-19}
		
		& \multicolumn{18}{c}{\bf $d=1000$, $n = 900$, dense $\beta_0^{de}$, and $x^{(4)}$ so that $m(x^{(4)})=x^{(4)T}\beta_0^{de}$} \\ \cline{2-19}
		Avg-Bias & 0.159 & 0.159 & 0.145 & 0.855 & 0.488 & 0.590 & NA & NA & NA & 0.787 & 0.716 & 0.620 & 0.235 & 0.617 & 0.576 & 1.244 & 0.982 & 0.755 \\
		Avg-Cov & 1.00 & 1.00 & 1.00 & 0.067 & 0.413 & 0.156 & NA & NA & NA & 0.080 & 0.066 & 0.074 & 0.306 & 0.054 & 0.450 & 0.00 & 0.040 & 0.110 \\
		Avg-Len & 3.939 & 4.073 & 3.449 & 1.060 & 0.880 & 0.806 & NA & NA & NA & 0.846 & 0.745 & 0.669 & 0.249 & 0.630 & 7.242 & 0.864 & 0.925 & 0.979 \\
		\cline{2-19}
		
		& \multicolumn{18}{c}{\bf $d=1000$, $n = 900$, pseudo-dense $\beta_0^{pd}$, and $x^{(4)}$ so that $m(x^{(4)})=x^{(4)T}\beta_0^{pd}$} \\ \cline{2-19}
		Avg-Bias & 0.076 & 0.075 & 0.082 & 0.239 & 0.166 & 0.170 & NA & NA & NA & 0.145 & 0.127 & 0.110 & 0.112 & 0.261 & 0.543 & 0.358 & 0.278 & 0.212 \\
		Avg-Cov & 1.00 & 1.00 & 1.00 & 0.100 & 0.324 & 0.170 & NA & NA & NA & 0.363 & 0.399 & 0.398 & 0.676 & 0.013 & 0.448 & 0.00 & 0.061 & 0.161 \\
		Avg-Len & 1.072 & 1.114 & 0.907 & 0.312 & 0.278 & 0.239 & NA & NA & NA & 0.239 & 0.221 & 0.189 & 0.285 & 0.223 & 7.460 & 0.253 & 0.271 & 0.287 \\
		\hline
		
		& \multicolumn{18}{c}{\bf $d=1000$, $n = 900$, sparse $\beta_0^{sp}$, and $x^{(5)}$ so that $m(x^{(5)})=x^{(5)T}\beta_0^{sp}$} \\ \cline{2-19}
		Avg-Bias & 0.035 & 0.035 & 0.031 & 0.042 & 0.040 & 0.032 & NA & NA & NA & 0.002 & 0.005 & 0.002 & 0.037 & 0.015 & 0.075 & 0.017 & 0.017 & 0.021 \\
		Avg-Cov & 1.00 & 1.00 & 1.00 & 0.940 & 0.945 & 0.939 & NA & NA & NA & 0.870 & 0.647 & 0.870 & 0.948 & 0.00 & 0.500 & 0.00 & 0.00 & 0.939 \\
		Avg-Len & 1.724 & 1.752 & 1.367 & 0.199 & 0.194 & 0.154 & NA & NA & NA & 0.010 & 0.011 & 0.008 & 0.189 & 0.00 & 8.699 & 0.00 & 0.006 & 0.100 \\
		\hline\hline
\end{tabular}}
\caption{Simulation results under the circulant symmetric covariance matrix $\Sigma^{\mathrm{cs}}$, Gaussian noise $\mathcal{N}(0,1)$, and the MCAR setting. Nearly all the standard errors for the above average quantities are less than 0.01 except for those associated with ``DL-vdG'', ``R-Proj'', and ``SSL-AIPW'' and are thus omitted.}
\label{table:cirsym_MCAR_res}
\end{center}
\end{table}

\begin{table}[hbtp!]
\begin{center}
	\resizebox{\textwidth}{!}{	
		\begin{tabular}{ccccccccccccccccccc}
			\hline \hline 
			& \begin{tabular}{@{}c@{}} DL-Jav \\ (CC)\end{tabular} & \begin{tabular}{@{}c@{}} DL-Jav \\ (IPW)\end{tabular} & \begin{tabular}{@{}c@{}} DL-Jav \\ (Oracle)\end{tabular} & \begin{tabular}{@{}c@{}} DL-vdG \\ (CC)\end{tabular} & \begin{tabular}{@{}c@{}} DL-vdG \\ (IPW)\end{tabular} & \begin{tabular}{@{}c@{}} DL-vdG \\ (Oracle)\end{tabular} & \begin{tabular}{@{}c@{}} R-Proj \\ (CC)\end{tabular} & \begin{tabular}{@{}c@{}} R-Proj \\ (IPW)\end{tabular} & \begin{tabular}{@{}c@{}} R-Proj \\ (Oracle)\end{tabular} & \begin{tabular}{@{}c@{}} Refit \\ (CC)\end{tabular} & \begin{tabular}{@{}c@{}} Refit \\ (IPW)\end{tabular} & \begin{tabular}{@{}c@{}} Refit \\ (Oracle)\end{tabular} & DDR & SAS & \begin{tabular}{@{}c@{}} SSL \\ AIPW\end{tabular} 
			& \begin{tabular}{@{}c@{}} \textbf{Debias} \\ (min-CV)\end{tabular} & \begin{tabular}{@{}c@{}} \textbf{Debias} \\ (1SE)\end{tabular} & \begin{tabular}{@{}c@{}} \textbf{Debias} \\ (min-feas)\end{tabular} \\ \cline{2-19}
			& \multicolumn{18}{c}{\bf $d=1000$, $n = 900$, sparse $\beta_0^{sp}$, and $x^{(1)}$ so that $m(x^{(1)})=x^{(1)T}\beta_0^{sp} = \beta_{0,1}^{(a)}$} \\ \cline{2-19}
			Avg-Bias & 0.072 & 0.072 & 0.061 & 0.073 & 0.073 & 0.061 & 0.275 & 0.444 & 0.608 & 0.073 & 0.074 & 0.060 & 0.074 & 0.075 & 0.223 & 0.074 & 0.076 & 0.089 \\
			Avg-Cov & 0.885 & 0.884 & 0.894 & 0.912 & 0.906 & 0.917 & 0.963 & 0.955 & 0.950 & 0.948 & 0.944 & 0.956 & 0.916 & 0.905 & 0.536 & 0.883 & 0.914 & 0.929 \\
			Avg-Len & 0.282 & 0.282 & 0.247 & 0.312 & 0.309 & 0.266 & 0.885 & 1.413 & 1.917 & 0.360 & 0.358 & 0.301 & 0.311 & 0.327 & 40.764 & 0.301 & 0.332 & 0.420 \\
			\cline{2-19}
			& \multicolumn{18}{c}{\bf $d=1000$, $n = 900$, dense $\beta_0^{de}$, and $x^{(1)}$ so that $m(x^{(1)})=x^{(1)T}\beta_0^{de} = \beta_{0,1}^{(b)}$} \\ \cline{2-19}
			Avg-Bias & 0.299 & 0.298 & 0.152 & 0.289 & 0.221 & 0.153 & 0.204 & 0.200 & 0.372 & 0.454 & 0.264 & 0.178 & 0.266 & 0.233 & 0.170 & 0.489 & 0.430 & 0.404 \\
			Avg-Cov & 0.740 & 0.748 & 0.927 & 0.552 & 0.389 & 0.559 & 0.941 & 0.759 & 0.945 & 0.772 & 0.751 & 0.824 & 0.650 & 0.963 & 0.485 & 0.850 & 0.932 & 0.983 \\
			Avg-Len & 0.847 & 0.862 & 0.671 & 0.562 & 0.302 & 0.301 & 1.004 & 0.588 & 0.833 & 1.242 & 0.775 & 0.593 & 0.616 & 1.237 & 14.396 & 1.760 & 1.940 & 2.459 \\
			\cline{2-19}
			& \multicolumn{18}{c}{\bf $d=1000$, $n = 900$, pseudo-dense $\beta_0^{pd}$, and $x^{(1)}$ so that $m(x^{(1)})=x^{(1)T}\beta_0^{pd} = \beta_{0,1}^{(c)}$} \\ \cline{2-19}
			Avg-Bias & 0.091 & 0.091 & 0.071 & 0.094 & 0.090 & 0.073 & 0.137 & 0.174 & 0.265 & 0.097 & 0.095 & 0.075 & 0.101 & 0.094 & 0.230 & 0.105 & 0.098 & 0.106 \\
			Avg-Cov & 0.868 & 0.867 & 0.906 & 0.845 & 0.817 & 0.863 & 0.954 & 0.920 & 0.925 & 0.911 & 0.900 & 0.921 & 0.824 & 0.941 & 0.511 & 0.867 & 0.929 & 0.964 \\
			Avg-Len & 0.345 & 0.345 & 0.300 & 0.334 & 0.298 & 0.275 & 0.659 & 0.647 & 1.490 & 0.426 & 0.395 & 0.335 & 0.346 & 0.468 & 32.172 & 0.400 & 0.440 & 0.558 \\
			\hline
			
			& \multicolumn{18}{c}{\bf $d=1000$, $n = 900$, sparse $\beta_0^{sp}$, and $x^{(2)}$ so that $m(x^{(2)})=x^{(2)T}\beta_0^{sp}$} \\ \cline{2-19}
			Avg-Bias & 0.057 & 0.057 & 0.047 & 0.068 & 0.066 & 0.055 & NA & NA & NA & 0.050 & 0.051 & 0.041 & 0.070 & 0.054 & 0.198 & 0.058 & 0.062 & 0.074 \\
			Avg-Cov & 0.986 & 0.995 & 0.994 & 0.911 & 0.923 & 0.927 & NA & NA & NA & 0.935 & 0.937 & 0.952 & 0.917 & 0.900 & 0.495 & 0.925 & 0.933 & 0.942 \\
			Avg-Len & 0.357 & 0.391 & 0.319 & 0.294 & 0.291 & 0.248 & NA & NA & NA & 0.241 & 0.241 & 0.202 & 0.295 & 0.221 & 51.810 & 0.260 & 0.292 & 0.354 \\
			\cline{2-19}
			& \multicolumn{18}{c}{\bf $d=1000$, $n = 900$, dense $\beta_0^{de}$, and $x^{(2)}$ so that $m(x^{(2)})=x^{(2)T}\beta_0^{de}$} \\ \cline{2-19}
			Avg-Bias & 0.244 & 0.246 & 0.126 & 0.241 & 0.197 & 0.137 & NA & NA & NA & 0.396 & 0.256 & 0.177 & 0.232 & 0.221 & 0.157 & 0.407 & 0.361 & 0.350 \\
			Avg-Cov & 0.974 & 0.980 & 0.999 & 0.602 & 0.414 & 0.577 & NA & NA & NA & 0.709 & 0.694 & 0.750 & 0.683 & 0.867 & 0.495 & 0.862 & 0.946 & 0.977 \\
			Avg-Len & 1.399 & 1.496 & 1.071 & 0.529 & 0.284 & 0.280 & NA & NA & NA & 1.067 & 0.662 & 0.519 & 0.584 & 0.836 & 17.159 & 1.519 & 1.708 & 2.071 \\
			\cline{2-19}
			& \multicolumn{18}{c}{\bf $d=1000$, $n = 900$, pseudo-dense $\beta_0^{pd}$, and $x^{(2)}$ so that $m(x^{(2)})=x^{(2)T}\beta_0^{pd}$} \\ \cline{2-19}
			Avg-Bias & 0.083 & 0.083 & 0.065 & 0.091 & 0.085 & 0.070 & NA & NA & NA & 0.087 & 0.084 & 0.067 & 0.089 & 0.087 & 0.197 & 0.091 & 0.087 & 0.091 \\
			Avg-Cov & 0.990 & 0.994 & 0.998 & 0.824 & 0.806 & 0.860 & NA & NA & NA & 0.915 & 0.908 & 0.921 & 0.863 & 0.843 & 0.517 & 0.867 & 0.914 & 0.951 \\
			Avg-Len & 0.529 & 0.564 & 0.456 & 0.315 & 0.280 & 0.256 & NA & NA & NA & 0.382 & 0.354 & 0.301 & 0.328 & 0.316 & 29.117 & 0.345 & 0.388 & 0.470 \\
			\hline
			
			& \multicolumn{18}{c}{\bf $d=1000$, $n = 900$, sparse $\beta_0^{sp}$, and $x^{(3)}$ so that $m(x^{(3)})=x^{(3)T}\beta_0^{sp} = \beta_{0,100}^{(a)}$} \\ \cline{2-19}
			Avg-Bias & 0.057 & 0.057 & 0.054 & 0.076 & 0.076 & 0.068 & 0.421 & 1.217 & 0.733 & 0.000 & 0.000 & 0.000 & 0.081 & 0.094 & 0.269 & 0.077 & 0.088 & 0.097 \\
			Avg-Cov & 0.985 & 0.985 & 0.978 & 0.943 & 0.943 & 0.939 & 0.972 & 0.967 & 0.958 & 1.000 & 1.000 & 1.000 & 0.947 & 0.956 & 0.498 & 0.953 & 0.956 & 0.955 \\
			Avg-Len & 0.352 & 0.352 & 0.312 & 0.372 & 0.368 & 0.322 & 3.970 & 8.027 & 8.075 & 0.000 & 0.000 & 0.000 & 0.390 & 0.472 & 49.953 & 0.380 & 0.436 & 0.489 \\
			\cline{2-19}
			& \multicolumn{18}{c}{\bf $d=1000$, $n = 900$, dense $\beta_0^{de}$, and $x^{(3)}$ so that $m(x^{(3)})=x^{(3)T}\beta_0^{de} = \beta_{0,100}^{(b)}$} \\ \cline{2-19}
			Avg-Bias & 0.177 & 0.176 & 0.130 & 0.225 & 0.189 & 0.167 & 0.179 & 0.221 & 0.279 & 0.255 & 0.217 & 0.200 & 0.254 & 0.325 & 0.181 & 0.399 & 0.371 & 0.392 \\
			Avg-Cov & 0.978 & 0.978 & 0.990 & 0.782 & 0.528 & 0.604 & 0.988 & 0.889 & 0.954 & 0.229 & 0.337 & 0.391 & 0.804 & 0.967 & 0.492 & 0.969 & 0.993 & 0.995 \\
			Avg-Len & 1.078 & 1.097 & 0.856 & 0.669 & 0.359 & 0.364 & 1.288 & 1.680 & 1.097 & 0.388 & 0.353 & 0.331 & 0.773 & 1.791 & 26.954 & 2.223 & 2.550 & 2.858 \\
			\cline{2-19}
			& \multicolumn{18}{c}{\bf $d=1000$, $n = 900$, pseudo-dense $\beta_0^{pd}$, and $x^{(3)}$ so that $m(x^{(3)})=x^{(3)T}\beta_0^{pd} = \beta_{0,100}^{(c)}$} \\ \cline{2-19}
			Avg-Bias & 0.071 & 0.071 & 0.067 & 0.093 & 0.086 & 0.082 & 0.302 & 0.209 & 0.394 & 0.059 & 0.057 & 0.056 & 0.100 & 0.130 & 0.239 & 0.103 & 0.107 & 0.117 \\
			Avg-Cov & 0.976 & 0.977 & 0.969 & 0.892 & 0.864 & 0.890 & 0.984 & 0.970 & 0.963 & 0.303 & 0.346 & 0.377 & 0.926 & 0.962 & 0.508 & 0.948 & 0.972 & 0.980 \\
			Avg-Len & 0.436 & 0.437 & 0.381 & 0.398 & 0.354 & 0.333 & 2.462 & 3.289 & 1.953 & 0.143 & 0.149 & 0.143 & 0.434 & 0.677 & 46.809 & 0.505 & 0.579 & 0.649 \\
			\hline
			
			& \multicolumn{18}{c}{\bf $d=1000$, $n = 900$, sparse $\beta_0^{sp}$, and $x^{(4)}$ so that $m(x^{(4)})=x^{(4)T}\beta_0^{sp}$} \\ \cline{2-19}
			Avg-Bias & 0.049 & 0.049 & 0.040 & 0.069 & 0.061 & 0.054 & NA & NA & NA & 0.045 & 0.045 & 0.037 & 0.069 & 0.045 & 0.138 & 0.045 & 0.046 & 0.048 \\
			Avg-Cov & 1.000 & 1.000 & 1.000 & 0.857 & 0.901 & 0.873 & NA & NA & NA & 0.944 & 0.938 & 0.960 & 0.835 & 0.844 & 0.530 & 0.890 & 0.924 & 0.931 \\
			Avg-Len & 0.736 & 0.766 & 0.568 & 0.248 & 0.246 & 0.203 & NA & NA & NA & 0.218 & 0.217 & 0.182 & 0.240 & 0.160 & 23.997 & 0.183 & 0.212 & 0.226 \\
			\cline{2-19}
			& \multicolumn{18}{c}{\bf $d=1000$, $n = 900$, dense $\beta_0^{de}$, and $x^{(4)}$ so that $m(x^{(4)})=x^{(4)T}\beta_0^{de}$} \\ \cline{2-19}
			Avg-Bias & 0.180 & 0.181 & 0.089 & 0.173 & 0.123 & 0.090 & NA & NA & NA & 0.376 & 0.228 & 0.159 & 0.168 & 0.163 & 0.119 & 0.436 & 0.323 & 0.285 \\
			Avg-Cov & 1.000 & 1.000 & 1.000 & 0.694 & 0.568 & 0.684 & NA & NA & NA & 0.616 & 0.560 & 0.613 & 0.760 & 0.844 & 0.488 & 0.650 & 0.867 & 0.944 \\
			Avg-Len & 3.697 & 3.793 & 2.527 & 0.446 & 0.240 & 0.229 & NA & NA & NA & 0.782 & 0.470 & 0.359 & 0.479 & 0.606 & 9.202 & 1.070 & 1.238 & 1.320 \\
			\cline{2-19}
			& \multicolumn{18}{c}{\bf $d=1000$, $n = 900$, pseudo-dense $\beta_0^{pd}$, and $x^{(4)}$ so that $m(x^{(4)})=x^{(4)T}\beta_0^{pd}$} \\ \cline{2-19}
			Avg-Bias & 0.066 & 0.066 & 0.054 & 0.140 & 0.103 & 0.103 & NA & NA & NA & 0.067 & 0.063 & 0.050 & 0.115 & 0.064 & 0.146 & 0.078 & 0.065 & 0.062 \\
			Avg-Cov & 1.000 & 1.000 & 1.000 & 0.474 & 0.625 & 0.528 & NA & NA & NA & 0.880 & 0.871 & 0.886 & 0.631 & 0.837 & 0.540 & 0.794 & 0.918 & 0.946 \\
			Avg-Len & 1.320 & 1.351 & 1.020 & 0.266 & 0.236 & 0.210 & NA & NA & NA & 0.258 & 0.239 & 0.203 & 0.270 & 0.230 & 23.784 & 0.243 & 0.281 & 0.300 \\
			\hline
			
			& \multicolumn{18}{c}{\bf $d=1000$, $n = 900$, sparse $\beta_0^{sp}$, and $x^{(5)}$ so that $m(x^{(5)})=x^{(5)T}\beta_0^{sp}$} \\ \cline{2-19}
			Avg-Bias & 0.028 & 0.028 & 0.020 & 0.039 & 0.036 & 0.029 & NA & NA & NA & 0.002 & 0.003 & 0.001 & 0.033 & 0.005 & 0.025 & 0.007 & 0.007 & 0.009 \\
			Avg-Cov & 1.000 & 1.000 & 1.000 & 0.956 & 0.968 & 0.951 & NA & NA & NA & 0.865 & 0.625 & 0.889 & 0.950 & 0.000 & 0.488 & 0.904 & 0.906 & 0.960 \\
			Avg-Len & 2.623 & 2.650 & 1.879 & 0.188 & 0.186 & 0.142 & NA & NA & NA & 0.006 & 0.007 & 0.005 & 0.167 & 0.000 & 7.529 & 0.031 & 0.035 & 0.045 \\
			
			\hline\hline
	\end{tabular}}
	\caption{Simulation results under the Toeplitz (or autoregressive) covariance matrix $\Sigma^{\mathrm{ar}}$, Gaussian noise $\mathcal{N}(0,1)$, and the MCAR setting. Nearly all the standard errors for the above average quantities are less than 0.01 except for those associated with ``DL-vdG'', ``R-Proj'', and ``SSL-AIPW'' and are thus omitted.}
	\label{table:AR_MCAR_res}
\end{center}
\end{table}

\section{Proofs}
\label{app:Proofs}

This section consists of the proofs for theorems and propositions in the main paper as well as other auxiliary results.

\subsection{Additional Notations}
\label{subapp:add_notation}

We denote by $\{X_i\}_{i=1}^{\infty}$ a sequence of independent $\mathcal{X}$-valued random elements with common probability measure $\mathrm{P}$. In most cases of this paper, $\mathcal{X}$ is the Euclidean space $\mathbb{R}^d$. For any measure $Q$ and any real-valued $Q$-integrable function $f$ on the measurable space $(\mathcal{X},\mathcal{A})$, we write $Qf=\int f\,dQ$ and denote by $L_p(\mathcal{X},\mathcal{A},Q)$, with $p\in[1,\infty)$, the space of all real-valued measurable functions $f$ on $(\mathcal{X},\mathcal{A})$ with finite $L_p(Q)$-norm, \emph{i.e.}, $\norm{f}_{Q,p} := \left(\int |f|^p\,dQ\right)^{\frac{1}{p}} < \infty$. The empirical measure $\mathbb{P}_n$ associated with observations $\{X_1,...,X_n\}$ is defined as the random measure on $(\mathcal{X},\mathcal{A})$ given by $\mathbb{P}_n=\frac{1}{n} \sum_{i=1}^n \delta_{X_i}$, where $\delta_x$ is the Dirac measure at point $x\in \mathcal{X}$. We denote the empirical process indexed by a class $\mathcal{F}$ by $\mathbb{G}_n(f)=\sqrt{n}\left(\mathbb{P}_n-\mathrm{P}\right) f = \frac{1}{\sqrt{n}} \sum_{i=1}^n \left(f(X_i)-\mathrm{P} f\right), f\in \mathcal{F}$. For any measure $Q$ on $(\mathcal{X},\mathcal{A})$, we set $\norm{Q}_{\mathcal{F}} = \sup_{f\in \mathcal{F}} |Qf|$. 

Analogous to the $\psi_{\alpha}$-Orlicz norm \eqref{Orlicz} for random variables, we define the $\psi_{\alpha}$-Orlicz norm for a real-valued $Q$-integrable function $f$ on $(\mathcal{X},\mathcal{A})$ as:
\begin{equation}
	\label{Orlicz_func}
	\norm{f}_{Q,\psi_{\alpha}} = \inf\left\{t>0: \int \psi_{\alpha}\left(\frac{|f|}{t} \right) dQ\leq 1\right\}
\end{equation}
for $\psi_{\alpha}(u)=\exp(u^{\alpha})-1$.

\subsection{Proofs of Propositions~\ref{prop:dual_debias} and \ref{prop:cond_MSE}}
\label{subapp:cond_MSE}

\begin{customprop}{2.1}
	The dual form of the debiasing program \eqref{debias_prog} is given by
	\begin{equation*}
		\min_{\ell(x) \in \mathbb{R}^d} \left\{\frac{1}{4n} \sum_{i=1}^n \hat{\pi}_i \left[X_i^T \ell(x)\right]^2 + x^T \ell(x) +\frac{\gamma}{n}\norm{\ell(x)}_1 \right\}.
	\end{equation*}
	If strong duality holds, then the solution $\hat{\bm{w}}(x) \in \mathbb{R}^n$ to the primal debiasing program \eqref{debias_prog} and the solution $\hat{\ell}(x) \in \mathbb{R}^d$ to its dual formulation \eqref{debias_prog_dual} satisfy the following identity as: 
	\begin{equation*}
		\hat{w}_i(x) = -\frac{1}{2\sqrt{n}} \cdot X_i^T\hat{\ell}(x), \quad i=1,...,n.
	\end{equation*}
\end{customprop}

\begin{proof}[Proof of Proposition~\ref{prop:dual_debias}]
	The proof is mainly adapted from Lemma 5 of \cite{giessing2021inference}. Recall from \eqref{debias_prog} that the debiasing program is a convex optimization problem and can be written in matrix form as:
	\begin{align}
		\label{debias_prog_mat}
		\begin{split}
			&\min_{\bm{w}:=\bm{w}(x)\in \mathbb{R}^n} \frac{1}{2}\bm{w}^T \hat{\Psi} \bm{w}\\
			&\text{subject to } \,A\bm{w} \preceq b,
		\end{split}
	\end{align}
	where $\hat{\Psi}=\Diag\left(2\hat{\pi}_1,...,2\hat{\pi}_n \right)\in \mathbb{R}^{n\times n}$, $A= \frac{1}{\sqrt{n}} \begin{pmatrix}
		-\bm{X}^T \hat{\bm{\Pi}}\\
		\bm{X}^T \hat{\bm{\Pi}}
	\end{pmatrix} \in \mathbb{R}^{2d\times n}$ with $\hat{\bm{\Pi}}=\Diag\left(\hat{\pi}_1,...,\hat{\pi}_n\right)\in \mathbb{R}^{n\times n}$ and $\bm{X}=(X_1,...,X_n)^T \in \mathbb{R}^{n\times d}$, as well as $b=\begin{pmatrix}
		\frac{\gamma}{n}\cdot \mathbbm{1}_d -x\\
		\frac{\gamma}{n}\cdot \mathbbm{1}_d +x
	\end{pmatrix} \in \mathbb{R}^{2d}$. This is because the box constraint $\norm{x- \frac{1}{\sqrt{n}}\sum_{i=1}^n w_i\cdot \hat{\pi}_i\cdot X_i }_{\infty} \leq \frac{\gamma}{n}$ in \eqref{debias_prog} is equivalent to
	\begin{align*}
		-\frac{\gamma}{n}\cdot\mathbbm{1}_d \leq x- \frac{1}{\sqrt{n}}\sum_{i=1}^n w_i\cdot \hat{\pi}_i\cdot X_i \leq \frac{\gamma}{n} \cdot \mathbbm{1}_d
		&\quad \text{ and } \quad \frac{1}{\sqrt{n}} \begin{pmatrix}
			-\bm{X}^T \hat{\bm{\Pi}}\\
			\bm{X}^T \hat{\bm{\Pi}}
		\end{pmatrix} \bm{w} \leq \begin{pmatrix}
			\frac{\gamma}{n}\cdot \mathbbm{1}_d -x\\
			\frac{\gamma}{n}\cdot \mathbbm{1}_d +x
		\end{pmatrix},
	\end{align*}
	where $\bm{w}=(w_1,...,w_n)^T \in \mathbb{R}^n$. The Lagrange dual function is given by
	\begin{equation}
		\label{lagrange_func}
		g(\xi) = \inf_{\bm{w}\in \mathbb{R}^n} \left[\frac{1}{2}\bm{w}^T \hat{\Psi} \bm{w} + \xi^T (A\bm{w}-b) \right] = -\frac{1}{2} \xi^T A\hat{\Psi}^{-1} A^T \xi - \xi^T b,
	\end{equation}
	where $\xi\in \mathbb{R}^{2d}$ and the infimum is attained at $\bm{w}=-\hat{\Psi}^{-1} A^T\xi$. This yields the dual of \eqref{debias_prog_mat} as:
	\begin{align*}
		&\max_{\xi:=\xi(x) \in \mathbb{R}^{2d}} \left[-\frac{1}{2} \xi^T A\hat{\Psi}^{-1} A^T \xi - \xi^T b \right] \\
		&\text{subject to } \,\xi \succeq \bm{0}.
	\end{align*}
	If strong duality holds, then the solution $\hat{\bm{w}}(x) \in \mathbb{R}^n$ to the primal problem is related to the solution $\hat{\xi}(x)\in \mathbb{R}^{2d}$ to the dual problem by
	\begin{equation}
		\label{primal_dual_sol}
		\hat{\bm{w}}(x) = -\hat{\Psi}^{-1} A^T \hat{\xi}(x).
	\end{equation}
	Since $\hat{\xi}(x)\in \mathbb{R}^{2d}$ is the vector of optimal Lagrangian multipliers associated with the constraints $A\bm{w}\preceq b$ in the primal problem \eqref{debias_prog_mat}, it satisfies the complementary slackness conditions
	$$\left[\hat{\xi}(x)\right]_k \geq 0 \quad \text{ and } \quad \left[A\hat{\bm{w}}(x)\right]_k < b_k \quad \implies \quad \left[\hat{\xi}(x)\right]_k=0,$$
	for any $ k=1,...,2d$. We write $\hat{\xi}(x) = \left(\bar{\mu}, \underline{\mu} \right)$ with $\bar{\mu},\underline{\mu} \in \mathbb{R}_+^d$. The above complementary slackness conditions indicate that
	\begin{equation}
		\label{comple_slack2}
		\bar{\mu}_k \cdot \underline{\mu}_k =0, \quad k=1,...,d,
	\end{equation}
	because by design, at least one of the inequalities $\left[A\hat{\bm{w}}(x)\right]_k \leq b_k$ and $\left[A\hat{\bm{w}}(x)\right]_{d+k} \leq b_{d+k}$ should be strict. Take $\hat{\ell}(x) := (-\bar{\mu}+\underline{\mu}) \in \mathbb{R}^d$. Equation \eqref{comple_slack2} implies that
	$$|\hat{\ell}_k(x)| = \bar{\mu}_k + \underline{\mu}_k, \quad k=1,...,d,$$
	and thus,
	$$-\frac{1}{2} \hat{\xi}(x)^T A\hat{\Psi}^{-1} A^T \hat{\xi}(x) - \hat{\xi}(x)^T b = -\frac{1}{2n} \hat{\ell}(x)^T \bm{X}^T \hat{\bm{\Pi}} \hat{\Psi}^{-1} \hat{\bm{\Pi}} \bm{X} \hat{\ell}(x) -x^T \hat{\ell}(x) - \frac{\gamma}{n} \norm{\hat{\ell}(x)}_1.$$
	Since every $\ell(x)\in \mathbb{R}^d$ can be written as $\ell(x)=\left(-\bar{\mu}+\underline{\mu}\right)$ for some $\bar{\mu},\underline{\mu} \in \mathbb{R}_+^d$, the above equation implies that $\hat{\ell}(x)$ solves the dual problem
	\begin{align*}
		&\min_{\ell(x) \in \mathbb{R}^d} \left[\frac{1}{2n} \ell(x)^T \bm{X}^T \hat{\bm{\Pi}} \hat{\Psi}^{-1} \hat{\bm{\Pi}}\bm{X} \ell(x) + x^T\ell(x) + \frac{\gamma}{n}\norm{\ell(x)}_1 \right] \\
		&\equiv \min_{\ell(x) \in \mathbb{R}^d} \left\{\frac{1}{4n} \sum_{i=1}^n \hat{\pi}_i\left[X_i^T \ell(x)\right]^2 + x^T\ell(x) + \frac{\gamma}{n}\norm{\ell(x)}_1 \right\}.
	\end{align*}
	Finally, we observe from \eqref{primal_dual_sol} that $\hat{\bm{w}}(x) = -\frac{1}{\sqrt{n}} \hat{\Psi}^{-1} \hat{\bm{\Pi}}\bm{X} \hat{\ell}(x)$ and therefore,
	$$\hat{w}_i(x) = -\frac{1}{2\sqrt{n}} \cdot X_i^T\hat{\ell}(x), \quad i=1,...,n.$$
	The proof is completed.
\end{proof}

\vspace{3mm}

\begin{customprop}{2.2}
	Under Assumption~\ref{assump:basic}, the conditional mean squared error of $\sqrt{n}\,m^{\mathrm{debias}}(x;\bm{w})$ given $\bm{X}=(X_1,...,X_n)^T$ has the following decomposition as:
	\begin{align*}
		&\mathrm{E}\left[\left(\sqrt{n}\,m^{\mathrm{debias}}(x;\bm{w}) - \sqrt{n}\, m_0(x) \right)^2 \Big| \bm{X} \right] \\
		&= \underbrace{\sigma_{\epsilon}^2\sum_{i=1}^n w_i^2 \pi(X_i)}_{\text{Conditional variance I}} + \underbrace{\left(\beta_0 - \beta\right)^T \left[\sum_{i=1}^n w_i^2\pi(X_i)\left(1-\pi(X_i)\right)X_iX_i^T \right] \left(\beta_0 - \beta\right)}_{\text{Conditional variance II}} \\
		&\quad + \underbrace{\left[\left(\frac{1}{\sqrt{n}}\sum_{i=1}^n w_i \pi(X_i) X_i -x \right)^T \sqrt{n}\left(\beta_0 - \beta\right) \right]^2}_{\text{Conditional bias}}.
	\end{align*}
\end{customprop}

\begin{proof}[Proof of Proposition~\ref{prop:cond_MSE}]
	Under the MAR condition in Assumption~\ref{assump:basic}, the conditional mean of \eqref{debias_generic} (scaled by a factor $\sqrt{n}$) given $\bm{X} = (X_1,...,X_n)^T$ can be written as
	$$\mathrm{E}\left[\sqrt{n}\, m^{\text{debias}}(x;\bm{w}) \big|\bm{X} \right] = \sqrt{n}\,x^T\beta + \sum_{i=1}^n w_i \pi(X_i) X_i^T\left(\beta_0 - \beta\right).$$
	We then calculate the conditional mean squared error of \eqref{debias_generic} under model \eqref{linear_reg}, with $m_0(x)=\mathrm{E}(Y|X=x)=x^T\beta_0$, as follows:
	\begin{align*}
		&\mathrm{E}\left[\left(\sqrt{n}\, m^{\text{debias}}(x;\bm{w}) - \sqrt{n}\, m_0(x) \right)^2 \Big| \bm{X} \right] \\
		&= \mathrm{E}\left\{\left[\sqrt{n} \, m^{\text{debias}}(x;\bm{w}) - \mathrm{E}\left(\sqrt{n}\, m^{\text{debias}}(x;\bm{w})\big|\bm{X} \right)\right]^2 \Big| \bm{X} \right\} \\
		&\quad + \left\{\mathrm{E}\left[\sqrt{n}\, m^{\text{debias}}(x;\bm{w}) \big|\bm{X}\right] -\sqrt{n}\, m_0(x) \right\}^2\\
		&= \mathrm{E}\left\{\left[\sum_{i=1}^n w_i (R_i-\pi(X_i)) X_i^T\left(\beta_0 - \beta\right) + \sum_{i=1}^n w_i R_i\epsilon_i \right]^2 \Bigg|\bm{X} \right\} \\
		&\quad + \left[\left(\frac{1}{\sqrt{n}}\sum_{i=1}^n w_i \pi(X_i) X_i -x \right)^T \sqrt{n}\left(\beta_0 - \beta\right) \right]^2\\
		&= \sum_{i=1}^n w_i^2 \pi(X_i) \left(1-\pi(X_i)\right) \left[X_i^T\left(\beta_0 - \beta\right) \right]^2 + \sigma_{\epsilon}^2 \sum_{i=1}^n w_i^2 \pi(X_i) \\
		&\quad + \left[\left(\frac{1}{\sqrt{n}}\sum_{i=1}^n w_i \pi(X_i) X_i -x \right)^T \sqrt{n}\left(\beta_0 - \beta\right) \right]^2\\
		&= \underbrace{\sigma_{\epsilon}^2\sum_{i=1}^n w_i^2 \pi(X_i)}_{\text{Conditional variance I}} + \underbrace{\left(\beta_0 - \beta\right)^T \left[\sum_{i=1}^n w_i^2\pi(X_i)\left(1-\pi(X_i)\right)X_iX_i^T \right] \left(\beta_0 - \beta\right)}_{\text{Conditional variance II}} \\
		&\quad + \underbrace{\left[\left(\frac{1}{\sqrt{n}}\sum_{i=1}^n w_i \pi(X_i) X_i -x \right)^T \sqrt{n}\left(\beta_0 - \beta\right) \right]^2}_{\text{Conditional bias}}.
	\end{align*}
	The result thus follows.
\end{proof}

\begin{remark}
	We can upper bound the ``Conditional variance II'' term as follows:
	\begin{align*}
		&\sqrt{n}\left(\beta_0 - \beta\right)^T \left[\frac{1}{n}\sum_{i=1}^n w_i^2\pi(X_i)\left(1-\pi(X_i)\right)X_iX_i^T \right] \sqrt{n}\left(\beta_0 - \beta\right) \\
		&\leq \left(\sqrt{n} \norm{\beta_0-\beta}_1 \right)^2 \max_{1\leq i \leq n}\norm{\sqrt{\pi(X_i)\left(1-\pi(X_i)\right)} \cdot w_i X_i}_{\infty}^2.
	\end{align*}
	Thus, in finite samples when the ``Conditional variance II'' term is not negligible, one may consider optimizing the weights $\bm{w}=(w_1,...,w_n)^T\in \mathbb{R}^n$ with the following more precise debiasing program:
	\begin{align*}
		&\min_{\bm{w} \in \mathbb{R}^n} \sum_{i=1}^n \hat{\pi}_iw_i^2 \\
		&\text{ subject to } \norm{x- \frac{1}{\sqrt{n}}\sum_{i=1}^n w_i \hat{\pi}_i X_i }_{\infty} \leq \frac{\gamma_1}{n} \quad\text{ and }\quad \max_{1\leq i \leq n}\norm{\sqrt{\hat{\pi}_i\left(1-\hat{\pi}_i\right)} \cdot w_i X_i}_{\infty} \leq \frac{\gamma_2}{n},
	\end{align*}
	where $\gamma_1,\gamma_2>0$ are two tuning parameters. However, it will be difficult to derive the dual form of the above debiasing program and tune the two parameters $\gamma_1,\gamma_2$ simultaneously. Given that the above program has no asymptotic efficiency gains relative to our proposed program \eqref{debias_prog} by Lemma~\ref{lem:negligible_cond_var}, we still recommend our approach even in finite samples.
\end{remark}

\subsection{Proof of \autoref{thm:lasso_const}}
\label{subapp:proof_consistency}

\begin{customthm}{3.2}[Consistency of the Lasso pilot estimate with complete-case data]
	Let $\delta \in (0,1)$ and $A_1,A_2>0$ be some absolute constants. Suppose that Assumptions~\ref{assump:basic}, \ref{assump:sub_gaussian}, and \ref{assump:gram_lower_bnd} (more precisely, Eq.\eqref{res_eigen_cond}) hold and $n$ is sufficiently large so that 
	$$\sqrt{\frac{s_{\beta}\log(ed/s_{\beta})}{n}} + \sqrt{\frac{\log(1/\delta)}{n}} < \min\left\{\frac{\kappa_R^2}{A_1 \phi_{s_{\beta}}}, 1 \right\}.$$
	\begin{enumerate}[label=(\alph*)]
		\item If $\lambda> \frac{4}{n}\norm{\sum_{i=1}^n R_iX_i\epsilon_i}_{\infty} >0$ and $d\geq s_{\beta}+2$, then with probability at least $1-\delta$,
		$$\norm{\hat{\beta}-\beta_0}_2 \lesssim \frac{\sqrt{s_{\beta}}\lambda}{\kappa_R^2}.$$ 
		
		\item If, in addition, Assumption~\ref{assump:subgau_error} holds and $\lambda = A_2\sigma_{\epsilon} \sqrt{\frac{\phi_1\log(d/\delta)}{n}}$, then with probability at least $1-\delta$,
		\begin{align*}
			\norm{\hat{\beta}-\beta_0}_2 &\lesssim \frac{ \sigma_{\epsilon}}{\kappa_R^2} \sqrt{\frac{\phi_1 s_{\beta}\log(d/\delta)}{n}}.
		\end{align*}
	\end{enumerate}
\end{customthm}

\begin{proof}[Proof of \autoref{thm:lasso_const}]
	(a) Since $\hat{\beta}$ minimizes the function $L(\beta) = \frac{1}{n}\sum_{i=1}^n R_i\left(Y_i -X_i^T\beta\right)^2 + \lambda\norm{\beta}_1$ with complete-case data, we know that 
	$$\frac{1}{n}\sum_{i=1}^n R_i\left(Y_i -X_i^T\hat{\beta}\right)^2 + \lambda\norm{\hat{\beta}}_1 \leq \frac{1}{n}\sum_{i=1}^n R_i\left(Y_i -X_i^T\beta_0\right)^2 + \lambda\norm{\beta_0}_1.$$
	Then, under the linear model \eqref{linear_reg} in Assumption~\ref{assump:basic}, direct calculation or Taylor's expansion implies that
	\begin{equation}
		\label{basic_inq}
		\frac{1}{n}\sum_{i=1}^n R_i\left(X_i^T\hat{v}\right)^2 \leq \frac{2}{n}\sum_{i=1}^n \epsilon_i R_i X_i^T \hat{v} + \lambda\left(\norm{\beta_0}_1 - \norm{\beta_0 + \hat{v}}_1 \right)
	\end{equation}
	with $\hat{v} = \hat{\beta}-\beta_0$. Notice that $\beta_{0,S_{\beta}^c} =0$, so we have $\norm{\beta_0}_1 = \norm{\beta_{0,S_{\beta}}}_1$ and 
	$$\norm{\beta_0 + \hat{v}}_1 = \norm{\beta_{0,S_{\beta}} + \hat{v}_{S_{\beta}}}_1 + \norm{\hat{v}_{S_{\beta}^c}}_1 \geq \norm{\beta_{0,S_{\beta}}}_1 - \norm{\hat{v}_{S_{\beta}}}_1 + \norm{\hat{v}_{S_{\beta}^c}}_1.$$
	Plugging this lower bound into \eqref{basic_inq}, we obtain that
	\begin{align}
		\label{basic_inq2}
		\begin{split}
			\frac{1}{n}\sum_{i=1}^n R_i\left(X_i^T\hat{v}\right)^2 &\leq \frac{2}{n}\sum_{i=1}^n \epsilon_i R_i X_i^T \hat{v} + \lambda \left(\norm{\hat{v}_{S_{\beta}}}_1 -\norm{\hat{v}_{S_{\beta}^c}}_1 \right)\\
			&\leq \frac{2\norm{\sum_{i=1}^n R_iX_i\epsilon_i}_{\infty}}{n} \cdot \norm{\hat{v}}_1 + \lambda \left(\norm{\hat{v}_{S_{\beta}}}_1 -\norm{\hat{v}_{S_{\beta}^c}}_1 \right),
		\end{split}
	\end{align}
	where we apply H{\"o}lder's inequality to obtain the second inequality. Under the condition that $\frac{4}{n} \norm{\sum_{i=1}^n R_iX_i\epsilon_i}_{\infty} < \lambda$, we can proceed from \eqref{basic_inq2} as follows:
	\begin{align}
		\label{basic_inq3}
		\begin{split}
			0\leq \frac{1}{n}\sum_{i=1}^n R_i\left(X_i^T\hat{v}\right)^2 &\leq \frac{\lambda}{2} \left(\norm{\hat{v}_{S_{\beta}}}_1 + \norm{\hat{v}_{S_{\beta}^c}}_1 \right) + \lambda \left(\norm{\hat{v}_{S_{\beta}}}_1 -\norm{\hat{v}_{S_{\beta}^c}}_1 \right)\\
			&= \frac{3\lambda}{2} \norm{\hat{v}_{S_{\beta}}}_1 - \frac{\lambda}{2} \norm{\hat{v}_{S_{\beta}^c}}_1.
		\end{split}
	\end{align}
	Rearranging this inequality shows that $\norm{\hat{v}_{S_{\beta}^c}}_1 \leq 3\norm{\hat{v}_{S_{\beta}}}_1$. Thus, the error $\hat{v} =\hat{\beta} -\beta_0$ belongs to the cone $\mathcal{C}_1(S_{\beta},3)$. By Assumption~\ref{assump:gram_lower_bnd} (more precisely, Eq.\eqref{res_eigen_cond}) and Lemma~\ref{lem:gram_mat_conc} (taking $\vartheta=3$), we know that there exists an absolute constant $A_0>0$ such that with probability at least $1-\delta$,
	\begin{align*}
		&\frac{1}{n} \sum_{i=1}^n R_i \left(X_i^T \hat{v} \right)^2\\
		&\geq \frac{1}{n}\sum_{i=1}^n \mathrm{E}\left[R_i\left(X_i^T\hat{v}\right)^2 \right]- A_0 \phi_{s_{\beta}} \left[\sqrt{\frac{\log(1/\delta) + s_{\beta}\log(ed/s_{\beta})}{n}} + \frac{\log(1/\delta)  + s_{\beta}\log(ed/s_{\beta})}{n} \right] \norm{\hat{v}}_2^2\\
		&\geq \left[\kappa_R^2 - A_0\phi_{s_{\beta}} \left(\sqrt{\frac{\log(1/\delta) + s_{\beta}\log(ed/s_{\beta})}{n}} + \frac{\log(1/\delta) + s_{\beta}\log(ed/s_{\beta})}{n} \right) \right] \norm{\hat{v}}_2^2\\
		&\geq \left[\kappa_R^2 - 2A_0\phi_{s_{\beta}} \sqrt{\frac{\log(1/\delta) + s_{\beta}\log(ed/s_{\beta})}{n}} \right] \norm{\hat{v}}_2^2\\
		&\geq \frac{\kappa_R^2}{2} \norm{\hat{v}}_2^2,
	\end{align*} 
	where the last two inequalities hold under our growth condition 
	$$\sqrt{\frac{s_{\beta}\log(ed/s_{\beta})}{n}} + \sqrt{\frac{\log(1/\delta)}{n}} < \min\left\{\frac{\kappa_R^2}{A_1 \phi_{s_{\beta}}}, 1 \right\}$$
	with $A_1=4A_0 >0$ when $n$ is sufficiently large. Substituting the above inequality into \eqref{basic_inq3} yields that with probability at least $1-\delta$,
	\begin{align*}
		\frac{\kappa_R^2}{2}\norm{\hat{v}}_2^2 &\leq\frac{3\lambda}{2} \norm{\hat{v}_{S_{\beta}}}_1 - \frac{\lambda}{2} \norm{\hat{v}_{S_{\beta}^c}}_1 \\
		&\leq \frac{3}{2} \lambda \sqrt{s_{\beta}} \norm{\hat{v}}_2,
	\end{align*}
	where the second inequality is due to the Cauchy-Schwarz inequality. The result now follows by dividing both sides of the inequality by $\norm{\hat{v}}_2$.\\
	
	\noindent(b) It suffices to show that $\frac{4}{n}\norm{\sum_{i=1}^n R_iX_i\epsilon_i}_{\infty}$ is upper bounded by $\sigma_{\epsilon} \sqrt{\frac{\phi_1 \log(d/\delta)}{n}}$ up to some absolute constant with probability at least $1-\delta$. Notice that
	$$\frac{4}{n}\norm{\sum_{i=1}^n R_iX_i\epsilon_i}_{\infty} = \max_{1\leq k \leq d} \left|\frac{1}{n} \sum_{i=1}^n 4\epsilon_i R_iX_{ik} \right|,$$
	where $X_{ik}$ is the $k$-th element of the covariate $X_i\in \mathbb{R}^d$ for $i=1,...,n$. By Assumptions~\ref{assump:sub_gaussian} and \ref{assump:subgau_error} together with Lemma~\ref{lem:Orlicz_prod} (see also Lemma 2.7.7 in \citealt{vershynin2018high}), we know that
	\begin{align*}
		\max_{1\leq k \leq d} \norm{4\epsilon_iR_i X_{ik}}_{\psi_1} &\leq \max_{1\leq k \leq d} \left[4\left(\max_{1\leq i\leq n} \norm{X_{ik}}_{\psi_2}\right) \norm{R_i \epsilon_i}_{\psi_2} \right]\\ 
		&\lesssim \sqrt{\phi_1} \cdot \sigma_{\epsilon}
	\end{align*}
	for all $i=1,...,n$. Hence, by the union bound and Bernstein's inequality (Corollary 2.8.3 in \citealt{vershynin2018high}), we obtain that for all $t>0$, there exists an absolute constant $A_3>0$ such that
	\begin{align*}
		\mathrm{P}\left(\frac{4}{n}\norm{\sum_{i=1}^n R_iX_i\epsilon_i}_{\infty} > t \right) &\leq d \cdot \max_{1\leq k \leq d} \mathrm{P}\left(\left|\frac{1}{n} \sum_{i=1}^n 4\epsilon_i R_i X_{ik} \right| > t \right)\\
		&\leq 2d \cdot \exp\left(-nA_3 \min\left\{\frac{t}{\sigma_{\epsilon} \sqrt{\phi_1}}, \frac{t^2}{\sigma_{\epsilon}^2\phi_1} \right\} \right).
	\end{align*}
	Thus, we conclude that with probability at least $1-\delta$,
	\begin{align*}
		\frac{4}{n}\norm{\sum_{i=1}^n R_iX_i\epsilon_i}_{\infty} &\lesssim \sigma_{\epsilon}\sqrt{\phi_1} \left[\sqrt{\frac{\log(d/\delta)}{n}} + \frac{\log(d/\delta)}{n} \right]\\
		&\lesssim \sigma_{\epsilon}\sqrt{\frac{\phi_1\log(d/\delta)}{n}},
	\end{align*}
	where we notice that $\sqrt{\frac{\log(d/\delta)}{n}}$ dominates the bound under our growth condition $\sqrt{\frac{s_{\beta}\log(ed/s_{\beta})}{n}} + \sqrt{\frac{\log(1/\delta)}{n}}<1$. Finally, the bound on $\norm{\hat{\beta} - \beta_0}_2$ follows from substituting $\lambda =A_2 \sigma_{\epsilon}\sqrt{\frac{\phi_1\log(d/\delta)}{n}}$ for some absolute constant into the upper bound of $\norm{\hat{\beta} -\beta_0}_2$ derived in (a).
\end{proof}

\subsection{Proof of \autoref{thm:dual_consist2}}
\label{subapp:proof_consist_dual}

Before presenting the proof of \autoref{thm:dual_consist2}, we first establish the consistency of $\hat{\ell}(x)$ to the sparse $r_{\ell}$-approximation $\tilde{\ell}(x)$ of the exact population dual solution \eqref{popul_dual_exact} under Assumption~\ref{assump:sparse_dual} and the known/oracle propensity scores $\pi_i=\pi(X_i), i=1,...,n$ in the following lemma. 

\begin{lemma}[Consistency of the dual solution $\hat{\ell}(x)$ when the propensity score is known]
	\label{lem:dual_consist}
	Let $\delta \in (0,1)$,  $x\in \mathbb{R}^d$ be a fixed covariate vector, and $A_1,A_2,A_3>0$ be some large absolute constants. Suppose that Assumptions~\ref{assump:basic}, \ref{assump:sub_gaussian}, \ref{assump:gram_lower_bnd}, and \ref{assump:sparse_dual} hold as well as
	$$\sqrt{\frac{s_{\ell}(x)\log(ed/s_{\ell}(x))}{n}} + \sqrt{\frac{\log(1/\delta)}{n}} < \min\left\{\frac{\kappa_R^2}{A_1\phi_{s_{\ell}(x)}}, 1 \right\} \quad \text{ and } \quad d\geq s_{\ell}(x)+2.$$
	If $\gamma = A_2 \frac{\sqrt{\phi_1}}{\kappa_R}\norm{x}_2 \sqrt{n\log(d/\delta)}$ and $r_{\ell} < \left(\frac{A_2^2 \,\kappa_R^6\, \phi_1}{A_3^2\norm{x}_2 \phi_{s_{\ell}(x)}}\right)^{\frac{1}{4}} \left(\frac{\log(d/\delta)}{n}\right)^{\frac{1}{4}}$, then with probability at least $1-\delta$,
	\begin{align*}
		\norm{\hat{\ell}(x)  - \tilde{\ell}(x)}_2 &\lesssim \frac{\norm{x}_2}{\kappa_R^2}\left[\sqrt{\frac{\phi_1}{\kappa_R^2}} \sqrt{\frac{s_{\ell}(x) \log(d/\delta)}{n}} + \frac{\phi_{s_{\ell}(x)}}{\kappa_R^2} \cdot r_{\ell} \right].
	\end{align*}
\end{lemma}

\begin{proof}[Proof of Lemma~\ref{lem:dual_consist}]
	Since $\hat{\ell}(x)$ is the solution to the dual program \eqref{debias_prog_dual}, we know that
	$$\frac{1}{4n} \sum_{i=1}^n \pi(X_i)\left[X_i^T \hat{\ell}(x)\right]^2 + x^T\hat{\ell}(x) + \frac{\gamma}{n}\norm{\hat{\ell}(x)}_1 \leq \frac{1}{4n} \sum_{i=1}^n \pi(X_i)\left[X_i^T \tilde{\ell}(x)\right]^2 + x^T\tilde{\ell}(x) + \frac{\gamma}{n}\norm{\tilde{\ell}(x)}_1.$$
	Some rearrangements (or Taylor's expansion) give us that
	\begin{align}
		\label{dual_taylor}
		\begin{split}
			&\left[\hat{\ell}(x) - \tilde{\ell}(x)\right]^T \left[\frac{1}{4n}\sum_{i=1}^n \pi(X_i) X_iX_i^T \right] \left[\hat{\ell}(x) - \tilde{\ell}(x)\right] \\
			&\quad + \left[\hat{\ell}(x) - \tilde{\ell}(x)\right]^T \left[\frac{1}{2n}\sum_{i=1}^n \pi(X_i) X_iX_i^T \right] \left[\tilde{\ell}(x) - \ell_0(x)\right]\\
			&\leq \left[\frac{1}{2n} \sum_{i=1}^n \pi(X_i) X_iX_i^T \ell_0(x) +x \right]^T \left[\tilde{\ell}(x) - \hat{\ell}(x) \right] + \frac{\gamma}{n}\left[\norm{\tilde{\ell}(x)}_1 - \norm{\hat{\ell}(x)}_1 \right].
		\end{split}
	\end{align}
	In the sequel, we write $\hat{\Delta}(x) = \hat{\ell}(x) - \tilde{\ell}(x)$. By H{\"o}lder's inequality and the (reverse) triangle inequality, together with the sparsity of $\tilde{\ell}(x)$ under Assumption~\ref{assump:sparse_dual}, the right-hand side of \eqref{dual_taylor} is upper bounded by
	\begin{align}
		\label{dual_taylor_bnd}
		\begin{split}
			&\left[\frac{1}{2n} \sum_{i=1}^n \pi(X_i) X_iX_i^T \ell_0(x) +x \right]^T \left[\tilde{\ell}(x) - \hat{\ell}(x) \right] + \frac{\gamma}{n}\left[\norm{\tilde{\ell}(x)}_1 - \norm{\hat{\ell}(x)}_1 \right] \\
			&\leq \norm{\frac{1}{2n} \sum_{i=1}^n \pi(X_i) X_iX_i^T \ell_0(x) +x}_{\infty} \norm{\hat{\Delta}(x)}_1 \\
			&\quad + \frac{\gamma}{n} \left[\norm{\left[\tilde{\ell}(x)\right]_{S_{\ell}(x)}}_1 - \norm{\left[\hat{\Delta}(x)\right]_{S_{\ell}(x)} + \left[\tilde{\ell}(x)\right]_{S_{\ell}(x)}}_1 - \norm{\left[\hat{\ell}(x)\right]_{S_{\ell}(x)^c}}_1\right]\\
			&\leq \frac{\gamma}{5n} \left[\norm{\left[\hat{\Delta}(x)\right]_{S_{\ell}(x)}}_1 + \norm{\left[\hat{\Delta}(x)\right]_{S_{\ell}(x)^c}}_1 \right] + \frac{\gamma}{n} \left[\norm{\left[\hat{\Delta}(x) \right]_{S_{\ell}(x)}}_1 - \norm{\left[\hat{\Delta}(x) \right]_{S_{\ell}(x)^c}}_1 \right]\\
			&= \frac{6\gamma}{5n} \norm{\left[\hat{\Delta}(x) \right]_{S_{\ell}(x)}}_1 - \frac{4\gamma}{5n} \norm{\left[\hat{\Delta}(x) \right]_{S_{\ell}(x)^c}}_1\\
			&\leq \frac{6\gamma}{5n} \sqrt{s_{\ell}(x)} \norm{\hat{\Delta}(x)}_2,
		\end{split}
	\end{align}
	where $\norm{\left[\hat{\Delta}(x) \right]_{S_{\ell}(x)}}_1 = \sum_{k\in S_{\ell}(x)} \left|\hat{\Delta}_k(x) \right|$. 
	
	Next, we know from Lemma~\ref{lem:res_cross_poly} with $a_0=5$ that 
	$$\hat{\ell}(x)-\tilde{\ell}(x) \in \mathcal{C}_1\left(S_{\ell}(x), 3 \right) \bigcup B_d^{(1)}\left(0,\, \frac{13 r_{\ell}^2 \norm{x}_2^2}{2 \kappa_R^4} \cdot \frac{\nu}{\gamma}\right) \equiv \mathcal{C}_1\left(S_{\ell}(x), 3 \right) \bigcup B_d^{(1)}(0,r_1),$$ 
	where $B_d^{(1)}(0,r) = \left\{v\in \mathbb{R}^d: \norm{v}_1 \leq r\right\}$ and $r_1\equiv \frac{13 r_{\ell}^2 \norm{x}_2^2}{2 \kappa_R^4} \cdot \frac{\nu}{\gamma}$. 
	
	Now, by Lemma~\ref{lem:gamma_rate}, we know that there exist absolute constants $A_2,A_2'>0$ such that with probability at least $1-\delta$,
	\begin{align*}
		\norm{\frac{1}{n}\sum_{i=1}^n \pi(X_i) X_iX_i^T \ell_0(x) + 2x}_{\infty} &\leq A_2' \frac{\sqrt{\phi_1}}{\kappa_R}\norm{x}_2 \left[\sqrt{\frac{\log(d/\delta)}{n}} + \frac{\log(d/\delta)}{n} \right] \\
		&\leq A_2 \frac{\sqrt{\phi_1}}{\kappa_R}\norm{x}_2 \sqrt{\frac{\log(d/\delta)}{n}}\equiv \frac{\gamma}{n},
	\end{align*}
	where we note that $\sqrt{\frac{\log(d/\delta)}{n}}$ asymptotically dominates the bound under our growth condition $\sqrt{\frac{s_{\ell}(x)\log(ed/s_{\ell}(x))}{n}} + \sqrt{\frac{\log(1/\delta)}{n}} < 1$. Similarly, we obtain from Lemma~\ref{lem:rel_poly_cost} that there exists an absolute constant $A_3>0$ such that with probability at least $1-\delta$,
	\begin{align*}
		\sup_{v\in \mathcal{C}_2\left(S_{\ell}(x), \vartheta_{\ell}\right)\cap B_d(0,1)} \left|\frac{1}{n} \sum_{i=1}^n \pi_i\left[X_i^Tv\right]^2 \right| &\leq \frac{2A_3}{13}\cdot \phi_{s_{\ell}(x)}\left[1+ \sqrt{\frac{\log(1/\delta) + s_{\ell}(x) \log\left(d/s_{\ell}(x)\right)}{n}}\right] \\
		&\leq \frac{2A_3}{13}\cdot \phi_{s_{\ell}(x)} \equiv \frac{\nu}{n}.
	\end{align*}
	Hence, with probability at least $1-\delta$,
	$$\frac{\kappa_R^2}{\norm{x}_2} \sqrt{\frac{2\gamma}{13\nu}} \geq \left(\frac{A_2^2 \,\kappa_R^6\, \phi_1}{A_3^2\norm{x}_2 \phi_{s_{\ell}(x)}}\right)^{\frac{1}{4}} \left(\frac{\log(d/\delta)}{n}\right)^{\frac{1}{4}} > r_{\ell},$$
	where the last inequality follows from our condition. Hence, with probability at least $1-\delta$, we have that $r_1= \frac{13 r_{\ell}^2 \norm{x}_2^2}{2 \kappa_R^4} \cdot \frac{\nu}{\gamma} < 1$, and 
	$$\left[\mathcal{C}_1\left(S_{\ell}(x), 3 \right) \bigcup B_d^{(1)}(0,r_1)\right] \bigcap \mathbb{S}^{d-1} = \mathcal{C}_1\left(S_{\ell}(x), 3 \right) \bigcap \mathbb{S}^{d-1}.$$
	
	Therefore, by Assumption~\ref{assump:gram_lower_bnd} together with Lemma~\ref{lem:dual_approx_cone} with $\vartheta_{\ell} = \frac{r_{\ell}}{\sqrt{1-r_{\ell}^2}} < 1$, the left-hand side of \eqref{dual_taylor} is lower bounded by
	\begin{align*}
		&\left[\hat{\ell}(x) - \tilde{\ell}(x)\right]^T \left[\frac{1}{4n}\sum_{i=1}^n \pi(X_i) X_iX_i^T \right] \left[\hat{\ell}(x) - \tilde{\ell}(x)\right] \\
		&+ \left[\hat{\ell}(x) - \tilde{\ell}(x)\right]^T \left[\frac{1}{2n}\sum_{i=1}^n \pi(X_i) X_iX_i^T \right] \left[\tilde{\ell}(x) - \ell_0(x)\right]\\
		&\geq \frac{\kappa_R^2}{4} \norm{\hat{\Delta}(x)}_2^2 - \underbrace{\sup_{u \in \mathcal{C}_1(S_{\ell}(x),3)\bigcap \mathbb{S}^{d-1}} \left|\frac{1}{4n} \sum_{i=1}^n \left\{\pi(X_i) \left(X_i^T u\right)^2 - \mathrm{E}\left[\pi(X_i) \left(X_i^T u\right)^2\right]\right\} \right| \norm{\hat{\Delta}(x)}_2^2}_{\textbf{Term I}}\\
		&\quad -\underbrace{\sup_{\substack{u \in \mathcal{C}_1(S_{\ell}(x),3)\bigcap \mathbb{S}^{d-1} \\v\in \mathcal{C}_2(S_{\ell}(x), 1) \cap \mathbb{S}^{d-1}}} \left|\frac{1}{2n} \sum_{i=1}^n \pi(X_i) \left(X_i^T v\right) \left(X_i^T u\right) \right|\cdot r_{\ell}\norm{\ell_0(x)}_2\norm{\hat{\Delta}(x)}_2}_{\textbf{Term II}}.
	\end{align*}
	We now derive bounds on the remaining two terms. Recall our growth condition $\sqrt{\frac{s_{\ell}(x)\log(ed/s_{\ell}(x))}{n}} + \sqrt{\frac{\log(1/\delta)}{n}} < 1$. Using the same arguments in Lemma~\ref{lem:cone_ball_bnd}, we know that with probability at least $1-\delta$,
	\begin{align*}
		\textbf{Term I} &\leq \sup_{u \in \mathcal{C}_1(S_{\ell}(x),3)\bigcap \mathbb{S}^{d-1}} \left|\frac{1}{4n} \sum_{i=1}^n \left\{\pi(X_i) \left(X_i^T u\right)^2 - \mathrm{E}\left[\pi(X_i) \left(X_i^T u\right)^2\right]\right\} \right| \norm{\hat{\Delta}(x)}_2^2 \\
		&= \norm{\hat{\Delta}(x)}_2^2\sup_{v \in \mathcal{C}_1(S_{\ell}(x),3) \cap \mathbb{S}^{d-1}} \left|\frac{1}{4n} \sum_{i=1}^n \left\{\pi(X_i) \left(X_i^T v\right)^2 - \mathrm{E}\left[\pi(X_i) \left(X_i^Tv\right)^2\right]\right\} \right| \\
		&\lesssim \norm{\hat{\Delta}(x)}_2^2 \phi_{s_{\ell}(x)} \sqrt{\frac{\log(1/\delta)+s_{\ell}(x)\log(ed/s_{\ell}(x))}{n}}.
	\end{align*}
	In addition, by Lemma~\ref{lem:cone_ball_bnd} again, we know that with probability at least $1-\delta$,
	\begin{align*}
		\textbf{Term II} \lesssim r_{\ell}\norm{\ell_0(x)}_2\norm{\hat{\Delta}(x)}_2 \left[\phi_{s_{\ell}(x)} \sqrt{\frac{\log(1/\delta)+s_{\ell}(x)\log(ed/s_{\ell}(x))}{n}} + 1\right].
	\end{align*}
	Now, combining \eqref{dual_taylor} and \eqref{dual_taylor_bnd} with the above two bounds yields that there exists an absolute constant $A_0>0$ such that with probability at least $1-\delta$, 
	\begin{align*}
		&\left[\frac{\kappa_R^2}{4} - A_0\phi_{s_{\ell}(x)} \sqrt{\frac{\log(1/\delta)+s_{\ell}(x)\log(ed/s_{\ell}(x))}{n}} \right]\norm{\hat{\Delta}(x)}_2^2 \\
		&\quad - A_0\left[\sqrt{\frac{\log(1/\delta)+s_{\ell}(x)\log(ed/s_{\ell}(x))}{n}} + 1\right]\phi_{s_{\ell}(x)} \cdot r_{\ell} \norm{\ell_0(x)}_2\norm{\hat{\Delta}(x)}_2\\
		&\leq\frac{6\gamma}{5n} \sqrt{s_{\ell}(x)} \norm{\hat{\Delta}(x)}_2.
	\end{align*}
	By our growth condition $\sqrt{\frac{s_{\ell}(x)\log(ed/s_{\ell}(x))}{n}} + \sqrt{\frac{\log(1/\delta)}{n}} < \min\left\{\frac{\kappa_R^2}{A_1 \phi_{s_{\ell}(x)}}, 1 \right\}$ with $A_1=8A_0$, it follows that
	$$\frac{\kappa_R^2}{4} - A_0\phi_{s_{\ell}(x)} \sqrt{\frac{\log(1/\delta)+s_{\ell}(x)\log(ed/s_{\ell}(x))}{n}} \geq \frac{\kappa_R^2}{8},$$
	and therefore,
	\begin{align*}
		\frac{\kappa_R^2}{8} \norm{\hat{\Delta}(x)}_2 &\leq \frac{6\gamma}{5n} \sqrt{s_{\ell}(x)} + 2A_0\left[\sqrt{\frac{\log(1/\delta)+s_{\ell}(x)\log(ed/s_{\ell}(x))}{n}} + 1\right]\phi_{s_{\ell}(x)} \cdot r_{\ell} \cdot \frac{\norm{x}_2}{\kappa_R^2}\\
		&\lesssim \frac{\gamma}{n} \sqrt{s_{\ell}(x)} + \phi_{s_{\ell}(x)} \cdot r_{\ell} \cdot \frac{\norm{x}_2}{\kappa_R^2}
	\end{align*}
	Dividing both sides of the above inequality by $\frac{\kappa_R^2}{8}$ gives us that
	$$\norm{\hat{\ell}(x)  - \tilde{\ell}(x)}_2 \lesssim \frac{\gamma \sqrt{s_{\ell}(x)}}{n \cdot \kappa_R^2} + \frac{\phi_{s_{\ell}(x)}}{\kappa_R^2} \cdot \frac{\norm{x}_2}{\kappa_R^2} \cdot r_{\ell}.$$
	Finally, the bound on $\norm{\hat{\ell}(x) - \tilde{\ell}(x)}_2$ follows from substituting $\frac{\gamma}{n} \equiv A_2 \frac{\sqrt{\phi_1}}{\kappa_R}\norm{x}_2 \sqrt{\frac{\log(d/\delta)}{n}}$ into the upper bound of $\norm{\hat{\ell}(x) - \tilde{\ell}(x)}_2$ above.
\end{proof}

\vspace{3mm}

With the above lemma, we can now formally prove the consistency of the dual solution $\hat{\ell}(x)$ (\autoref{thm:dual_consist2}). Notice that the major difference between the proof of Lemma~\ref{lem:dual_consist} and \autoref{thm:dual_consist2} is that the propensity scores now are estimated in \autoref{thm:dual_consist2}.

\begin{customthm}{3.3}[Consistency of the dual solution $\hat{\ell}(x)$]
	Let $\delta \in (0,1)$, $x\in \mathbb{R}^d$ be a fixed covariate vector, and $A_1,A_2,A_3>0$ be some large absolute constants. Suppose that Assumptions~\ref{assump:basic}, \ref{assump:sub_gaussian}, \ref{assump:gram_lower_bnd}, \ref{assump:prop_consistent_rate}, and \ref{assump:sparse_dual} hold
	as well as
	$$\sqrt{\frac{s_{\ell}(x)\log(ed/s_{\ell}(x))}{n}} + \sqrt{\frac{\log(1/\delta)}{n}} + r_{\pi} < \min\left\{\frac{\kappa_R^2}{A_1 \phi_{s_{\ell}(x)}}, 1 \right\} \quad \text{ and } \quad d\geq s_{\ell}(x) + 2.$$
	We denote $r_{\gamma}\equiv r_{\gamma}(n,\delta) = \frac{\sqrt{\phi_1}}{\kappa_R}\norm{x}_2\sqrt{\frac{\log(d/\delta)}{n}} + \frac{\sqrt{\phi_1 \phi_{s_{\ell}(x)}}\norm{x}_2}{\kappa_R^2} \cdot r_{\pi}$. If $\frac{\gamma}{n} = A_2\cdot r_{\gamma}$ and 
	$$r_{\ell} \leq \min\left\{\frac{1}{2},\, \left[\left(\frac{A_2^2 \,\kappa_R^6\, \phi_1}{A_3^2\norm{x}_2 \phi_{s_{\ell}(x)}}\right)^{\frac{1}{4}} \left(\frac{\log(d/\delta)}{n}\right)^{\frac{1}{4}} + \left(\frac{A_2^2 \phi_1 \kappa_R^4}{A_3^2 \phi_{s_{\ell}(x)} \norm{x}_2^2}\right)^{\frac{1}{4}} \sqrt{r_{\pi}}\right]\right\},$$
	then with probability at least $1-\delta$,
	\begin{align*}
		\norm{\hat{\ell}(x) - \tilde{\ell}(x)}_2
		&\lesssim \frac{\norm{x}_2}{\kappa_R^3}\left[ \sqrt{\frac{\phi_1 s_{\ell}(x)\log(d/\delta)}{n}} + \frac{\phi_{s_{\ell}(x)}}{\kappa_R} \cdot r_{\ell} + \frac{\sqrt{\phi_1 \phi_{s_{\ell}(x)} s_{\ell}(x)}}{\kappa_R} \cdot  r_{\pi} \right].
	\end{align*}
\end{customthm}

\begin{proof}[Proof of \autoref{thm:dual_consist2}]
	Since $\hat{\ell}(x)$ minimizes the objective function of the dual problem \eqref{debias_prog_dual}, we know that
	$$\frac{1}{4n} \sum_{i=1}^n \hat{\pi}_i\left[X_i^T \hat{\ell}(x)\right]^2 + x^T\hat{\ell}(x) + \frac{\gamma}{n}\norm{\hat{\ell}(x)}_1 \leq \frac{1}{4n} \sum_{i=1}^n \hat{\pi}_i\left[X_i^T \tilde{\ell}(x)\right]^2 + x^T\tilde{\ell}(x) + \frac{\gamma}{n}\norm{\tilde{\ell}(x)}_1.$$
	Some rearrangements (or Taylor's expansion) give us that
	\begin{align}
		\label{dual_taylor2}
		\begin{split}
			&\left[\hat{\ell}(x) - \tilde{\ell}(x)\right]^T \left[\frac{1}{4n}\sum_{i=1}^n \hat{\pi}_i X_iX_i^T \right] \left[\hat{\ell}(x) - \tilde{\ell}(x)\right] + \left[\hat{\ell}(x) - \tilde{\ell}(x)\right]^T \left[\frac{1}{2n}\sum_{i=1}^n \hat{\pi}_i X_iX_i^T \right] \left[\tilde{\ell}(x) - \ell_0(x)\right]\\
			&\leq \left[\frac{1}{2n} \sum_{i=1}^n \hat{\pi}_i X_iX_i^T \ell_0(x) +x \right]^T \left[\tilde{\ell}(x) - \hat{\ell}(x) \right] + \frac{\gamma}{n}\left[\norm{\tilde{\ell}(x)}_1 - \norm{\hat{\ell}(x)}_1 \right].
		\end{split}
	\end{align}
	We now reduce the above inequality to \eqref{dual_taylor} in Lemma~\ref{lem:dual_consist} so that we can apply parts of the results in Lemma~\ref{lem:dual_consist}. Some algebra shows that the above inequality is equivalent to
	{\small\begin{align}
		\label{dual_taylor3}
		\begin{split}
			&\underbrace{\left[\hat{\ell}(x) - \tilde{\ell}(x)\right]^T \left[\frac{1}{4n}\sum_{i=1}^n \pi(X_i) X_iX_i^T \right] \left[\hat{\ell}(x) - \tilde{\ell}(x)\right] + \left[\hat{\ell}(x) - \tilde{\ell}(x)\right]^T \left[\frac{1}{2n}\sum_{i=1}^n \pi(X_i) X_iX_i^T \right] \left[\tilde{\ell}(x) - \ell_0(x)\right]}_{\textbf{Term I}}\\
			&\quad + \underbrace{\left[\hat{\ell}(x) - \tilde{\ell}(x)\right]^T \left[\frac{1}{4n}\sum_{i=1}^n \left(\hat{\pi}_i-\pi_i\right) X_iX_i^T \right] \left[\hat{\ell}(x) - \tilde{\ell}(x)\right]}_{\textbf{Term II}}\\
			&\quad + \underbrace{\left[\hat{\ell}(x) - \tilde{\ell}(x)\right]^T \left[\frac{1}{2n}\sum_{i=1}^n \left(\hat{\pi}_i - \pi_i\right) X_iX_i^T \right] \left[\tilde{\ell}(x) - \ell_0(x)\right]}_{\textbf{Term III}}\\
			&\leq \left[\frac{1}{2n} \sum_{i=1}^n \hat{\pi}_i X_iX_i^T \ell_0(x) +x \right]^T \left[\tilde{\ell}(x) - \hat{\ell}(x) \right] + \frac{\gamma}{n}\left[\norm{\tilde{\ell}(x)}_1 - \norm{\hat{\ell}(x)}_1 \right].
		\end{split}
	\end{align}}%
	We first obtain an upper bound on the right-hand side of the above inequality via the same arguments as in the proof of Lemma~\ref{lem:dual_consist}. We write $\hat{\Delta}(x) = \hat{\ell}(x) - \tilde{\ell}(x)$. By H{\"o}lder's inequality and the (reverse) triangle inequality, together with the sparsity of $\tilde{\ell}(x)$ under Assumption~\ref{assump:sparse_dual}, the right-hand side of \eqref{dual_taylor2} is upper bounded by
	\begin{align*}
		&\left[\frac{1}{2n} \sum_{i=1}^n \hat{\pi}_i X_iX_i^T \ell_0(x) +x \right]^T \left[\tilde{\ell}(x) - \hat{\ell}(x) \right] + \frac{\gamma}{n}\left[\norm{\tilde{\ell}(x)}_1 - \norm{\hat{\ell}(x)}_1 \right] \\
		&\leq \norm{\frac{1}{2n} \sum_{i=1}^n \hat{\pi}_i X_iX_i^T \ell_0(x) +x}_{\infty} \norm{\hat{\Delta}(x)}_1 \\
		&\quad + \frac{\gamma}{n} \left[\norm{\left[\tilde{\ell}(x)\right]_{S_{\ell}(x)}}_1 - \norm{\left[\hat{\Delta}(x)\right]_{S_{\ell}(x)} + \left[\tilde{\ell}(x)\right]_{S_{\ell}(x)}}_1 - \norm{\left[\hat{\ell}(x)\right]_{S_{\ell}(x)^c}}_1\right]\\
		&\leq \frac{\gamma}{5n} \left[\norm{\left[\hat{\Delta}(x)\right]_{S_{\ell}(x)}}_1 + \norm{\left[\hat{\Delta}(x)\right]_{S_{\ell}(x)^c}}_1 \right] + \frac{\gamma}{n} \left[\norm{\left[\hat{\Delta}(x) \right]_{S_{\ell}(x)}}_1 - \norm{\left[\hat{\Delta}(x) \right]_{S_{\ell}(x)^c}}_1 \right]\\
		&\leq \frac{6\gamma}{5n} \sqrt{s_{\ell}(x)} \norm{\hat{\Delta}(x)}_2
	\end{align*}
	where $\norm{\left[\hat{\Delta}(x) \right]_{S_{\ell}(x)}}_1 = \sum_{k\in S_{\ell}(x)} \left|\hat{\Delta}_k(x) \right|$.
	
	Next, we derive lower bounds for the terms on the left-hand side of \eqref{dual_taylor3}. We know from Lemma~\ref{lem:res_cross_poly} by $a_0=5$ that 
	$$\hat{\ell}(x)-\tilde{\ell}(x) \in \mathcal{C}_1\left(S_{\ell}(x), 3 \right) \bigcup B_d^{(1)}\left(0,\, \frac{13 r_{\ell}^2 \norm{x}_2^2}{2 \kappa_R^4} \cdot \frac{\nu}{\gamma}\right) \equiv \mathcal{C}_1\left(S_{\ell}(x), 3 \right) \bigcup B_d^{(1)}(0,r_1),$$ 
	where $B_d^{(1)}(0,r) = \left\{v\in \mathbb{R}^d: \norm{v}_1 \leq r\right\}$ and $r_1 \equiv \frac{13 r_{\ell}^2 \norm{x}_2^2}{2 \kappa_R^4} \cdot \frac{\nu}{\gamma}$.
	
	Now, by Lemma~\ref{lem:rel_poly_cost}, there exist absolute constants $A_2,A_2'>0$ such that 
	\begin{align*}
		&\norm{\frac{1}{n}\sum_{i=1}^n \hat{\pi}_i X_i X_i^T \ell_0(x) + 2x}_{\infty} \\
		&\leq A_2'\cdot \frac{\sqrt{\phi_1}}{\kappa_R}\norm{x}_2 \left[\sqrt{\frac{\log(d/\delta)}{n}} + \frac{\log(d/\delta)}{n} \right]\\
		&\quad + A_2'\cdot \frac{\sqrt{\phi_1 \phi_{s_{\ell}(x)}}\norm{x}_2}{\kappa_R^2} \left[1 + \sqrt{\frac{\log(d/\delta)}{n}} + \frac{\log(d/\delta)}{n}\right] r_{\pi}(n,\delta)\\
		&\lesssim A_2\cdot \frac{\sqrt{\phi_1}}{\kappa_R}\norm{x}_2\sqrt{\frac{\log(d/\delta)}{n}} + A_2\cdot \frac{\sqrt{\phi_1 \phi_{s_{\ell}(x)}}\norm{x}_2}{\kappa_R^2} \cdot r_{\pi}(n,\delta) \equiv \frac{\gamma}{n},
	\end{align*}
	where we apply our growth condition $\sqrt{\frac{s_{\ell}(x)\log(ed/s_{\ell}(x))}{n}} + \sqrt{\frac{\log(1/\delta)}{n}} < 1$. Similarly, we obtain from Lemma~\ref{lem:rel_poly_cost} that there exist absolute constants $A_3,A_3'>0$ such that with probability at least $1-\delta$,
	\begin{align*}
		&\sup_{v\in \mathcal{C}_2\left(S_{\ell}(x), \vartheta_{\ell}\right)\cap B_d(0,1)} \left|\frac{1}{n} \sum_{i=1}^n \hat{\pi}_i\left[X_i^Tv\right]^2 \right| \\
		&\leq \frac{2A_3'}{13}\cdot \phi_{s_{\ell}(x)}\left[1+ \sqrt{\frac{\log(1/\delta) + s_{\ell}(x) \log\left(d/s_{\ell}(x)\right)}{n}}\right] \left[1 + r_{\pi}(n,\delta)\right] \\
		&\leq \frac{2A_3}{13}\cdot \phi_{s_{\ell}(x)} \equiv \frac{\nu}{n},
	\end{align*}
	where we apply our growth condition $\sqrt{\frac{s_{\ell}(x)\log(ed/s_{\ell}(x))}{n}} + \sqrt{\frac{\log(1/\delta)}{n}} + r_{\pi} < 1$ again. Hence, with probability at least $1-\delta$,
	$$\frac{\kappa_R^2}{\norm{x}_2} \sqrt{\frac{2\gamma}{13\nu}} \geq \left(\frac{A_2^2 \,\kappa_R^6\, \phi_1}{A_3^2\norm{x}_2 \phi_{s_{\ell}(x)}}\right)^{\frac{1}{4}} \left(\frac{\log(d/\delta)}{n}\right)^{\frac{1}{4}} + \left(\frac{A_2^2 \phi_1 \kappa_R^4}{A_3^2 \phi_{s_{\ell}(x)} \norm{x}_2^2}\right)^{\frac{1}{4}} \sqrt{r_{\pi}} > r_{\ell},$$
	where the last inequality follows from our condition. In other words, with probability at least $1-\delta$, we have that $r_1= \frac{13 r_{\ell}^2 \norm{x}_2^2}{2 \kappa_R^4} \cdot \frac{\nu}{\gamma} < 1$, and 
	$$\left[\mathcal{C}_1\left(S_{\ell}(x), 3 \right) \bigcup B_d^{(1)}(0,r_1)\right] \bigcap \mathbb{S}^{d-1} = \mathcal{C}_1\left(S_{\ell}(x), 3 \right) \bigcap \mathbb{S}^{d-1}.$$
	
	$\bullet$ {\bf Term I:} This term coincides with the left-hand side of \eqref{dual_taylor} in the proof of Lemma~\ref{lem:dual_consist}. Thus, we conclude that there exist absolute constants $A,A_0>0$ such that with probability at least $1-\delta$, 
	\begin{align*}
		&\textbf{Term I} \\
		&\geq \left[\frac{\kappa_R^2}{4} - A\phi_{s_{\ell}(x)} \sqrt{\frac{\log(1/\delta)+s_{\ell}(x)\log(ed/s_{\ell}(x))}{n}} \right]\norm{\hat{\Delta}(x)}_2^2 \\
		&\quad - A\left[\sqrt{\frac{\log(1/\delta)+s_{\ell}(x)\log(ed/s_{\ell}(x))}{n}} + 1\right]\phi_{s_{\ell}(x)} \cdot r_{\ell} \norm{\ell_0(x)}_2\norm{\hat{\Delta}(x)}_2\\
		&\geq \left[\frac{\kappa_R^2}{4} - A_0\phi_{s_{\ell}(x)} \sqrt{\frac{\log(1/\delta)+s_{\ell}(x)\log(ed/s_{\ell}(x))}{n}} \right]\norm{\hat{\Delta}(x)}_2^2 - A_0 \phi_{s_{\ell}(x)} \cdot r_{\ell} \norm{\ell_0(x)}_2\norm{\hat{\Delta}(x)}_2,
	\end{align*}
	where we use our growth condition $\sqrt{\frac{s_{\ell}(x)\log(ed/s_{\ell}(x))}{n}} + \sqrt{\frac{\log(1/\delta)}{n}} < 1$ to obtain the second inequality.
	
	$\bullet$ {\bf Term II:} By Lemmas~\ref{lem:gram_mat_conc} and \ref{lem:cone_ball_bnd}, we have that with probability at least $1-\delta$,
	\begin{align*}
		\left|\textbf{Term II}\right| &\leq \norm{\hat{\Delta}(x)}_2^2 \sup_{u \in \mathcal{C}_1(S_{\ell}(x),3) \bigcap \mathbb{S}^{d-1}} \frac{1}{4n} \sum_{i=1}^n \left|\hat{\pi}_i - \pi_i\right| \left(X_i^T u\right)^2\\
		&\leq \norm{\hat{\Delta}(x)}_2^2 \sup_{u \in \mathcal{C}_1(S_{\ell}(x),3) \bigcap \mathbb{S}^{d-1}} \frac{r_{\pi}(n,\delta)}{4n} \sum_{i=1}^n  \left(X_i^T u\right)^2 \\
		&\lesssim \phi_{s_{\ell}(x)} \cdot r_{\pi}(n,\delta) \norm{\hat{\Delta}(x)}_2^2 \left[\sqrt{\frac{\log(1/\delta)+s_{\ell}(x)\log(ed/s_{\ell}(x))}{n}} + 1\right]\\
		&\lesssim \phi_{s_{\ell}(x)} \cdot r_{\pi}(n,\delta) \norm{\hat{\Delta}(x)}_2^2,
	\end{align*}
	where we again use our growth condition $\sqrt{\frac{s_{\ell}(x)\log(ed/s_{\ell}(x))}{n}} + \sqrt{\frac{\log(1/\delta)}{n}} < 1$.
	
	$\bullet$ {\bf Term III:} By Assumption~\ref{assump:sparse_dual}, Lemmas~\ref{lem:dual_approx_cone} and \ref{lem:cone_ball_bnd} with $\vartheta_{\ell}=\frac{r_{\ell}}{\sqrt{1-r_{\ell}^2}} < 1$ and $r_{\ell} \leq \frac{1}{2}$, 
	\begin{align*}
		&\left|\textbf{Term III}\right| \\
		&\leq \sup_{\substack{u \in \mathcal{C}_1(S_{\ell}(x),3) \bigcap \mathbb{S}^{d-1}\\v\in \mathcal{C}_2(S_{\ell}(x), 1) \cap \mathbb{S}^{d-1}}} \left[\frac{1}{2n} \sum_{i=1}^n \left|\hat{\pi}_i - \pi_i\right| \left|X_i^Tu \right| |X_i^Tv| \right] r_{\ell} \norm{\ell_0(x)}_2 \norm{\hat{\Delta}(x)}_2\\
		&\leq r_{\pi}(n,\delta) \sup_{\substack{u \in \mathcal{C}_1(S_{\ell}(x),3) \bigcap \mathbb{S}^{d-1}\\v\in \mathcal{C}_2(S_{\ell}(x), 1) \cap \mathbb{S}^{d-1}}} \left[\frac{1}{2n} \sum_{i=1}^n \left|X_i^Tu \right| |X_i^Tv| \right] r_{\ell} \norm{\ell_0(x)}_2 \norm{\hat{\Delta}(x)}_2\\
		&\leq r_{\pi}(n,\delta) \cdot r_{\ell} \norm{\ell_0(x)}_2 \norm{\hat{\Delta}(x)}_2 \cdot \phi_{s_{\ell}(x)} \left[\sqrt{\frac{\log(1/\delta) + s_{\ell}(x)\log(ed/s_{\ell}(x))}{n}} + 1 \right]
	\end{align*}
	Under the growth condition $\sqrt{\frac{s_{\ell}(x)\log(ed/s_{\ell}(x))}{n}} + \sqrt{\frac{\log(1/\delta)}{n}} < 1$, we have that with probability at least $1-\delta$, 
	$$\left|\textbf{Term III}\right| \lesssim \phi_{s_{\ell}(x)} \cdot r_{\pi}(n,\delta) \cdot r_{\ell} \norm{\ell_0(x)}_2 \norm{\hat{\Delta}(x)}_2.$$
	
	Combining \eqref{dual_taylor3} with the above three bounds yields that there exists an absolute constant $A_1>0$ such that with probability at least $1-\delta$,
	\begin{align*}
		&\left[\frac{\kappa_R^2}{4} - \frac{A_1}{16}\phi_{s_{\ell}(x)} \sqrt{\frac{\log(1/\delta)+s_{\ell}(x)\log(ed/s_{\ell}(x))}{n}} -\frac{A_1}{16} \cdot \phi_{s_{\ell}(x)} \cdot r_{\pi}(n,\delta) \right]\norm{\hat{\Delta}(x)}_2^2 \\
		&\quad - A_1 \left[1+r_{\pi}(n,\delta)\right] \phi_{s_{\ell}(x)} r_{\ell} \norm{\ell_0(x)}_2\norm{\hat{\Delta}(x)}_2\\
		&\leq \frac{6\gamma}{5n} \sqrt{s_{\ell}(x)} \norm{\hat{\Delta}(x)}_2.
	\end{align*}
	By our growth condition $\sqrt{\frac{s_{\ell}(x)\log(ed/s_{\ell}(x))}{n}} + \sqrt{\frac{\log(1/\delta)}{n}} + r_{\pi} < \min\left\{\frac{\kappa_R^2}{A_1 \phi_{s_{\ell}(x)}}, 1 \right\}$, we conclude by dividing both sides of the inequality by $\norm{\hat{\Delta}(x)}_2$ and $\frac{\kappa_R^2}{8}$ that
	\begin{align*}
		\norm{\hat{\ell}(x) - \tilde{\ell}(x)}_2 &\leq \frac{48\gamma}{5\kappa_R^2 n} \sqrt{s_{\ell}(x)} + \frac{8A_1}{\kappa_R^2} \left[1+r_{\pi}(n,\delta)\right] \phi_{s_{\ell}(x)} r_{\ell} \norm{\ell_0(x)}_2\\
		&\lesssim \frac{\gamma \sqrt{s_{\ell}(x)}}{n \cdot \kappa_R^2} + \frac{\phi_{s_{\ell}(x)} \norm{x}_2}{\kappa_R^4} \cdot r_{\ell}.
	\end{align*}
	Finally, the bound for $\norm{\hat{\ell}(x)-\tilde{\ell}(x)}_2$ follows by substituting the asymptotic rate of $\frac{\gamma}{n}$ into the upper bound of $\norm{\hat{\ell}(x)-\tilde{\ell}(x)}_2$.
\end{proof}

\subsection{Proofs of Lemma~\ref{lem:strong_duality} and Lemma~\ref{lem:negligible_cond_var}}
\label{subapp:proofs_dual}

\begin{customlem}{3.4}[Sufficient conditions for strong duality]
	Let $\delta \in (0,1)$. Under the assumptions of \autoref{thm:dual_consist2}, there exists $\gamma>0$ such that $\frac{\gamma}{n} \asymp r_{\gamma}\equiv r_{\gamma}(n,\delta)$ and strong duality as well as the relation \eqref{primal_dual_sol2} in Proposition~\ref{prop:dual_debias} hold with probability at least $1-\delta$.
\end{customlem}

\begin{proof}[Proof of Lemma~\ref{lem:strong_duality}]
	Given that the primal problem \eqref{debias_prog} is convex, it suffices to verify that Slater's condition holds with probability at least $1-\delta$. That is, it suffices to show that the interior of the constraint set defined as 
	$$\left\{\bm{w}\in \mathbb{R}^n:\norm{x-\frac{1}{\sqrt{n}} \sum_{i=1}^n w_i\cdot\hat{\pi}_i\cdot X_i}_{\infty} < \frac{\gamma}{n}\right\}$$ 
	is nonempty for some $\gamma>0$ satisfying $\frac{\gamma}{n}\asymp r_{\gamma}$. To this end, we define
	\begin{align*}
		\bm{w}^*(x)=\left(w_1^*(x),...,w_n^*(x)\right)^T &= \frac{1}{\sqrt{n}} \bm{X}\left[\mathrm{E}\left(RXX^T\right) \right]^{-1} x \\
		&=\left(\frac{1}{\sqrt{n}} X_1^T\left[\mathrm{E}\left(RXX^T\right) \right]^{-1} x,..., \frac{1}{\sqrt{n}} X_n^T\left[\mathrm{E}\left(RXX^T\right) \right]^{-1} x \right)^T
	\end{align*}
	with $\bm{X}=(X_1,...,X_n)^T \in \mathbb{R}^{n\times d}$. The inverse $\left[\mathrm{E}(RXX^T) \right]^{-1}=\left[\mathrm{E}\left(\pi(X) XX^T \right) \right]^{-1}$ is valid under Assumption~\ref{assump:gram_lower_bnd}. Then,
	\begin{align*}
		&\norm{x-\frac{1}{\sqrt{n}} \sum_{i=1}^n w_i^*(x)\cdot\hat{\pi}_i\cdot X_i}_{\infty} \\
		&\leq \underbrace{\norm{x-\frac{1}{\sqrt{n}} \sum_{i=1}^n w_i^*(x)\cdot\pi_i\cdot X_i}_{\infty}}_{\textbf{Term I}} + \underbrace{\norm{\frac{1}{\sqrt{n}} \sum_{i=1}^n w_i^*(x)\left(\hat{\pi}_i-\pi_i\right) X_i}_{\infty}}_{\textbf{Term II}}.
	\end{align*}
	
	$\bullet$ \textbf{Term I:} We know that $\mathrm{E}\left[\frac{1}{\sqrt{n}} \sum_{i=1}^n w_i^*(x)\cdot\pi_i\cdot X_i \right]=x$. Thus, by Assumption~\ref{assump:sub_gaussian} and Lemma 2.7.7 in \cite{vershynin2018high}, all the random variables 
	$$x_1 - \frac{1}{\sqrt{n}} \sum_{i=1}^n \left(e_1^TX_i\right) \pi_i\cdot w_i^*(x)\,,\,...\,,\, x_d - \frac{1}{\sqrt{n}} \sum_{i=1}^n \left(e_d^TX_i\right) \pi_i\cdot w_i^*(x)$$
	have zero means and sub-exponential tails (but neither identically distributed nor independent), where $e_1,...,e_d$ are the standard basis vectors in $\mathbb{R}^d$. Moreover, we can write
	$$x_k - \frac{1}{\sqrt{n}} \sum_{i=1}^n \left(e_k^TX_i\right) \pi_i\cdot w_i^*(x) = \frac{1}{n}\sum_{i=1}^n\left(x_k - \pi(X_i)X_{ik} X_i^T \left[\mathrm{E}(X_1R_1X_1^T) \right]^{-1} x \right)$$
	for all $k=1,...,d$. By Assumption~\ref{assump:sub_gaussian} and Lemma~\ref{lem:Orlicz_prod}, we know that 
	\begin{align*}
		&\max_{1\leq k\leq d} \norm{x_k - \pi(X_i)X_{ik} X_i^T \left[\mathrm{E}(X_1R_1X_1^T) \right]^{-1} x}_{\psi_1} \\
		&\leq \max_{1\leq k\leq d}\left(\norm{\pi(X_i)X_{ik} X_i^T \left[\mathrm{E}(X_1R_1X_1^T) \right]^{-1} x}_{\psi_1} + \frac{1}{\log 2} \left|\mathrm{E}\left\{\pi(X_i)X_{ik} X_i^T \left[\mathrm{E}(X_1R_1X_1^T) \right]^{-1} x \right\} \right|\right)\\
		& \leq \left(1+\frac{1}{\log 2}\right) \max_{1\leq k\leq d} \norm{\pi(X_i)X_{ik} X_i^T \left[\mathrm{E}(X_1R_1X_1^T) \right]^{-1} x}_{\psi_1}\\
		&\lesssim \max_{1\leq k\leq d} \max_{1\leq i\leq n} \norm{X_{ik}}_{\psi_2}\norm{\pi(X_i)X_i^T \left[\mathrm{E}(X_1R_1X_1^T) \right]^{-1} x}_{\psi_2}\\
		&\lesssim \frac{\sqrt{\phi_1}}{\kappa_R} \norm{x}_2
	\end{align*}
	for all $1\leq i \leq n$. Therefore, by the union bound and Bernstein's inequality (Corollary 2.8.3 in \citealt{vershynin2018high}), we know that there exists an absolute constant $A_0>0$ such that for all $t>0$,
	\begin{align*}
		&\mathrm{P}\left(\norm{x - \frac{1}{\sqrt{n}} \sum_{i=1}^n w_i^*(x)\cdot \pi_i\cdot X_i}_{\infty} > t \right) \\
		& \leq d\cdot \max_{1\leq k\leq d} \mathrm{P}\left\{\left| \frac{1}{n}\sum_{i=1}^n\left(x_k - \pi(X_i)X_{ik} X_i^T \left[\mathrm{E}(X_1R_1X_1^T) \right]^{-1} x \right) \right| > t \right\}\\
		&\leq 2d \cdot \exp\left(-nA_0 \min\left\{\frac{t \kappa_R}{\sqrt{\phi_1}\norm{x}_2},\, \frac{t^2 \kappa_R^2}{\phi_1\norm{x}_2^2} \right\} \right).
	\end{align*}
	In other words, with probability at least $1-\delta$,
	\begin{align*}
		\textbf{Term I}=\norm{x - \frac{1}{\sqrt{n}} \bm{X}^T \bm{\Pi}\bm{w}^*(x)}_{\infty} &\lesssim \frac{\sqrt{\phi_1} \norm{x}_2}{\kappa_R} \left(\sqrt{\frac{\log(2d/\delta)}{n}} + \frac{\log(2d/\delta)}{n} \right)\\
		&\lesssim \frac{\sqrt{\phi_1} \norm{x}_2}{\kappa_R} \sqrt{\frac{\log(d/\delta)}{n}},
	\end{align*}
	where we apply our growth condition  $\sqrt{\frac{s_{\ell}(x)\log(ed/s_{\ell}(x))}{n}} + \sqrt{\frac{\log(1/\delta)}{n}} < 1$ in the last asymptotic inequality.
	
	\noindent $\bullet$ \textbf{Term II:} We know from the proof of Lemma~\ref{lem:gamma_rate} (specifically \eqref{dual_para_term_II}) that
	\begin{align*}
		\textbf{Term II} &= \norm{\frac{1}{\sqrt{n}} \sum_{i=1}^n w_i^*(x) \left(\hat{\pi}_i -\pi_i\right) X_i}_{\infty}\\
		&=\norm{\frac{1}{n} \sum_{i=1}^n \left(\hat{\pi}_i-\pi_i\right)X_iX_i^T\left[\mathrm{E}\left(RXX^T\right) \right]^{-1} x}_{\infty}\\
		&= \norm{\frac{1}{2n} \sum_{i=1}^n \left(\hat{\pi}_i-\pi_i\right)X_iX_i^T\ell_0(x)}_{\infty}\\
		&\lesssim \frac{\sqrt{\phi_1 \phi_{s_{\ell}(x)}}\norm{x}_2}{\kappa_R^2} \cdot  r_{\pi}(n,\delta),
	\end{align*}
	where we again apply our growth condition  $\sqrt{\frac{s_{\ell}(x)\log(ed/s_{\ell}(x))}{n}} + \sqrt{\frac{\log(1/\delta)}{n}} < 1$ in the last asymptotic inequality.
	
	In summary, we conclude that there exists an absolute constant $A_0>0$ such that with probability at least $1-\delta$, the constraint set is nonempty for all $\gamma >0$ satisfying
	\begin{align*}
		\frac{\gamma}{n} &\gtrsim A_0\left[\frac{\sqrt{\phi_1}\norm{x}_2}{\kappa_R} \sqrt{\frac{\log(d/\delta)}{n}} + \frac{\sqrt{\phi_1 \cdot \phi_{s_{\ell}(x)}}\norm{x}_2}{\kappa_R^2} \cdot r_{\pi}(n,\delta) \right]\\
		&\equiv A_0\cdot r_{\gamma}(n,\delta)
	\end{align*}
	that is defined in \autoref{thm:dual_consist2}. In particular, there exists $\gamma >0$ such that $\frac{\gamma}{n} \asymp r_{\gamma}(n,\delta)$ and Slater's condition holds.
\end{proof}

\vspace{3mm}

\begin{customlem}{3.5}[Negligible conditional variance term]
	Suppose that Assumptions~\ref{assump:basic}, \ref{assump:sub_gaussian}, \ref{assump:subgau_error}, \ref{assump:gram_lower_bnd}, \ref{assump:prop_consistent_rate}, and \ref{assump:sparse_dual} hold
	as well as 
	$$\left[s_{\beta} \log(d/s_{\beta}) + s_{\ell}(x) \log(d/s_{\ell}(x))\right]^2 =O\left(\sqrt{n}\right), \quad \frac{\norm{x}_2}{\kappa_R^2}\sqrt{\phi_{2s_{\ell}(x)} \phi_{2s_{\beta}}} = O(1), \quad \text{ and } \quad r_{\ell}=O(1).$$
	If $\norm{\hat{\beta}-\beta_0}_2=o_P(1)$ and $\norm{\hat{\ell}(x) - \tilde{\ell}(x)}_2 =o_P(1)$, then the ``Conditional variance II'' term in Proposition~\ref{prop:cond_MSE} is asymptotically negligible in the sense that
	\begin{align*}
		\left(\beta_0 - \hat{\beta}\right)^T \left[\sum_{i=1}^n \hat{w}_i(x)^2\hat{\pi}_i\left(1-\hat{\pi}_i\right)X_iX_i^T \right] \left(\beta_0 - \hat{\beta}\right) =o_P(1).
	\end{align*}
\end{customlem}

\begin{proof}[Proof of Lemma~\ref{lem:negligible_cond_var}]
	By Lemma~\ref{lem:strong_duality}, strong duality holds with probability at least $1-\delta$ under our assumptions. Moreover, by the proofs of \autoref{thm:lasso_const} and \autoref{thm:dual_consist2}, we know that with probability at least $1-\delta$,
	$$\hat{\beta} -\beta_0 \in \mathcal{C}_1(S_{\beta}, 3), \quad \hat{\ell}(x) -\tilde{\ell}(x) \in \mathcal{C}_1\left(S_{\ell}(x), 3 \right) \bigcup B_d^{(1)}\left(0,\, \frac{13 r_{\ell}^2 \norm{x}_2^2}{2 \kappa_R^4} \cdot \frac{\nu}{\gamma}\right)\equiv \mathcal{C}_1\left(S_{\ell}(x), 3 \right) \bigcup B_d^{(1)}(0,r_1),$$
	where $r_1\equiv \frac{13 r_{\ell}^2 \norm{x}_2^2}{2 \kappa_R^4} \cdot \frac{\nu}{\gamma} < 1$ under the condition in \autoref{thm:dual_consist2}. In other words, with probability at least $1-\delta$,
	$$\left[\mathcal{C}_1\left(S_{\ell}(x), 3 \right) \bigcup B_d^{(1)}(0,r_1)\right] \bigcap \mathbb{S}^{d-1} = \mathcal{C}_1\left(S_{\ell}(x), 3 \right) \bigcap \mathbb{S}^{d-1}.$$
	Thus, we derive from \eqref{primal_dual_sol2} that
	\begin{align}
		\label{neg_varI}
		\begin{split}
			&\left(\beta_0 - \hat{\beta}\right)^T \left[\sum_{i=1}^n \hat{w}_i(x)^2\hat{\pi}_i\left(1-\hat{\pi}_i\right)X_iX_i^T \right] \left(\beta_0 - \hat{\beta}\right) \\
			&= \left(\beta_0 - \hat{\beta}\right)^T \left[\frac{1}{4n}\sum_{i=1}^n \hat{\pi}_i\left(1-\hat{\pi}_i\right) \left[X_i^T \hat{\ell}(x)\right]^2 X_iX_i^T \right] \left(\beta_0 - \hat{\beta}\right).
		\end{split}
	\end{align}
	Notice that we do not assume the estimated propensity scores $\hat{\pi}_i, i=1,...,n$ to be within $[0,1]$. In order to eliminate the factor $\hat{\pi}_i\left(1-\hat{\pi}_i\right)$ in each summand, we observe the identity that $a(1-a)-b(1-b) = (1-2b)(a-b) - (a-b)^2$ for all $a,b\in \mathbb{R}$. Hence, with probability at least $1-\delta$, we have that
	$$\left|\hat{\pi}_i\left(1-\hat{\pi}_i\right) - \pi_i(1-\pi_i)\right| \leq \left| \hat{\pi}_i -\pi_i \right| + \left| \hat{\pi}_i -\pi_i \right|^2 \lesssim r_{\pi}$$
	for all $i=1,...,n$, where the last asymptotic inequality follows from Assumption~\ref{assump:prop_consistent_rate}. Moreover, since $\max_{1\leq i \leq n} \left|\pi_i (1-\pi_i) \right| \leq \frac{1}{4}$, we can upper bound \eqref{neg_varI} (up to an absolute constant) with probability at least $1-\delta$ by
	$$\left(\beta_0 - \hat{\beta}\right)^T \left[\frac{(1+r_{\pi})}{16n}\sum_{i=1}^n \left[X_i^T \hat{\ell}(x)\right]^2 X_iX_i^T \right] \left(\beta_0 - \hat{\beta}\right).$$
	We further upper bound this quantity as follows:
	\begin{align}
		\label{neg_varII}
		&\left(\beta_0 - \hat{\beta}\right)^T \left[\frac{(1+r_{\pi})}{16n}\sum_{i=1}^n \left[X_i^T \hat{\ell}(x)\right]^2 X_iX_i^T \right] \left(\beta_0 - \hat{\beta}\right) \nonumber\\
		&\lesssim (1+r_{\pi}) \left(\beta_0 - \hat{\beta}\right)^T \left[\frac{1}{16n}\sum_{i=1}^n \left[X_i^T \left(\hat{\ell}(x) - \tilde{\ell}(x)\right)\right]^2 X_iX_i^T + \frac{1}{16n}\sum_{i=1}^n \left[X_i^T \tilde{\ell}(x)\right]^2 X_iX_i^T \right] \left(\beta_0 - \hat{\beta}\right) \nonumber\\
		\begin{split}
			&\lesssim (1+r_{\pi}) \norm{\hat{\beta}-\beta_0}_2^2 \Biggl[\sup_{\substack{u\in \mathcal{C}_1\left(S_{\ell}(x), 3 \right) \bigcap \mathbb{S}^{d-1} \\v\in \mathcal{C}_1(S_{\beta},3) \cap \mathbb{S}^{d-1}}} \frac{1}{n}\sum_{i=1}^n \left(X_i^Tu\right)^2 \left(X_i^T v\right)^2 \norm{\hat{\Delta}(x)}_2^2\\
			&\hspace{35mm} + \sup_{\substack{u\in \mathbb{S}^{d-1},\norm{u}_0\leq s_{\ell}(x)\\v\in \mathcal{C}_1(S_{\beta},3) \cap \mathbb{S}^{d-1}}} \frac{\norm{\tilde{\ell}(x)}_2^2}{n}\sum_{i=1}^n \left(X_i^Tu\right)^2 \left(X_i^T v\right)^2\Biggr],
		\end{split}
	\end{align}
	where $\hat{\Delta}(x) = \hat{\ell}(x) -\tilde{\ell}(x)$. 
	
	Now, by Lemma~\ref{lem:size_dom_cone}, there exist discrete sets $M_{d,\ell}, M_{d,\beta} \subset B_d(0,1)$ with cardinality $|M_{d,\ell}| \leq \frac{3}{2} \left(\frac{5ed}{s_{\ell}(x)} \right)^{s_{\ell}(x)}$ and $|M_{d,\beta}| \leq \frac{3}{2} \left(\frac{5ed}{s_{\beta}} \right)^{s_{\beta}}$ such that $\mathcal{C}_1(S_{\ell}(x),3) \cap \mathbb{S}^{d-1} \subset 10 \cdot \mathrm{conv}(M_{d,\ell})$, $\mathcal{C}_1(S_{\beta},3) \cap \mathbb{S}^{d-1} \subset 10 \cdot \mathrm{conv}(M_{d,\beta})$, and $\norm{u}_0 \leq s_{\ell}(x)$ and $\norm{v}_0 \leq s_{\beta}$ for all $u\in M_{d,\ell}$ and $v\in M_{d,\beta}$. Moreover, we know from Lemma~\ref{lem:max_biconvex} that the supremum of a biconvex function on $\mathrm{conv}(M_{d,\ell})$ or $\mathrm{conv}(M_{d,\beta})$ is attained at its vertices, which are contained in $M_{d,\ell}$ or $M_{d,\beta}$. Hence, we can upper-bound the above display by
	\begin{align*}
		&(1+r_{\pi}) \norm{\hat{\beta}-\beta_0}_2^2 \Bigg[\sup_{\substack{u\in M_{d,\ell} \\v\in M_{d,\beta}}} \frac{\norm{\hat{\ell}(x) -\tilde{\ell}(x)}_2^2}{n}\sum_{i=1}^n \left(X_i^Tu\right)^2 \left(X_i^T v\right)^2  \\
		&\hspace{35mm} + \sup_{\substack{u\in M_{d,\ell}\\v\in M_{d,\beta}}} \frac{\norm{\tilde{\ell}(x)}_2^2}{n}\sum_{i=1}^n \left(X_i^Tu\right)^2 \left(X_i^T v\right)^2\Bigg]\\
		&\leq (1+r_{\pi}) \norm{\hat{\beta}-\beta_0}_2^2 \left[\norm{\hat{\ell}(x) -\tilde{\ell}(x)}_2^2 + (1+r_{\ell})^2 \norm{\ell_0(x)}_2^2\right] \\
		&\quad \times \left[\sup_{\substack{u\in M_{d,\ell}\\v\in M_{d,\beta}}}\mathrm{E}\left[(X^Tu)^2 (X^Tv)^2\right] + \frac{1}{\sqrt{n}}\norm{\mathbb{G}_n}_{\mathcal{F}} \right],
	\end{align*}
	where
	$$\mathcal{F}=\left\{f_{u,v}: \mathbb{R}^d \to \mathbb{R}: f_{u,v}(X) = \left(X^Tu\right)^2 \left(X^Tv\right)^2 - \mathrm{E}\left[\left(X^Tu\right)^2 \left(X^Tv\right)^2 \right],\, u\in M_{d,\ell}, v\in M_{d,\beta}\right\}$$
	and we utilize Assumption~\ref{assump:sparse_dual} to obtain that
	$$\norm{\tilde{\ell}(x)}_2 \leq \norm{\tilde{\ell}(x)-\ell_0(x)}_2 + \norm{\ell_0(x)}_2 \leq (1+r_{\ell}) \norm{\ell_0(x)}_2$$
	in the last inequality. Now, by Assumption~\ref{assump:sub_gaussian}, we know from the Cauchy-Schwarz inequality that
	\begin{align*}
		\sup_{\substack{u\in M_{d,\ell}\\v\in M_{d,\beta}}}\mathrm{E}\left[(X^Tu)^2 (X^Tv)^2\right] & \leq \sup_{u\in M_{d,\ell}} \sqrt{\mathrm{E}\left[(X^Tu)^4\right]}  \sup_{v\in M_{d,\beta}} \sqrt{\mathrm{E}\left[(X^Tv)^4\right]}\\
		&\lesssim \sup_{u\in M_{d,\ell}}\norm{X^Tu}_{\psi_2}^2 \sup_{v\in M_{d,\beta}}\norm{X^Tv}_{\psi_2}^2\\
		&\lesssim \phi_{s_{\ell}(x)}  \phi_{s_{\beta}}.
	\end{align*}
	Furthermore, for any $f_{u_1,v_1},f_{u_2,v_2} \in \mathcal{F}$ and $\alpha=\frac{1}{2}$, we have that 
	\begin{align*}
		&\norm{f_{u_1,v_1} - f_{u_2,v_2}}_{\mathrm{P},\psi_{\alpha}} \\
		&= \norm{\left(X^Tu_1\right)^2 \left(X^Tv_1\right)^2 - \mathrm{E}\left[\left(X^Tu_1\right)^2 \left(X^Tv_1\right)^2 \right] - \left(X^Tu_2\right)^2 \left(X^Tv_2\right)^2 + \mathrm{E}\left[\left(X^Tu_2\right)^2 \left(X^Tv_2\right)^2 \right]}_{\psi_{\alpha}}\\
		&\stackrel{\text{(v)}}{\lesssim} \norm{\left(X^Tu_1\right)^2 \left(X^Tv_1\right)^2 - \left(X^Tu_2\right)^2 \left(X^Tv_2\right)^2}_{\psi_{\alpha}} \\
		&\lesssim \norm{\left[\left(X^Tu_1\right)^2 - \left(X^Tu_2\right)^2 \right] \left(X^Tv_1\right)^2}_{\psi_{\alpha}} + \norm{\left(X^Tu_2\right)^2\left[\left(X^Tv_1\right)^2 - \left(X^Tv_2\right)^2\right]}_{\psi_{\alpha}}\\
		&\stackrel{\text{(vi)}}{\leq} \norm{X^T(u_1+u_2)}_{\psi_2} \norm{X^T(u_1-u_2)}_{\psi_2} \norm{X^Tv_1}_{\psi_2}^2 + \norm{X^Tu_2}_{\psi_2}^2 \norm{X^T(v_1+v_2)}_{\psi_2} \norm{X^T(v_1-v_2)}_{\psi_2}\\
		&\stackrel{\text{(vii)}}{\lesssim} \phi_{2s_{\ell}(x)} \phi_{s_{\beta}} \norm{u_1-u_2}_2 + \phi_{s_{\ell}(x)} \phi_{2s_{\beta}} \norm{v_1-v_2}_2\\
		&\lesssim \phi_{2s_{\ell}(x)} \phi_{2s_{\beta}} \sqrt{\norm{u_1-u_2}_2^2 + \norm{v_1-v_2}_2^2},
	\end{align*}
	where the asymptotic inequality (v) follows from the fact that
	$$\norm{\mathrm{E}\left[\left(X^Tu_1\right)^2 \left(X^Tv_1\right)^2 - \left(X^Tu_2\right)^2 \left(X^Tv_2\right)^2\right]}_{\psi_{\alpha}} \lesssim \norm{\left(X^Tu_1\right)^2 \left(X^Tv_1\right)^2 - \left(X^Tu_2\right)^2 \left(X^Tv_2\right)^2}_{\psi_{\alpha}}$$
	because for any random variable $X$ with $\norm{X}_{\psi_{\alpha}}<\infty$, we know that $\mathrm{P}\left(|X| > t\right) \leq 2e^{-\frac{t^{\alpha}}{\norm{X}_{\psi_{\alpha}}^{\alpha}}}$ and
	\begin{align*}
		\mathrm{E}|X|^p &= \int_0^{\infty} \mathrm{P}\left(|X| > t\right) pt^{p-1} dt\\
		&\leq 2\int_0^{\infty} pt^{p-1} e^{-\frac{t^{\alpha}}{\norm{X}_{\psi_{\alpha}}^{\alpha}}} dt\\
		&=\frac{2p}{\alpha} \norm{X}_{\psi_{\alpha}}^p\int_0^{\infty} u^{\frac{p}{\alpha}-1} e^{-u}du \quad \text{ as we take } u=\frac{t^{\alpha}}{\norm{X}_{\psi_{\alpha}}^{\alpha}}\\
		&= \frac{2p}{\alpha} \cdot\Gamma\left(\frac{p}{\alpha}\right)\cdot \norm{X}_{\psi_{\alpha}}^p,
	\end{align*}
	where $\Gamma(\cdot)$ is the Gamma function. In addition, the inequality (vi) is due to Lemma~\ref{lem:Orlicz_prod} and (vii) uses the fact that $\norm{u_1\pm u_2}_0\leq 2s_{\ell}(x)$ and $\norm{v_1\pm v_2}_0 \leq 2s_{\beta}$ for any $u_1,u_2 \in M_{d,\ell}$ and $v_1,v_2\in M_{d,\beta}$. In particular, $\norm{f_{u,v}-f_{0,0}}_{\mathrm{P},\psi_{\alpha}} = \norm{f_{u,v}}_{\mathrm{P},\psi_{\alpha}}$ also satisfies the above bound. The diameter of $\mathcal{F}$ with respect to the norm
	$$\varrho\left((u_1,v_1), (u_2,v_2)\right) := \sqrt{\norm{u_1-u_2}_2^2 + \norm{v_1-v_2}_2^2}$$
	is at most 2. Moreover, by Lemma~\ref{lem:size_dom_cone}, the covering number satisfies that for any $\varepsilon>0$,
	$$N\left(\varepsilon, M_{d,\ell}\times M_{d,\beta}, \varrho \right) \leq |M_{d,\ell}| \times |M_{d,\beta}| \leq \frac{9}{4} \left(\frac{5ed}{s_{\ell}(x)}\right)^{s_{\ell}(x)} \left(\frac{5ed}{s_{\beta}}\right)^{s_{\beta}}.$$
	Thus, by Lemma~\ref{lem:empirical_bound} and Remark~\ref{remark:index_class}, we obtain that with probability at least $1-\delta$, 
	\begin{align*}
		\frac{1}{\sqrt{n}} \norm{\mathbb{G}_n}_{\mathcal{F}} & \lesssim \frac{1}{\sqrt{n}} \norm{\mathbb{G}_n}_{\mathcal{F}_{\eta}}\\
		&\lesssim \frac{\phi_{2s_{\ell}(x)}  \phi_{2s_{\beta}}}{\sqrt{n}}\Biggl[\int_0^2 \sqrt{\log N\left(\varepsilon, M_{d,\ell}\times M_{d,\beta}, \varrho\right)} d\varepsilon \\
		&\quad + \frac{1}{\sqrt{n}} \int_0^2 \left[\log N\left(\varepsilon, M_{d,\ell}\times M_{d,\beta}, \varrho\right)\right]^2 d\varepsilon + 2\left(\sqrt{\log(1/\delta)} + \frac{\left[\log(1/\delta)\right]^2}{\sqrt{n}}\right)\Biggr]\\
		&\lesssim \frac{\phi_{2s_{\ell}(x)}  \phi_{2s_{\beta}}}{\sqrt{n}}  \Biggl[\sqrt{s_{\ell}(x)\log\left(\frac{ed}{s_{\ell}(x)}\right) + s_{\beta}\log\left(\frac{ed}{s_{\beta}}\right)} \\
		&\quad + \frac{1}{\sqrt{n}} \left(s_{\ell}(x)\log\left(\frac{ed}{s_{\ell}(x)}\right) + s_{\beta}\log\left(\frac{ed}{s_{\beta}}\right) \right)^2 + \sqrt{\log(1/\delta)} + \frac{\left[\log(1/\delta)\right]^2}{\sqrt{n}}\Biggr]\\
		&\lesssim \phi_{2s_{\ell}(x)} \phi_{2s_{\beta}} \left[\sqrt{ \frac{s_{\ell}(x)\log(ed/s_{\ell}(x)) }{n} + \frac{s_\beta\log(ed/s_\beta ) }{n} } \right. \\
		&\quad \left. +  \left(\frac{s_{\ell}(x)\log(ed/s_{\ell}(x)) }{\sqrt{n}} + \frac{s_\beta\log(ed/s_\beta ) }{\sqrt{n}}\right)^2 + \sqrt{\frac{\log(1/\delta)}{n}} + \left(\frac{\log(1/\delta)}{\sqrt{n}} \right)^2\right].
	\end{align*}
	Set $\delta = \frac{1}{\max\{d,n\}}$. By our growth conditions $\left[s_{\beta} \log(d/s_{\beta}) + s_{\ell}(x) \log(d/s_{\ell}(x))\right]^2 =O\left(\sqrt{n}\right)$ and $\frac{\norm{x}_2}{\kappa_R^2}\sqrt{\phi_{2s_{\ell}(x)} \phi_{2s_{\beta}}} = O(1)$, we conclude that $\frac{1}{\sqrt{n}}\norm{\mathbb{G}_n}_{\mathcal{F}} = O_P(1)$.
	
	Finally, combining our above results with \autoref{thm:lasso_const} and \autoref{thm:dual_consist2}, we conclude that
	\begin{align*}
		&\left(\beta_0 - \hat{\beta}\right)^T \left[\sum_{i=1}^n \hat{w}_i(x)^2\hat{\pi}_i\left(1-\hat{\pi}_i\right)X_iX_i^T \right] \left(\beta_0 - \hat{\beta}\right)\\
		&\lesssim \norm{\hat{\beta}-\beta_0}_2^2 \left[\norm{\hat{\ell}(x) -\tilde{\ell}(x)}_2^2 + (1+r_{\ell})^2 \frac{\norm{x}_2^2}{\kappa_R^4}\right]\left[\phi_{s_{\ell}(x)} \phi_{s_{\beta}} + \frac{1}{\sqrt{n}}\norm{\mathbb{G}_n}_{\mathcal{F}} \right]\\
		&=o_P(1) \left[o_P(1) + (1+r_{\ell})^2 \frac{\norm{x}_2^2}{\kappa_R^4}\right] \left[\phi_{s_{\ell}(x)} \phi_{s_{\beta}} + O_P(1)\right]\\
		&=o_P(1)
	\end{align*}
	by our growth condition $\frac{\norm{x}_2}{\kappa_R^2}\sqrt{\phi_{2s_{\ell}(x)} \phi_{2s_{\beta}}} = O(1)$ again.
\end{proof}

\subsection{Proof of \autoref{thm:asym_normal}}
\label{subapp:proof_asym_normal}

\begin{customthm}{3.6}[Asymptotic normality of the debiased estimator]
	Let $x\in \mathbb{R}^d$ be a fixed query point with $x\neq 0$ and $s_{\max}=\max\left\{s_{\beta},s_{\ell}(x)\right\}$. Suppose that Assumptions~\ref{assump:basic}, \ref{assump:sub_gaussian}, \ref{assump:subgau_error}, \ref{assump:gram_lower_bnd}, \ref{assump:prop_consistent_rate}, and \ref{assump:sparse_dual} hold and $\lambda, \gamma, r_{\ell} >0$ are specified as in \autoref{thm:lasso_const}(b) and \autoref{thm:dual_consist2}, respectively. Furthermore, we assume that $\sigma_{\epsilon}\phi_1\phi_{s_{\max}}^{3/2}\norm{x}_2=O(1)$, 
	\begin{align*}
		\frac{(1+\kappa_R^2)s_{\max}\log(nd)}{\kappa_R^4} = o\left(\sqrt{n}\right), \quad \text{ and } \quad \frac{(1+\kappa_R^4)\sqrt{\log(nd)}}{\kappa_R^6}\left(\sqrt{s_{\max}}\cdot r_{\ell} +s_{\max}r_{\pi}\right)=o(1).
	\end{align*} 
	Then, when $\sigma_m^2(x) = \lim_{n\to\infty} \sigma_{\epsilon}^2 \cdot x^T\left[\mathrm{E}\left(R XX^T\right)\right]^{-1}x$ exists, we have that
	$$\sqrt{n}\left[\hat{m}^{\mathrm{debias}}(x;\hat{\bm{w}}) - m_0(x)\right] \stackrel{d}{\to} \mathcal{N}\left(0,\, \sigma_m^2(x) \right).$$
\end{customthm}

\begin{proof}[Proof of \autoref{thm:asym_normal}]
	We continue the calculation in \eqref{bias_var_dual} and decompose the ``Bias terms'' as follows:
	\begin{align*}
		\sqrt{n}\left[\hat{m}^{\text{debias}}(x;\hat{\bm{w}}) - m_0(x)\right]
		&= -\frac{1}{2\sqrt{n}} \sum_{i=1}^n R_i\epsilon_i X_i^T\hat{\ell}(x) + \left[\frac{1}{2n} \sum_{i=1}^n R_i X_i^T\hat{\ell}(x) X_i +x \right]^T \sqrt{n}\left(\hat{\beta} -\beta_0 \right)\\
		&= \underbrace{-\frac{1}{2\sqrt{n}} \sum_{i=1}^n R_i\epsilon_i X_i^T\ell_0(x)}_{\text{``Stochastic term''}} + \underbrace{\frac{1}{2\sqrt{n}} \sum_{i=1}^n R_i\epsilon_i X_i^T\left[\ell_0(x) - \hat{\ell}(x)\right]}_{\textbf{Bias I}} \\
		&\quad + \underbrace{\left[\frac{1}{2n} \sum_{i=1}^n R_i X_i^T\ell_0(x) X_i +x \right]^T \sqrt{n}\left(\hat{\beta} -\beta_0 \right)}_{\textbf{Bias II}}\\
		&\quad + \underbrace{\left[\frac{1}{2n} \sum_{i=1}^n R_i X_i^T\left[\hat{\ell}(x) - \ell_0(x)\right] X_i \right]^T \sqrt{n}\left(\hat{\beta} -\beta_0 \right)}_{\textbf{Bias III}}.
	\end{align*}
	Notice that the ``Stochastic term'' can be written as:
	$$-\frac{1}{2\sqrt{n}} \sum_{i=1}^n R_i\epsilon_i X_i^T\ell_0(x) = \frac{1}{\sqrt{n}}\sum_{i=1}^n R_i\epsilon_i X_i^T \left[\mathrm{E}\left(R XX^T\right)\right]^{-1} x$$
	according to $\ell_0(x) = -2\left[\mathrm{E}\left(R XX^T\right)\right]^{-1} x$ in \eqref{popul_dual_exact}, where $R_i\epsilon_i X_i^T\left[\mathrm{E}\left(R XX^T\right)\right]^{-1} x,i=1,...,n$ are i.i.d. random variables with mean 0 and variance $\sigma_{m,d}^2(x) = \sigma_{\epsilon}^2\cdot x^T\left[\mathrm{E}(R XX^T)\right]^{-1}x < \infty$. Moreover, the limit $\sigma_m^2(x) = \lim_{n\to\infty}\sigma_{\epsilon}^2\cdot x^T\left[\mathrm{E}(R XX^T)\right]^{-1}x$ exists by assumption and is finite under Assumption~\ref{assump:gram_lower_bnd}. We now verify the condition of Lyapunov central limit theorem. Specifically, each random variable $R_i\epsilon_i X_i^T \left[\mathrm{E}\left(R XX^T\right)\right]^{-1} x$ has mean zero, is subexponential, and satisfies that
	\begin{align*}
		\left[\mathbb{E}\left|R_i\epsilon_i X_i^T \left[\mathrm{E}\left(R XX^T\right)\right]^{-1} x\right|^3 \right]^{\frac{1}{3}} &\lesssim \norm{R_i\epsilon_i X_i^T \left[\mathrm{E}\left(R XX^T\right)\right]^{-1} x}_{\psi_1} \\
		&\lesssim \norm{\epsilon_i}_{\psi_2} \norm{R_i X_i^T \left[\mathrm{E}\left(R XX^T\right)\right]^{-1} x}_{\psi_2}\\
		&\lesssim \sigma_{\epsilon} \sqrt{x^T\left[\mathrm{E}(R XX^T)\right]^{-1}x} = \sigma_{m,d}(x)
	\end{align*}
	by Proposition 2.7.1 and Lemma 2.7.7 in \cite{vershynin2018high}. Thus, 
	\begin{align*}
		\frac{1}{\sigma_{m,d}^3(x)}\sum_{i=1}^n \mathbb{E}\left|\frac{1}{\sqrt{n}}  R_i\epsilon_i X_i^T \left[\mathrm{E}\left(R XX^T\right)\right]^{-1} x\right|^3 &\lesssim \frac{1}{\sqrt{n}} = o(1), 
	\end{align*}
	and the Lyapunov condition holds. As a result,
	$$\frac{1}{\sqrt{n}} \sum_{i=1}^n R_i\epsilon_i X_i^T\left[\mathrm{E}\left(R XX^T\right)\right]^{-1} x \stackrel{d}{\to} \mathcal{N}\left(0,\, \sigma_m^2(x) \right).$$
	It remains to show that ``$\textbf{Bias I} + \textbf{Bias II} + \textbf{Bias III}$'' is of order $o_P(1)$, \emph{i.e.}, converging to 0 in probability.
	
	Let $\delta \in (0,1)$. Notice that our assumptions guarantee that with probability at least $1-\delta$, the strong duality holds as well as 
	\begin{align*}
		&\hat{\beta}-\beta_0 \in \mathcal{C}_1(S_{\beta}, 3), \quad \tilde{\ell}(x) - \ell_0(x) \in \mathcal{C}_2(S_{\ell}(x),1)\\ 
		&\text{ and } \quad \hat{\ell}(x) - \tilde{\ell}(x) \in \mathcal{C}_1(S_{\ell}(x),3) \bigcup B_d^{(1)}\left(0,\, \frac{13 r_{\ell}^2 \norm{x}_2^2}{2 \kappa_R^4} \cdot \frac{\nu}{\gamma}\right) \equiv \mathcal{C}_1\left(S_{\ell}(x),3\right) \bigcup B_d^{(1)}(0,r_1), 
	\end{align*}
	recalling our arguments in \autoref{thm:lasso_const}, \autoref{thm:dual_consist2}, Lemma~\ref{lem:dual_approx_cone}, and Lemma~\ref{lem:res_cross_poly}. Moreover, with probability at least $1-\delta$, we have that $r_1\equiv \frac{13 r_{\ell}^2 \norm{x}_2^2}{2 \kappa_R^4} \cdot \frac{\nu}{\gamma} < 1$ under the condition in \autoref{thm:dual_consist2}. In other words, 
	$$\left[\mathcal{C}_1\left(S_{\ell}(x), 3 \right) \bigcup B_d^{(1)}(0,r_1)\right] \bigcap \mathbb{S}^{d-1} = \mathcal{C}_1\left(S_{\ell}(x), 3 \right) \bigcap \mathbb{S}^{d-1}$$
	with probability at least $1-\delta$. 
	We will use these results repeatedly in the following proof.\\

	\noindent$\bullet$ \textbf{Bias I:} We compute that
	\begin{align}
		\label{bias_I}
		\begin{split}
			\left|\textbf{Bias I}\right| &=\left|\frac{1}{2\sqrt{n}} \sum_{i=1}^n R_i\epsilon_i X_i^T\left[\ell_0(x) - \hat{\ell}(x)\right]\right|\\
			&\leq \underbrace{\sqrt{n} \left|\frac{1}{2n} \sum_{i=1}^n R_i\epsilon_i X_i^T\left[\ell_0(x) - \tilde{\ell}(x)\right]\right|}_{\textbf{Term I}}\\
			&\quad +\sqrt{n} \underbrace{\sup_{u \in \left[\mathcal{C}_1(S_{\ell}(x),3) \bigcup B_d^{(1)}\left(0,\, r_1\right)\right] \bigcap \mathbb{S}^{d-1}} \left|\frac{1}{2n} \sum_{i=1}^n R_i\epsilon_i X_i^T u \right|\norm{ \hat{\Delta}(x)}_2}_{\textbf{Term II}},
		\end{split}
	\end{align}
	where $\hat{\Delta}(x) = \hat{\ell}(x) -\tilde{\ell}(x)$. 
	
	As for \textbf{Term I}, we note that under Assumptions~\ref{assump:sub_gaussian} and \ref{assump:subgau_error}, $R_i\epsilon_i X_i^T\left[\ell_0(x) - \tilde{\ell}(x)\right], i=1,...,n$ are i.i.d. centered sub-exponential random variables. In addition,
	\begin{align*}
		\max_{1\leq i \leq n} \norm{R_i\epsilon_i X_i^T\left[\ell_0(x) - \tilde{\ell}(x)\right]}_{\psi_1} &\leq \max_{1\leq i \leq n} \norm{R_i\epsilon_i}_{\psi_2} \norm{X_i^T\left[\ell_0(x) - \tilde{\ell}(x)\right]}_{\psi_2}\\
		&\lesssim \max_{1 \leq i \leq n} \sigma_{\epsilon} \left\{\mathrm{E}\left[\left(\ell_0(x) - \tilde{\ell}(x)\right)^T X_iX_i^T \left(\ell_0(x) - \tilde{\ell}(x)\right)\right]\right\}^{\frac{1}{2}}\\
		&\leq \max_{1 \leq i \leq n}\sigma_{\epsilon} \cdot \sup_{v \in \mathcal{C}_2(S_{\ell}(x),1)\cap B_d\left(0, 10r_{\ell}\norm{\ell_0(x)}_2\right)} \left[\mathrm{E}\left(X_i^T v\right)^2\right]^{\frac{1}{2}}\\
		&\lesssim \max_{1 \leq i \leq n} \sigma_{\epsilon}r_{\ell} \norm{\ell_0(x)}_2 \sup_{v \in \mathcal{C}_2(S_{\ell}(x),1)\cap B_d\left(0,1\right)} \left[\mathrm{E}\left(X_i^T v\right)^2\right]^{\frac{1}{2}}\\
		&\lesssim \sigma_{\epsilon} \sqrt{\phi_{s_{\ell}(x)}} \cdot r_{\ell} \norm{\ell_0(x)}_2,
	\end{align*}
	where we use Lemma~\ref{lem:Orlicz_prod} in the first inequality, Lemma~\ref{lem:dual_approx_cone} in the third inequality, and Lemma~\ref{lem:size_dom_cone} in the fifth inequality. By Bernstein's inequality, we bound the probability
	\begin{align*}
		&\mathrm{P}\left(\left|\frac{1}{2n} \sum_{i=1}^n R_i\epsilon_i X_i^T\left[\ell_0(x) - \tilde{\ell}(x)\right]\right| > t \right) \\
		& \leq 2\exp\left(-nA_0\min\left\{\frac{t^2}{\sigma_{\epsilon}^2 \phi_{s_{\ell}(x)} \cdot r_{\ell}^2 \norm{\ell_0(x)}_2^2}, \frac{t}{\sigma_{\epsilon} \sqrt{\phi_{s_{\ell}(x)}} \cdot r_{\ell} \norm{\ell_0(x)}_2} \right\} \right)
	\end{align*}
	for some absolute constant $A_0 >0$. It implies that
	\begin{align*}
		&\textbf{Term I} \lesssim \sigma_{\epsilon} \sqrt{\phi_{s_{\ell}(x)}} \cdot r_{\ell} \norm{\ell_0(x)}_2 \left[\sqrt{\log(1/\delta)} + \frac{\log(1/\delta)}{\sqrt{n}} \right].
	\end{align*}
	
	As for \textbf{Term II}, by Assumptions~\ref{assump:sub_gaussian} and \ref{assump:subgau_error}, $R_i\epsilon_i X_i^Tv, i=1,...,n$ are independent, centered, and sub-exponential random variables for all $v\in \left[\mathcal{C}_1(S_{\ell}(x),3) \bigcup B_d^{(1)}\left(0,\, r_1\right)\right] \bigcap \mathbb{S}^{d-1}$ with
	\begin{align*}
		\max_{1\leq i \leq n} \norm{R_i\epsilon_i X_i^Tv}_{\psi_1} &\leq \max_{1\leq i \leq n} \norm{R_i\epsilon_i}_{\psi_2} \norm{X_i^Tv}_{\psi_2}\\
		&\lesssim \sigma_{\epsilon} \sqrt{\phi_{s_{\ell}(x)}}.
	\end{align*}
	Thus, by Lemma~\ref{lem:gram_mat_conc}, union bound, and Bernstein's inequality, we obtain that with probability at least $1-\delta$,
	\begin{align*}
		&\sup_{v \in \left[\mathcal{C}_1(S_{\ell}(x),3) \bigcup B_d^{(1)}\left(0,\, r_1\right)\right] \bigcap \mathbb{S}^{d-1}} \left|\frac{1}{2n} \sum_{i=1}^n R_i\epsilon_i X_i^Tv \right| \\
		&\lesssim \sigma_{\epsilon} \sqrt{\phi_{s_{\ell}(x)}} \left[\sqrt{\frac{\log(1/\delta) + s_{\ell}(x) \log(ed/s_{\ell}(x))}{n}} + \frac{\log(1/\delta) + s_{\ell}(x) \log(ed/s_{\ell}(x))}{n} \right].
	\end{align*}
	
	Plugging the result back into \eqref{bias_I}, we conclude that with probability at least $1-\delta$,
	\begin{align*}
		&\left|\textbf{Bias I}\right| \\
		&\lesssim \sigma_{\epsilon} \sqrt{\phi_{s_{\ell}(x)}} \cdot r_{\ell} \norm{\ell_0(x)}_2 \left[\sqrt{\log(1/\delta)} + \frac{\log(1/\delta)}{\sqrt{n}} \right]\\
		&\quad + \sqrt{n} \norm{\hat{\ell}(x) -\tilde{\ell}(x)}_2 \sigma_{\epsilon} \sqrt{\phi_{s_{\ell}(x)}} \left[\sqrt{\frac{\log(1/\delta) + s_{\ell}(x) \log(ed/s_{\ell}(x))}{n}} \right. \\
		&\quad \quad\quad \quad  \left.+ \frac{\log(1/\delta) + s_{\ell}(x) \log(ed/s_{\ell}(x))}{n} \right]\\
		&\lesssim \sigma_{\epsilon} \sqrt{\phi_{s_{\ell}(x)}}\left[\frac{\norm{x}_2 r_{\ell}}{\kappa_R^2} \sqrt{\log(1/\delta)} + \norm{\hat{\ell}(x) - \tilde{\ell}(x)}_2 \sqrt{\log(1/\delta) + s_{\ell}(x) \log(d/s_{\ell}(x))}\right] 
	\end{align*}
	when $\log(1/\delta) =o(n)$ and $s_{\max}\left(\log n + \log d\right) = o(\sqrt{n})$.\\

	\noindent$\bullet$ \textbf{Bias II:} Notice that
	$$\left|\textbf{Bias II}\right| \leq \norm{\frac{1}{2n} \sum_{i=1}^n R_i X_i^T\ell_0(x) X_i +x}_{\infty} \sqrt{n}\norm{\hat{\beta} -\beta_0}_1.$$
	By \autoref{thm:lasso_const}, we know that with probability at least $1-\delta$,
	\begin{align*}
		\sqrt{n}\norm{\hat{\beta}-\beta_0}_1 \leq 4\sqrt{n s_{\beta}}\norm{\hat{\beta} - \beta_0}_2 & \lesssim \frac{\sigma_{\epsilon}\sqrt{\phi_1} s_{\beta}}{\kappa_R^2} \sqrt{\log\left(d/\delta\right)}.
	\end{align*} 
	Additionally, by Lemma~\ref{lem:gamma_rate}, we know that with probability at least $1-\delta$,
	\begin{align*}
		\norm{\frac{1}{2n} \sum_{i=1}^n R_i X_i^T\ell_0(x) X_i +x}_{\infty} &\lesssim \frac{\sqrt{\phi_1}}{\kappa_R} \norm{x}_2 \left[\sqrt{\frac{\log(d/\delta)}{n}} + \frac{\log(d/\delta)}{n} \right]
	\end{align*}
	Thus, we obtain that with probability at least $1-\delta$,
	\begin{align*}
		\left|\textbf{Bias II}\right| &\lesssim \frac{\sigma_{\epsilon}\phi_1 s_{\beta} \norm{x}_2}{\kappa_R^3} \cdot \frac{\log(d/\delta)}{\sqrt{n}},
	\end{align*}
	provided that $s_{\max}\left(\log n + \log d\right) = o(\sqrt{n})$.\\
	
	\noindent $\bullet$ \textbf{Bias III:} Notice that 
	\begin{align}
		\label{bias_III}
		&\textbf{Bias III} \nonumber\\ 
		&= \left[\frac{1}{2n} \sum_{i=1}^n R_i X_i^T\left[\hat{\ell}(x) - \ell_0(x)\right] X_i \right]^T \sqrt{n}\left(\hat{\beta} -\beta_0 \right) \nonumber\\
		&= \left[\hat{\ell}(x) - \tilde{\ell}(x)\right]^T \left[\frac{1}{2n}\sum_{i=1}^n \left(R_iX_iX_i^T - \mathrm{E}\left[R_iX_iX_i^T\right]\right)\right] \sqrt{n}\left(\hat{\beta} -\beta_0 \right) \nonumber\\
		&\quad + \left[\tilde{\ell}(x) - \ell_0(x)\right]^T \left[\frac{1}{2n}\sum_{i=1}^n \left(R_iX_iX_i^T - \mathrm{E}\left[R_iX_iX_i^T\right]\right)\right] \sqrt{n}\left(\hat{\beta} -\beta_0 \right) \nonumber\\
		&\quad + \frac{\sqrt{n}}{2} \left[\hat{\ell}(x) - \tilde{\ell}(x)\right]^T \mathrm{E}\left(R_iX_iX_i^T\right) \left(\hat{\beta} -\beta_0 \right) + \frac{\sqrt{n}}{2} \left[\tilde{\ell}(x) - \ell_0(x)\right]^T \mathrm{E}\left(R_iX_iX_i^T\right) \left(\hat{\beta} -\beta_0 \right) \nonumber\\
		&\left|\textbf{Bias III}\right| \leq \norm{\hat{\ell}(x)-\tilde{\ell}(x)}_2 \sqrt{n}\norm{\hat{\beta}-\beta_0}_2 \underbrace{\sup_{\substack{u\in \mathcal{C}_1(S_{\ell}(x),3) \cap \mathbb{S}^{d-1}\\v\in \mathcal{C}_1(S_{\beta},3) \cap \mathbb{S}^{d-1}}} \left|\frac{1}{2n}\sum_{i=1}^n \left\{R_i(X_i^Tu)(X_i^Tv) -\mathrm{E}\left[R_i(X_i^Tu)(X_i^Tv)\right]\right\} \right|}_{\textbf{Term III}} \nonumber\\
		&\quad + r_{\ell}\norm{\ell_0(x)}_2 \sqrt{n}\norm{\hat{\beta}-\beta_0}_2 \underbrace{\sup_{\substack{u\in \mathcal{C}_2(S_{\ell}(x),1) \cap \mathbb{S}^{d-1}\\v\in \mathcal{C}_1(S_{\beta},3) \cap \mathbb{S}^{d-1}}} \left|\frac{1}{2n}\sum_{i=1}^n \left\{R_i(X_i^Tu)(X_i^Tv) -\mathrm{E}\left[R_i(X_i^Tu)(X_i^Tv)\right]\right\} \right|}_{\textbf{Term IV}} \nonumber\\
		&\quad + \frac{\sqrt{n}}{2} \norm{\hat{\ell}(x)-\tilde{\ell}(x)}_2 \norm{\hat{\beta}-\beta_0}_2 \underbrace{\sup_{\substack{u\in \mathcal{C}_1(S_{\ell}(x),3) \cap \mathbb{S}^{d-1}\\v\in \mathcal{C}_1(S_{\beta},3) \cap \mathbb{S}^{d-1}}} \left|\mathrm{E}\left[R_i(X_i^Tu) (X_i^Tv)\right]\right|}_{\textbf{Term V}}  \nonumber\\
		\begin{split}
			&\quad + \frac{\sqrt{n}}{2} r_{\ell}\norm{\ell_0(x)}_2 \norm{\hat{\beta}-\beta_0}_2 \underbrace{\sup_{\substack{u\in \mathcal{C}_2(S_{\ell}(x),1) \cap \mathbb{S}^{d-1}\\v\in \mathcal{C}_1(S_{\beta},3) \cap \mathbb{S}^{d-1}}} \left|\mathrm{E}\left[R_i(X_i^Tu) (X_i^Tv)\right]\right|}_{\textbf{Term VI}}.
		\end{split}
	\end{align}
	By Lemmas~\ref{lem:gram_mat_conc} and \ref{lem:cone_ball_bnd}, we know that with probability at least $1-\delta$,
	\begin{align*}
		\textbf{Term III},\, \textbf{Term IV} \lesssim \sqrt{\phi_{s_{\ell}(x)} \phi_{s_{\beta}}} \sqrt{\frac{\log(1/\delta) + s_{\beta}\log(ed/s_{\beta}) + s_{\ell}(x) \log(ed/s_{\ell}(x))}{n}}
	\end{align*}
	when $s_{\max}\left(\log n + \log d\right) = o(\sqrt{n})$ and $\log(1/\delta) = o(n)$. By Lemma~\ref{lem:size_dom_cone}, we can find discrete sets $M_{d,\ell}, M_{d,\beta} \subset B_d(0,1)$ with cardinalities $|M_{d,\ell}|\leq \frac{3}{2}\left(\frac{5ed}{s_{\ell}(x)}\right)^{s_{\ell}(x)}$ and $|M_{d,\beta}|\leq \frac{3}{2}\left(\frac{5ed}{s_{\beta}}\right)^{s_{\beta}}$ such that 
	\begin{align*}
		&\mathcal{C}_1(S_{\ell}(x),3) \cap \mathbb{S}^{d-1} \subset 10 \cdot \mathrm{conv}(M_{d,\ell}), \quad \mathcal{C}_2(S_{\ell}(x),1) \cap \mathbb{S}^{d-1} \subset 6 \cdot \mathrm{conv}(M_{d,\ell}),\\
		& \text{ and } \quad  \mathcal{C}_1(S_{\beta},3) \cap \mathbb{S}^{d-1} \subset 10\cdot \mathrm{conv}(M_{d,\beta})
	\end{align*}
	with $\norm{u}_0\leq s_{\ell}(x), \norm{v}_0\leq s_{\beta}$ for all $u\in M_{d,\ell}, v\in M_{d,\beta}$. Thus, 
	\begin{align*}
		\textbf{Term V}, \textbf{Term VI} \lesssim \sqrt{\phi_{s_{\ell}(x)} \phi_{s_{\beta}}}.
	\end{align*}
	Again, we know from \autoref{thm:lasso_const} that with probability at least $1-\delta$
	\begin{align*}
		\sqrt{n}\norm{\hat{\beta}-\beta_0}_2 \lesssim \frac{\sigma_{\epsilon}}{\kappa_R^2} \sqrt{\phi_1s_{\beta}\log(d/\delta)}.
	\end{align*}
	Thus, from \eqref{bias_III}, we conclude that with probability at least $1-\delta$,
	\begin{align*}
		\left|\textbf{Bias III}\right| 
		&\lesssim \frac{\sigma_{\epsilon}\sqrt{\phi_1 \phi_{s_{\beta}} \phi_{s_{\ell}(x)} s_{\beta}\log(d/\delta)}}{\kappa_R^2} \left[\norm{\hat{\ell}(x) -\tilde{\ell}(x)}_2 + \frac{\norm{x}_2 r_{\ell}}{\kappa_R^2} \right] \\
		&\quad \times \left[1 + \sqrt{\frac{\log(1/\delta) + s_{\beta}\log(ed/s_{\beta}) + s_{\ell}(x) \log(ed/s_{\ell}(x))}{n}} \right]\\
		&\lesssim \frac{\sigma_{\epsilon}\sqrt{\phi_1 \phi_{s_{\beta}} \phi_{s_{\ell}(x)} s_{\beta}\log(d/\delta)}}{\kappa_R^2} \left[\norm{\hat{\ell}(x) -\tilde{\ell}(x)}_2 + \frac{\norm{x}_2 r_{\ell}}{\kappa_R^2} \right]
	\end{align*}
	when $s_{\max}\left(\log n + \log d\right) = o(\sqrt{n})$ and $\log(1/\delta) = o(n)$.\\
	
	In summary, we can bound the bias terms with probability at least $1-\delta$ as follows:
	\begin{align*}
		&\left|\textbf{Bias I}\right| + \left|\textbf{Bias II}\right| + \left|\textbf{Bias III}\right| \\
		&\lesssim \sigma_{\epsilon} \sqrt{\phi_{s_{\ell}(x)}}\left[\frac{\norm{x}_2 r_{\ell}}{\kappa_R^2} \sqrt{\log(1/\delta)} + \norm{\hat{\ell}(x) - \tilde{\ell}(x)}_2 \sqrt{\log(1/\delta) + s_{\ell}(x) \log(d/s_{\ell}(x))}\right] \\
		&\quad + \frac{\sigma_{\epsilon}\phi_1 s_{\beta} \norm{x}_2}{\kappa_R^3} \cdot \frac{\log(d/\delta)}{\sqrt{n}} + \frac{\sigma_{\epsilon}\sqrt{\phi_1 \phi_{s_{\beta}} \phi_{s_{\ell}(x)} s_{\beta}\log(d/\delta)}}{\kappa_R^2} \left[\norm{\hat{\ell}(x) -\tilde{\ell}(x)}_2 + \frac{\norm{x}_2 r_{\ell}}{\kappa_R^2} \right].
	\end{align*}
	Recalling the result of \autoref{thm:dual_consist2} that
	\begin{align*}
		&\norm{\hat{\ell}(x) - \tilde{\ell}(x)}_2 \lesssim \frac{\norm{x}_2}{\kappa_R^3}\left[ \sqrt{\frac{\phi_1 s_{\ell}(x)\log(d/\delta)}{n}} + \frac{\phi_{s_{\ell}(x)}}{\kappa_R} \cdot r_{\ell} + \frac{\sqrt{\phi_1 \phi_{s_{\ell}(x)} s_{\ell}(x)}}{\kappa_R} \cdot  r_{\pi} \right]
	\end{align*}
	and setting $\delta=\frac{1}{n}$, we obtain that
	\begin{align*}
		&\left|\textbf{Bias I}\right| + \left|\textbf{Bias II}\right| + \left|\textbf{Bias III}\right| \\
		&=O_P\left(\frac{\sigma_{\epsilon}\phi_1 \norm{x}_2 s_{\beta}}{\kappa_R^3} \cdot \frac{\log(nd)}{\sqrt{n}} \right) \\
		&\quad + O_P\left(\frac{\sigma_{\epsilon}\norm{x}_2\sqrt{\phi_{s_{\ell}(x)} \log n}}{\kappa_R^2} \cdot r_{\ell} \right) + O_P\left(\frac{\sigma_{\epsilon}\norm{x}_2\sqrt{\phi_1 \phi_{s_{\beta}} \phi_{s_{\ell}(x)} s_{\beta} \log(nd)}}{\kappa_R^4} \cdot r_{\ell}\right) \\
		&\quad + O_P\left(\sigma_{\epsilon}\sqrt{\phi_{s_{\ell}(x)}\left[\log n + s_{\ell}(x) \log(d/s_{\ell}(x)) \right]} + \frac{\sigma_{\epsilon} \sqrt{\phi_1 \phi_{s_{\beta}} \phi_{s_{\ell}(x)} s_{\beta} \log(nd)}}{\kappa_R^2} \right) \norm{\hat{\ell}(x) - \tilde{\ell}(x)}_2\\
		&= O_P\left(\frac{\sigma_{\epsilon}\phi_1 \norm{x}_2 s_{\max}}{\kappa_R^3} \cdot \frac{\log(nd)}{\sqrt{n}} \right) + O_P\left(\frac{\sigma_{\epsilon}\phi_1 \phi_{s_{\max}}(1+\kappa_R^2)\norm{x}_2}{\kappa_R^4} \cdot \frac{s_{\max}\log(nd)}{\sqrt{n}}\right)\\
		&\quad + O_P\left(\frac{\sigma_{\epsilon}\norm{x}_2\sqrt{\phi_{s_{\max}} \log n}}{\kappa_R^2} \cdot r_{\ell} \right) + O_P\left(\frac{\sigma_{\epsilon}\norm{x}_2\left(1+\kappa_R^2 \right)\sqrt{\phi_1 \phi_{s_{\max}}^2 s_{\max} \log(nd)}}{\kappa_R^6} \cdot r_{\ell}\right)\\
		&\quad + O_P\left(\frac{\sigma_{\epsilon}\phi_1 \phi_{s_{\max}}^{\frac{3}{2}} \left(1+\kappa_R^2 \right) \norm{x}_2 s_{\max} \sqrt{\log(nd)}}{\kappa_R^6} \cdot r_{\pi}\right)\\
		&= O_P\left(\frac{\sigma_{\epsilon}\phi_1 \phi_{s_{\max}}(1 +\kappa_R^2)\norm{x}_2}{\kappa_R^4} \cdot \frac{s_{\max}\log(nd)}{\sqrt{n}}\right) \\
		&\quad + O_P\left(\frac{\sigma_{\epsilon}\sqrt{\phi_1} \phi_{s_{\max}}\left(1+ \kappa_R^4 \right) \norm{x}_2\sqrt{s_{\max} \log(nd)}}{\kappa_R^6} \cdot r_{\ell}\right)\\
		&\quad + O_P\left(\frac{\sigma_{\epsilon}\phi_1 \phi_{s_{\max}}^{\frac{3}{2}} \left(1+\kappa_R^2 \right) \norm{x}_2 s_{\max} \sqrt{\log(nd)}}{\kappa_R^6} \cdot r_{\pi}\right)\\
		&=o_P(1)
	\end{align*}
	as long as $\sigma_{\epsilon}\phi_1\phi_{s_{\max}}^{\frac{3}{2}}\norm{x}_2=O(1)$, $\frac{(1+\kappa_R^2)s_{\max}\log(nd)}{\kappa_R^4} = o\left(\sqrt{n}\right)$, and $$\frac{(1+\kappa_R^4)\sqrt{\log(nd)}}{\kappa_R^6}\left(\sqrt{s_{\max}}\cdot r_{\ell} +s_{\max}r_{\pi}\right) =o(1).$$ 
	The result follows.
\end{proof}

\subsection{Proof of Proposition~\ref{prop:var_const_est}}
\label{subapp:proof_var_const_est}

\begin{customprop}{3.7}[Consistent estimate of the asymptotic variance]
	Let $x\in \mathbb{R}^d$ be a fixed query point. Suppose that Assumptions~\ref{assump:basic}, \ref{assump:sub_gaussian}, \ref{assump:gram_lower_bnd}, \ref{assump:prop_consistent_rate}, and \ref{assump:sparse_dual} hold and $\gamma, r_{\ell} >0$ are specified as in \autoref{thm:dual_consist2}. Furthermore, we assume that $\norm{x}_2^2 \phi_{s_{\ell}(x)}^2 = O(1)$, $\frac{(1+\kappa_R^3)}{\kappa_R^5} \sqrt{\frac{s_{\ell}(x)\log(nd)}{n}} = o(1)$, and $\frac{(1+\kappa_R^4)}{\kappa_R^6} \left[r_{\ell} + r_{\pi}\sqrt{s_{\ell}(x)}\right]=o(1)$. Then,
	\begin{align*}
		&\left|\sum_{i=1}^n \hat{\pi}_i \hat{w}_i(x)^2 - x^T\left[\mathrm{E}\left(RXX^T\right)\right]^{-1}x \right|=o_P(1).
	\end{align*}
\end{customprop}

\begin{proof}[Proof of Proposition~\ref{prop:var_const_est}]
	By Taylor's expansion and the relations between the primal solution $\hat{\bm{w}}(x)=\left(\hat{w}_1(x),...,\hat{w}_n(x)\right)^T \in \mathbb{R}^n$ and the dual solution $\hat{\ell}(x) \in \mathbb{R}^d$ in \eqref{primal_dual_sol2}, we know that
	\begin{align*}
		&\left|\sum_{i=1}^n \hat{\pi}_i \hat{w}_i(x)^2 - x^T\left[\mathrm{E}\left(RXX^T\right)\right]^{-1}x \right|\\ 
		&= \left|\frac{1}{4n}\sum_{i=1}^n \hat{\pi}_i \left[X_i^T \hat{\ell}(x)\right]^2 - x^T\left[\mathrm{E}\left(RXX^T\right)\right]^{-1}x \right|\\
		&\leq \left|\frac{1}{4n} \sum_{i=1}^n \hat{\pi}_i \left[X_i^T \ell_0(x)\right]^2 - x^T\left[\mathrm{E}\left(RXX^T\right)\right]^{-1}x \right| + \left|\left[\frac{1}{2n}\sum_{i=1}^n \hat{\pi}_i X_iX_i^T \ell_0(x)\right]^T \left[\hat{\ell}(x)-\ell_0(x) \right] \right| \\
		&\quad + \left|\frac{1}{4n}\sum_{i=1}^n \hat{\pi}_i \left[X_i^T\left(\hat{\ell}(x) -\ell_0(x)\right)\right]^2\right|\\
		&\leq \underbrace{\left|\frac{1}{4n} \sum_{i=1}^n \pi(X_i) \left[X_i^T \ell_0(x)\right]^2 - x^T\left[\mathrm{E}\left(RXX^T\right)\right]^{-1}x \right|}_{\textbf{Term I}} + \underbrace{\left|\frac{1}{4n}\sum_{i=1}^n \left(\hat{\pi}_i -\pi_i\right)\left[X_i^T\ell_0(x)\right]^2\right|}_{\textbf{Term II}}\\
		&\quad + \underbrace{\left|\left[\frac{1}{2n}\sum_{i=1}^n \pi(X_i) X_iX_i^T \ell_0(x)\right]^T \left[\hat{\ell}(x)-\ell_0(x) \right] \right|}_{\textbf{Term III}} + \underbrace{\left|\left[\frac{1}{2n}\sum_{i=1}^n \left(\hat{\pi}_i -\pi_i\right) X_iX_i^T \ell_0(x)\right]^T \left[\hat{\ell}(x)-\ell_0(x) \right] \right|}_{\textbf{Term IV}} \\
		&\quad + \underbrace{\left|\frac{1}{4n}\sum_{i=1}^n \pi(X_i) \left[X_i^T\left(\hat{\ell}(x) -\tilde{\ell}(x)\right)\right]^2\right|}_{\textbf{Term V}} + \underbrace{\left|\frac{1}{4n}\sum_{i=1}^n \left(\hat{\pi}_i -\pi_i\right) \left[X_i^T\left(\hat{\ell}(x) -\tilde{\ell}(x)\right)\right]^2\right|}_{\textbf{Term VI}}\\
		&\quad + \underbrace{\left|\frac{1}{2n}\sum_{i=1}^n \pi(X_i)\left[X_i^T\left(\hat{\ell}(x) -\tilde{\ell}(x)\right)\right] \left[X_i^T\left(\tilde{\ell}(x)-\ell_0(x)\right)\right]\right|}_{\textbf{Term VII}} \\
		&\quad + \underbrace{\left|\frac{1}{2n}\sum_{i=1}^n \left(\hat{\pi}_i -\pi_i\right)\left[X_i^T\left(\hat{\ell}(x) -\tilde{\ell}(x)\right)\right] \left[X_i^T\left(\tilde{\ell}(x)-\ell_0(x)\right)\right]\right|}_{\textbf{Term VIII}}\\
		&\quad + \underbrace{\left|\frac{1}{4n}\sum_{i=1}^n \pi(X_i) \left[X_i^T\left(\tilde{\ell}(x)-\ell_0(x)\right)\right]^2\right|}_{\textbf{Term IX}} + \underbrace{\left|\frac{1}{4n}\sum_{i=1}^n \left(\hat{\pi}_i -\pi_i\right) \left[X_i^T\left(\tilde{\ell}(x)-\ell_0(x)\right)\right]^2\right|}_{\textbf{Term X}}.         
	\end{align*}
	We know from Assumption~\ref{assump:sparse_dual}, Lemma~\ref{lem:dual_approx_cone} and Lemma~\ref{lem:res_cross_poly} with $a_0=5$ that 
	\begin{equation}
		\label{geom_const1}
		\norm{\tilde{\ell}(x) -\ell_0(x)}_2 \leq r_{\ell} \norm{\ell_0(x)}_2, \quad \ell_0(x) \in \mathcal{C}_2\left(S_{\ell}(x),\vartheta_{\ell}\right)
	\end{equation}
	with $\vartheta_{\ell}=\frac{r_{\ell}}{\sqrt{1-r_{\ell}^2}}$, and 
	\begin{equation}
		\label{geom_const2}
		\hat{\ell}(x)-\tilde{\ell}(x) \in \mathcal{C}_1\left(S_{\ell}(x), 3 \right) \bigcup B_d^{(1)}\left(0,\, \frac{13 r_{\ell}^2 \norm{x}_2^2}{2 \kappa_R^4} \cdot \frac{\nu}{\gamma}\right) \equiv \mathcal{C}_1\left(S_{\ell}(x), 3 \right) \bigcup B_d^{(1)}(0,r_1),
	\end{equation}
	where $B_d^{(1)}(0,r) = \left\{v\in \mathbb{R}^d: \norm{v}_1 \leq r\right\}$. Moreover, with probability at least $1-\delta$, we have that $r_1\equiv \frac{13 r_{\ell}^2 \norm{x}_2^2}{2 \kappa_R^4} \cdot \frac{\nu}{\gamma} < 1$ under the condition in \autoref{thm:dual_consist2}. In other words, 
	$$\left[\mathcal{C}_1\left(S_{\ell}(x), 3 \right) \bigcup B_d^{(1)}(0,r_1)\right] \bigcap \mathbb{S}^{d-1} = \mathcal{C}_1\left(S_{\ell}(x), 3 \right) \bigcap \mathbb{S}^{d-1}$$
	with probability at least $1-\delta$. Now, we bound these ten terms separately.\\
	
	\noindent$\bullet$ \textbf{Term I:} Note that $\ell_0(x)=-2\left[\mathrm{E}\left(R XX^T\right)\right]^{-1}x$ and 
	$$\pi(X_i) \left[X_i^T\ell_0(x)\right]^2 - 4x^T\left[\mathrm{E}\left(RXX^T\right)\right]^{-1}x, i=1,...,n$$ 
	are independent, centered, and sub-exponential random variables with
	\begin{align*}
		&\max_{1\leq i \leq n} \norm{\pi(X_i)\left[X_i^T\ell_0(x)\right]^2 - 4x^T\left[\mathrm{E}\left(RXX^T\right)\right]^{-1}x}_{\psi_1}\\
		&\lesssim \max_{1\leq i \leq n}\norm{\sqrt{\pi(X_i)} \cdot X_i^T \left[\mathrm{E}\left(\pi(X)XX^T\right)\right]^{-1}x}_{\psi_2}^2\\
		&\lesssim \frac{\norm{x}_2^2}{\kappa_R^2}.
	\end{align*}
	By Bernstein's inequality (Corollary 2.8.3 in \citealt{vershynin2018high}), we know that with probability at least $1-\delta$,
	\begin{align*}
		\textbf{Term I} &=\left|\frac{1}{4n} \sum_{i=1}^n \pi(X_i) \left[X_i^T \ell_0(x)\right]^2 -x^T\left\{\mathrm{E}\left[\pi(X) XX^T\right]\right\}^{-1}x \right| \\
		&\lesssim \frac{\norm{x}_2^2}{\kappa_R^2} \left[\sqrt{\frac{\log(1/\delta)}{n}} + \frac{\log(1/\delta)}{n} \right]\\
		&\lesssim \frac{\norm{x}_2^2}{\kappa_R^2} \sqrt{\frac{\log(1/\delta)}{n}}
	\end{align*}
	when $\log(1/\delta)=o(n)$.\\
	
	\noindent$\bullet$ \textbf{Term II:} By \eqref{geom_const1} and Lemma~\ref{lem:gram_mat_conc} that with probability at least $1-\delta$,
	\begin{align*}
		&\textbf{Term II} = \left|\frac{1}{4n}\sum_{i=1}^n \left(\hat{\pi}_i -\pi_i\right)\left[X_i^T\ell_0(x)\right]^2\right| \\
		&\lesssim r_{\pi}(n,\delta) \norm{\ell_0(x)}_2^2 \sup_{v\in \mathcal{C}_2(S_{\ell}(x),\vartheta_{\ell})\cap \mathbb{S}^{d-1}} \left|\frac{1}{4n}\sum_{i=1}^n \left(X_i^Tv\right)^2\right|\\
		&\lesssim \frac{r_{\pi}(n,\delta)\norm{x}_2^2}{\kappa_R^4} \left\{\sup_{v \in \mathcal{C}_2(S_{\ell}(x),\vartheta_{\ell})\cap \mathbb{S}^{d-1}}\left|\frac{1}{4n}\sum_{i=1}^n \left\{\left(X_i^Tv\right)^2 - \mathrm{E}\left[(X_i^Tv)^2\right]\right\}\right| \right.\\
		&\quad\quad \left.+ \sup_{v \in \mathcal{C}_2(S_{\ell}(x),\vartheta_{\ell})\cap \mathbb{S}^{d-1}} \mathrm{E}\left[(X_1^Tv)^2\right] \right\}\\
		&\lesssim \frac{r_{\pi}(n,\delta)\norm{x}_2^2}{\kappa_R^4} \left[\phi_{s_{\ell}(x)} \sqrt{\frac{\log(1/\delta) + s_{\ell}(x)\log(ed/s_{\ell}(x))}{n}} + \phi_{s_{\ell}(x)} \right]\\
		&\lesssim \frac{\phi_{s_{\ell}(x)}\norm{x}_2^2}{\kappa_R^4} \cdot r_{\pi}(n,\delta),
	\end{align*}
	when $\log d=o(n)$ and $\log(1/\delta)=o(n)$.\\
	
	\noindent$\bullet$ \textbf{Term III:} We know from the geometric constraints \eqref{geom_const1} and \eqref{geom_const2} as well as Lemmas~\ref{lem:gram_mat_conc}, \ref{lem:dual_approx_cone}, and \ref{lem:cone_ball_bnd} that with probability at least $1-\delta$,
	\begin{align*}
		&\textbf{Term III} \\
		&= \left|\left[\frac{1}{2n}\sum_{i=1}^n \pi(X_i) X_iX_i^T \ell_0(x)\right]^T \left[\hat{\ell}(x)-\ell_0(x) \right] \right|\\
		&\leq \left|\left[\frac{1}{2n}\sum_{i=1}^n \pi(X_i) X_iX_i^T \ell_0(x)\right]^T \left[\hat{\ell}(x)-\tilde{\ell}(x) \right] \right| + \left|\left[\frac{1}{2n}\sum_{i=1}^n \pi(X_i) X_iX_i^T \ell_0(x)\right]^T \left[\tilde{\ell}(x)-\ell_0(x) \right] \right|\\
		&\leq \left|\left[\frac{1}{2n}\sum_{i=1}^n \left\{\pi(X_i) X_iX_i^T \ell_0(x) -\mathrm{E}\left[\pi(X_i)X_iX_i^T \ell_0(x)\right] \right\}\right]^T \left[\hat{\ell}(x)-\tilde{\ell}(x) \right] \right| \\
		&\quad + \left|\left[\hat{\ell}(x) -\tilde{\ell}(x)\right]^T \mathrm{E}\left[\pi(X_i)X_iX_i^T\ell_0(x)\right] \right|\\
		&\quad + \left|\left[\frac{1}{2n}\sum_{i=1}^n \left\{\pi(X_i) X_iX_i^T \ell_0(x) -\mathrm{E}\left[\pi(X_i)X_iX_i^T \ell_0(x)\right] \right\}\right]^T \left[\tilde{\ell}(x) -\ell_0(x)\right] \right| \\
		&\quad + \left|\left[\tilde{\ell}(x) -\ell_0(x)\right]^T \mathrm{E}\left[\pi(X_i)X_iX_i^T\ell_0(x)\right] \right|\\
		&\leq \norm{\hat{\ell}(x)-\tilde{\ell}(x)}_2 \norm{\ell_0(x)}_2\\
		&\quad \times  \sup_{\substack{u\in \mathcal{C}_1(S_{\ell}(x),3) \cap \mathbb{S}^{d-1}\\ v\in \mathcal{C}_2(S_{\ell}(x),\vartheta_{\ell}) \cap \mathbb{S}^{d-1}}} \Bigg\{\left|\frac{1}{2n}\sum_{i=1}^n \left\{\pi(X_i) (X_i^Tu)(X_i^Tv) -\mathrm{E}\left[\pi(X_i) (X_i^Tu) (X_i^Tv)\right] \right\} \right| \\
		&\hspace{45mm} + \left|\mathrm{E}\left[\pi(X_i) (X_i^Tu) (X_i^Tv)\right]\right| \Bigg\}\\
		&\quad + \norm{\tilde{\ell}(x) - \ell_0(x)}_2 \norm{\ell_0(x)}_2\\
		&\quad\quad \times \sup_{\substack{u\in \mathcal{C}_1(S_{\ell}(x),3) \cap \mathbb{S}^{d-1}\\ v\in \mathcal{C}_2(S_{\ell}(x),1) \cap \mathbb{S}^{d-1}}} \Bigg\{\left|\frac{1}{2n}\sum_{i=1}^n \left\{\pi(X_i) (X_i^Tu)(X_i^Tv) -\mathrm{E}\left[\pi(X_i) (X_i^Tu) (X_i^Tv)\right] \right\} \right| \\
		&\hspace{45mm} + \left|\mathrm{E}\left[\pi(X_i) (X_i^Tu) (X_i^Tv)\right]\right| \Bigg\}\\
		&\lesssim \left[\norm{\ell_0(x)}_2 \norm{\hat{\ell}(x) -\tilde{\ell}(x)}_2 + r_{\ell} \norm{\ell_0(x)}_2^2 \right]\left[\phi_{s_{\ell}(x)}\sqrt{\frac{\log(1/\delta)+s_{\ell}(x)\log(ed/s_{\ell}(x))}{n}} + \phi_{s_{\ell}(x)} \right]\\
		&\lesssim \frac{\phi_{s_{\ell}(x)} \norm{x}_2}{\kappa_R^2} \left[\norm{\hat{\ell}(x) -\tilde{\ell}(x)}_2 + \frac{\norm{x}_2 r_{\ell}}{\kappa_R^2} \right]
	\end{align*}
	when $\log d=o(n)$ and $\log(1/\delta)=o(n)$. \\

	\noindent$\bullet$ \textbf{Term IV:} By Assumption~\ref{assump:prop_consistent_rate}, Lemmas~\ref{lem:size_dom_cone}, \ref{lem:gram_mat_conc}, and \ref{lem:cone_ball_bnd}, we obtain that with probability at least $1-\delta$,
	\begin{align*}
		&\textbf{Term IV} = \left|\left[\frac{1}{2n}\sum_{i=1}^n \left(\hat{\pi}_i -\pi_i\right) X_iX_i^T \ell_0(x)\right]^T \left[\hat{\ell}(x)-\ell_0(x) \right] \right|\\
		&\lesssim r_{\pi}(n,\delta)\Biggl[\norm{\hat{\ell}(x) -\tilde{\ell}(x)}_2 \sup_{v \in \mathcal{C}_1(S_{\ell}(x),3)\cap \mathbb{S}^{d-1}} \frac{1}{2n}\sum_{i=1}^n \left|X_i^T\ell_0(x) \cdot X_i^Tv\right| \\
		&\quad + \frac{1}{2n}\sum_{i=1}^n \left|X_i^T\ell_0(x) \cdot X_i^T\left[\tilde{\ell}(x) -\ell_0(x)\right]\right| \Biggr]\\
		&\lesssim \frac{r_{\pi}(n,\delta) \cdot \norm{x}_2}{\kappa_R^2} \Biggl[\norm{\hat{\ell}(x) -\tilde{\ell}(x)}_2 \sup_{\substack{v \in \mathcal{C}_1(S_{\ell}(x),3) \cap \mathbb{S}^{d-1}\\ u\in \mathcal{C}_2(S_{\ell}(x),\vartheta_{\ell})\cap \mathbb{S}^{d-1}}} \frac{1}{2n}\sum_{i=1}^n \left|X_i^Tu \cdot X_i^Tv\right| \\
		&\quad\quad + \sup_{u \in \mathcal{C}_2(S_{\ell}(x),\vartheta_{\ell})\cap \mathbb{S}^{d-1}} \frac{1}{2n}\sum_{i=1}^n \left|X_i^Tu \cdot X_i^T\left[\tilde{\ell}(x) -\ell_0(x)\right]\right| \Biggr]\\
		&\stackrel{\text{(i)}}{\lesssim} \frac{r_{\pi}(n,\delta) \cdot \norm{x}_2}{\kappa_R^2} \left[\norm{\hat{\ell}(x) -\tilde{\ell}(x)}_2 + \norm{\tilde{\ell}(x) -\ell_0(x)}_2 \right]\\
		&\quad \times  \left[\sup_{v \in \mathcal{C}_1(S_{\ell}(x),3) \cap \mathbb{S}^{d-1}} \frac{1}{n}\sum_{i=1}^n \left(X_i^Tv \right)^2 + \sup_{u\in \mathcal{C}_2(S_{\ell}(x),\vartheta_{\ell})\cap \mathbb{S}^{d-1}} \frac{1}{n}\sum_{i=1}^n \left(X_i^Tu\right)^2 \right] \\
		&\lesssim \frac{r_{\pi}(n,\delta) \cdot \norm{x}_2}{\kappa_R^2} \left[\norm{\hat{\ell}(x) -\tilde{\ell}(x)}_2 + \norm{\tilde{\ell}(x) -\ell_0(x)}_2 \right] \\
		&\quad \times \Biggl\{\sup_{v \in \mathcal{C}_1(S_{\ell}(x),3)\cap \mathbb{S}^{d-1}} \frac{1}{n}\sum_{i=1}^n \left\{\left(X_i^Tv \right)^2 - \mathrm{E}\left[\left(X_i^Tv \right)^2\right] \right\} + \sup_{v \in \mathcal{C}_1(S_{\ell}(x),3) \cap \mathbb{S}^{d-1}} \mathrm{E}\left[\left(X_1^Tv \right)^2\right] \\
		&\quad \quad + \sup_{u \in \mathcal{C}_2(S_{\ell}(x),\vartheta_{\ell})\cap \mathbb{S}^{d-1}} \frac{1}{n}\sum_{i=1}^n \left\{\left(X_i^Tu \right)^2 - \mathrm{E}\left[\left(X_i^Tu \right)^2\right] \right\} + \sup_{u \in \mathcal{C}_2(S_{\ell}(x),\vartheta_{\ell})\cap \mathbb{S}^{d-1}} \mathrm{E}\left[\left(X_1^Tu \right)^2\right] \Biggr\} \\
		&\lesssim \frac{r_{\pi}(n,\delta) \cdot \norm{x}_2}{\kappa_R^2} \left(\norm{\hat{\ell}(x) -\tilde{\ell}(x)}_2 + \frac{r_{\ell} \norm{x}_2}{\kappa_R^2} \right) \phi_{s_{\ell}(x)} \left[\sqrt{\frac{s_{\ell}(x)\log(ed/s_{\ell}(x))}{n}} + 1\right] \\
		&\lesssim \frac{r_{\pi}(n,\delta) \cdot \norm{x}_2}{\kappa_R^2} \left(\norm{\hat{\ell}(x) -\tilde{\ell}(x)}_2 + \frac{r_{\ell} \norm{x}_2}{\kappa_R^2} \right) \phi_{s_{\ell}(x)}
	\end{align*}
	when $\log d=o(n)$ and $\log(1/\delta)=o(n)$, where we use the fact that $|a|\cdot|b|\leq \frac{a^2+b^2}{2}$ in the (asymptotic) inequality (i). \\

	\noindent$\bullet$ \textbf{Term V and VI:} Under \eqref{geom_const2}, Assumption~\ref{assump:prop_consistent_rate}, Lemmas~\ref{lem:size_dom_cone}, \ref{lem:gram_mat_conc}, and \ref{lem:cone_ball_bnd}, we have that with probability at least $1-\delta$,
	\begin{align*}
		&\textbf{Term V} + \textbf{Term VI} \\
		&= \left|\frac{1}{4n}\sum_{i=1}^n \pi(X_i) \left[X_i^T\left(\hat{\ell}(x) -\tilde{\ell}(x)\right)\right]^2\right| + \left|\frac{1}{4n}\sum_{i=1}^n \left(\hat{\pi}_i -\pi_i\right) \left[X_i^T\left(\hat{\ell}(x) -\tilde{\ell}(x)\right)\right]^2\right|\\
		&\leq \left(1+r_{\pi}(n,\delta)\right) \left|\frac{1}{4n}\sum_{i=1}^n \left[X_i^T\left(\hat{\ell}(x) -\tilde{\ell}(x)\right)\right]^2\right|\\
		&\leq \left(1+r_{\pi}(n,\delta)\right)\norm{\hat{\ell}(x) -\tilde{\ell}(x)}_2^2 \\
		&\quad \quad \times \sup_{v\in \left[\mathcal{C}_1(S_{\ell}(x),3) \cup \{e_1,...,e_d\} \right] \cap \mathbb{S}^{d-1}}\left[\left|\frac{1}{4n}\sum_{i=1}^n \left\{(X_i^Tv)^2 -\mathrm{E}\left[(X_i^Tv)^2\right] \right\}\right| + \left|\frac{1}{4n}\mathrm{E}\left[(X_1^T v)^2\right]\right| \right]\\
		&\lesssim \left(1+r_{\pi}(n,\delta)\right)\norm{\hat{\ell}(x) -\tilde{\ell}(x)}_2^2\phi_{s_{\ell}(x)}\left[\sqrt{\frac{\log(1/\delta) + s_{\ell}(x)\log(ed/s_{\ell}(x))}{n}} + 1 \right] \\
		&\lesssim \left(1+r_{\pi}(n,\delta)\right)\norm{\hat{\ell}(x) -\tilde{\ell}(x)}_2^2\phi_{s_{\ell}(x)}
	\end{align*}
	when $\log d=o(n)$ and $\log(1/\delta)=o(n)$.\\

	\noindent$\bullet$ \textbf{Term VII:} Under the geometric constraints \eqref{geom_const1}, \eqref{geom_const2}, Lemmas~\ref{lem:gram_mat_conc}, \ref{lem:dual_approx_cone}, and \ref{lem:cone_ball_bnd}, we know that with probability at least $1-\delta$,
	\begin{align*}
		&\textbf{Term VII} =\left|\frac{1}{2n}\sum_{i=1}^n \pi(X_i)\left[X_i^T\left(\hat{\ell}(x) -\tilde{\ell}(x)\right)\right] \left[X_i^T\left(\tilde{\ell}(x)-\ell_0(x)\right)\right]\right|\\
		&\lesssim \norm{\hat{\ell}(x)-\tilde{\ell}(x)}_2 r_{\ell}\norm{\ell_0(x)}_2 \\
		&\quad \times \sup_{\substack{u\in \mathcal{C}_1(S_{\ell}(x),3)\cap \mathbb{S}^{d-1}\\v\in \mathcal{C}_2(S_{\ell}(x),1) \cap \mathbb{S}^{d-1}}}\left[\left|\frac{1}{2n}\sum_{i=1}^n \left\{\pi(X_i)(X_i^Tu)(X_i^Tv) -\mathrm{E}\left[\pi(X_i)(X_i^Tu)(X_i^Tv)\right] \right\}\right| + \left|\mathrm{E}\left[\pi(X_1)(X_1^Tu)(X_1^Tv)\right]\right| \right]\\
		&\lesssim \norm{\hat{\ell}(x) -\tilde{\ell}(x)}_2 \frac{\norm{x}_2r_{\ell}}{\kappa_R^2} \cdot \phi_{s_{\ell}(x)}  \left[\sqrt{\frac{\log(1/\delta) + s_{\ell}(x)\log(ed/s_{\ell}(x))}{n}} + 1\right]\\
		& \lesssim \norm{\hat{\ell}(x) -\tilde{\ell}(x)}_2 \frac{\norm{x}_2r_{\ell}}{\kappa_R^2} \cdot \phi_{s_{\ell}(x)}
	\end{align*}
	when $\log d=o(n)$ and $\log(1/\delta)=o(n)$.\\
	
	\noindent$\bullet$ \textbf{Term VIII:} We bound \textbf{Term VIII} using arguments similar to those for \textbf{Term IV} as:
	\begin{align*}
		&\textbf{Term VIII} = \left|\frac{1}{2n}\sum_{i=1}^n \left(\hat{\pi}_i -\pi_i\right)\left[X_i^T\left(\hat{\ell}(x) -\tilde{\ell}(x)\right)\right] \left[X_i^T\left(\tilde{\ell}(x)-\ell_0(x)\right)\right]\right|\\
		&\lesssim r_{\pi}(n,\delta) \norm{\hat{\ell}(x) - \tilde{\ell}(x)}_2 \\
		&\quad \times \Biggl\{\sup_{v \in \mathcal{C}_1(S_{\ell}(x),3)\cap \mathbb{S}^{d-1}} \frac{1}{2n}\sum_{i=1}^n \left[\left|X_i^T v\cdot X_i^T\left[\tilde{\ell}(x)-\ell_0(x)\right]\right| - \mathrm{E}\left|X_i^T v \cdot X_i^T\left[\tilde{\ell}(x)-\ell_0(x)\right]\right|\right] \\
		&\quad + \sup_{v \in \mathcal{C}_1(S_{\ell}(x),3) \cap \mathbb{S}^{d-1}} \mathrm{E}\left|X_i^T v \cdot X_i^T\left[\tilde{\ell}(x)-\ell_0(x)\right]\right| \Biggr\}\\
		&\lesssim r_{\pi}(n,\delta)\cdot \norm{\hat{\ell}(x)-\tilde{\ell}(x)}_2\frac{r_{\ell}\norm{x}_2}{\kappa_R^2} \cdot \phi_{s_{\ell}(x)} 
	\end{align*}
	with probability at least $1-\delta$ when $\log d=o(n)$ and $\log(1/\delta)=o(n)$.\\
	
	\noindent$\bullet$ \textbf{Term IX and X:} Under \eqref{geom_const1}, Assumption~\ref{assump:prop_consistent_rate}, Lemma~\ref{lem:dual_approx_cone}, and Lemma~\ref{lem:gram_mat_conc}, we have that with probability at least $1-\delta$,
	\begin{align*}
		&\textbf{Term IX}+\textbf{Term X} \\
		&= \left|\frac{1}{4n}\sum_{i=1}^n \pi(X_i) \left[X_i^T\left(\tilde{\ell}(x)-\ell_0(x)\right)\right]^2\right| + \left|\frac{1}{4n}\sum_{i=1}^n \left(\hat{\pi}_i -\pi_i\right) \left[X_i^T\left(\tilde{\ell}(x)-\ell_0(x)\right)\right]^2\right|\\
		&\leq \left(1+r_{\pi}(n,\delta)\right) \left|\frac{1}{4n}\sum_{i=1}^n \left[X_i^T\left(\tilde{\ell}(x)-\ell_0(x)\right)\right]^2\right|\\
		&\leq \left(1+r_{\pi}(n,\delta)\right) r_{\ell}^2\norm{\ell_0(x)}_2^2 \sup_{v\in \mathcal{C}_2(S_{\ell}(x),1) \cap \mathbb{S}^{d-1}}\left[\left|\frac{1}{4n}\sum_{i=1}^n \left\{(X_i^Tv)^2 -\mathrm{E}\left[(X_i^Tv)^2\right] \right\}\right| + \left|\frac{1}{4n}\mathrm{E}\left[(X_1^T v)^2\right]\right| \right]\\
		&\lesssim \left(1+r_{\pi}(n,\delta)\right) \frac{\norm{x}_2^2 r_{\ell}^2}{\kappa_R^4} \cdot \phi_{s_{\ell}(x)} \left[\sqrt{\frac{\log(1/\delta) + s_{\ell}(x)\log(ed/s_{\ell}(x))}{n}} + 1 \right]\\
		&\lesssim \left(1+r_{\pi}(n,\delta)\right) \frac{\norm{x}_2^2 r_{\ell}^2}{\kappa_R^4} \cdot \phi_{s_{\ell}(x)}
	\end{align*}
	when $\log d=o(n)$ and $\log(1/\delta)=o(n)$.\\
	
	\noindent$\bullet$ \textbf{Conclusion:} Combining our above calculations yields that with probability at least $1-\delta$,
	\begin{align*}
		&\left|\sum_{i=1}^n \hat{\pi}_i \hat{w}_i(x)^2 - x^T\left\{\mathrm{E}\left[\pi(X_1) X_1X_1^T\right]\right\}^{-1}x \right| \\
		&\lesssim \frac{\norm{x}_2^2}{\kappa_R^2}\sqrt{\frac{\log(1/\delta)}{n}} + \frac{\phi_{s_{\ell}(x)} \norm{x}_2^2}{\kappa_R^4} \left(r_{\pi} + r_{\ell} + r_{\pi} r_{\ell} + r_{\ell}^2 + r_{\pi} r_{\ell}^2\right) \\
		&\quad + \frac{\phi_{s_{\ell}(x)} \norm{x}_2}{\kappa_R^2} \norm{\hat{\ell}(x) - \tilde{\ell}(x)}_2\left(1 + r_{\ell} + r_{\pi} r_{\ell}\right) + \phi_{s_{\ell}(x)}\left(1+r_{\pi}\right) \norm{\hat{\ell}(x) - \tilde{\ell}(x)}_2^2.
	\end{align*}
	By setting $\delta=\frac{1}{nd}$ with our bound for $\norm{\hat{\ell}(x) - \tilde{\ell}(x)}_2$ in \autoref{thm:dual_consist2}, we conclude that
	\begin{align*}
		&\left|\sum_{i=1}^n \hat{\pi}_i \hat{w}_i(x)^2 - x^T\left\{\mathrm{E}\left[\pi(X_1) X_1X_1^T\right]\right\}^{-1}x \right| \\
		&= O_P\left(\frac{\norm{x}_2^2}{\kappa_R^2}\sqrt{\frac{\log(nd)}{n}} \right) + O_P\left(\frac{\phi_{s_{\ell}(x)} \norm{x}_2^2}{\kappa_R^4} \left(r_{\pi} + r_{\ell}\right) \right) \\
		&\quad + O_P\left(\phi_{s_{\ell}(x)}\left[1+r_{\pi} + \frac{\norm{x}_2}{\kappa_R^2} (1+r_{\ell})\right]\right) \norm{\hat{\ell}(x) - \tilde{\ell}(x)}_2\\
		&= O_P\left(\frac{\norm{x}_2^2}{\kappa_R^2}\sqrt{\frac{\log(nd)}{n}} \right) + O_P\left(\frac{\phi_{s_{\ell}(x)} \norm{x}_2^2}{\kappa_R^4} \left(r_{\pi} + r_{\ell}\right) \right) \\
		&\quad + O_P\left(\frac{\phi_{s_{\ell}(x)}\norm{x}_2}{\kappa_R^3}\left[1+r_{\pi} + \frac{\norm{x}_2}{\kappa_R^2} (1+r_{\ell})\right] \right.\\
		&\quad \quad \quad \quad \left. \left[\sqrt{\frac{\phi_1 s_{\ell}(x) \log(nd)}{n}} + \frac{\phi_{s_{\ell}(x)} r_{\ell}}{\kappa_R} + \frac{\sqrt{\phi_1\phi_{s_{\ell}(x)} s_{\ell}(x)}}{\kappa_R} \cdot r_{\pi}\right]\right)\\
		&= o_P(1)
	\end{align*}
	when $\norm{x}_2^2 \phi_{s_{\ell}(x)}^2 = O(1)$, $\frac{(1+\kappa_R^3)}{\kappa_R^5} \sqrt{\frac{s_{\ell}(x)\log(nd)}{n}} = o(1)$, and $\frac{(1+\kappa_R^4)}{\kappa_R^6} \left[r_{\ell} + r_{\pi}\sqrt{s_{\ell}(x)}\right]=o(1)$.
	Therefore, $\sum_{i=1}^n \hat{\pi}_i \hat{w}_i(x)^2$ is a consistent estimator of the variance factor $x^T\left[\mathrm{E}\left(R X_1X_1^T\right)\right]^{-1}x$.
\end{proof}

\subsection{Proof of Proposition~\ref{prop:rel_const_glm}}
\label{subapp:proof_prop_score}

Before presenting the proof of Proposition~\ref{prop:rel_const_glm}, we first establish the consistency of $\hat{\theta}$ from the Lasso-type generalized linear regression in the following lemma, which is of independent interest. 

\begin{lemma}[Consistency of propensity score estimation with Lasso-type generalized linear regression]
	\label{lem:prop_score_const}
	Let $\delta \in (0,1)$ and $A_1,A_2,A_3 >0$ be some absolute constants. Suppose that Assumptions~\ref{assump:sub_gaussian}, \ref{assump:glm_prop_score}, \ref{assump:res_eigen_theta}, and \ref{assump:link_func} hold as well as
	$$d \geq s_{\theta} + 2, \quad A_1\phi_{s_{\theta}} \log\left(\frac{\phi_{s_{\theta}}}{\rho}\right) \sqrt{\frac{\phi_1 s_{\theta}\log d}{n}} \leq \frac{\rho}{8}, \; \text{ and }\;   A_2\exp\left[-\frac{n\rho^2}{\phi_{s_{\theta}}^2 \log^2\left(\frac{\phi_{s_{\theta}}}{\rho}\right)} - \phi_1 s_{\theta}\log d\right] \leq \delta.$$
	\begin{enumerate}[label=(\alph*)]
		\item If $\zeta \geq \norm{\frac{2}{n}\sum\limits_{i=1}^n \frac{\varphi'(X_i^T\theta_0)\left[\varphi(X_i^T\theta_0) -R_i \right]X_i}{\varphi(X_i^T\theta_0) \left[1-\varphi(X_i^T\theta_0) \right]}}_{\infty}$, then with probability at least $1-\delta$,
		$$\norm{\hat{\theta} - \theta_0}_2 \lesssim \frac{\sqrt{s_{\theta}} \zeta}{\rho}.$$ 
		
		\item If, in addition, $\zeta = A_3 A_{\varphi} \sqrt{\frac{\phi_1\log(d/\delta)}{n}}$, then with probability at least $1-\delta$,
		$$\norm{\hat{\theta} - \theta_0}_2 \lesssim \frac{A_{\varphi}}{\rho} \sqrt{\frac{\phi_1 s_{\theta} \log(d/\delta)}{n}}.$$
	\end{enumerate}
\end{lemma}

\begin{proof}[Proof of Lemma~\ref{lem:prop_score_const}]
	The proof is partially adapted from Theorem 2 in \cite{cai2023statistical}.
	
	\noindent (a) Since $\hat{\theta}$ minimizes the objective function in \eqref{glm_lasso}, we know that
	\begin{align*}
		&-\frac{1}{n} \sum_{i=1}^n R_i \log\left[\varphi(X_i^T\hat{\theta})\right] -\frac{1}{n}\sum_{i=1}^n (1-R_i)\log\left[1-\varphi(X_i^T \hat{\theta}) \right]  + \zeta\norm{\hat{\theta}}_1\\ 
		&\leq -\frac{1}{n} \sum_{i=1}^n R_i \log\left[\varphi(X_i^T\theta_0)\right] -\frac{1}{n}\sum_{i=1}^n (1-R_i)\log\left[1-\varphi(X_i^T \theta_0) \right]  + \zeta\norm{\theta_0}_1.
	\end{align*}
	By Taylor's expansion of $\mathcal{L}_{\varphi}(\theta)$ about $\theta_0$ in the direction of $\hat{\Delta}=\hat{\theta} -\theta_0$, we have
	\begin{align}
		\label{glm_taylor}
		\begin{split}
			\int_0^1 (1-t)\cdot \hat{\Delta}^T \left[\nabla\nabla \mathcal{L}_{\varphi}(\theta_0+t\hat{\Delta})\right] \hat{\Delta} \,dt 
			& \leq -\nabla\mathcal{L}_{\varphi}(\theta_0)^T \hat{\Delta} + \zeta \left[\norm{\theta_0}_1 - \norm{\theta_0 + \hat{\Delta}}_1\right].
		\end{split}
	\end{align}
	By Assumption~\ref{assump:link_func}(c), we know that the left-hand side of \eqref{glm_taylor} can be lower bounded as follows:
	\begin{align*}
		\int_0^1 (1-t)\cdot \hat{\Delta}^T \left[\nabla\nabla \mathcal{L}_{\varphi}(\theta_0+t\hat{\Delta})\right] \hat{\Delta} \,dt
		&=\frac{1}{n} \sum_{i=1}^n \int_0^1 (1-t) \cdot H(\theta_0+t\hat{\Delta};R_i,X_i) \left(X_i^T\hat{\Delta}\right)^2 dt >0.
	\end{align*}
	It thus implies that
	\begin{equation}
		\label{glm_taylor2}
		-\nabla\mathcal{L}_{\varphi}(\theta_0)^T \hat{\Delta} + \zeta \left[\norm{\theta_0}_1 - \norm{\theta_0 + \hat{\Delta}}_1\right] >0.
	\end{equation}
	In addition, we know that
	\begin{align*}
		\nabla\mathcal{L}_{\varphi}(\theta_0) &= -\frac{1}{n}\sum_{i=1}^n \frac{R_iX_i \cdot \varphi'(X_i^T \theta_0)}{\varphi(X_i^T\theta_0)} +\frac{1}{n} \sum_{i=1}^n \frac{(1-R_i)X_i \cdot \varphi'(X_i^T\theta_0)}{1-\varphi(X_i^T\theta_0)}\\
		&= \frac{1}{n} \sum_{i=1}^n \frac{\varphi'(X_i^T\theta_0)\left[\varphi(X_i^T\theta_0) -R_i\right] X_i}{\varphi(X_i^T\theta_0) \left[1- \varphi(X_i^T\theta_0)\right]}.
	\end{align*}
	By H{\"o}lder's inequality and the (reverse) triangle inequality, together with the sparsity of $\theta_0$ under Assumption~\ref{assump:glm_prop_score}, we know from \eqref{glm_taylor2} that
	\begin{align*}
		0 & \leq \norm{\nabla\mathcal{L}_{\varphi}(\theta_0)}_{\infty} \norm{\hat{\Delta}}_1 + \zeta\left[\norm{\theta_{0,S_{\theta}}}_1 - \norm{\theta_{0,S_{\theta}} + \hat{\Delta}_{S_{\theta}}}_1 - \norm{\hat{\Delta}_{S_{\theta}^c}}_1 \right]\\
		&\leq \frac{\zeta}{2} \left[\norm{\hat{\Delta}_{S_{\theta}}}_1 + \norm{\hat{\Delta}_{S_{\theta}^c}}_1\right] + \zeta \left[\norm{\hat{\Delta}_{S_{\theta}}}_1 - \norm{\hat{\Delta}_{S_{\theta}^c}}_1 \right]\\
		&= \frac{3\zeta}{2} \norm{\hat{\Delta}_{S_{\theta}}}_1 - \frac{\zeta}{2} \norm{\hat{\Delta}_{S_{\theta}^c}}_1,
	\end{align*}
	where $\norm{\left[\hat{\Delta}_{S_{\theta}} \right]}_1 = \sum_{k \in S_{\theta}} \left|\hat{\Delta}_k \right|$. It implies that $\norm{\hat{\Delta}_{S_{\theta}^c}}_1 \leq 3 \norm{\hat{\Delta}_{S_{\theta}}}_1$ and thus, the estimation error $\hat{\Delta}=\hat{\theta} -\theta_0$ lies in the cone $\mathcal{C}_1(S_{\theta},3)$. 
	
	Now, by Assumption~\ref{assump:link_func}(c), we know that for all $1\leq i\leq n$,
	$$\left|\log H(\theta_0+t\hat{\Delta}; R_i,X_i) - \log H(\theta_0;R_i,X_i) \right| \leq A_H\left(|X_i^T\theta_0|^2 + t^2|X_i^T \hat{\Delta}|^2 + t|X_i^T\hat{\Delta}|\right),$$
	which implies that for all $1\leq i\leq n$,
	\begin{align*}
		\exp\left[-A_H\left(|X_i^T\theta_0|^2 + t^2|X_i^T \hat{\Delta}|^2 + t|X_i^T\hat{\Delta}|\right) \right] &\leq \frac{H(\theta_0+t\hat{\Delta}; R_i,X_i)}{H(\theta_0; R_i,X_i)} \\
		&\leq \exp\left[A_H\left(|X_i^T\theta_0|^2 + t^2|X_i^T \hat{\Delta}|^2 + t|X_i^T\hat{\Delta}|\right) \right].
	\end{align*}
	Let $K_d(S_{\theta}) = \mathcal{C}_1(S_{\theta},3)\cap \mathbb{S}^{d-1}$. Thus, for two constants $T_1,T_2>0$ that will be determined in Lemma~\ref{lem:prop_score_cont}, we obtain that
	{\scriptsize\begin{align*}
			&\frac{1}{n} \sum_{i=1}^n \int_0^1 (1-t) \cdot H(\theta_0+t\hat{\Delta};R_i,X_i) \left(X_i^T\hat{\Delta}\right)^2 dt \\
			&\geq \inf_{v\in K_d(S_{\theta})}\left[\frac{1}{n} \sum_{i=1}^n \int_0^1 (1-t) \cdot H(\theta_0+tv;R_i,X_i) \cdot \left(X_i^Tv\right)^2 dt\right] \norm{\hat{\Delta}}_2^2\\
			&\geq \inf_{v\in K_d(S_{\theta})} \left\{\frac{1}{n} \sum_{i=1}^n \int_0^1 (1-t) \cdot H(\theta_0;R_i,X_i) \exp\left[-A_H\left(|X_i^T\theta_0|^2 + t^2|X_i^T v|^2 + t|X_i^Tv|\right) \right]\left(X_i^Tv\right)^2 dt \right\} \norm{\hat{\Delta}}_2^2\\
			&\stackrel{\text{(i)}}{\geq} \inf_{v\in K_d(S_{\theta})}\Big\{\frac{1}{n} \sum_{i=1}^n \int_0^1 (1-t) \cdot H(\theta_0;R_i,X_i) \exp\left[-A_H\left(|X_i^T\theta_0|^2 + t^2|X_i^T v|^2 + t|X_i^Tv|\right) \right]\left(X_i^Tv\right)^2\\
			&\quad \times \mathbbm{1}_{\{\sqrt{A_H}\cdot t|X_i^T v| \leq T_1\}}dt \Big\} \norm{\hat{\Delta}}_2^2\\
			&\stackrel{\text{(ii)}}{\geq} \inf_{v\in K_d(S_{\theta})}\left\{\frac{1}{n} \sum_{i=1}^n \int_0^1 (1-t) \cdot H(\theta_0;R_i,X_i) \exp\left[-A_H|X_i^T\theta_0|^2 -T_1^2 - \sqrt{A_H}T_1 \right]\left(X_i^Tv\right)^2 \cdot \mathbbm{1}_{\left\{\sqrt{A_H} \cdot t |X_i^T v| \leq T_1 \right\}}dt \right\} \norm{\hat{\Delta}}_2^2\\
			&\stackrel{\text{(iii)}}{\geq} \inf_{v\in K_d(S_{\theta})}\left\{\frac{1}{n} \sum_{i=1}^n \int_0^1 (1-t) \cdot H(\theta_0;R_i,X_i) \exp\left[-A_H|X_i^T\theta_0|^2 -T_1^2 - \sqrt{A_H}T_1 \right]\left(X_i^Tv\right)^2 \cdot \mathbbm{1}_{\left\{\sqrt{A_H} |X_i^T v| \leq T_1\right\}}dt \right\} \norm{\hat{\Delta}}_2^2\\
			& =\inf_{v\in K_d(S_{\theta})}\left\{\frac{1}{2n} \sum_{i=1}^n H(\theta_0;R_i,X_i) \exp\left[-A_H|X_i^T\theta_0|^2 -T_1^2 - \sqrt{A_H}T_1 \right]\left(X_i^Tv\right)^2 \cdot \mathbbm{1}_{\left\{\sqrt{A_H} |X_i^T v| \leq T_1\right\}} \right\} \norm{\hat{\Delta}}_2^2\\
			& \geq \inf_{v\in K_d(S_{\theta})}\left\{\frac{1}{2n} \sum_{i=1}^n H(\theta_0;R_i,X_i) \exp\left[-A_HT_2^2 -T_1^2 - \sqrt{A_H}T_1 \right]\left(X_i^Tv\right)^2 \cdot \mathbbm{1}_{\left\{\sqrt{A_H} |X_i^T v| \leq T_1\right\}} \cdot \mathbbm{1}_{\{|X_i^T \theta_0| \leq T_2\}} \right\} \norm{\hat{\Delta}}_2^2,
	\end{align*}}%
	where we multiply by the indicator $\mathbbm{1}_{\{\sqrt{A_H}\cdot t|X_i^T v| \leq T_1\}} \in \{0,1\}$ to derive (i), plug the upper bound $T_1$ from this indicator into (ii), and take the largest value $t=1$ in the indicator to obtain (iii). Notice that by Assumption~\ref{assump:link_func}(a,c), there exists a constant $c_{\bm{T}}'>0$ that depends only on $\bm{T}=(T_1,T_2)$ such that in the region $\left\{X_i\in \mathbb{R}^d:|X_i^T\theta_0| \leq T_2\right\}$,
	$$H(\theta_0;R_i,X_i) \geq c_{\bm{T}}'$$
	for all $i=1,...,n$. In detail, this is because 
	\begin{align*}
		H(\theta_0;R_i,X_i) \geq H(0;R_i,X_i) \exp\left(-A_H |X_i^T\theta_0|^2\right),
	\end{align*}
	while $H(0;R_i,X_i) >0$ and $X_i^T\theta_0$ is bounded within the region $\left\{X_i\in \mathbb{R}^d:|X_i^T\theta_0| \leq T_2\right\}$. Therefore, we can proceed the above display as:
	\begin{align*}
		&\frac{1}{n} \sum_{i=1}^n \int_0^1 (1-t) \cdot H(\theta_0+t\hat{\Delta};R_i,X_i) \left(X_i^T\hat{\Delta}\right)^2 dt \\
		&\geq \inf_{v\in K_d(S_{\theta})}\left\{\frac{c_{\bm{T}}'}{2n} \sum_{i=1}^n \exp\left(-A_HT_2^2 -T_1^2 - \sqrt{A_H}T_1 \right)\left(X_i^Tv\right)^2 \cdot \mathbbm{1}_{\left\{\sqrt{A_H} |X_i^T v| \leq T_1 \right\}} \cdot \mathbbm{1}_{\left\{|X_i^T \theta_0| \leq T_2\right\}} \right\} \norm{\hat{\Delta}}_2^2\\
		&\equiv \inf_{v\in K_d(S_{\theta})}\left[\frac{C_{\bm{T}}}{n} \sum_{i=1}^n \left(X_i^Tv\right)^2 \cdot \mathbbm{1}_{\left\{\sqrt{A_H} |X_i^T v| \leq T_1 \right\}} \cdot \mathbbm{1}_{\left\{|X_i^T \theta_0| \leq T_2\right\}} \right] \norm{\hat{\Delta}}_2^2,
	\end{align*}
	where $C_{\bm{T}} = \frac{c_{\bm{T}}'}{2} \exp\left(-A_HT_2^2 - T_1^2 - \sqrt{A_H} T_1 \right)>0$ is a constant that depends on $A_H>0$ (thus, the link function $\varphi$) and the values of $\bm{T}=(T_1,T_2)$. In addition, we consider the truncation function with $T_1>0$ such that
	\begin{equation}
		\label{trunc_func}
		h_{T_1}(u) = 
		\begin{cases}
			u^2 & \text{ if } 0\leq |u| \leq \frac{T_1}{2\sqrt{A_H}},\\
			\left(\frac{T_1}{\sqrt{A_H}} -|u| \right)^2 & \text{ if } \frac{T_1}{2\sqrt{A_H}} \leq |u| \leq \frac{T_1}{\sqrt{A_H}},\\
			0 & \text{ otherwise},
		\end{cases}
	\end{equation}
	and denote $f_v(x) = x^T v \cdot \mathbbm{1}_{\left\{|x^T \theta_0| \leq T_2\right\}}$ and $g_v(x) = h_{T_1}\left(f_v(x)\right)$. Then, we know that
	\begin{align*}
		&\inf_{v\in K_d(S_{\theta})}\left[\frac{C_{\bm{T}}}{n} \sum_{i=1}^n \left(X_i^Tv\right)^2 \cdot \mathbbm{1}_{\left\{\sqrt{A_H} |X_i^T v| \leq T_1 \right\}} \cdot \mathbbm{1}_{\left\{|X_i^T \theta_0| \leq T_2\right\}} \right] \norm{\hat{\Delta}}_2^2 \\
		&\geq \inf_{v\in K_d(S_{\theta})}\left[\frac{C_{\bm{T}}}{n} \sum_{i=1}^n g_v(X_i) \right] \norm{\hat{\Delta}}_2^2.
	\end{align*}
	By Lemma~\ref{lem:prop_score_cont} below, we can take 
	$$T_1= \bar{A}_3\sqrt{A_H \phi_{s_{\theta}}\log\left(\frac{\phi_{s_{\theta}}}{\rho} \right)} \quad \text{ and } \quad T_2 = \bar{A}_4 \norm{\theta_0}_2 \sqrt{\phi_{s_{\theta}}\log\left(\frac{\phi_{s_{\theta}}}{\rho} \right)}$$ 
	so that
	\begin{align*}
		&\mathrm{P}\left(\inf_{v\in K_d(S_{\theta}),\norm{v}_1=t} \mathbb{P}_n\left[g_v(X)\right] \leq \frac{\rho}{4}- \bar{A}_1 K_1 t\sqrt{\frac{\phi_1\log d}{n}} \right) \leq \bar{A}_2 \exp\left(-\frac{n\rho^2}{K_1^2} -\phi_1 t^2\log d\right)
	\end{align*}
	with $K_1=\bar{A}_3^2\phi_{s_{\theta}} \log\left(\frac{\phi_{s_{\theta}}}{\rho}\right)$, where $\bar{A}_1,\bar{A}_2,\bar{A}_3,\bar{A}_4>0$ are some absolute constants. Notice that for any $v\in K_d(S_{\theta})$, we know that
	$$1\leq \norm{v}_1\leq 4 \norm{v_{S_{\theta}}}_1 \leq 4\sqrt{s_{\theta}}\norm{v}_2 = 4\sqrt{s_{\theta}}.$$
	Therefore, under the condition that
	$$A_1\phi_{s_{\theta}} \log\left(\frac{\phi_{s_{\theta}}}{\rho}\right) \sqrt{\frac{\phi_1 s_{\theta}\log d}{n}} \leq \frac{\rho}{8} \quad \text{ and }\quad   A_2\exp\left[-\frac{n\rho^2}{\phi_{s_{\theta}}^2 \log^2\left(\frac{\phi_{s_{\theta}}}{\rho}\right)} - \phi_1 s_{\theta}\log d\right] \leq \delta$$
	for some absolute constants $A_1,A_2>0$, we have that with probability at least $1-\delta$
	\begin{align*}
		&\frac{1}{n} \sum_{i=1}^n \int_0^1 (1-t) \cdot H(\theta_0+t\hat{\Delta};R_i,X_i) \left(X_i^T\hat{\Delta}\right)^2 dt \\
		&\geq \inf_{v\in K_d(S_{\theta})}\left[\frac{C_{\bm{T}}}{n} \sum_{i=1}^n \left(X_i^Tv\right)^2 \cdot \mathbbm{1}_{\left\{\sqrt{A_H} |X_i^T v| \leq T_1 \right\}} \cdot \mathbbm{1}_{\left\{|X_i^T \theta_0| \leq T_2\right\}} \right] \norm{\hat{\Delta}}_2^2\\ 
		&\geq \frac{\rho}{8} \norm{\hat{\Delta}}_2^2.
	\end{align*}
	Substituting the above inequality to \eqref{glm_taylor} yields that with probability at least $1-\delta$,
	\begin{align*}
		\frac{\rho}{8} \norm{\hat{\Delta}}_2^2 & \leq \frac{\zeta}{2} \norm{\hat{\Delta}}_1 + \zeta \left[\norm{\hat{\Delta}_{S_{\theta}}}_1 - \norm{\hat{\Delta}_{S_{\theta}^c}}_1 \right]\\
		&\leq \frac{3\zeta}{2} \norm{\hat{\Delta}_{S_{\theta}}}_1\\
		&\leq \frac{3\zeta \sqrt{s_{\theta}}}{2} \norm{\hat{\Delta}}_2.
	\end{align*}
	The result follows by dividing both sides of the above inequality by $\norm{\hat{\Delta}}_2$ and rearranging.\\

	\noindent (b) Notice that by law of total expectation,
	$$\mathrm{E}\left\{\frac{2}{n}\sum_{i=1}^n \frac{\varphi'(X_i^T\theta_0)\left[\varphi(X_i^T\theta_0) -R_i \right]X_i}{\varphi(X_i^T\theta_0) \left[1-\varphi(X_i^T\theta_0) \right]}\right\} = \mathrm{E}\left\{\frac{2}{n}\sum_{i=1}^n \frac{X_i \cdot\varphi'(X_i^T\theta_0)\cdot \mathrm{E}\left[\varphi(X_i^T\theta_0) -R_i |X_i\right]}{\varphi(X_i^T\theta_0) \left[1-\varphi(X_i^T\theta_0) \right]}\right\} =0.$$
	In addition, 
	$$\norm{\frac{2}{n}\sum_{i=1}^n \frac{\varphi'(X_i^T\theta_0)\left[\varphi(X_i^T\theta_0) -R_i \right]X_i}{\varphi(X_i^T\theta_0) \left[1-\varphi(X_i^T\theta_0) \right]}}_{\infty} = \max_{1\leq k \leq d} \left|\frac{2}{n}\sum_{i=1}^n \frac{\varphi'(X_i^T\theta_0)\left[\varphi(X_i^T\theta_0) -R_i \right]}{\varphi(X_i^T\theta_0) \left[1-\varphi(X_i^T\theta_0) \right]} \cdot X_{ik} \right|$$
	and by (a) and (b) in Assumption~\ref{assump:link_func} together with Remark~\ref{link_func_bound}, we know that
	\begin{align*}
		\left|\frac{\varphi'(X_i^T\theta_0)\left[\varphi(X_i^T\theta_0) -R_i \right]}{\varphi(X_i^T\theta_0) \left[1-\varphi(X_i^T\theta_0) \right]} \cdot X_{ik}\right| &= \left|\left[\frac{(1-R_i) \varphi'(X_i^T\theta_0)}{1-\varphi(X_i^T\theta_0)} - \frac{R_i \varphi'(X_i^T\theta_0)}{\varphi(X_i^T\theta_0)}\right]X_{ik}\right|\\
		&\leq \max\left\{\frac{\varphi'(X_i^T\theta_0)}{\varphi(X_i^T\theta_0)},\, \frac{\varphi'(X_i^T\theta_0)}{1-\varphi(X_i^T\theta_0)} \right\}|X_{ik}|\\
		&\leq A_{\varphi}|X_{ik}|,
	\end{align*}
	showing that for any fixed $k$, $\frac{\varphi'(X_i^T\theta_0)\left[\varphi(X_i^T\theta_0) -R_i \right]}{\varphi(X_i^T\theta_0) \left[1-\varphi(X_i^T\theta_0) \right]} \cdot X_{ik}, i=1,...,n$ are independent and sub-Gaussian random variables with zero means. By Assumptions~\ref{assump:sub_gaussian}, we also have that
	\begin{align*}
		\max_{1\leq k \leq d} \norm{\frac{\varphi'(X_i^T\theta_0)\left[\varphi(X_i^T\theta_0) -R_i \right]}{\varphi(X_i^T\theta_0) \left[1-\varphi(X_i^T\theta_0) \right]} X_{ik}}_{\psi_2} &\lesssim \max_{1\leq k \leq d} A_{\varphi} \norm{X_{ik}}_{\psi_2}\\
		&\lesssim A_{\varphi} \sqrt{\phi_1}
	\end{align*}
	for all $1\leq i \leq n$. Therefore, by the general Hoeffding inequality (Theorem 2.6.2 in \citealt{vershynin2018high}), we know that there exists an absolute constant $A_0>0$ such that for all $t>0$,
	\begin{align*}
		&\mathrm{P}\left(\norm{\frac{2}{n}\sum_{i=1}^n \frac{X_i \varphi'(X_i^T\theta_0)\left[\varphi(X_i^T\theta_0) -R_i \right]}{\varphi(X_i^T\theta_0) \left[1-\varphi(X_i^T\theta_0) \right]}}_{\infty} >t \right) \\
		&\leq d \cdot \max_{1\leq k \leq d} \mathrm{P}\left(\left|\frac{2}{n}\sum_{i=1}^n \frac{\varphi'(X_i^T\theta_0)\left[\varphi(X_i^T\theta_0) -R_i \right]}{\varphi(X_i^T\theta_0) \left[1-\varphi(X_i^T\theta_0) \right]} \cdot X_{ik} \right| > t \right)\\
		&\leq 2d \cdot \exp\left(-\frac{nA_0 t^2}{A_{\varphi}^2 \phi_1} \right).
	\end{align*}
	Hence, we conclude that with probability at least $1-\delta$,
	\begin{align*}
		\norm{\frac{2}{n}\sum_{i=1}^n \frac{X_i \varphi'(X_i^T\theta_0)\left[\varphi(X_i^T\theta_0) -R_i \right]}{\varphi(X_i^T\theta_0) \left[1-\varphi(X_i^T\theta_0) \right]}}_{\infty} &\lesssim A_{\varphi} \sqrt{\frac{\phi_1\log(d/\delta)}{n}}.
	\end{align*}
	Finally, the bound on $\norm{\hat{\theta} - \theta_0}_2$ follows from (a), where we substitute $\zeta = A_3 A_{\varphi} \sqrt{\frac{\phi_1\log(d/\delta)}{n}}$ for some absolute constant $A_3>0$ into the upper bound of $\norm{\hat{\theta} - \theta_0}_2$.
\end{proof}

\vspace{3mm}

\begin{lemma}
	\label{lem:prop_score_cont}
	Let $\bar{A}_1,\bar{A}_2, \bar{A}_3>0$ be three absolute constants. Under Assumptions~\ref{assump:sub_gaussian} and \ref{assump:res_eigen_theta}, we have
	\begin{align*}
		&\mathrm{P}\left(\inf_{v\in K_d(S_{\theta}), \norm{v}_1=t} \mathbb{P}_n\left[g_v(X)\right] \leq \frac{\rho}{4}- \bar{A}_1 K_1 t\sqrt{\frac{\phi_1 \log d}{n}} \right) \leq \bar{A}_2 \exp\left(-\frac{n\rho^2}{K_1^2} - \phi_1 t^2\log d\right),
	\end{align*}
	where $\mathbb{P}_n$ is the empirical expectation over $n$ covariate vectors $\{X_1,...,X_n\}$, $K_d(S_{\theta}) = \mathcal{C}_1(S_{\theta},3) \cap \mathbb{S}^{d-1}$, $K_1=\bar{A}_3^2\phi_{s_{\theta}} \log\left(\frac{\phi_{s_{\theta}}}{\rho}\right)$, and the function $g_v$ is defined in \eqref{trunc_func}. 
\end{lemma}

\begin{proof}[Proof of Lemma~\ref{lem:prop_score_cont}]
	The proof is partially adapted from Theorem 2 and Lemma 4 in \cite{cai2023statistical} as well as Proposition 3 in \cite{negahban2010unified}. We prove this concentration bound in the following two steps.
	\begin{enumerate}
		\item For any $v\in K_d(S_{\theta}) = \mathcal{C}_1(S_{\theta},3) \cap \mathbb{S}^{d-1}$, we show that
		\begin{equation}
			\label{bnd_step1}
			\mathrm{E}\left[g_v(X)\right] \geq \frac{\rho}{2}.
		\end{equation}
		\item Define the random variable 
		$$Z(t) = \sup_{v\in K_d(S_{\theta}), \norm{v}_1=t} \big|\mathbb{P}_n\left[g_v(X)\right] - \mathrm{E}\left[g_v(X)\right] \big|.$$
		Then, it holds that
		\begin{equation}
			\label{bnd_step2}
			\mathrm{P}\left(Z(t) \geq \frac{\rho}{4} + \bar{A}_1K_1 t \sqrt{\frac{\phi_1\log d}{n}} \right) \leq \bar{A}_2 \exp\left(-\frac{n\rho^2}{K_1^2} - \phi_1 t^2\log d\right),
		\end{equation}
		where $K_1 = \bar{A}_3^2\phi_{s_{\theta}} \log\left(\frac{\phi_{s_{\theta}}}{\rho}\right)$ and $\bar{A}_1,\bar{A}_2,\bar{A}_3 >0$ are three absolute constants.
	\end{enumerate}
	Combining \eqref{bnd_step1} and \eqref{bnd_step2} yields the final result, so it remains to prove these two steps.
	
	\noindent {\bf Proof of \eqref{bnd_step1}.} By Assumption~\ref{assump:res_eigen_theta} on $\mathcal{C}_1(S_{\theta},3)$, we know that $\mathrm{E}\left[(X^T v)^2\right] \geq \rho \norm{v}_2^2 = \rho$ for any $v \in K_d(S_{\theta}) = \mathcal{C}_1(S_{\theta},3) \cap \mathbb{S}^{d-1}$. Then, it suffices to show that
	$$\mathrm{E}\left[(X^T v)^2 - g_v(X)\right] \leq \frac{\rho}{2}.$$
	Note that $g_v(X) = (X^Tv)^2$ when $|X^T\theta_0| \leq T_2$ and $|X^T v| \leq \frac{T_1}{2\sqrt{A_H}}$. Thus, by the union bound, we have that
	$$\mathrm{E}\left[(X^T v)^2 - g_v(X)\right] \leq \mathrm{E}\left[\left(X^T v\right)^2 \cdot \mathbbm{1}_{\left\{|X^T\theta_0| \geq T_2\right\}}\right] + \mathrm{E}\left[\left(X^T v\right)^2 \cdot \mathbbm{1}_{\left\{|X^Tv| \geq \frac{T_1}{2\sqrt{A_H}}\right\}} \right].$$
	Under Assumption~\ref{assump:sub_gaussian}, both $X^Tv$ and $X^T\theta_0$ are sub-Gaussian random variables. By the Cauchy-Schwarz inequality and the sub-Gaussian tail bound, we obtain that
	\begin{align*}
		\mathrm{E}\left[\left(X^T v\right)^2 \cdot \mathbbm{1}_{\left\{|X^T\theta_0| \geq T_2\right\}} \right] &\leq \sqrt{\mathrm{E}\left[\left(X^T v\right)^4\right]} \cdot \sqrt{\mathrm{P}\left(|X^T\theta_0| \geq T_2\right)}\\
		&\lesssim \sqrt{\left(2 \sqrt{\phi_{s_{\theta}}} \right)^4} \cdot \sqrt{2\exp\left(- \frac{T_2^2}{ \norm{\theta_0}_2^2 \phi_{s_{\theta}}} \right)}\\
		&\lesssim \phi_{s_{\theta}} \exp\left(-\frac{T_2^2}{ \norm{\theta_0}_2^2 \phi_{s_{\theta}}} \right)
	\end{align*}
	and
	\begin{align*}
		\mathrm{E}\left[\left(X^T v\right)^2 \cdot \mathbbm{1}_{\left\{|X^Tv| \geq \frac{T_1}{2\sqrt{A_H}}\right\}} \right] &\leq \sqrt{\mathrm{E}\left[\left(X^T v\right)^4\right]} \cdot \sqrt{\mathrm{P}\left(\left|X^T v\right| \geq \frac{T_1}{2\sqrt{A_H}}\right)}\\
		&\lesssim \phi_{s_{\theta}} \exp\left(-\frac{T_1^2}{A_H  \phi_{s_{\theta}}} \right)
	\end{align*}
	for $v\in K_d(S_{\theta})$, where Lemma~\ref{lem:size_dom_cone} implies that $\sup_{v\in K_d(S_{\theta})} \mathrm{E}\left[\left(X_i^T v\right)^2\right] \lesssim \phi_{s_{\theta}}$.  
	Therefore, we can take $T_1 =\bar{A}_3 \sqrt{A_H \phi_{s_{\theta}} \log\left(\frac{\phi_{s_{\theta}}}{\rho} \right)}$ and $T_2 = \bar{A}_4\norm{\theta_0}_2 \sqrt{\phi_{s_{\theta}}\log\left(\frac{\phi_{s_{\theta}}}{\rho} \right)}$ for some absolute constants $\bar{A}_3,\bar{A}_4>0$ in order for both of the above two expectations to be less than $\frac{\rho}{4}$. It thus proves \eqref{bnd_step1}.\\
	
	\noindent {\bf Proof of \eqref{bnd_step2}.} For any $v\in K_d(S_{\theta})$, we know that $\norm{g_v(X)}_{\infty} \leq \frac{T_1^2}{A_H} =\bar{A}_3^2 \phi_{s_{\theta}} \log\left(\frac{\phi_{s_{\theta}}}{\rho}\right) := K_1$ for some absolute constant $\bar{A}_3>0$. By the bounded differences inequality (Theorem 2.9.1 in \citealt{vershynin2018high}), we know that
	$$\mathrm{P}\left(Z(t) \geq \mathrm{E}\left[Z(t)\right] + z \right) \leq \exp\left(-\frac{nz^2}{2K_1^2} \right)$$
	for any $z>0$. Taking $z^* = \frac{\rho}{4} + K_1 t \sqrt{\frac{\phi_1\log d}{n}}$ yields that
	$$\mathrm{P}\left(Z(t) \geq \mathrm{E}\left[Z(t)\right] + z^* \right) \leq \exp\left[-\frac{n}{2K_1^2} \left(\frac{\rho}{4} + K_1 t \sqrt{\frac{\phi_1\log d}{n}} \right)^2 \right].$$
	Hence, in order to prove \eqref{bnd_step2}, it suffices to show that
	\begin{equation}
		\label{Expect_Zt}
		\mathrm{E}\left[Z(t)\right] \lesssim K_1 t \sqrt{\frac{\phi_1\log d}{n}}.
	\end{equation}
	Given a sequence $\{\omega_i\}_{i=1}^n$ of i.i.d. Rademacher random variables, the standard symmetrization argument (Lemma 6.4.2 in \citealt{vershynin2018high}) gives that
	\begin{align*}
		\mathrm{E}\left[Z(t)\right] & \leq 2 \mathrm{E}\left[\sup_{v\in \mathbb{S}^{d-1},\norm{v}_1=t} \left|\frac{1}{n} \sum_{i=1}^n \omega_i\, g_v(X_i) \right| \right]\\
		&= 2 \mathrm{E}\left[\sup_{v\in \mathbb{S}^{d-1},\norm{v}_1=t} \left|\frac{1}{n} \sum_{i=1}^n \omega_i\, h_{T_1}\left(f_v(X_i)\right) \right| \right].
	\end{align*}
	By the definition \eqref{trunc_func}, the function $h_{T_1}$ is Lipschitz with parameter at most $\frac{T_1}{\sqrt{A_H}} = \sqrt{K_1}$, which can be chosen to be smaller than $K_1$ when we potentially enlarge $K_1$ to $\max\left\{1, \bar{A}_3^2\phi_{s_{\theta}} \log\left(\frac{4 \phi_{s_{\theta}}}{\rho}\right) \right\}$ up to an absolute constant. In addition, $h_{T_1}(0)=0$. Therefore, by the contraction principle (Theorem 6.7.1 in \citealt{vershynin2018high}), we have that
	\begin{align}
		\label{contraction_prin}
		\begin{split}
			\mathrm{E}\left[Z(t)\right] & \leq 4K_1 \mathrm{E}\left[\sup_{v\in \mathbb{S}^{d-1},\norm{v}_1=t} \left|\frac{1}{n} \sum_{i=1}^n \omega_i\, f_v(X_i) \right| \right]\\
			&= 4K_1 \mathrm{E}\left[\sup_{v\in \mathbb{S}^{d-1},\norm{v}_1=t} \left|\frac{1}{n} \sum_{i=1}^n \omega_i\cdot X_i^Tv \cdot \mathbbm{1}_{\left\{|X_i^T\theta_0| \leq T_2 \right\}} \right| \right]\\
			&\leq 4K_1 t \cdot \mathrm{E}\left[\norm{\frac{1}{n} \sum_{i=1}^n \omega_i\cdot X_i \cdot \mathbbm{1}_{\left\{|X_i^T\theta_0| \leq T_2 \right\}}}_{\infty} \right],
		\end{split}
	\end{align}
	where the last inequality is due to H{\"o}lder's inequality. Finally, since $\omega_i \cdot X_i \cdot \mathbbm{1}_{\left\{|X_i^T\theta_0| \leq T_2 \right\}}, i=1,...,n$ are centered independent sub-Gaussian random vectors, we apply the standard bounds for sub-Gaussian maxima to obtain that
	$$\mathrm{E}\left[\norm{\frac{1}{n} \sum_{i=1}^n \omega_i \cdot X_i \cdot \mathbbm{1}_{\left\{|X_i^T\theta_0| \leq T_2 \right\}}}_{\infty} \right] \lesssim \sqrt{\frac{ \phi_1\log d}{n}}.$$
	Combining with \eqref{contraction_prin} leads to \eqref{Expect_Zt}, and \eqref{bnd_step2} is thus proved.
\end{proof}

\vspace{3mm}

\begin{customprop}{C.1}
	Let $\delta \in (0,1)$ and $A_1,A_2,A_3,A_4>0$ be some absolute constants. Suppose that Assumptions~\ref{assump:sub_gaussian}, \ref{assump:glm_prop_score}, \ref{assump:res_eigen_theta}, and \ref{assump:link_func} hold as well as
	$$d \geq s_{\theta} + 2, \quad A_1\phi_{s_{\theta}} \log\left(\frac{\phi_{s_{\theta}}}{\rho}\right) \sqrt{\frac{\phi_1 s_{\theta}\log d}{n}} \leq \frac{\rho}{8}, \; \text{ and }\;   A_2\exp\left[-\frac{n\rho^2}{\phi_{s_{\theta}}^2 \log^2\left(\frac{\phi_{s_{\theta}}}{\rho}\right)} - \phi_1 s_{\theta}\log d\right] \leq \delta.$$
	If $\zeta =A_3 A_{\varphi} \sqrt{\frac{\phi_1\log(d/\delta)}{n}}$, then the rate of consistency $r_{\pi}=r_{\pi}(n,\delta)$ in Assumption~\ref{assump:prop_consistent_rate} can be taken as:
	$$r_{\pi}(n,\delta) = \frac{A_4 L_{\varphi} A_{\varphi}}{\rho} \sqrt{\frac{\phi_1 \phi_{s_{\theta}} s_{\theta} \log(d/\delta)}{n}\left[\log\left(\frac{n}{\delta}\right) + s_{\theta} \log\left(\frac{e d}{s_{\theta}} \right) \right]}.$$
\end{customprop}

\begin{proof}[Proof of Proposition~\ref{prop:rel_const_glm}]
	Under the differentiability and Lipschitz conditions on the link function $\varphi$ guaranteed by Assumption~\ref{assump:link_func}, we can utilize the Taylor's theorem to obtain that
	\begin{align}
		\label{rel_const_bnd}
		\begin{split}
			\left|\hat{\pi}_i -\pi_i\right| &= \left|\varphi(X_i^T\hat{\theta}) -\varphi(X_i^T\theta_0)\right|\\ 
			&\leq L_{\varphi} \left(\sup_{v\in \mathcal{C}_1(S_{\theta},3)\cap B_d(0,1)}\left|X_i^Tv\right|\right)\norm{\hat{\theta}-\theta_0}_2
		\end{split}
	\end{align}
	for all $i=1,...,n$ with probability at least $1-\delta$, where the last inequality follows from the proof of Lemma~\ref{lem:prop_score_const} that $\hat{\theta}-\theta_0 \in \mathcal{C}_1(S_{\theta},3)$ with probability at least $1-\delta$ when $\zeta =A_3 A_{\varphi} \sqrt{\frac{\phi_1\log(d/\delta)}{n}}$. By Lemma~\ref{lem:size_dom_cone}, whenever $d\geq s_{\theta}+2$, there exists a discrete set $M_{d,\theta} \subset B_d(0,1)$ with cardinality $|M_{d,\theta}| \leq \frac{3}{2} \left(\frac{5ed}{s_{\theta}} \right)^{s_{\theta}}$ such that $\mathcal{C}_1(S_{\theta},3) \cap B_d(0,1) \subset 10 \cdot \mathrm{conv}(M_{d,\theta})$ and $\norm{u}_0 \leq s_{\theta}$ for all $u\in M_{d,\theta}$. For any $t>0$, we bound the probability
	\begin{align*}
		\mathrm{P}\left(\max_{1\leq i \leq n}\sup_{v\in \mathcal{C}_1(S_{\theta},3)\cap B_d(0,1)}\left|X_i^Tv\right| > t \right) &\stackrel{\text{(i)}}{\leq} n\cdot \mathrm{P}\left(\sup_{v \in 10\cdot \mathrm{conv}(M_{d,\theta})} \left|X_i^Tv\right| > t \right)\\
		&\stackrel{\text{(ii)}}{=} n\cdot \mathrm{P}\left(\sup_{v \in 10M_{d,\theta}} \left|X_i^Tv\right| > t \right)\\
		&\stackrel{\text{(iii)}}{\leq} n|M_{d,\theta}| \sup_{v \in 10M_{d,\theta}} \mathrm{P}\left(\left|X_i^Tv\right| > t \right)\\
		&\stackrel{\text{(iv)}}{\leq} \frac{3n}{2}\left(\frac{5ed}{s_{\theta}}\right)^{s_{\theta}} \cdot 2\exp\left(-\frac{A_0 t^2}{\phi_{s_{\theta}}} \right)
	\end{align*}
	for some absolute constant $A_0>0$, where we apply the union bound in inequalities (i) and (iii), utilize Lemma~\ref{lem:max_biconvex} that the supremum of a convex function on $\mathrm{conv}(M_{d,\theta})$ is always attained at its vertices which are contained in $M_{d,\theta}$ to obtain the equality (ii), and leverage the sub-Gaussian bound under Assumption~\ref{assump:sub_gaussian} to obtain (iv). It implies that
	$$\max_{1\leq i \leq n}\sup_{v\in \mathcal{C}_1(S_{\theta},3)\cap B_d(0,1)}\left|X_i^Tv\right| \lesssim \sqrt{\phi_{s_{\theta}}\left[\log\left(\frac{n}{\delta}\right) + s_{\theta} \log\left(\frac{e d}{s_{\theta}} \right) \right]}$$
	with probability at least $1-\delta$. By Lemma~\ref{lem:prop_score_const}, we also know that 
	$$\norm{\hat{\theta} - \theta_0}_2 \lesssim \frac{A_{\varphi}}{\rho} \sqrt{\frac{\phi_1 s_{\theta} \log(d/\delta)}{n}}$$
	with probability at least $1-\delta$. Thus, combining \eqref{rel_const_bnd} with Lemma~\ref{lem:prop_score_const} gives us that
	\begin{align*}
		\left|\hat{\pi}_i - \pi_i \right| &\lesssim L_{\varphi} \sqrt{\phi_{s_{\theta}}\left[\log\left(\frac{n}{\delta}\right) + s_{\theta} \log\left(\frac{e d}{s_{\theta}} \right) \right]} \cdot \norm{\hat{\theta}-\theta_0}_2\\
		&\lesssim \frac{L_{\varphi} A_{\varphi}}{\rho} \sqrt{\frac{\phi_1 \phi_{s_{\theta}} s_{\theta} \log(d/\delta)}{n}\left[\log\left(\frac{n}{\delta}\right) + s_{\theta} \log\left(\frac{e d}{s_{\theta}} \right) \right]}
	\end{align*}
	for all $i=1,...,n$ with probability at least $1-\delta$. The result follows.
\end{proof}

\subsection{Auxiliary Results}
\label{subapp:auxi_res}

Recall that a cone of dominant coordinates is defined as:
$$\mathcal{C}_q(S,\vartheta)=\left\{v\in \mathbb{R}^d: \norm{v_{S^c}}_q \leq \vartheta \norm{v_S}_q \right\},$$
where $\vartheta \in [0,\infty]$. When $\vartheta=\infty$, $\mathcal{C}_q(S,\infty)=\mathbb{R}^d$; when $\vartheta=0$, $\mathcal{C}_q(S,0)$ reduces to the linear subspace with $|S|$-sparse vectors whose nonzero entries are indexed by $S$ in $\mathbb{R}^d$. The following lemma sheds light on the geometric properties of $\mathcal{C}_q(S,\vartheta)$.

\begin{lemma}[Size of cones of dominant coordinates with respect to the norm $\norm{\cdot}_q$]
	\label{lem:size_dom_cone}
	For any $\vartheta \in [0,\infty]$ and $S \subseteq \{1,...,d\}$, we denote the cardinality of $S$ by $s=|S|$. If $d\geq s+2$, then there exists a set $M_d \subset B_d(0,1)$ such that $\norm{m}_0 \leq s$ for all $m\in M_d$ and
	$$\mathcal{C}_q(S,\vartheta) \cap B_d(0,1) \subset 2(2+\vartheta) \mathrm{conv}(M_d)\quad \text{ with } \quad |M_d| \leq \frac{3}{2}\left(\frac{5ed}{s} \right)^s,$$
	where $|M_d|$ is the cardinality of $M_d$ and $q\geq 1$. 
\end{lemma} 

\begin{proof}[Proof of Lemma~\ref{lem:size_dom_cone}]
	The main part of the proof follows from Lemma 7.1 in \cite{koltchinskii2011oracle} with generalization to the norm $\norm{\cdot}_q, q\geq 1$. It yields that
	$$|M_d| \leq 5^s \sum_{k=0}^s \binom{d}{k} \leq 5^s \cdot\frac{3 d^s}{2s!} \leq \frac{3}{2}\left(\frac{5ed}{s} \right)^s,$$
	where the second inequality follows from Proposition 3.6.4 in \cite{gine2021mathematical} when $d\geq s+2$, and the nonasymptotic Stirling formula is applied in the third inequality; see Page 52 in \cite{feller1968prob}. Here, the factor $5^s$ comes from the covering number of the unit ball; see Corollary 4.2.13 in \cite{vershynin2018high}.
\end{proof}

\vspace{3mm}

\begin{lemma}[Maxima of biconvex functions]
	\label{lem:max_biconvex}
	Let $f: \mathcal{X} \times \mathcal{Y} \to \mathbb{R}$ be a biconvex function, \emph{i.e.}, for any fixed $x\in \mathcal{X}, y\in \mathcal{Y}$, the functions $f(\cdot,y), f(x,\cdot)$ are convex. Then, 
	$$\sup_{x\in \mathrm{conv}(\mathcal{X})} \sup_{y\in \mathrm{conv}(\mathcal{Y})} f(x,y) = \sup_{x\in \mathcal{X}} \sup_{y\in \mathcal{Y}} f(x,y).$$
	Moreover, the identity still holds if we replace $f$ with $|f|$.
\end{lemma}

\begin{proof}[Proof of Lemma~\ref{lem:max_biconvex}]
	The proof is adapted from Lemma 1 in \cite{geng2000robust}; see also Section 3.3 in \citet{gorski2007biconvex}. For any $x\in \mathrm{conv}(\mathcal{X})$, there exists an integer $k$ and $\lambda_1,...,\lambda_k \geq 0$ with $\sum_{i=1}^k \lambda_i=1$ such that $x=\sum_{i=1}^k \lambda_ix_i$ for some $x_1,...,x_k \in \mathcal{X}$. Similarly, for any $y\in \mathrm{conv}(\mathcal{Y})$, there exists an integer $m$ and $\mu_1,...,\mu_m\geq 0$ with $\sum_{j=1}^m \mu_j=1$ such that $y=\sum_{j=1}^m \mu_j y_j$ for some $y_1,...,y_m \in \mathcal{Y}$. Thus, by the biconvexity of $f$ and Jensen's inequality, we have that, for any $x\in \mathrm{conv}(\mathcal{X}), y\in \mathrm{conv}(\mathcal{Y})$,
	$$f(x,y)\leq \sum_{i=1}^k \sum_{j=1}^m \lambda_i \mu_j f(x_i,y_j) \leq \sum_{i=1}^k \sum_{j=1}^m \lambda_i \mu_j \sup_{x\in \mathcal{X}} \sup_{y\in \mathcal{Y}} f(x,y) = \sup_{x\in \mathcal{X}} \sup_{y\in \mathcal{Y}} f(x,y).$$
	It implies that $\sup_{x\in \mathrm{conv}(\mathcal{X})}\sup_{y\in \mathrm{conv}(\mathcal{Y})} f(x,y) \leq \sup_{x\in \mathcal{X}} \sup_{y\in \mathcal{Y}} f(x,y)$. The reverse inequality holds trivially because $\mathcal{X} \subseteq \mathrm{conv}(\mathcal{X}), \mathcal{Y} \subseteq \mathrm{conv}(\mathcal{Y})$. The same arguments still apply when we replace $f$ with $|f|$. The proof is completed.
\end{proof}

\vspace{3mm}

\begin{lemma}
	\label{lem:expect_bound}
	Let $a\geq \sqrt{e}$ and $b, K>0$ be some constants. If a random variable $Z$ is positive almost surely with the tail bound $\mathrm{P}(Z> t) \leq a\exp\left(-b\min\left\{\frac{t^2}{K^2}, \frac{t}{K}\right\} \right)$ for all $t>0$, then 
	$$\mathrm{E}(Z) \leq 2K \sqrt{\frac{\log a}{b}} + \frac{K\log a}{b} + \frac{K}{b}.$$
\end{lemma}

\begin{proof}[Proof of Lemma~\ref{lem:expect_bound}]
	By Lemma 1.2.1 in \cite{vershynin2018high}, we can bound $\mathrm{E}(Z)$ by
	\begin{align*}
		\mathrm{E}(Z) &= \int_0^{\infty} \mathrm{P}(Z > t) \,dt\\
		&\leq \inf_{u_1>0}\left[\int_0^{u_1} dt + \int_{u_1}^{\infty} a\exp\left(-\frac{bt^2}{K^2} \right)dt \right] + \inf_{u_2>0}\left[\int_0^{u_2} dt + \int_{u_2}^{\infty} a\exp\left(-\frac{bt}{K} \right)dt \right]\\
		&:= \inf_{u_1>0} f(u_1) + \inf_{u_2>0} g(u_2),
	\end{align*}
	where $u_1,u_2>0$ are two (nonrandom) variables to be optimized, $f(u_1) = u_1+\int_{u_1}^{\infty} a\exp\left(-\frac{bt^2}{K^2} \right)dt$, and $g(u_2)=u_2 + \frac{aK}{b} \exp\left(-\frac{bu_2}{K} \right)$. By calculating $f'(u_1)$ and $g'(u_2)$, we know that $f$ attains its minimum at $u_1^* = K\sqrt{\frac{\log a}{b}}$ and $g$ has its minimum at $u_2^*=\frac{K\log a}{b}$. Thus,
	\begin{align*}
		\mathrm{E}(Z) &\leq K\sqrt{\frac{\log a}{b}} + \int_{u_1^*}^{\infty} \frac{a t}{u_1^*}\exp\left(-\frac{bt^2}{K^2} \right)dt + \frac{K\log a}{b} + \frac{K}{b}\\
		&= K\sqrt{\frac{\log a}{b}} + \frac{aK^2}{2u_1^* b} \exp\left(-\frac{b\left(u_1^* \right)^2}{K^2} \right) + \frac{K\log a}{b} + \frac{K}{b}\\
		&= K\sqrt{\frac{\log a}{b}} + \frac{K}{2\sqrt{b\log a}} + \frac{K\log a}{b} + \frac{K}{b}\\
		&= K\sqrt{\frac{\log a}{b}} \left(1+\frac{1}{2\log a} \right) + \frac{K\log a}{b} + \frac{K}{b}\\
		&\leq 2K \sqrt{\frac{\log a}{b}} + \frac{K\log a}{b} + \frac{K}{b},
	\end{align*}
	where we leverage the fact that $a\geq \sqrt{e}$ to obtain the last inequality. The result follows.
\end{proof}

\vspace{3mm}

\begin{lemma}
	\label{lem:gram_mat_conc}
	Let $\vartheta \in [0,\infty]$ and $\delta \in (0,1)$ be arbitrary. Suppose that Assumption~\ref{assump:sub_gaussian} holds. Let $S\subset \{1,...,d\}$ with $|S|=s$. If $d\geq s+2$, then 
	$$\mathrm{E}\left[\sup_{v\in \mathcal{C}_1(S,\vartheta) \cap B_d(0,1)} \left|\frac{1}{n}\sum_{i=1}^n \Big\{(X_i^T v)^2 - \mathrm{E}\left[(X_i^T v)^2 \right]\Big\}\right|\right] \lesssim (2+\vartheta)^2 \phi_s \left[\sqrt{\frac{s\log(ed/s)}{n}} + \frac{s\log(ed/s)}{n} \right]$$
	and with probability at least $1-\delta$,
	\begin{align*}
		&\sup_{v\in \mathcal{C}_1(S,\vartheta) \cap B_d(0,1)} \left|\frac{1}{n}\sum_{i=1}^n \Big\{(X_i^T v)^2 - \mathrm{E}\left[(X_i^T v)^2 \right]\Big\}\right| \\
		&\lesssim (2+\vartheta)^2 \phi_s \left[\sqrt{\frac{\log(1/\delta) + s\log(ed/s)}{n}} + \frac{\log(1/\delta)  + s\log(ed/s)}{n} \right].
	\end{align*}
	In addition, these bounds also apply to
	$$\sup_{v\in \mathcal{C}_1(S,\vartheta) \cap B_d(0,1)} \left|\frac{1}{n}\sum_{i=1}^n \Big\{R_i(X_i^T v)^2 - \mathrm{E}\left[R_i(X_i^T v)^2 \right]\Big\}\right|$$
	for the complete-case analysis and 
	$$\sup_{v\in \mathcal{C}_1(S,\vartheta) \cap B_d(0,1)} \left|\frac{1}{n}\sum_{i=1}^n \left\{\frac{R_i}{\pi_i}(X_i^T v)^2 - \mathrm{E}\left[\frac{R_i}{\pi_i}(X_i^T v)^2 \right]\right\}\right|$$
	for the IPW analysis with the true propensity scores $\pi_i = \mathrm{P}(R_i=1|X_i), i=1,...,n$ under Assumption~\ref{assump:basic}.
	Finally, the above asymptotic upper bound remains almost identical for 
	\begin{align*}
		&\sup_{\substack{u \in \mathcal{C}_2(S_1,\vartheta_2) \cap \mathbb{S}^{d-1}\\v\in \mathcal{C}_1(S_2,\vartheta_1) \cap \mathbb{S}^{d-1}}} \left|\frac{1}{n}\sum_{i=1}^n \Big\{(X_i^T u)(X_i^T v) - \mathrm{E}\left[(X_i^Tu)(X_i^T v) \right]\Big\}\right|\\
		&\lesssim (2+\vartheta_1) (2+\vartheta_2) \sqrt{\phi_{s_1}\phi_{s_2}}\\
		&\quad \times  \left[\sqrt{\frac{\log(1/\delta) + s_1\log(ed/s_1) + s_2\log(ed/s_2)}{n}} + \frac{\log(1/\delta) + s_1\log(ed/s_1) + s_2\log(ed/s_2)}{n} \right]
	\end{align*}
	with $|S_1|=s_1$ and $|S_2|=s_2$.
\end{lemma}

\begin{proof}[Proof of Lemma~\ref{lem:gram_mat_conc}]
	By Lemma~\ref{lem:size_dom_cone}, whenever $d\geq s+2$, there exists a discrete set $M_d\subset B_d(0,1)$ with cardinality $|M_d| \leq \frac{3}{2} \left(\frac{5ed}{s} \right)^s$ such that $\mathcal{C}_1(S,\vartheta) \cap B_d(0,1) \subset 2(2+\vartheta)\, \mathrm{conv}(M_d)$ and $\norm{u}_0 \leq s$ for all $u\in M_d$. For any $t>0$, we consider the probability
	\begin{align}
		\label{RE_bound}
		\begin{split}
			&\mathrm{P}\left(\sup_{v\in \mathcal{C}_1(S,\vartheta) \cap B_d(0,1)}\left|\frac{1}{n}\sum_{i=1}^n \Big\{(X_i^T v)^2 - \mathrm{E}\left[(X_i^T v)^2 \right]\Big\}\right| > t \right) \\
			&\leq \mathrm{P}\left(\sup_{u,v\in 2(2+\vartheta) \mathrm{conv}(M_d)} \left|\frac{1}{n} \sum_{i=1}^n v^T\left[X_iX_i^T - \mathrm{E}\left(X_iX_i^T \right) \right]u\right| > t \right).
		\end{split}
	\end{align}
	Notice that $\frac{1}{n} \sum_{i=1}^n v^T\left[X_iX_i^T - \mathrm{E}\left(X_iX_i^T \right) \right]u$ is an affine function of $u$ (or $v$) for any fixed $v\in 2(2+\vartheta) M_d$ (or $u\in 2(2+\vartheta) M_d$) and is thus biconvex. By Lemma~\ref{lem:max_biconvex}, we know that the supremum of a biconvex function on $\mathrm{conv}(M_d)$ is always attained at its vertices, which are contained in $M_d$. Hence, we continue from \eqref{RE_bound} using the union bound as:
	\begin{align*}
		&\mathrm{P}\left(\sup_{u,v\in 2(2+\vartheta)M_d} \left|\frac{1}{n} \sum_{i=1}^n \Big\{(X_i^T v)(X_i^T u) - \mathrm{E}\left[(X_i^T v)(X_i^T u) \right] \Big\}\right| > t \right)\\
		&\leq |M_d|^2 \sup_{u,v\in 2(2+\vartheta)M_d}\mathrm{P}\left(\left|\frac{1}{n} \sum_{i=1}^n \Big\{(X_i^T v)(X_i^T u) - \mathrm{E}\left[(X_i^T v)(X_i^T u) \right] \Big\}\right| > t \right)\\
		&\leq \frac{9}{4} \left(\frac{5ed}{s} \right)^{2s} \sup_{u,v\in 2(2+\vartheta)M_d}\mathrm{P}\left(\left|\frac{1}{n} \sum_{i=1}^n \Big\{(X_i^T v)(X_i^T u) - \mathrm{E}\left[(X_i^T v)(X_i^T u) \right] \Big\}\right| > t \right).
	\end{align*}
	Now, we note that under Assumption~\ref{assump:sub_gaussian}, $(X_i^T v)(X_i^T u) - \mathrm{E}\left[(X_i^T v)(X_i^T u) \right], i=1,...,n$ are independent, centered, and sub-exponential random variables for all $u,v \in 2(2+\vartheta)M_d$, and
	\begin{align*}
		&\max_{1\leq i\leq n} \norm{v^T X_iX_i^T u - v^T \mathrm{E}\left(X_iX_i^T\right) u}_{\psi_1} \\
		&\leq \max_{1\leq i \leq n}\left(\norm{v^TX_iX_i^T u}_{\psi_1} + \norm{v^T \mathrm{E}\left(X_iX_i^T\right) u}_{\psi_1} \right)\\
		&\lesssim \max_{1\leq i\leq n} \left(\norm{X_i^Tv}_{\psi_2} \norm{X_i^Tu}_{\psi_2} + \norm{X_i^Tv}_{\psi_2} \norm{X_i^T u}_{\psi_2} \right)\\
		&\lesssim (2+\vartheta)^2 \phi_s,
	\end{align*}
	where we use Lemma~\ref{lem:Orlicz_prod} in the second (asymptotic) inequality. Therefore, by Bernstein's inequality (Corollary 2.8.3 in \citealt{vershynin2018high}), we conclude that there exists an absolute constant $A_0 >0$ such that for any $t>0$,
	\begin{align*}
		&\mathrm{P}\left(\sup_{v\in \mathcal{C}_1(S,\vartheta) \cap B_d(0,1)}\left|\frac{1}{n}\sum_{i=1}^n \Big\{(X_i^T v)^2 - \mathrm{E}\left[(X_i^T v)^2 \right]\Big\}\right| > t \right) \\
		&\leq \frac{9}{2} \left(\frac{5ed}{s} \right)^{2s} \exp\left(-nA_0\min\left\{\frac{t^2}{(2+\vartheta)^4 \phi_s^2},\, \frac{t}{(2+\vartheta)^2 \phi_s} \right\} \right).
	\end{align*}
	Hence, for any $\delta \in (0,1)$, we have that with probability at least $1-\delta$,
	\begin{align*}
		&\sup_{v\in \mathcal{C}_1(S,\vartheta) \cap B_d(0,1)}\left|\frac{1}{n}\sum_{i=1}^n \Big\{(X_i^T v)^2 - \mathrm{E}\left[(X_i^T v)^2 \right]\Big\}\right| \\
		&\lesssim (2+\vartheta)^2 \phi_s \left[\sqrt{\frac{\log(1/\delta) + s\log(ed/s)}{n}} + \frac{\log(1/\delta)  + s\log(ed/s)}{n} \right].
	\end{align*}
	Finally, we bound the expectation of the above supremum using Lemma~\ref{lem:expect_bound} as:
	\begin{align*}
		&\mathrm{E}\left[\sup_{v\in \mathcal{C}_1(S,\vartheta) \cap B_d(0,1)}\left|\frac{1}{n}\sum_{i=1}^n \Big\{(X_i^T v)^2 - \mathrm{E}\left[(X_i^T v)^2 \right]\Big\}\right| \right] \\
		&\lesssim (2+\vartheta)^2 \phi_s \left[\sqrt{\frac{s\log(ed/s)}{n}} + \frac{s\log(ed/s)}{n} \right],
	\end{align*}
	where we set $a=\frac{9}{2} \left(\frac{5ed}{s} \right)^{2s} \geq \sqrt{e}$, $b=nC_0$, and $K=(2+\vartheta)^2 \phi_s$ in Lemma~\ref{lem:expect_bound}.\\
	
	As for the complete-case analysis, one can follow the same proof and notice that $R_iX_i, i=1,...,n$ are also sub-Gaussian under Assumption~\ref{assump:sub_gaussian} and for all $u,v \in 2(2+\vartheta)M_d$,
	\begin{align*}
		&\max_{1\leq i\leq n}\norm{R_i(X_i^T v)(X_i^T u) - \mathrm{E}\left[R_i(X_i^T v)(X_i^T u) \right]}_{\psi_1}\\
		&\leq \max_{1\leq i\leq n} \left(\norm{v^T X_iR_iX_i^T u}_{\psi_1} + \norm{v^T \mathrm{E}\left(X_iR_iX_i^T\right) u}_{\psi_1} \right)\\
		&\lesssim \max_{1\leq i\leq n} \left(\norm{X_i^T v}_{\psi_2}\norm{X_i^T u}_{\psi_2} + \norm{X_i^Tv}_{\psi_2}\norm{X_i^T u}_{\psi_2} \right)\\
		&\lesssim (2+\vartheta)^2 \phi_s.
	\end{align*}
	As for the IPW analysis under the true propensity scores $\pi_i\equiv\pi(X_i) = \mathrm{P}(R_i=1|X_i), i=1,...,n$, we can also calculate that 
	\begin{align*}
		&\max_{1\leq i\leq n}\norm{\frac{R_i}{\pi(X_i)}(X_i^T v)(X_i^T u) - \mathrm{E}\left[\frac{R_i}{\pi(X_i)}(X_i^T v)(X_i^T u) \right]}_{\psi_1}\\
		&\leq \max_{1\leq i\leq n} \left(\norm{v^T X_i\cdot \frac{R_i}{\pi(X_i)} \cdot X_i^T u}_{\psi_1} + \norm{v^T \mathrm{E}\left(X_i X_i^T\right) u}_{\psi_1} \right)\\
		&\lesssim \max_{1\leq i\leq n} \left(\norm{\frac{R_i}{\sqrt{\pi(X_i)}} \cdot X_i^T v}_{\psi_2}\norm{\frac{R_i}{\sqrt{\pi(X_i)}}\cdot X_i^T u}_{\psi_2} + \norm{X_i^Tv}_{\psi_2}\norm{X_i^T u}_{\psi_2} \right)\\
		&\lesssim \left[\mathrm{E}\left(v^T X_i \cdot \frac{R_i}{\pi(X_i)} \cdot X_i^T v\right)^{\frac{1}{2}} \right] \left[\mathrm{E}\left(u^T X_i \cdot \frac{R_i}{\pi(X_i)} \cdot X_i^T u\right)^{\frac{1}{2}} \right] + \left[\mathrm{E}\left(v^T X_i X_i^T v\right)^{\frac{1}{2}} \right] \left[\mathrm{E}\left(u^T X_i X_i^T u\right)^{\frac{1}{2}} \right]\\
		&\lesssim \left[\mathrm{E}\left(v^T X_i X_i^T v\right)^{\frac{1}{2}} \right] \left[\mathrm{E}\left(u^T X_i X_i^T u\right)^{\frac{1}{2}} \right]\\
		&\lesssim (2+\vartheta)^2 \phi_s,
	\end{align*}
	where we leverage the fact that propensity scores $\pi(X_i),i=1,...,n$ are nonnegative and apply the law of total expectation.
\end{proof}

\vspace{3mm}

\begin{lemma}
	\label{lem:dual_approx_cone}
	Under Definition~\ref{defn:popul_dual} and Assumption~\ref{assump:sparse_dual}, we have that
	$$\ell_0(x) \in \mathcal{C}_2\left(S_{\ell}(x),\vartheta_{\ell}\right)=\left\{v\in \mathbb{R}^d: \norm{v_{S_{\ell}(x)^c}}_2 \leq \vartheta_{\ell} \norm{v_{S_{\ell}(x)}}_2 \right\}\quad \text{ with } \quad \vartheta_{\ell}=\frac{r_{\ell}}{\sqrt{1-r_{\ell}^2}}$$
	for any $r_{\ell}\in [0,1)$. Moreover, if $0\leq r_{\ell} \leq \frac{1}{2}$, then $$\norm{M^T\left(\tilde{\ell}(x)-\ell_0(x)\right)}_2 \leq \sup_{v\in \mathcal{C}_2(S_{\ell}(x),1) \cap B_d\left(0,10r_{\ell}\norm{\ell_0(x)}_2\right)} \norm{M^T v}_2$$ 
	for any matrix $M\in \mathbb{R}^{d\times k}$.
\end{lemma}

\begin{proof}[Proof of Lemma~\ref{lem:dual_approx_cone}]
	The proof is modified from Lemma 24 in \cite{giessing2021inference}. By Assumption~\ref{assump:sparse_dual}, we know that
	\begin{align*}
		\norm{\tilde{\ell}(x) -\ell_0(x)}_2^2 &= \sum_{k\in S_{\ell}(x)^c} \left[\ell_0(x)\right]_k^2 \leq r_{\ell}^2\norm{\ell_0(x)}_2^2 = r_{\ell}^2 \left(\sum_{k\in S_{\ell}(x)} \left[\ell_0(x)\right]_k^2 + \sum_{k\in S_{\ell}(x)^c} \left[\ell_0(x)\right]_k^2 \right).
	\end{align*}
	Some rearrangements show that $\sum_{k\in S_{\ell}(x)^c} \left[\ell_0(x)\right]_k^2 \leq \frac{r_{\ell}^2}{1-r_{\ell}^2} \sum_{k\in S_{\ell}(x)} \left[\ell_0(x)\right]_k^2$, that is,
	$$\norm{\left[\ell_0(x)\right]_{S_{\ell}(x)^c}}_2 \leq \frac{r_{\ell}}{\sqrt{1-r_{\ell}^2}} \norm{\left[\ell_0(x)\right]_{S_{\ell}(x)}}_2.$$
	Let $S_{\ell} := S_{\ell}(x)$ in the sequel. We know that
	$$\ell_0(x) -\tilde{\ell}(x) = \left(0_{S_{\ell}}, \left[\ell_0(x)\right]_{S_{\ell}^c}\right) \quad \text{ with }\quad \ell_0(x) \in \mathcal{C}_2\left(S_{\ell}(x), \vartheta_{\ell}\right).$$
	In addition, for any $v\in \mathcal{C}_2\left(S_{\ell}(x), \vartheta_{\ell}\right)$, we have that 
	$$u=t\left(v_{S_{\ell}}, 0_{S_{\ell}^c}\right) + \left(0_{S_{\ell}}, v_{S_{\ell}^c}\right) \in \mathcal{C}_2\left(S_{\ell}(x), 1\right) \quad \text{ when }\quad |t| \geq \vartheta_{\ell}=\frac{r_{\ell}}{\sqrt{1-r_{\ell}^2}},$$
	because
	$$\norm{u_{S_{\ell}^c}}_2 = \norm{v_{S_{\ell}^c}}_2 \leq \vartheta_{\ell} \norm{v_{S_{\ell}}}_2 \leq |t|\cdot \norm{v_{S_{\ell}}}_2 = \norm{u_{S_{\ell}}}_2.$$
	As a result, we obtain that
	\begin{align*}
		\frac{1}{1+r_{\ell}} \cdot\tilde{\ell}(x) - \frac{1}{1-r_{\ell}} \cdot\ell_0(x) &= -\frac{2r_{\ell}}{1-r_{\ell}^2} \cdot \tilde{\ell}(x) - \frac{1}{1-r_{\ell}}\left[\ell_0(x) -\tilde{\ell}(x)\right]\\
		&= -\frac{2r_{\ell}}{1-r_{\ell}^2} \left(\left[\ell_0(x)\right]_{S_{\ell}}, 0_{S_{\ell}^c}\right) - \frac{1}{1-r_{\ell}}\left(0_{S_{\ell}}, \left[\ell_0(x)\right]_{S_{\ell}^c}\right)\\
		&\in \mathcal{C}_2\left(S_{\ell}(x), 1\right),
	\end{align*}
	since $\frac{2r_{\ell}}{1-r_{\ell}^2} (1-r_{\ell}) \geq \frac{r_{\ell}}{\sqrt{1-r_{\ell}^2}}$ when $0\leq r_{\ell} \leq \frac{3}{5}$. Similarly, 
	\begin{align*}
		\frac{1}{1+2r_{\ell}} \cdot\tilde{\ell}(x) - \ell_0(x) &= -\frac{2r_{\ell}}{1+2r_{\ell}} \cdot \tilde{\ell}(x) - \left[\ell_0(x) -\tilde{\ell}(x)\right]\\
		&= -\frac{2r_{\ell}}{1+2r_{\ell}} \left(\left[\ell_0(x)\right]_{S_{\ell}}, 0_{S_{\ell}^c}\right) - \left(0_{S_{\ell}}, \left[\ell_0(x)\right]_{S_{\ell}^c}\right)\\
		&\in \mathcal{C}_2\left(S_{\ell}(x), 1\right),
	\end{align*}
	since $\frac{2r_{\ell}}{1+2r_{\ell}} \geq \frac{r_{\ell}}{\sqrt{1-r_{\ell}^2}}$ when $0\leq r_{\ell} \leq \frac{1}{2}$. Now, we notice that
	\begin{equation}
		\label{M_ind}
		\norm{M^T\left[\ell_0(x)-\tilde{\ell}(x)\right]}_2^2 = \norm{M^T\ell_0(x)}_2^2 + \norm{M^T \tilde{\ell}(x)}_2^2 - 2\left[M^T\tilde{\ell}(x)\right]^T\left[M^T\ell_0(x)\right].
	\end{equation}
	\emph{Case I:} If $\tilde{\ell}(x)^T MM^T \ell_0(x) \leq 0$, then the right-hand side of \eqref{M_ind}  can be upper bounded by
	\begin{align*}
		&\norm{M^T\ell_0(x)}_2^2 + \norm{M^T \tilde{\ell}(x)}_2^2 - 2\left[M^T\tilde{\ell}(x)\right]^T\left[M^T\ell_0(x)\right] \\
		&\stackrel{\text{(i)}}{\leq} (1+r_{\ell})^2\left[ \frac{1}{(1+r_{\ell})^2}\norm{M^T\tilde{\ell}(x)}_2^2 + \frac{1}{(1-r_{\ell})^2} \norm{M^T\ell_0(x)}_2^2 - \frac{2}{1-r_{\ell}^2} \cdot \tilde{\ell}(x)^T MM^T \ell_0(x)\right]\\
		&= (1+r_{\ell})^2 \left[\frac{1}{1+r_{\ell}} \cdot \tilde{\ell}(x) - \frac{1}{1-r_{\ell}} \cdot \ell_0(x)\right]^T MM^T \left[\frac{1}{1+r_{\ell}} \cdot \tilde{\ell}(x) - \frac{1}{1-r_{\ell}} \cdot \ell_0(x)\right]\\
		&\stackrel{\text{(ii)}}{\leq} (1+r_{\ell})^2 \sup_{v\in \mathcal{C}_2(S_{\ell}(x),1) \cap B_d\left(0,5r_{\ell}\norm{\ell_0(x)}_2\right)} \norm{M^Tv}_2^2\\
		&\leq \sup_{v\in \mathcal{C}_2(S_{\ell}(x),1) \cap B_d\left(0,10r_{\ell}\norm{\ell_0(x)}_2\right)} \norm{M^Tv}_2^2,
	\end{align*}
	where the inequality (i) follows from the facts that $(1+r_{\ell})^2 \geq 1, (1+r_{\ell})^2 \geq 1-r_{\ell}^2$, and we compute that
	\begin{align*}
		\norm{\frac{1}{1+r_{\ell}} \cdot \tilde{\ell}(x) - \frac{1}{1-r_{\ell}} \cdot \ell_0(x)}_2 &= \norm{-\frac{2r_{\ell}}{1-r_{\ell}^2} \cdot \tilde{\ell}(x) - \frac{1}{1-r_{\ell}}\left[\ell_0(x) -\tilde{\ell}(x)\right]}_2\\
		&\leq r_{\ell}\left(\frac{2}{1-r_{\ell}^2} + \frac{1}{1-r_{\ell}} \right)\norm{\ell_0(x)}_2\\
		&\leq 5r_{\ell} \norm{\ell_0(x)}_2
	\end{align*}
	for any $0\leq r_{\ell} \leq \frac{1}{2}$ to obtain the inequality (ii).\\
	\emph{Case II:} If $\tilde{\ell}(x)^T MM^T \ell_0(x) > 0$, then the right-hand side of \eqref{M_ind}  can be upper bounded by
	\begin{align*}
		&\norm{M^T\ell_0(x)}_2^2 + \norm{M^T \tilde{\ell}(x)}_2^2 - 2\left[M^T\tilde{\ell}(x)\right]^T\left[M^T\ell_0(x)\right] \\
		&\stackrel{\text{(iii)}}{\leq} (1+2r_{\ell})^2\left[ \frac{1}{(1+2r_{\ell})^2}\norm{M^T\tilde{\ell}(x)}_2^2 + \norm{M^T\ell_0(x)}_2^2 - \frac{2}{1+2r_{\ell}} \cdot \tilde{\ell}(x)^T MM^T \ell_0(x)\right]\\
		&= (1+2r_{\ell})^2 \left[\frac{1}{1+2r_{\ell}} \cdot \tilde{\ell}(x) -  \ell_0(x)\right]^T MM^T \left[\frac{1}{1+2r_{\ell}} \cdot \tilde{\ell}(x) - \ell_0(x)\right]\\
		&\stackrel{\text{(iv)}}{\leq} (1+2r_{\ell})^2 \sup_{v\in \mathcal{C}_2(S_{\ell}(x),1) \cap B_d\left(0,3r_{\ell}\norm{\ell_0(x)}_2\right)} \norm{M^Tv}_2^2\\
		&\leq \sup_{v\in \mathcal{C}_2(S_{\ell}(x),1) \cap B_d\left(0,6r_{\ell}\norm{\ell_0(x)}_2\right)} \norm{M^Tv}_2^2,
	\end{align*}
	where the inequality (iii) follows from 
	\begin{align*}
		\left[(1+2r_{\ell})^2 - 1\right] \norm{M^T\ell_0(x)}_2^2 &= (4r_{\ell}^2 + 4r_{\ell}) \norm{M^T\ell_0(x)}_2^2\\
		&\geq 4r_{\ell} \left[M^T\tilde{\ell}(x)\right]^T\left[M^T\ell_0(x)\right]
	\end{align*}
		by the Cauchy-Schwarz inequality, and we calculate that
	\begin{align*}
		\norm{\frac{1}{1+2r_{\ell}} \cdot \tilde{\ell}(x) - \ell_0(x)}_2 &= \norm{-\frac{2r_{\ell}}{1+2r_{\ell}} \cdot \tilde{\ell}(x) - \left[\ell_0(x) -\tilde{\ell}(x)\right]}_2\\
		&\leq r_{\ell}\left(\frac{2}{1+2r_{\ell}} + 1\right) \norm{\ell_0(x)}_2\\
		&\leq 3r_{\ell} \norm{\ell_0(x)}_2
	\end{align*}
	to obtain the inequality (iv). Therefore, we combine the above two cases with \eqref{M_ind} to conclude that 
	$$\norm{M^T\left(\tilde{\ell}(x)-\ell_0(x)\right)}_2 \leq \sup_{v\in \mathcal{C}_2(S_{\ell}(x),1) \cap B_d\left(0,10r_{\ell}\norm{\ell_0(x)}_2\right)} \norm{M^T v}_2$$ 
	for any matrix $M\in \mathbb{R}^{d\times k}$ and $0\leq r_{\ell} \leq \frac{1}{2}$.
\end{proof}

\vspace{3mm}

\begin{lemma}[Restricted cone and cross polytope property]
	\label{lem:res_cross_poly}
	Let $a_0 >1$ be some absolute constant. Under Assumptions~\ref{assump:gram_lower_bnd} and \ref{assump:sparse_dual}, if 
	\begin{equation}
		\label{res_cross_poly_cond}
		\norm{\frac{1}{2n}\sum_{i=1}^n \hat{\pi}_i\cdot \ell_0(x)^T X_i X_i +x}_{\infty} \leq \frac{\gamma}{a_0 n} \quad \text{ and } \quad \sup_{v\in \mathcal{C}_2\left(S_{\ell}(x), \vartheta_{\ell}\right)\cap B_d(0,1)} \left|\frac{1}{n} \sum_{i=1}^n \hat{\pi}_i\left[X_i^Tv\right]^2 \right| \leq \frac{\nu}{a_0 n}
	\end{equation}
	with $\vartheta_{\ell} = \frac{r_{\ell}}{\sqrt{1-r_{\ell}^2}}$, then
	$$\sum_{k\in S_{\ell}(x)^c} \left|\hat{\ell}_k(x) \right| \leq \frac{4r_{\ell}^2 \norm{x}_2^2}{(a_0-1) \kappa_R^4}\cdot \frac{\nu}{\gamma} + \left(\frac{a_0+1}{a_0-1}\right) \sum_{k\in S_{\ell}(x)^c} \left|\hat{\ell}_k(x) -\tilde{\ell}_k(x) \right|$$
	so that 
	$$\hat{\ell}(x)-\tilde{\ell}(x) \in \mathcal{C}_1\left(S_{\ell}(x), 2\left(\frac{a_0+1}{a_0-1}\right) \right) \bigcup B_d^{(1)}\left(0,\, \frac{(a_0^2+1)r_{\ell}^2 \norm{x}_2^2}{(a_0-1) \kappa_R^4} \cdot \frac{\nu}{\gamma}\right),$$ where $B_d^{(1)}(0,r) = \left\{v\in \mathbb{R}^d: \norm{v}_1 \leq r\right\}$.
\end{lemma}

\begin{remark}
	The same results hold true when the estimated propensity scores $\hat{\pi}_i,i=1,...,n$ are replaced by the true ones $\pi_i,i=1,...,n$.
\end{remark}

\begin{proof}[Proof of Lemma~\ref{lem:res_cross_poly}]
	The proof is modified from the original arguments of Lemma 24 in \cite{giessing2021inference}. Since $\hat{\ell}(x)$ minimizes the objective function of the dual problem \eqref{debias_prog_dual}, we know that
	$$\frac{1}{4n} \sum_{i=1}^n \hat{\pi}_i\left[X_i^T \hat{\ell}(x)\right]^2 + x^T\hat{\ell}(x) + \frac{\gamma}{n}\norm{\hat{\ell}(x)}_1 \leq \frac{1}{4n} \sum_{i=1}^n \hat{\pi}_i\left[X_i^T \tilde{\ell}(x)\right]^2 + x^T\tilde{\ell}(x) + \frac{\gamma}{n}\norm{\tilde{\ell}(x)}_1.$$
	Taking Taylor's expansion about $\ell_0(x)$ on both sides of the above inequality gives us that
	\begin{align*}
		&\frac{1}{4n} \sum_{i=1}^n \hat{\pi}_i \left[\hat{\ell}(x) -\ell_0(x)\right]^T X_iX_i^T \left[\hat{\ell}(x) -\ell_0(x)\right] \\
		&\quad + \left[\frac{1}{2n} \sum_{i=1}^n \hat{\pi}_i X_iX_i^T\ell_0(x) + x\right]^T \left[\hat{\ell}(x) -\ell_0(x)\right] + \frac{\gamma}{n} \norm{\hat{\ell}(x)}_1\\
		&\leq \frac{1}{4n} \sum_{i=1}^n \hat{\pi}_i \left[\tilde{\ell}(x) -\ell_0(x)\right]^T X_iX_i^T \left[\tilde{\ell}(x) -\ell_0(x)\right] \\
		&\quad + \left[\frac{1}{2n} \sum_{i=1}^n \hat{\pi}_i X_iX_i^T\ell_0(x) + x\right]^T \left[\tilde{\ell}(x) -\ell_0(x)\right] + \frac{\gamma}{n} \norm{\tilde{\ell}(x)}_1.
	\end{align*}
	Some rearrangements show that
	\begin{align*}
		0&\leq \frac{1}{4n} \sum_{i=1}^n \hat{\pi}_i \left[\hat{\ell}(x) -\ell_0(x)\right]^T X_iX_i^T \left[\hat{\ell}(x) -\ell_0(x)\right]\\
		&\leq \frac{1}{4n} \sum_{i=1}^n \hat{\pi}_i \left[\tilde{\ell}(x) -\ell_0(x)\right]^T X_iX_i^T \left[\tilde{\ell}(x) -\ell_0(x)\right] + \left[\frac{1}{2n}\sum_{i=1}^n \hat{\pi}_i X_i X_i^T\ell_0(x) + x\right]^T\left[\tilde{\ell}(x) - \hat{\ell}(x)\right] \\
		&\quad + \frac{\gamma}{n} \norm{\tilde{\ell}(x)}_1 - \frac{\gamma}{n} \norm{\hat{\ell}(x)}_1\\
		&\stackrel{\text{(i)}}{\leq} \frac{r_{\ell}^2 \nu}{4a_0n} \norm{\ell_0(x)}_2^2 + \frac{\gamma}{a_0n} \norm{\hat{\ell}(x) - \tilde{\ell}(x)}_1 + \frac{\gamma}{n} \norm{\tilde{\ell}(x)}_1 - \frac{\gamma}{n} \norm{\hat{\ell}(x)}_1\\
		&\stackrel{\text{(ii)}}{\leq} \frac{r_{\ell}^2 \nu \norm{x}_2^2}{a_0\cdot n\cdot \kappa_R^4} + \frac{\gamma}{a_0n} \left[\norm{\left[\hat{\Delta}(x)\right]_{S_{\ell}(x)}}_1 + \norm{\left[\hat{\Delta}(x)\right]_{S_{\ell}(x)^c}}_1 \right] \\
		&\quad + \frac{\gamma}{n} \norm{\left[\tilde{\ell}(x)\right]_{S_{\ell}(x)}}_1 - \frac{\gamma}{n}\norm{\left[\hat{\Delta}(x)\right]_{S_{\ell}(x)} + \left[\tilde{\ell}(x)\right]_{S_{\ell}(x)}}_1 - \frac{\gamma}{n} \norm{\left[\hat{\Delta}(x)\right]_{S_{\ell}(x)^c}}_1\\
		&\stackrel{\text{(iii)}}{\leq} \frac{r_{\ell}^2 \nu \norm{x}_2^2}{a_0\cdot n\cdot \kappa_R^4} + \frac{\gamma}{a_0n} \left[\norm{\left[\hat{\Delta}(x)\right]_{S_{\ell}(x)}}_1 + \norm{\left[\hat{\Delta}(x)\right]_{S_{\ell}(x)^c}}_1 \right] \\
		&\quad + \frac{\gamma}{n} \norm{\left[\hat{\Delta}(x)\right]_{S_{\ell}(x)}}_1 - \frac{\gamma}{n} \norm{\left[\hat{\Delta}(x)\right]_{S_{\ell}(x)^c}}_1,
	\end{align*}
	where we utilize our condition \eqref{res_cross_poly_cond} to obtain the inequality (i), denote $\hat{\Delta}(x) = \hat{\ell}(x) -\tilde{\ell}(x)$ and apply Assumption~\ref{assump:gram_lower_bnd} in the inequality (ii), and leverage the reverse triangle's inequality in (iii). Hence,
	$$0\leq \frac{r_{\ell}^2 \nu \norm{x}_2^2}{a_0\cdot n\cdot \kappa_R^4} + \frac{(1+a_0)\gamma}{a_0 \cdot n} \norm{\left[\hat{\Delta}(x)\right]_{S_{\ell}(x)}}_1 - \frac{(a_0-1) \gamma}{a_0 \cdot n} \norm{\left[\hat{\Delta}(x)\right]_{S_{\ell}(x)^c}}_1,$$
	which implies that
	$$\norm{\left[\hat{\Delta}(x)\right]_{S_{\ell}(x)^c}}_1 \leq \frac{r_{\ell}^2 \norm{x}_2^2}{(a_0-1) \kappa_R^4}\cdot \frac{\nu}{\gamma} + \left(\frac{a_0+1}{a_0-1}\right) \norm{\left[\hat{\Delta}(x)\right]_{S_{\ell}(x)}}_1.$$
	On the one hand, if $\frac{r_{\ell}^2 \norm{x}_2^2}{(a_0-1) \kappa_R^4}\cdot \frac{\nu}{\gamma} \leq \left(\frac{a_0+1}{a_0-1}\right) \sum_{k\in S_{\ell}(x)} \left|\hat{\ell}_k(x) -\tilde{\ell}_k(x) \right|$, then
	$$\norm{\left[\hat{\Delta}(x)\right]_{S_{\ell}(x)^c}}_1 \leq 2\left(\frac{a_0+1}{a_0-1}\right) \norm{\left[\hat{\Delta}(x)\right]_{S_{\ell}(x)}}_1 \quad \text{ so that }\quad \hat{\ell}(x)-\tilde{\ell}(x) \in \mathcal{C}_1\left(S_{\ell}(x), 2\left(\frac{a_0+1}{a_0-1}\right)\right).$$
	On the other hand, if $\frac{r_{\ell}^2 \norm{x}_2^2}{(a_0-1) \kappa_R^4}\cdot \frac{\nu}{\gamma} > \left(\frac{a_0+1}{a_0-1}\right) \sum_{k\in S_{\ell}(x)} \left|\hat{\ell}_k(x) -\tilde{\ell}_k(x) \right|$, then
	$$\norm{\hat{\ell}(x) - \tilde{\ell}(x)}_1 \leq \frac{(a_0^2+1)r_{\ell}^2 \norm{x}_2^2}{(a_0-1) \kappa_R^4} \cdot \frac{\nu}{\gamma}.$$
	Thus, we conclude that
	$$\hat{\ell}(x)-\tilde{\ell}(x) \in \mathcal{C}_1\left(S_{\ell}(x), 2\left(\frac{a_0+1}{a_0-1}\right) \right) \bigcup B_d^{(1)}\left(0,\, \frac{(a_0^2+1)r_{\ell}^2 \norm{x}_2^2}{(a_0-1) \kappa_R^4} \cdot \frac{\nu}{\gamma}\right).$$
	The proof is completed.
\end{proof}

\vspace{3mm}

\begin{lemma}
	\label{lem:cone_ball_bnd}
	Let $\vartheta,r_0,r_1 \in (0,\infty)$ and $\delta \in (0,1)$ be arbitrary. Suppose that Assumption~\ref{assump:sub_gaussian} holds. Let $S,\tilde{S}\subset \{1,...,d\}$ with $|S|=s$ and $|\tilde{S}|=\tilde{s}$. If $d\geq s+2$, then with probability at least $1-\delta$,
	\begin{align*}
		&\sup_{v\in \left[\mathcal{C}_1(S,\vartheta) \bigcup B_d^{(1)}(0,r_1) \right] \bigcap B_d(0,r_0)} \left|\frac{1}{n} \sum_{i=1}^n \left\{(X_i^Tv)^2 - \mathrm{E}\left[(X_i^Tv)^2\right] \right\}\right| \\
		&\lesssim r_0^2 (2+\vartheta)^2 \phi_s\left[\sqrt{\frac{\log(1/\delta) + s\log(ed/s)}{n}} + \frac{\log(1/\delta)  + s\log(ed/s)}{n} \right] \\
		&\quad + r_1^2\phi_1 \left[\sqrt{\frac{\log(d/\delta)}{n}} + \frac{\log(d/\delta)}{n}\right].
	\end{align*}
	In particular, when $r_0\geq r_1$, we further have that with probability at least $1-\delta$,
	\begin{align*}
		&\sup_{v\in \left[\mathcal{C}_1(S,\vartheta) \bigcup B_d^{(1)}(0,r_1) \right] \bigcap B_d(0,r_0)} \left|\frac{1}{n} \sum_{i=1}^n \left\{(X_i^Tv)^2 - \mathrm{E}\left[(X_i^Tv)^2\right] \right\}\right| \\
		&\lesssim r_0^2 (2+\vartheta)^2 \phi_s\left[\sqrt{\frac{\log(1/\delta) + s\log(ed/s)}{n}} + \frac{\log(1/\delta)  + s\log(ed/s)}{n} \right].
	\end{align*}
	Similarly, if $d\geq \max\{s,\tilde{s}\} +2$, then with probability at least $1-\delta$,
	\begin{align*}
		&\sup_{\substack{u \in \mathcal{C}_2(\tilde{S},\vartheta) \cap B_d(0,1)\\v\in \left[\mathcal{C}_1(S,\vartheta) \bigcup B_d^{(1)}(0,r_1) \right] \bigcap B_d(0,r_0)}} \frac{1}{n} \sum_{i=1}^n \left|X_i^Tu \right| \left|X_i^Tv \right| \\
		&\lesssim \left[r_0(2+\vartheta)^2+r_1\right] \sqrt{\phi_{\tilde{s}} \cdot \phi_1} \left[\sqrt{\frac{\log(1/\delta) + s\log(ed/s)}{n}} + \frac{\log(1/\delta)  + s\log(ed/s)}{n} + 1\right].
	\end{align*}
\end{lemma}

\begin{remark}
	\label{remark:cone_ball_union}
	Notice that the structure of $\mathcal{C}_1(S,\vartheta) \cup B_d^{(1)}(0,r_1)$ is no more complicated than the cone of dominant coordinates $\mathcal{C}_1(S,\vartheta)$. The supremum of a (bi)convex function over the $L_1$-ball $B_d^{(1)}(0,r_1)$ is attained at its vertex set $r_1\cdot \{\pm e_1,...,\pm e_d\}$, where $e_1,...,e_d$ are the standard basis vectors in $\mathbb{R}^d$, leading to a simpler discretization than the one for $\mathcal{C}_1(S,\vartheta)$ in Lemma~\ref{lem:size_dom_cone}. Hence, the probabilistic bounds for taking the supremum over $\mathcal{C}_1(S,\vartheta) \cup B_d^{(1)}(0,r_1)$ in Lemma~\ref{lem:cone_ball_bnd} are dominated by the probabilistic bounds for taking the supremum over a cone of dominant coordinates $\mathcal{C}_1(S,\vartheta)$ in Lemma~\ref{lem:gram_mat_conc}.
\end{remark}

\begin{proof}[Proof of Lemma~\ref{lem:cone_ball_bnd}]
	We first notice that
	\begin{align}
		\label{cone_union_ball1}
		\begin{split}
			&\sup_{v\in \left[\mathcal{C}_1(S,\vartheta) \bigcup B_d^{(1)}(0,r_1) \right] \bigcap B_d(0,r_0)} \left|\frac{1}{n} \sum_{i=1}^n \left\{(X_i^Tv)^2 - \mathrm{E}\left[(X_i^Tv)^2\right] \right\}\right| \\
			&\leq \underbrace{\sup_{v\in \mathcal{C}_1(S,\vartheta) \cap B_d(0,r_0)} \left|\frac{1}{n} \sum_{i=1}^n \left\{(X_i^Tv)^2 - \mathrm{E}\left[(X_i^Tv)^2\right] \right\}\right|}_{\textbf{Term I}} \\
			&\quad + \underbrace{\sup_{v\in  B_d^{(1)}(0,r_1) \cap B_d(0,r_0)} \left|\frac{1}{n} \sum_{i=1}^n \left\{(X_i^Tv)^2 - \mathrm{E}\left[(X_i^Tv)^2\right] \right\}\right|}_{\textbf{Term II}}.
		\end{split}
	\end{align}
	
	$\bullet$ {\bf Term I:} We know from Lemma~\ref{lem:gram_mat_conc} that
	\begin{align*}
		&\sup_{v\in \mathcal{C}_1(S,\vartheta) \cap B_d(0,r_0)} \left|\frac{1}{n} \sum_{i=1}^n \left\{(X_i^Tv)^2 - \mathrm{E}\left[(X_i^Tv)^2\right] \right\}\right| \\
		&\leq r_0^2 \sup_{v\in \mathcal{C}_1(S,\vartheta) \cap B_d(0,1)} \left|\frac{1}{n} \sum_{i=1}^n \left\{(X_i^Tv)^2 - \mathrm{E}\left[(X_i^Tv)^2\right] \right\}\right| \\
		&\lesssim r_0^2 (2+\vartheta)^2 \phi_s\left[\sqrt{\frac{\log(1/\delta) + s\log(ed/s)}{n}} + \frac{\log(1/\delta)  + s\log(ed/s)}{n} \right].
	\end{align*}
	
	$\bullet$ {\bf Term II:} Note that
	\begin{align}
		\label{L1_ball_term1}
		\nonumber&\sup_{v\in  B_d^{(1)}(0,r_1) \cap B_d(0,r_0)} \left|\frac{1}{n} \sum_{i=1}^n \left\{(X_i^Tv)^2 - \mathrm{E}\left[(X_i^Tv)^2\right] \right\}\right| \\ \nonumber
		&\leq \sup_{u,v\in  B_d^{(1)}(0,r_1) \cap B_d(0,r_0)} \left|\frac{1}{n} \sum_{i=1}^n \left\{(X_i^Tu) (X_i^Tv) - \mathrm{E}\left[(X_i^Tu) (X_i^Tv)\right] \right\}\right|\\ 
		&\leq \sup_{u,v \in B_d^{(1)}(0,r_1)} \left|\frac{1}{n} \sum_{i=1}^n \left\{(X_i^Tu) (X_i^Tv) - \mathrm{E}\left[(X_i^Tu) (X_i^Tv)\right] \right\}\right|.
	\end{align}
	Since $\frac{1}{n} \sum_{i=1}^n \left\{(X_i^Tu) (X_i^Tv) - \mathrm{E}\left[(X_i^Tu) (X_i^Tv)\right] \right\}$ is a biconvex function with respect to $u,v$, we know from Lemma~\ref{lem:max_biconvex} that the supremum is attained at the vertex set 
	$$\left\{x\in \mathbb{R}^d: \norm{x}_0=1, \norm{x}_1=\norm{x}_2=r_1 \right\} = r_1 \cdot \left\{\pm e_1,...,\pm e_d\right\}$$ 
	of the $L_1$-ball $B_d^{(1)}(0,r_1)$, where $e_1,...,e_d$ are the standard basis vectors in $\mathbb{R}^d$. We further upper bound \eqref{L1_ball_term1} as follows:  
	\begin{align*}
		&\sup_{v\in  B_d^{(1)}(0,r_1)} \left|\frac{1}{n} \sum_{i=1}^n \left\{(X_i^Tv)^2 - \mathrm{E}\left[(X_i^Tv)^2\right] \right\}\right| \\
		&\leq r_1^2 \sup_{u,v \in \{e_1,...,e_d\}} \left|\frac{1}{n} \sum_{i=1}^n \left\{(X_i^Tu) (X_i^Tv) - \mathrm{E}\left[(X_i^Tu) (X_i^Tv)\right] \right\}\right|\\
		&= r_1^2 \max_{1\leq j, k \leq d} \left|\frac{1}{n} \sum_{i=1}^n \left[ X_{ij}X_{ik} - \mathrm{E}\left(X_{ij}X_{ik}\right)\right]\right|.
	\end{align*}
	Under Assumption~\ref{assump:sub_gaussian}, we note that $X_{ij}X_{ik} - \mathrm{E}\left(X_{ij}X_{ik}\right), i=1,...,n$ are independent, centered, and sub-exponential random variables for all $j,k=1,...,d$, and
	\begin{align*}
		\max_{1\leq j, k \leq d}\norm{X_{ij}X_{ik} - \mathrm{E}\left(X_{ij}X_{ik}\right)}_{\psi_1} &\leq \max_{1\leq j, k \leq d}\left[\norm{X_{ij} X_{ik}}_{\psi_1} + \norm{\mathrm{E}\left(X_{ij} X_{ik}\right)}_{\psi_1}\right] \\
		&\lesssim \max_{1\leq j, k \leq d}\norm{X_{ij}}_{\psi_2} \norm{X_{ik}}_{\psi_2}\\
		&\lesssim \phi_1.
	\end{align*}
	Thus, by Bernstein's inequality (Corollary 2.8.3 in \citealt{vershynin2018high}) and the union bound, we obtain that there exists an absolute constant $A_0>0$ such that for any $t>0$,
	\begin{align*}
		\mathrm{P}\left(\max_{1\leq j, k \leq d} \left|\frac{1}{n} \sum_{i=1}^n \left[ X_{ij}X_{ik} - \mathrm{E}\left(X_{ij}X_{ik}\right)\right]\right| > t\right) &\leq 2d^2 \exp\left(-nA_0 \min\left\{\frac{t^2}{\phi_1^2}, \frac{t}{\phi_1}\right\}\right).
	\end{align*}
	Hence, for any $\delta \in (0,1)$, we have that with probability at least $1-\delta$,
	$$\max_{1\leq j, k \leq d} \left|\frac{1}{n} \sum_{i=1}^n \left[ X_{ij}X_{ik} - \mathrm{E}\left(X_{ij}X_{ik}\right)\right]\right| \lesssim \phi_1 \left(\sqrt{\frac{\log(d/\delta)}{n}} + \frac{\log(d/\delta)}{n}\right).$$
	Plugging the bounds for \textbf{Term I} and \textbf{Term II} into \eqref{cone_union_ball1}, we conclude that with probability at least $1-\delta$,
	\begin{align*}
		&\sup_{v\in \left[\mathcal{C}_1(S,\vartheta) \bigcup B_d^{(1)}(0,r_1) \right] \bigcap B_d(0,r_0)} \left|\frac{1}{n} \sum_{i=1}^n \left\{(X_i^Tv)^2 - \mathrm{E}\left[(X_i^Tv)^2\right] \right\}\right| \\
		&\lesssim r_0^2 (2+\vartheta)^2 \phi_s\left[\sqrt{\frac{\log(1/\delta) + s\log(ed/s)}{n}} + \frac{\log(1/\delta) + s\log(ed/s)}{n} \right] \\
		&\quad + r_1^2\phi_1 \left(\sqrt{\frac{\log(d/\delta)}{n}} + \frac{\log(d/\delta)}{n}\right).
	\end{align*}
	In particular, when $d\geq s+2$ and $r_0 \geq r_1$, we have that $\log\left(\frac{1}{\delta}\right) + s\log\left(\frac{ed}{s}\right) \geq \log\left(\frac{d}{\delta}\right)$ and therefore, with probability at least $1-\delta$,
	\begin{align*}
		&\sup_{v\in \left[\mathcal{C}_1(S,\vartheta) \bigcup B_d^{(1)}(0,r_1) \right] \bigcap B_d(0,r_0)} \left|\frac{1}{n} \sum_{i=1}^n \left\{(X_i^Tv)^2 - \mathrm{E}\left[(X_i^Tv)^2\right] \right\}\right| \\
		&\lesssim r_0^2 (2+\vartheta)^2 \phi_s\left[\sqrt{\frac{\log(1/\delta) + s\log(ed/s)}{n}} + \frac{\log(1/\delta)  + s\log(ed/s)}{n} \right].
	\end{align*}
	Similarly, we notice that
	\begin{align}
		\label{cone_union_ball2}
		\begin{split}
			&\sup_{\substack{u \in \mathcal{C}_2(\tilde{S},\vartheta) \cap B_d(0,1)\\v\in \left[\mathcal{C}_1(S,\vartheta) \bigcup B_d^{(1)}(0,r_1) \right] \bigcap B_d(0,r_0)}} \frac{1}{n} \sum_{i=1}^n \left|X_i^Tu\right| \left|X_i^Tv\right| \\
			&\leq \underbrace{\sup_{\substack{u \in \mathcal{C}_2(\tilde{S},\vartheta) \cap B_d(0,1)\\v\in \mathcal{C}_1(S,\vartheta) \cap B_d(0,r_0)}} \frac{1}{n} \sum_{i=1}^n \left|X_i^Tu\right| \left|X_i^Tv\right|}_{\textbf{Term III}} + \underbrace{\sup_{\substack{u \in \mathcal{C}_2(\tilde{S},\vartheta) \cap B_d(0,1)\\v\in B_d^{(1)}(0,r_1) \cap B_d(0,r_0)}} \frac{1}{n} \sum_{i=1}^n \left|X_i^Tu\right| \left|X_i^Tv\right|}_{\textbf{Term IV}}.
		\end{split}
	\end{align}
	
	$\bullet$ {\bf Term III:} We know from Lemma~\ref{lem:gram_mat_conc} that 
	\begin{align*}
		&\sup_{\substack{u \in \mathcal{C}_2(\tilde{S},\vartheta) \cap B_d(0,1)\\v\in \mathcal{C}_1(S,\vartheta) \cap B_d(0,r_0)}} \frac{1}{n} \sum_{i=1}^n \left|X_i^Tu\right| \left|X_i^Tv\right| \\
		&\leq r_0 \sup_{\substack{u \in \mathcal{C}_2(\tilde{S},\vartheta) \cap B_d(0,1)\\v\in \mathcal{C}_1(S,\vartheta) \cap B_d(0,1)}} \left|\frac{1}{n} \sum_{i=1}^n \left[\left|X_i^Tu\right| \left|X_i^Tv\right| - \mathrm{E}\left(|X_i^Tu| |X_i^Tv|\right)\right]\right| \\
		&\quad + r_0 \sup_{\substack{u \in \mathcal{C}_2(\tilde{S},\vartheta) \cap B_d(0,1)\\v\in \mathcal{C}_1(S,\vartheta) \cap B_d(0,1)}} \mathrm{E}\left(|X_i^Tu| |X_i^Tv|\right)\\
		&\lesssim r_0 (2+\vartheta)^2 \sqrt{\phi_s \cdot \phi_{\tilde{s}}} \left[\sqrt{\frac{\log(1/\delta) + s\log(ed/s)}{n}} + \frac{\log(1/\delta)  + s\log(ed/s)}{n} + 1 \right].
	\end{align*}
	
	$\bullet$ {\bf Term IV:} Note that 
	\begin{align*}
		\sup_{\substack{u \in \mathcal{C}_2(\tilde{S},\vartheta) \cap B_d(0,1)\\v\in B_d^{(1)}(0,r_1) \cap B_d(0,r_0)}} \frac{1}{n} \sum_{i=1}^n \left|X_i^Tu\right| \left|X_i^Tv\right| &\leq \sup_{\substack{u \in \mathcal{C}_2(\tilde{S},\vartheta) \cap B_d(0,1)\\v\in B_d^{(1)}(0,r_1)}} \frac{1}{n} \sum_{i=1}^n \left|X_i^Tu\right| \left|X_i^Tv\right|.
	\end{align*}
	Since $\frac{1}{n} \sum_{i=1}^n \left|X_i^Tu\right| \left|X_i^Tv\right|$ is a biconvex function with respect to $u,v$, we know from Lemma~\ref{lem:max_biconvex} that the supremum is attained at the vertex set 
	$$\left\{x\in \mathbb{R}^d: \norm{x}_0=1, \norm{x}_1=\norm{x}_2=r_1 \right\} = r_1 \cdot \left\{\pm e_1,...,\pm e_d\right\}$$ 
	of the $L_1$-ball $B_d^{(1)}(0,r_1)$, where $e_1,...,e_d$ are the standard basis vectors in $\mathbb{R}^d$. Hence,
	\begin{align*}
		&\sup_{\substack{u \in \mathcal{C}_2(\tilde{S},\vartheta) \cap B_d(0,1)\\v\in B_d^{(1)}(0,r_1) \cap B_d(0,r_0)}} \frac{1}{n} \sum_{i=1}^n \left|X_i^Tu\right| \left|X_i^Tv\right| \\
		&\leq \sup_{\substack{u \in \mathcal{C}_2(\tilde{S},\vartheta) \cap B_d(0,1)\\v\in B_d^{(1)}(0,r_1)}} \frac{1}{n} \sum_{i=1}^n \left|X_i^Tu\right| \left|X_i^Tv\right| \\
		&\leq r_1 \sup_{\substack{u \in \mathcal{C}_2(\tilde{S},\vartheta) \cap B_d(0,1)\\v\in \{e_1,...,e_d\}}} \frac{1}{n} \sum_{i=1}^n \left|X_i^Tu\right| \left|X_i^Tv\right| \\
		&\leq r_1 \sup_{u \in \mathcal{C}_2(\tilde{S},\vartheta) \cap B_d(0,1)} \max_{1\leq k \leq d} \frac{1}{n} \sum_{i=1}^n \left[|X_i^Tu| |X_{ik}| - \mathrm{E}\left(|X_i^Tu| |X_{ik}|\right) \right] \\
		&\quad + r_1 \sup_{u \in \mathcal{C}_2(\tilde{S},\vartheta) \cap B_d(0,1)} \max_{1\leq k \leq d} \mathrm{E}\left(|X_i^Tu| |X_{ik}|\right).
	\end{align*}
	Again, under Assumption~\ref{assump:sub_gaussian}, $|X_i^Tu| |X_{ik}| - \mathrm{E}\left(|X_i^Tu| |X_{ik}|\right),i=1,...,n$ are independent, centered, and sub-exponential random variables for all $k=1,...,d$, and
	\begin{align*}
		\max_{1\leq k \leq d}\norm{|X_i^Tu| |X_{ik}| - \mathrm{E}\left(|X_i^Tu| |X_{ik}|\right)}_{\psi_1} &\leq \max_{1\leq k \leq d} \left[\norm{(X_i^Tu) X_{ik}}_{\psi_1} + \norm{\mathrm{E}\left[(X_i^Tu) X_{ik}\right]}_{\psi_1} \right]\\
		&\lesssim \max_{1\leq k \leq d} \norm{X_i^Tu}_{\psi_2} \norm{X_{ik}}_{\psi_2} \\
		&\lesssim \sqrt{\phi_{\tilde{s}} \cdot \phi_1}.
	\end{align*}
	Thus, by Bernstein's inequality, Lemma~\ref{lem:size_dom_cone}, and the union bound, we obtain that there exists an absolute constant $A_0>0$ such that for any $t>0$,
	\begin{align*}
		&\mathrm{P}\left(\sup_{u \in \mathcal{C}_2(\tilde{S},\vartheta) \cap B_d(0,1)} \max_{1\leq k \leq d} \left|\frac{1}{n} \sum_{i=1}^n \left[|X_i^Tu| |X_{ik}| - \mathrm{E}\left(|X_i^Tu| |X_{ik}|\right) \right]\right| > t\right) \\
		&\leq 3d\left(\frac{5ed}{\tilde{s}}\right)^{\tilde{s}} \exp\left(-nA_0 \min\left\{\frac{t^2}{\phi_{\tilde{s}} \cdot \phi_1}, \frac{t}{\sqrt{\phi_{\tilde{s}} \cdot \phi_1}}\right\}\right).
	\end{align*}
	Hence, for any $\delta \in (0,1)$, we have that with probability at least $1-\delta$, 
	\begin{align*}
		&\sup_{u \in \mathcal{C}_2(\tilde{S},\vartheta) \cap B_d(0,1)} \max_{1\leq k \leq d} \left|\frac{1}{n} \sum_{i=1}^n \left[|X_i^Tu| |X_{ik}| - \mathrm{E}\left(|X_i^Tu| |X_{ik}|\right) \right]\right| \\
		&\lesssim \sqrt{\phi_{\tilde{s}} \cdot \phi_1} \left[\sqrt{\frac{\log(1/\delta) + \tilde{s}\log(ed/\tilde{s})}{n}} + \frac{\log(1/\delta)  + \tilde{s}\log(ed/\tilde{s})}{n} \right].
	\end{align*}
	In addition, by Lemma~\ref{lem:size_dom_cone}, we also have that
	\begin{align*}
		\sup_{u \in \mathcal{C}_2(\tilde{S},\vartheta) \cap B_d(0,1)} \max_{1\leq k \leq d} \mathrm{E}\left(|X_i^Tu| |X_{ik}|\right) \lesssim \sqrt{\phi_{\tilde{s}} \cdot \phi_1}.
	\end{align*}
	Therefore, 
	\begin{align*}
		&\sup_{\substack{u \in \mathcal{C}_2(\tilde{S},\vartheta) \cap B_d(0,1)\\v\in B_d^{(1)}(0,r_1) \cap B_d(0,r_0)}} \frac{1}{n} \sum_{i=1}^n \left|X_i^Tu \right| \left|X_i^Tv \right| \\
		&\lesssim r_1 \sqrt{\phi_{\tilde{s}} \cdot \phi_1} \left[\sqrt{\frac{\log(1/\delta) + \tilde{s}\log(ed/\tilde{s})}{n}} + \frac{\log(1/\delta)  + \tilde{s}\log(ed/\tilde{s})}{n} + 1\right].
	\end{align*}
	Finally, plugging the bounds for \textbf{Term III} and \textbf{Term IV} into \eqref{cone_union_ball2}, we conclude that with probability at least $1-\delta$, 
	\begin{align*}
		&\sup_{\substack{u \in \mathcal{C}_2(\tilde{S},\vartheta) \cap B_d(0,1)\\v\in \left[\mathcal{C}_1(S,\vartheta) \bigcup B_d^{(1)}(0,r_1) \right] \bigcap B_d(0,r_0)}} \frac{1}{n} \sum_{i=1}^n \left|X_i^Tu \right| \left|X_i^Tv \right| \\
		&\lesssim \left[r_0(2+\vartheta)^2+r_1\right] \sqrt{\phi_{\tilde{s}} \cdot \phi_1} \left[\sqrt{\frac{\log(1/\delta) + s\log(ed/s)}{n}} + \frac{\log(1/\delta)  + s\log(ed/s)}{n} + 1\right].
	\end{align*}
	The results follow.
\end{proof}

\vspace{3mm}

We know from \autoref{thm:dual_consist2} that
\begin{align*}
	\frac{\gamma}{a_0n} &\asymp r_{\gamma}(n,\delta) = \frac{\sqrt{\phi_1}}{\kappa_R}\norm{x}_2\sqrt{\frac{\log(d/\delta)}{n}} + \frac{\sqrt{\phi_1 \phi_{s_{\ell}(x)}}\norm{x}_2}{\kappa_R^2} \cdot r_{\pi}
\end{align*}
will satisfy the bound in \eqref{res_cross_poly_cond} with probability at least $1-\delta$. The following lemma suggests a (probabilistic) lower bound for $\frac{\nu}{a_0n}$ in \eqref{res_cross_poly_cond} using the similar arguments. It is worth noting that the lower bound for $\frac{\gamma}{a_0n}$ can converge to 0 as $n\to \infty$ under certain conditions, while the lower bound for $\frac{\nu}{a_0n}$ is asymptotically non-vanishing.  

\begin{lemma}
	\label{lem:rel_poly_cost}
	Let $\delta\in (0,1)$ and $x\in \mathbb{R}^d$ be a fixed covariate vector. Suppose that Assumption~\ref{assump:sub_gaussian} holds. Then, with probability at least $1-\delta$,
	\begin{align*}
		&\sup_{v\in \mathcal{C}_2\left(S_{\ell}(x), \vartheta_{\ell}\right)\cap B_d(0,1)} \left|\frac{1}{n} \sum_{i=1}^n \pi_i\left[X_i^Tv\right]^2 \right| \\
		&\lesssim (2+\vartheta_{\ell})^2 \phi_{s_{\ell}(x)}\left[1+\sqrt{\frac{\log(\frac{1}{\delta})+s_{\ell}(x)\log\left(\frac{ed}{s_{\ell}(x)}\right)}{n}} + \frac{\log(\frac{1}{\delta})+s_{\ell}(x)\log\left(\frac{ed}{s_{\ell}(x)}\right)}{n}\right].
	\end{align*}
	If, in addition, Assumption~\ref{assump:prop_consistent_rate} holds, then with probability at least $1-\delta$,
	\begin{align*}
		&\sup_{v\in \mathcal{C}_2\left(S_{\ell}(x), \vartheta_{\ell}\right)\cap B_d(0,1)} \left|\frac{1}{n} \sum_{i=1}^n \hat{\pi}_i\left[X_i^Tv\right]^2 \right| \\
		&\lesssim \left[1+r_{\pi}(n,\delta)\right] (2+\vartheta_{\ell})^2 \phi_{s_{\ell}(x)}\left[1+\sqrt{\frac{\log(\frac{1}{\delta})+s_{\ell}(x)\log\left(\frac{ed}{s_{\ell}(x)}\right)}{n}} + \frac{\log(\frac{1}{\delta})+s_{\ell}(x)\log\left(\frac{ed}{s_{\ell}(x)}\right)}{n}\right].
	\end{align*}
\end{lemma}

\begin{proof}[Proof of Lemma~\ref{lem:rel_poly_cost}]
	We only prove the second statement, as the proof of the first statement is already contained in the following proof. By Assumption~\ref{assump:prop_consistent_rate}, we know that with probability at least $1-\delta$,
	\begin{align*}
		&\sup_{v\in \mathcal{C}_2\left(S_{\ell}(x), \vartheta_{\ell}\right)\cap B_d(0,1)} \left|\frac{1}{n} \sum_{i=1}^n \hat{\pi}_i\left[X_i^Tv\right]^2 \right| \\
		&\leq \sup_{v\in \mathcal{C}_2\left(S_{\ell}(x), \vartheta_{\ell}\right)\cap B_d(0,1)} \left|\frac{1}{n} \sum_{i=1}^n \pi_i\left[X_i^Tv\right]^2 \right| + \sup_{v\in \mathcal{C}_2\left(S_{\ell}(x), \vartheta_{\ell}\right)\cap B_d(0,1)} \left|\frac{1}{n} \sum_{i=1}^n \left(\hat{\pi}_i-\pi_i\right)\left[X_i^Tv\right]^2 \right| \\
		&\leq \left[1+r_{\pi}(n,\delta)\right] \sup_{v\in \mathcal{C}_2\left(S_{\ell}(x), \vartheta_{\ell}\right)\cap B_d(0,1)} \left|\frac{1}{n} \sum_{i=1}^n \left(X_i^Tv\right)^2 \right|\\
		&\leq \left[1+r_{\pi}(n,\delta)\right] \sup_{v\in \mathcal{C}_2\left(S_{\ell}(x), \vartheta_{\ell}\right)\cap B_d(0,1)}\left[\left|\frac{1}{n}\sum_{i=1}^n \left\{\left(X_i^Tv\right)^2 - \mathrm{E}\left[ \left(X_i^Tv\right)^2\right]\right\} \right| + \left|v^T \mathrm{E}\left[ X_1X_1^T\right] v\right| \right].
	\end{align*}
	By Lemma~\ref{lem:gram_mat_conc}, we know that 
	\begin{align*}
		&\sup_{v\in \mathcal{C}_2\left(S_{\ell}(x), \vartheta_{\ell}\right)\cap B_d(0,1)} \left|\frac{1}{n}\sum_{i=1}^n \left\{\left(X_i^Tv\right)^2 - \mathrm{E}\left[\left(X_i^Tv\right)^2\right]\right\} \right| \\
		&\lesssim (2+\vartheta_{\ell})^2 \phi_{s_{\ell}(x)}\left[\sqrt{\frac{\log(\frac{1}{\delta})+s_{\ell}(x)\log\left(\frac{ed}{s_{\ell}(x)}\right)}{n}} + \frac{\log(\frac{1}{\delta})+s_{\ell}(x)\log\left(\frac{ed}{s_{\ell}(x)}\right)}{n}\right]
	\end{align*}
	with probability at least $1-\delta$. In addition, we know from Lemma~\ref{lem:size_dom_cone} that there exists a discrete set $M_{d,\ell} \subset B_d(0,1)$ such that $\mathcal{C}_2\left(S_{\ell}(x),\vartheta_{\ell}\right) \cap B_d(0,1) \subset 2(2+\vartheta_{\ell}) \cdot \mathrm{conv}(M_{d,\ell})$ with $\norm{v}_0\leq s_{\ell}(x)$ for all $v\in M_{d,\ell}$. By Lemma~\ref{lem:max_biconvex}, we obtain that
	\begin{align*}
		\sup_{v\in \mathcal{C}_2\left(S_{\ell}(x), \vartheta_{\ell}\right)\cap B_d(0,1)} \left|v^T \mathrm{E}\left[X_1X_1^T\right] v\right| &\leq \sup_{v \in 2(2+\vartheta_{\ell}) \mathrm{conv}(M_{d,\ell})} \left|v^T \mathrm{E}\left[X_1X_1^T\right] v\right|\\
		&= \sup_{v \in 2(2+\vartheta_{\ell}) \cdot M_{d,\ell}} \left|v^T \mathrm{E}\left[X_1X_1^T\right] v\right|\\
		&\leq 4(2+\vartheta_{\ell})^2 \phi_{s_{\ell}(x)}.
	\end{align*}
	Therefore, we conclude that
	\begin{align*}
		&\sup_{v\in \mathcal{C}_2\left(S_{\ell}(x), \vartheta_{\ell}\right)\cap B_d(0,1)} \left|\frac{1}{n} \sum_{i=1}^n \hat{\pi}_i\left[X_i^Tv\right]^2 \right| \\
		&\lesssim \left[1+r_{\pi}(n,\delta)\right] (2+\vartheta_{\ell})^2 \phi_{s_{\ell}(x)}\left[\sqrt{\frac{\log(\frac{1}{\delta})+s_{\ell}(x)\log\left(\frac{ed}{s_{\ell}(x)}\right)}{n}} + \frac{\log(\frac{1}{\delta})+s_{\ell}(x)\log\left(\frac{ed}{s_{\ell}(x)}\right)}{n}\right]\\
		&\quad + 4(2+\vartheta_{\ell})^2 \phi_{s_{\ell}(x)}\\
		&\lesssim \left[1+r_{\pi}(n,\delta)\right] (2+\vartheta_{\ell})^2 \phi_{s_{\ell}(x)}\left[1+\sqrt{\frac{\log(\frac{1}{\delta})+s_{\ell}(x)\log\left(\frac{ed}{s_{\ell}(x)}\right)}{n}} + \frac{\log(\frac{1}{\delta})+s_{\ell}(x)\log\left(\frac{ed}{s_{\ell}(x)}\right)}{n}\right]
	\end{align*}
	with probability at least $1-\delta$. The result follows.
\end{proof}

\vspace{3mm}

\begin{lemma}
	\label{lem:gamma_rate}
	Let $\delta \in (0,1)$ and $x$ be a fixed covariate vector. Suppose that Assumptions~\ref{assump:basic}, \ref{assump:sub_gaussian} and \ref{assump:gram_lower_bnd} hold. Then, with probability at least $1-\delta$, 
	$$\norm{\frac{1}{n}\sum_{i=1}^n \tilde{Z}_i X_iX_i^T \ell_0(x) + 2x}_{\infty} \lesssim \frac{\sqrt{\phi_1}}{\kappa_R}\norm{x}_2 \left[\sqrt{\frac{\log(d/\delta)}{n}} + \frac{\log(d/\delta)}{n} \right],$$
	where the above statement holds when $\tilde{Z}_i = \pi(X_i), i=1,...,n$ or $\tilde{Z}_i =R_i, i=1,...,n$. If, in addition, Assumption~\ref{assump:prop_consistent_rate} holds, then with probability at least $1-\delta$,
	\begin{align*}
		\norm{\frac{1}{n}\sum_{i=1}^n \hat{\pi}_i X_i X_i^T \ell_0(x) + 2x}_{\infty} &\lesssim \frac{\sqrt{\phi_1}}{\kappa_R}\norm{x}_2 \left[\sqrt{\frac{\log(d/\delta)}{n}} + \frac{\log(d/\delta)}{n} \right]\\
		&\quad + \frac{\sqrt{\phi_1 \phi_{s_{\ell}(x)}}\norm{x}_2}{\kappa_R^2} \left[1 + \sqrt{\frac{\log(d/\delta)}{n}} + \frac{\log(d/\delta)}{n}\right] r_{\pi}(n,\delta).
	\end{align*}
\end{lemma}

\begin{proof}[Proof of Lemma~\ref{lem:gamma_rate}]
	We only prove the second statement, as the proof of the first statement is already contained in the following proof. Notice that 
	\begin{align*}
		\norm{\frac{1}{n}\sum_{i=1}^n \hat{\pi}_i X_i X_i^T \ell_0(x) + 2x}_{\infty} &\leq \underbrace{\norm{\frac{1}{n}\sum_{i=1}^n \pi_i X_i X_i^T \ell_0(x) + 2x}_{\infty}}_{\textbf{Term I}} + \underbrace{\norm{\frac{1}{n}\sum_{i=1}^n \left(\hat{\pi}_i - \pi_i\right) X_i X_i^T \ell_0(x)}_{\infty}}_{\textbf{Term II}}.
	\end{align*}
	
	$\bullet$ \textbf{Term I:} Recall from \eqref{popul_dual_exact} that $\ell_0(x) := -2\left[\mathrm{E}\left(R XX^T\right)\right]^{-1} x = -2\left[\mathrm{E}\left(\pi(X) XX^T\right)\right]^{-1} x$ and thus,
	\begin{align*}
		\norm{\frac{1}{n}\sum_{i=1}^n \pi(X_i) X_iX_i^T \ell_0(x) + 2x}_{\infty} &= \max_{1\leq k \leq d}\left|-\frac{2}{n}\sum_{i=1}^n \pi(X_i) X_{ik}X_i^T \mathrm{E}\left[\pi(X_i)X_iX_i^T\right]^{-1} x + 2x_k \right|,
	\end{align*}
	and all the random variables
	$$2x_1 - \frac{2}{n}\sum_{i=1}^n \pi(X_i) \cdot X_{i1}X_i^T \mathrm{E}\left[\pi(X_i)X_iX_i^T\right]^{-1} x,\dots, 2x_d - \frac{2}{n}\sum_{i=1}^n \pi(X_i) \cdot X_{id}X_i^T \mathrm{E}\left[\pi(X_i)X_iX_i^T\right]^{-1} x$$
	have zero means and sub-exponential tails (but neither identically distributed nor independent). By Assumption~\ref{assump:sub_gaussian} and Lemma~\ref{lem:Orlicz_prod}, we know that
	\begin{align*}
		&\max_{1\leq k \leq d} \norm{x_k - \pi(X_i) \cdot X_{ik}X_i^T \mathrm{E}\left[\pi(X_i)X_iX_i^T\right]^{-1} x}_{\psi_1}\\
		&\lesssim \max_{1\leq k \leq d} \norm{\pi(X_i) \cdot X_{ik}X_i^T \mathrm{E}\left[\pi(X_i)X_iX_i^T\right]^{-1} x}_{\psi_1}\\
		&\lesssim \max_{1\leq k \leq d} \max_{1\leq i \leq n}\norm{X_{ik}}_{\psi_2} \norm{\pi(X_i)X_i^T \mathrm{E}\left[\pi(X_i)X_iX_i^T\right]^{-1} x}_{\psi_2}\\
		&\lesssim \frac{\sqrt{\phi_1}}{\kappa_R} \norm{x}_2
	\end{align*}
	for all $1\leq i \leq n$. Therefore, by the union bound and Bernstein's inequality (Corollary 2.8.3 in \citealt{vershynin2018high}), we know that for all $t>0$, there exists an absolute constant $A_1>0$ such that
	\begin{align*}
		&\mathrm{P}\left(\norm{\frac{1}{n}\sum_{i=1}^n \pi(X_i) X_iX_i^T \ell_0(x) + 2x}_{\infty} > t \right) \\
		& \leq d \cdot \max_{1\leq k \leq d} \mathrm{P}\left(\left|\frac{2}{n}\sum_{i=1}^n \left[\pi(X_i) \cdot X_{ik}X_i^T \mathrm{E}\left[\pi(X_i)X_iX_i^T\right]^{-1} x - x_k\right] \right| >t \right)\\
		&\leq 2d \cdot \exp\left(-nA_1\min\left\{\frac{t\kappa_R}{\sqrt{\phi_1}\norm{x}_2},\, \frac{t^2\kappa_R^2}{\phi_1\norm{x}_2^2} \right\} \right).
	\end{align*}
	Hence, we conclude that with probability at least $1-\delta$,
	\begin{align*}
		\norm{\frac{1}{n}\sum_{i=1}^n \pi(X_i) X_iX_i^T \ell_0(x) + 2x}_{\infty} &\lesssim \frac{\sqrt{\phi_1}}{\kappa_R}\norm{x}_2 \left[\sqrt{\frac{\log(d/\delta)}{n}} + \frac{\log(d/\delta)}{n} \right].
	\end{align*}
	
	$\bullet$ \textbf{Term II:} Note that under Assumption~\ref{assump:prop_consistent_rate},
	\begin{align*}
		&\norm{\frac{1}{n}\sum_{i=1}^n \left(\hat{\pi}_i - \pi_i\right) X_i X_i^T \ell_0(x)}_{\infty} \\
		&= \max_{1\leq k \leq d} \left|\frac{1}{n}\sum_{i=1}^n \left(\hat{\pi}_i -\pi_i\right)X_{ik}X_i^T\ell_0(x) \right|\\
		&\lesssim r_{\pi}(n,\delta) \cdot \max_{1\leq k \leq d} \left(\frac{1}{n}\sum_{i=1}^n \left|X_{ik} X_i^T\ell_0(x) \right| \right)\\
		&\lesssim r_{\pi}(n,\delta) \left[\max_{1\leq k \leq d} \mathrm{E}\left|X_{ik}X_i^T\ell_0(x)\right| + \max_{1\leq k \leq d}\left|\frac{1}{n}\sum_{i=1}^n \Big(\left|X_{ik} X_i^T\ell_0(x) \right| -\mathrm{E}\left|X_{ik} X_i^T\ell_0(x) \right| \Big) \right|\right]
	\end{align*}
	with probability at least $1-\delta$. By Lemma~\ref{lem:dual_approx_cone}, we know that $\ell_0(x)\in \mathcal{C}_2(S_{\ell}(x),\vartheta_{\ell})$ with $\vartheta_{\ell}=\frac{r_{\ell}}{\sqrt{1-r_{\ell}^2}}$. Hence, by the Cauchy-Schwarz inequality,
	\begin{align*}
		\max_{1\leq k \leq d} \mathrm{E}\left|X_{ik}X_i^T\ell_0(x)\right| &\leq \max_{1\leq k \leq d} \sqrt{\mathrm{E} X_{ik}^2} \cdot \sqrt{\mathrm{E}\left[(X_i^T\ell_0(x))^2\right]} \\
		&\lesssim \sqrt{\phi_1} \cdot \sqrt{\sup_{v \in \mathcal{C}_2(S_{\ell}(x),\vartheta_{\ell})\cap \mathbb{S}^{d-1}} \mathrm{E}\left[(X_i^Tv)^2\right]} \cdot\norm{\ell_0(x)}_2\\
		&\stackrel{\text{(i)}}{\lesssim}\sqrt{\phi_1} \cdot \sqrt{\sup_{v \in 2(2+\vartheta_{\ell})\cdot \mathrm{conv}(M_{d,\ell})} \mathrm{E}\left[(X_i^Tv)^2\right]} \cdot \frac{\norm{x}_2}{\kappa_R^2}\\
		&\stackrel{\text{(ii)}}{\lesssim} \sqrt{\phi_1} \cdot \sqrt{\sup_{v \in M_{d,\ell}} \mathrm{E}\left[(X_i^Tv)^2\right]} \cdot \frac{\norm{x}_2}{\kappa_R^2}\\
		&\lesssim \sqrt{\phi_1 \phi_{s_{\ell}(x)}} \cdot \frac{\norm{x}_2}{\kappa_R^2},
	\end{align*}
	where, for inequality (i), Lemma~\ref{lem:size_dom_cone} guarantees the existence of a discrete set $M_{d,\ell} \subset B_d(0,1)$ with cardinality $|M_{d,\ell}| \leq \frac{3}{2} \left(\frac{5ed}{s_{\ell}(x)} \right)^{s_{\ell}(x)}$ such that $\mathcal{C}_2(S_{\ell}(x),\vartheta_{\ell}) \cap B_d(0,1) \subset 2(2+\vartheta_{\ell}) \cdot \mathrm{conv}(M_{d,\ell})$ and $\norm{u}_0 \leq s_{\ell}(x)$ for all $u\in M_{d,\ell}$. For inequality (ii), we use Lemma~\ref{lem:max_biconvex} to show that the supremum of a convex function on $\mathrm{conv}(M_{d,\ell})$ is attained at its vertices, which are contained in $M_{d,\ell}$.
	
	Now, we also know that all the random variables 
	$$\frac{1}{n}\sum_{i=1}^n \Big[\left|X_{i1} X_i^T\ell_0(x) \right| -\mathrm{E}\left|X_{i1} X_i^T\ell_0(x) \right| \Big],\,...,\, \frac{1}{n}\sum_{i=1}^n \Big[\left|X_{id} X_i^T\ell_0(x) \right| -\mathrm{E}\left|X_{id} X_i^T\ell_0(x) \right| \Big]$$
	have zero means and sub-exponential tails (but neither identically distributed nor independent). By Assumption~\ref{assump:sub_gaussian} and Lemma~\ref{lem:Orlicz_prod}, we know that for all $1\leq i \leq n$,
	\begin{align*}
		&\max_{1\leq k \leq d} \norm{\left|X_{ik} X_i^T\ell_0(x) \right| -\mathrm{E}\left|X_{ik} X_i^T\ell_0(x) \right|}_{\psi_1}\\
		&\lesssim \max_{1\leq k \leq d} \norm{X_{ik} X_i^T \mathrm{E}\left[\pi(X_i)X_iX_i^T\right]^{-1} x}_{\psi_1}\\
		&\lesssim \max_{1\leq k \leq d} \max_{1\leq i \leq n}\norm{X_{ik}}_{\psi_2} \norm{X_i^T \mathrm{E}\left[\pi(X_i)X_iX_i^T\right]^{-1} x}_{\psi_2}\\
		&\lesssim \frac{\sqrt{\phi_1 \phi_{s_{\ell}(x)}}}{\kappa_R^2} \norm{x}_2.
	\end{align*}
	Therefore, by the union bound and Bernstein's inequality (Corollary 2.8.3 in \citealt{vershynin2018high}), we know that there exists an absolute constant $A_2>0$ such that for all $t>0$, 
	\begin{align*}
		&\mathrm{P}\left(\max_{1\leq k \leq d}\left|\frac{1}{n}\sum_{i=1}^n \Big(\left|X_{ik} X_i^T\ell_0(x) \right| -\mathrm{E}\left|X_{ik} X_i^T\ell_0(x) \right| \Big) \right| > t \right) \\
		& \leq d \cdot \max_{1\leq k \leq d} \mathrm{P}\left(\left|\frac{1}{n}\sum_{i=1}^n \Big(\left|X_{ik} X_i^T\ell_0(x) \right| -\mathrm{E}\left|X_{ik} X_i^T\ell_0(x) \right| \Big) \right| >t \right)\\
		&\leq 2d \cdot \exp\left(-nA_2\min\left\{\frac{t\kappa_R^2}{\sqrt{\phi_1 \phi_{s_{\ell}(x)}}\norm{x}_2},\, \frac{t^2\kappa_R^4}{\phi_1 \phi_{s_{\ell}(x)}\norm{x}_2^2} \right\} \right).
	\end{align*}
	Hence, we conclude that with probability at least $1-\delta$,
	\begin{align*}
		&\max_{1\leq k \leq d}\left|\frac{1}{n}\sum_{i=1}^n \Big(\left|X_{ik} X_i^T\ell_0(x) \right| -\mathrm{E}\left|X_{ik} X_i^T\ell_0(x) \right| \Big) \right| \\
		&\lesssim \frac{\sqrt{\phi_1 \phi_{s_{\ell}(x)}}}{\kappa_R^2}\norm{x}_2 \left[\sqrt{\frac{\log(d/\delta)}{n}} + \frac{\log(d/\delta)}{n} \right]
	\end{align*} 
	Therefore, we conclude that with probability at least $1-\delta$,
	\begin{align}
		\label{dual_para_term_II}
		\begin{split}
			\textbf{Term II} &= \norm{\frac{1}{n}\sum_{i=1}^n \left(\hat{\pi}_i - \pi_i\right) X_i X_i^T \ell_0(x)}_{\infty}\\
			&\lesssim \frac{\sqrt{\phi_1 \phi_{s_{\ell}(x)}}\norm{x}_2}{\kappa_R^2} \left[1 + \sqrt{\frac{\log(d/\delta)}{n}} + \frac{\log(d/\delta)}{n}\right] r_{\pi}(n,\delta).
		\end{split}
	\end{align}
	Combining the bounds for \textbf{Term I} and \textbf{Term II} yields that with probability at least $1-\delta$,
	\begin{align*}
		\norm{\frac{1}{n}\sum_{i=1}^n \hat{\pi}_i X_i X_i^T \ell_0(x) + 2x}_{\infty} &\lesssim \frac{\sqrt{\phi_1}}{\kappa_R}\norm{x}_2 \left[\sqrt{\frac{\log(d/\delta)}{n}} + \frac{\log(d/\delta)}{n} \right]\\
		&\quad + \frac{\sqrt{\phi_1 \phi_{s_{\ell}(x)}}\norm{x}_2}{\kappa_R^2} \left[1 + \sqrt{\frac{\log(d/\delta)}{n}} + \frac{\log(d/\delta)}{n}\right] r_{\pi}(n,\delta).
	\end{align*}
	The proof is completed.
\end{proof}

\vspace{3mm}

\begin{lemma}[Orlicz-norm of products of random variables]
	\label{lem:Orlicz_prod}
	Let $X_1,...,X_N$ be random variables with finite $\psi_{\alpha}$-Orlicz norm \eqref{Orlicz}. Then, for any $N\geq 1$, $\norm{\prod_{i=1}^N X_i}_{\psi_{\frac{\alpha}{N}}} \leq \prod_{i=1}^N \norm{X_i}_{\psi_{\alpha}}$.
\end{lemma}

\begin{proof}[Proof of Lemma~\ref{lem:Orlicz_prod}]
	The result generalizes Lemma 2.7.7 in \cite{vershynin2018high} and follows from Young's inequality, $\prod_{i=1}^N x_i^{\lambda_i} \leq \sum_{i=1}^N \lambda_i x_i$, for all $x_i,\lambda_i\geq 0$, $i=1,...,N$, with $\sum_{i=1}^N \lambda_i=1$. Without loss of generality, assume that $\norm{X_i}_{\psi_{\alpha}} =1$ for all $i=1,...,N$. Then,
	\begin{align*}
		\mathrm{E}\left[\exp\left(\prod_{i=1}^N |X_i|^{\frac{\alpha}{N}} \right)\right] \leq \mathrm{E}\left[\exp\left(\frac{1}{N} \sum_{i=1}^N |X_i|^{\alpha} \right)\right] &= \mathrm{E}\left[\prod_{i=1}^N \exp\left(\frac{1}{N} |X_i|^{\alpha}\right) \right] \\
		&\leq \frac{1}{N} \sum_{i=1}^N \mathrm{E}\left[\exp\left(|X_i|^{\alpha}\right) \right] \leq 2,
	\end{align*}
	where we apply Young's inequality to obtain the first and last inequalities.
\end{proof}

\vspace{3mm}

\begin{lemma}[\citealt{giessing2022inequalities}]
	\label{lem:empirical_bound}
	Let $\mathcal{F} \subset L_1(\mathcal{X},\mathcal{A},\mathrm{P})$ and $\varrho$ be a pseudo-metric on $\mathcal{F}\times \mathcal{F}$. For $\eta >0$, we define $\mathcal{F}_{\eta}=\left\{f-g: f,g \in \mathcal{F},\varrho(f,g) \leq \eta \right\}$. Suppose that for $\alpha\in (0,2]$, there exists a constant $K>0$ such that for all $f,g\in \mathcal{F}$,
	$$\norm{\left(f-\mathrm{P}f\right) -\left(g-\mathrm{P} g\right)}_{\mathrm{P},\psi_{\alpha}} \leq K \cdot \varrho(f,g).$$
	Then, the following statements hold true:
	\begin{enumerate}[label=(\alph*)]
		\item If $\alpha \in (0,1]$, then for all $t>0$, with probability at least $1-e^{-t}$,
		\begin{align*}
			\norm{\mathbb{G}_n}_{\mathcal{F}_{\eta}} &\leq C_{\alpha} K\left[\int_0^{\eta} \sqrt{\log N(\epsilon,\mathcal{F},\varrho)} \, d\epsilon \right.\\
			&\quad \left.+ n^{-\frac{1}{2}} \int_0^{\eta} \big[\log N(\epsilon,\mathcal{F},\varrho) \big]^{\frac{1}{\alpha}} d\epsilon + \eta\left(\sqrt{t} + n^{-\frac{1}{2}} t^{\frac{1}{\alpha}} \right)\right],
		\end{align*}
		where $C_{\alpha} >0$ is a constant depending only on $\alpha$.
		
		\item If $\alpha \in (1,2]$, then for all $t>0$, with probability at least $1-e^{-t}$,
		\begin{align*}
			\norm{\mathbb{G}_n}_{\mathcal{F}_{\eta}} &\leq C_{\alpha} K\left[\int_0^{\eta} \sqrt{\log N(\epsilon,\mathcal{F},\varrho)} \, d\epsilon  \right.\\
			&\quad \left. + n^{-\frac{1}{2}+\frac{1}{\beta}} \int_0^{\eta} \big[\log N(\epsilon,\mathcal{F},\varrho) \big]^{\frac{1}{\alpha}} d\epsilon + \eta\left(\sqrt{t} + n^{-\frac{1}{2}+\frac{1}{\beta}} t^{\frac{1}{\alpha}} \right)\right],
		\end{align*}
		where $C_{\alpha} >0$ is a constant depending only on $\alpha$ and $\frac{1}{\alpha}+\frac{1}{\beta}=1$.
	\end{enumerate}
\end{lemma}

\begin{remark}[\citealt{giessing2021inference}]
	\label{remark:index_class}
	If $\mathcal{F}=\{f_t:t\in T\}$ and $\norm{\left(f_s-\mathrm{P} f_s\right) - \left(f_t -\mathrm{P}f_t\right)}_{\mathrm{P},\psi_{\alpha}} \leq K\cdot \varrho(s,t)$ for all $f_s,f_t\in \mathcal{F}$ and a pseudo-metric $\varrho$ on $T\times T$, where $K>0$ is some constant, then we can replace $N(\epsilon,\mathcal{F},\varrho)$ by $N(\epsilon,T,\varrho)$.
\end{remark}

\begin{acks}[Acknowledgments]
We thank Armeen Taeb for helpful discussions about the computational part of this paper and Thomas S. Richardson for insightful comments on causal inference applications.
\end{acks}

\begin{funding}
YZ was supported in part by YC's NSF grant DMS-2141808. AG is supported by the Ministry of Education, Singapore, under the Academic Research Fund Tier 1 (A-8004834-00-00) and NSF grant DMS-2310578. YC is supported by NSF grants DMS-1952781, 2112907, 2141808, and NIH U24-AG07212.

Funding for the Sloan Digital Sky Survey IV has been provided by the 
Alfred P. Sloan Foundation, the U.S. Department of Energy Office of 
Science, and the Participating Institutions. SDSS-IV acknowledges support and resources from the Center for High Performance Computing  at the University of Utah. The SDSS website is \url{www.sdss.org}.

SDSS-IV is managed by the Astrophysical Research Consortium for the Participating Institutions of the SDSS Collaboration including 
the Brazilian Participation Group, the Carnegie Institution for Science, 
Carnegie Mellon University, Center for Astrophysics | Harvard \& Smithsonian, the Chilean Participation Group, the French Participation Group, Instituto de Astrof\'isica de Canarias, The Johns Hopkins University, Kavli Institute for the Physics and Mathematics of the 
Universe (IPMU) / University of Tokyo, the Korean Participation Group, Lawrence Berkeley National Laboratory, Leibniz Institut f\"ur Astrophysik Potsdam (AIP), Max-Planck-Institut f\"ur Astronomie (MPIA Heidelberg), Max-Planck-Institut f\"ur Astrophysik (MPA Garching), Max-Planck-Institut f\"ur Extraterrestrische Physik (MPE), National Astronomical Observatories of China, New Mexico State University, 
New York University, University of Notre Dame, Observat\'orio Nacional / MCTI, The Ohio State University, Pennsylvania State 
University, Shanghai Astronomical Observatory, United Kingdom Participation Group, Universidad Nacional Aut\'onoma de M\'exico, University of Arizona, University of Colorado Boulder, 
University of Oxford, University of Portsmouth, University of Utah, 
University of Virginia, University of Washington, University of Wisconsin, Vanderbilt University, and Yale University.
\end{funding}


\bibliographystyle{imsart-nameyear} 
\bibliography{HighD_Inf}       


\end{document}